\documentclass[12pt]{article}

\usepackage[english]{babel}
\usepackage[utf8]{inputenc}
\usepackage{amsmath}
\usepackage{graphicx}
\usepackage{subfigure}
\usepackage{amssymb}
\usepackage{amsthm}
\usepackage{bm}
\usepackage{tikz-cd}
\usepackage{mathrsfs}
\usepackage[colorinlistoftodos]{todonotes}
\usepackage{enumitem}
\usepackage{yfonts}
\usepackage{ dsfont }
\usepackage{geometry}
\usepackage{setspace}
\usepackage{tikz}
\usetikzlibrary{positioning}
\usetikzlibrary{decorations.pathreplacing}
\usetikzlibrary{decorations,arrows}
\usetikzlibrary{decorations.markings}
\usetikzlibrary{patterns}

\numberwithin{equation}{section}
\title{A Mathematical Theory of Integer Quantum Hall Effect in Photonics\thanks{This work was partially supported by Hong Kong RGC grant GRF 16304621, 16307024 and NSFC grant 12371425.}}

\author{Jiayu Qiu\thanks{Department of Mathematics, 
 HKUST,  Clear Water Bay, Kowloon, Hong Kong S.A.R., China.
    \tt jqiuaj@connect.ust.hk.}\,\,\,, Hai Zhang\thanks{Department of Mathematics, 
 HKUST,  Clear Water Bay, Kowloon, Hong Kong S.A.R., China.
    \tt haizhang@ust.hk.}}

\date{\today}

\newtheorem{theorem}{Theorem}[section]
\newtheorem{lemma}[theorem]{Lemma}

\newtheorem{definition}[theorem]{Definition}

\newtheorem{corollary}[theorem]{Corollary}

\newtheorem{remark}[theorem]{Remark}
\newtheorem{proposition}[theorem]{Proposition}
\newtheorem{assumption}[theorem]{Assumption}
\allowdisplaybreaks[4]

\begin{document}

\maketitle
\begin{abstract}
This paper investigates interface modes in a square lattice of photonic crystal composed of gyromagnetic particles with $C_{4v}$ point group symmetry. The study shows that Dirac or linear degenerate points cannot occur at the three high-symmetry points in the Brillouin zone where two Bloch bands touch. Instead, a touch point at the M-point has a quadratic degeneracy in the generic case. It
is further proved that when a magnetic field is applied to the two sides of an interface in opposite directions, two interface modes supported along that interface can bifurcate from the quadratic degenerate point. 
These results provide a mathematical foundation for the first experimental realization of the integer quantum Hall effect in the context of photonics.
%in the work [Observation of unidirectional backscattering-immune topological electromagnetic states. Zheng Wang, Yidong Chong, John Joannopoulos, and Marin Soljaci\'c, Nature, 2009].
%\cite{wang09observation}.


% MSC:  35Q60, 35P05, 47A10
\end{abstract}

\section{Introduction}

\subsection{Background}
Topological photonics and phononics are rapidly growing fields that apply the principles of topological phases of matter, originally discovered in solid-state physics for electronic systems, to the contexts of optics and phononics. The study of topological phases of matter in condensed-matter systems began with the discovery of the integer quantum Hall effect (IQHE) in 1980 \cite{klitzing80IQHE,Klitzing86QHE}. Thouless et al. discovered in 1982 that the integer in the quantized Hall conductance is related to a topological invariant of the system, i.e. the Chern number of Bloch bundles \cite{thouless82QHC}. Haldane and Raghu first proposed an analog of the IQHE in photonics in their seminal 2008 work \cite{Haldane08IQHE_photonic}, and the first experimental realization was achieved one year later by Wang et al. \cite{wang09observation}. 
Since then, there have been great activities in the study of a variety of photonic and phononic systems realizing nontrivial topological properties, leading to the emerging research field of topological photonics and phononics \cite{lu2014topological,Khanikaev2017twod_topological_photonics}.

An important feature of topological insulators or topological phases is the existence of gapless edge/interface modes that are spatially localized at the boundary of the bulk insulator or the interface separating it from another bulk insulator that has a different topological phase. This is known as the bulk-edge/interface correspondence \cite{Rebbi76soliton,hatsugai93BEC_1,Qi06BEC}. A rigorous justification of the bulk-edge/interface correspondence is one of the most interesting and challenging mathematical problems in the study of topological materials. Since the first proof for the Harper model by Hatsugai in 1993 \cite{hatsugai93BEC_1,hatsugai93BEC_2}, much progress has been made for discrete models in electronic systems by the K-theory approach \cite{KELLENDONK02QHE,Elbau2002EqualityOB,bourne2017k-theory,braverman2019spectral,kubota2017controlled} and functional analysis approach \cite{Elgart2005equality,Graf2012BulkEdgeCF,Avila2013topological_invariants}. For continuous models where partial differential equations are involved, we refer to \cite{kellendonk2004quantization,taarabt2014equality,Combes2005edge_impurity,Bourne2016ChernNL} for electronic systems, and \cite{bal22topo_invariant,bal2023topological_charge,quinn2024approximations} for Dirac systems. There are also a few concerned with the photonic/phononic systems; see \cite{thiang23bec} for a result in one dimension and \cite{Drouot2021microlocal} in two dimensions. We note that the smoothness assumption of the coefficients of the concerned partial differential operator in \cite{Drouot2021microlocal} may hinder a direct application of its result to photonics structures where the coefficients are typically piecewise constant.
With rapid experiment developments, an applicable mathematical theory for photonic systems that illustrates the connection between the existence of interface modes and topological phases of bulk materials is highly desirable. 

%This motivates this work.

In addition to the bulk-interface correspondence principle, the existence of interface modes can be proved through the bifurcation of Dirac points. A Dirac point is a special degenerate point in the spectral bands of a periodic operator, where two bands intersect in a linear or conic manner \cite{Ammari-F-H-L-Y-2020, fefferman12, Li23, berkolaiko2018symmetry,ammari2020topological,Ammari2020high_freq}. Dirac points occur in electronic and photonic/phononic systems with a honeycomb structure. Typically, a topological phase transition occurs near a Dirac point, and an in-gap eigenvalue can be generated by applying proper perturbations to the periodic operator. The bifurcation of eigenvalues from Dirac points was rigorously analyzed for one-dimensional Schr\"{o}dinger operators \cite{fefferman2017topologically}, two-dimensional Schr\"{o}dinger operators \cite{fefferman2016honeycomb_edge}, and two-dimensional elliptic operators with smooth coefficients \cite{lee2019elliptic}, all using domain wall models. It is worth mentioning that Raghu and Haldane's original proposal on IQHE in photonics can be interpreted as an example of Dirac point bifurcation. The need for a Dirac point led them to focus on TE modes in a triangular lattice of gyroelectric particles. However, such a proposal was not adopted in the experiment \cite{wang09observation}, since the gyroelectric effect that can break the time-reversal symmetry in realistic materials is too weak, resulting in a too-small band gap that is not robust against disorder. Instead, they used a 2D square lattice of gyromagnetic particles whose band diagram contains a quadratic degenerate M-point. In that case, the bifurcation of eigenvalues from a quadratic degenerate point cannot be treated as in the linear case of Dirac points; see \cite{weinstein18quadratic, chaban2024instability} for discussion on those two types of degenerate points. This suggests the need for a new framework to study interface modes that bifurcate from quadratic degenerate points, as established in this paper.

This paper develops a mathematical theory to explain the first experimental realization of IQHE in photonics \cite{wang09observation}. Specifically, we rigorously prove the existence of interface modes in a 2D square lattice of gyromagnetic particles under an appropriate perturbation. Our primary contribution is a mathematical framework that reveals the mechanism of creating interface modes from a quadratic degenerate point. We demonstrate that a topological phase transition occurs at this degenerate point, characterized by a change in parity of the Bloch modes and their momentum derivatives when the direction of the applied magnetic field is reversed. This phase transition leads to the emergence of interface modes that bifurcate from the spectral degenerate point. We note that the original argument for the existence of interface modes in \cite{wang09observation} is based on the bulk-edge correspondence principle, whose applicability to realistic continuous photonic systems remains unclear. Furthermore, the two-scale asymptotic expansion approach developed in \cite{fefferman2017topologically, fefferman2016honeycomb_edge, drouot2020edge, lee2019elliptic} for the domain-wall models is not applied to our setting, where an interface sharply separates two differential periodic structures. 

In contrast to the case of Dirac points, studying interface modes that bifurcate from a quadratic degenerate point requires a more dedicated approach, including higher-order perturbation analysis and treatment of momentum derivatives of Bloch modes at the quadratic degenerate point, as established in this paper. Additionally, we highlight that the Chern number argument provides limited information for structures with time-reversal symmetry due to its triviality. Our method, however, can be applied to systems with time-reversal symmetry and used to study the existence of interface modes, as is done in \cite{qiu2023mathematical,li2024interface}.

\subsection{Model description and main results}
We start with the following periodic elliptic operator in $L^2(\mathbf{R}^2)$
\begin{equation*}
\mathcal{L}^{A}:D(\mathcal{L}^{A})=\{u\in H^1(\mathbf{R}^2):\, -\nabla\cdot A\nabla u\in {L^2(\mathbf{R}^2)}\}\subset L^2(\mathbf{R}^2)\to L^2(\mathbf{R}^2),\quad
u\mapsto -\nabla\cdot A\nabla u,
\end{equation*}
where the coefficient matrix $A=A(\bm{x})$ satisfies the following assumption:  
\begin{assumption} \label{assum_coefficient_matrix}
$A(\bm{x})=(1+a(\bm{x}))\cdot I_{2\times 2}$, where $a(\bm{x})=c\cdot \sum_{n_1,n_2\in\mathbf{Z}}\chi_{D_{n_1,n_2}}(\bm{x})$. Here $c>0$ is a positive constant, $\chi$ is the indicator function and $D_{n_1,n_2}:=D+n_1\bm{e}_1+n_2\bm{e}_2$ are inclusions. We assume that $D\subset (0,1)\times (0,1)$ is simply connected and the square lattice structure is symmetric under the $C_{4v}$ point group in the sense that
\begin{equation*}
\bm{x}\in \cup_{n_1,n_2\in\mathbf{Z}}D_{n_1,n_2}\Longrightarrow R\bm{x},M_2\bm{x}\in \cup_{n_1,n_2\in\mathbf{Z}}D_{n_1,n_2},
\end{equation*}
where $R=
    \begin{pmatrix}
    0 & -1 \\ 1 & 0
    \end{pmatrix}
    $ ($\pi/2-$rotation) and $M_{2}=
    \begin{pmatrix}
    1 & 0 \\ 0 & -1
    \end{pmatrix}
$ (reflection).
\end{assumption}
\noindent The operator $\mathcal{L}^A$ models the propagation of time-harmonic TE polarized electromagnetic waves in a 2D square lattice of gyromagnetic particles without external magnetic field (see Figure \ref{fig_domain} where the inclusions are disks, as in the original setting of \cite{wang09observation}). It can be equivalently defined through the following sesquilinear form on $L^2(\mathbf{R}^2)$
\begin{equation*} \label{eq_A_form}
\mathfrak{a}^{A}(u,v)=
\int_{\mathbf{R}^3}\big(A(\bm{x})\nabla u(\bm{x})\big)\cdot \overline{\nabla v(\bm{x})}d\bm{x},\quad
u,v\in H^1(\mathbf{R}^2)\subset L^2(\mathbf{R}^2).
\end{equation*}

\begin{figure}
\centering
\begin{tikzpicture}[scale=0.45]
%lambda real-> interface mode step 1:pv
\tikzset{->-/.style={decoration={markings,mark=at position #1 with 
{\arrow{latex}}},postaction={decorate}}};
\draw[dashed] (-5,0)--(5,0);
\draw[dashed] (-5,2)--(5,2);
\draw[dashed] (-5,4)--(5,4);
\draw[dashed] (-5,-2)--(5,-2);
\draw[dashed] (-5,-4)--(5,-4);
\draw[dashed] (0,-5)--(0,5);
\draw[dashed] (2,-5)--(2,5);
\draw[dashed] (4,-5)--(4,5);
\draw[dashed] (-2,-5)--(-2,5);
\draw[dashed] (-4,-5)--(-4,5);
\draw[fill=black,opacity=1] (1,1) ellipse(0.5 and 0.5);
\draw[fill=black,opacity=1] (3,1) ellipse(0.5 and 0.5);
\draw[fill=black,opacity=1] (-1,1) ellipse(0.5 and 0.5);
\draw[fill=black,opacity=1] (-3,1) ellipse(0.5 and 0.5);
\draw[fill=black,opacity=1] (1,-1) ellipse(0.5 and 0.5);
\draw[fill=black,opacity=1] (3,-1) ellipse(0.5 and 0.5);
\draw[fill=black,opacity=1] (-1,-1) ellipse(0.5 and 0.5);
\draw[fill=black,opacity=1] (-3,-1) ellipse(0.5 and 0.5);
\draw[fill=black,opacity=1] (1,3) ellipse(0.5 and 0.5);
\draw[fill=black,opacity=1] (3,3) ellipse(0.5 and 0.5);
\draw[fill=black,opacity=1] (-1,3) ellipse(0.5 and 0.5);
\draw[fill=black,opacity=1] (-3,3) ellipse(0.5 and 0.5);
\draw[fill=black,opacity=1] (1,-3) ellipse(0.5 and 0.5);
\draw[fill=black,opacity=1] (3,-3) ellipse(0.5 and 0.5);
\draw[fill=black,opacity=1] (-1,-3) ellipse(0.5 and 0.5);
\draw[fill=black,opacity=1] (-3,-3) ellipse(0.5 and 0.5);
\node[above,scale=0.6] at (1,1.35) {$D_{0,0}$};
\node[above,scale=0.6] at (-1,1.35) {$D_{-1,0}$};
\node[below,scale=0.6] at (1,-1.3) {$D_{0,-1}$};
\node[below,scale=0.6] at (-1,-1.3) {$D_{-1,-1}$};
\draw[dashed,very thick,blue] (1,1)--(-1,1)--(-1,-1)--(1,-1)--(1,1);

\draw[dashed] (8,-2)--(12,-2);
\draw[dashed] (8,2)--(12,2);
\draw[dashed] (8,-2)--(8,2);
\draw[dashed] (12,-2)--(12,2);
\filldraw [black, draw = black] (12,2) -- (10.75,2) arc(180:270:1.25) -- cycle;
\filldraw [black, draw = black] (8,2) -- (8,0.75) arc(270:360:1.25) -- cycle;
\filldraw [black, draw = black] (8,-2) -- (9.25,-2) arc(0:90:1.25) -- cycle;
\filldraw [black, draw = black] (12,-2) -- (12,-0.75) arc(90:180:1.25) -- cycle;
\node[below,scale=0.8] at (10,-2.5) {$Y=(-\frac{1}{2},\frac{1}{2})\times (-\frac{1}{2},\frac{1}{2})$};
\end{tikzpicture}
\caption{Left: structure of the problem in the case the inclusions are taken as disks. Right: illustration of the unit cell.}
\label{fig_domain}
\end{figure}
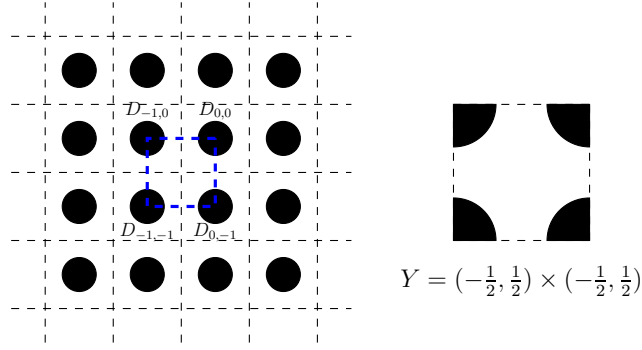
\noindent Assumption \ref{assum_coefficient_matrix}
%, the coefficient matrix $A(\bm{x})$ is symmetric with respect to the $C_{4v}$ point group,
%\begin{equation*}
%A(R\bm{x})=A(\bm{x}),\quad
%A(M_{2}\bm{x})=A(\bm{x}).
%\end{equation*}
implies that $\mathcal{L}^A$ is invariant under translation, rotation, and reflection in the sense that
\begin{equation}
\label{eq_LA_symmetry}
[\mathcal{L}^A,\mathcal{T}_i]=[\mathcal{L}^A,\mathcal{R}]=
[\mathcal{L}^A,\mathcal{M}_2]=0
\quad \text{on $D(\mathcal{L}^A)$},
\end{equation}
where
\begin{equation*}
(\mathcal{T}_{i}u)(\bm{x}):=u(\bm{x}+\bm{e}_i) \enspace (i=1,2),\quad
(\mathcal{R}u)(\bm{x}):=u(R\bm{x}),\quad
(\mathcal{M}_2 u)(\bm{x}):=u(M_2\bm{x}).
\end{equation*}
The spectral theory of periodic operators implies that
$$
\sigma(\mathcal{L}^{A}) = \bigcup_{\bm{\kappa}\in Y^*}\sigma(\mathcal{L}^{A}(\bm{\kappa})),
$$ 
where $Y_*:=[0,2\pi]^2$ is the Brillouin zone, $\mathcal{L}^{A}(\bm{\kappa})$ is the restriction of $\mathcal{L}^{A}$ on the space
\begin{equation*} 
L^2_{\bm{\kappa}}(\mathbf{R}^2):=
\{u\in L^2_{loc}(\mathbf{R}^2):\, u(\bm{x}+n_1\bm{e}_1+n_2\bm{e}_2)=e^{i(n_1\kappa_1+n_2\kappa_2)}u(\bm{x})\}.
\end{equation*}
The following identities demonstrate the symmetry properties of $\mathcal{L}^{A}(\bm{\kappa})$:
\begin{equation} \label{eq_symm_momen_sapce_3}
\mathcal{L}^{A}(\bm{\kappa}+n_1\cdot 2\pi \bm{e}_1+n_2\cdot 2\pi \bm{e}_2)=\mathcal{L}^{A}(\bm{\kappa}),\quad 
n_1,n_2\in\mathbf{Z},
\end{equation}
\begin{equation} \label{eq_symm_momen_sapce_1}
\mathcal{R}\mathcal{L}^{A}(\bm{\kappa})\mathcal{R}^{-1}=\mathcal{L}^{A}(R^{-1}\bm{\kappa}),
\end{equation}
\begin{equation} \label{eq_symm_momen_sapce_2}
\mathcal{M}_2\mathcal{L}^{A}(\bm{\kappa})\mathcal{M}_{2}^{-1}=\mathcal{L}^{A}(M_2\bm{\kappa}).
\end{equation}
Let $\{\lambda_n(\bm{\kappa}),u_n(\bm{x};\bm{\kappa})\}_{n\geq 1}$ be the Floquet-Bloch eigenpair of $\mathcal{L}^{A}(\bm{\kappa})$. It's known that $\lambda_n(\bm{\kappa})$ and $u_n(\bm{x};\bm{\kappa})\}_{n\geq 1}$ are piecewisely smooth on $\bm{\kappa}\in Y^*$ \cite{kuchment2016overview}. Singular points may occur at degenerate points in the spectrum of $\sigma(\mathcal{L}^{A}(\bm{\kappa}))$. The most common degenerate points are linear, referred to as Dirac points, defined below.
\begin{definition}[Linear degenerate point or Dirac point in square lattice]
\label{def_linear_degenracy}
A point $(\bm{\kappa_*},\lambda_*)\in Y^*\times \mathbf{R}$ is called a Dirac point if there exist $n_*\in \mathbf{Z}$, $\alpha_*>0$ such that

(1)$\lambda_{n_*}(\bm{\kappa_*})=\lambda_{n_*+1}(\bm{\kappa_*})=\lambda_*$. 

(2)For $|\bm{\kappa}-\bm{\kappa}_*|\ll 1$, 
    \begin{equation} \label{eq_conical_dispersion}
    \begin{aligned}
    \lambda_{n_*}(\bm{\kappa})-\lambda_{n_*}(\bm{\kappa}_*)&=-\alpha_*|\bm{\kappa}-\bm{\kappa}_*|+\mathcal{O}(|\bm{\kappa}-\bm{\kappa}_*|^2), \\
    \lambda_{n_*+1}(\bm{\kappa})-\lambda_{n_*}(\bm{\kappa}_*)&=\alpha_*|\bm{\kappa}-\bm{\kappa}_*|+\mathcal{O}(|\bm{\kappa}-\bm{\kappa}_*|^2).
    \end{aligned}
    \end{equation}
\end{definition}

%$\lambda_{n_*}(\bm{\kappa})\neq \lambda_*$, $\lambda_{n_*+1}(\bm{\kappa})\neq \lambda_*$ for any $\bm{\kappa}\in Y^*\backslash \{\bm{\kappa}_*\}$, and $\lambda_{n}(\bm{\kappa})\neq \lambda_*$ for any $n\notin \{n_*,n_*+1\}$ and $\bm{\kappa}\in Y^*$ (spectral no-fold condition);
The band structure featuring a Dirac point is illustrated in Figure \ref{fig_band_structure_dirac_cone}, plotted along the slice $\bm{\kappa}(\kappa_1)=\bm{\kappa}_*-\bm{e}_1+\kappa_1\bm{e}_1$.
Dirac points typically occur in \textit{triangular lattices} \cite{Kane10insulators,berkolaiko2018symmetry,fefferman2016honeycomb_edge}. In contrast, we prove the following theorem on the absence of Dirac points in the \textit{square lattices} at the three high-symmetry points in $Y_*$: $\bm{\kappa}^{(1)}=(0,0)$, $\bm{\kappa}^{(2)}=(\pi,\pi)$ and $\bm{\kappa}^{(3)}=(0,\pi)$; see Section 3. Note that $\bm{\kappa}^{(2)}$ is also called the M-point. Those high-symmetry points arise naturally as the only fixed points of the symmetry operations \eqref{eq_symm_momen_sapce_3}-\eqref{eq_symm_momen_sapce_2}.
\begin{theorem} \label{thm_non_exist_dirac_point}
The linear degenerate points cannot appear at $\bm{\kappa}^{(i)}$ ($i=1,2,3$).
\end{theorem}

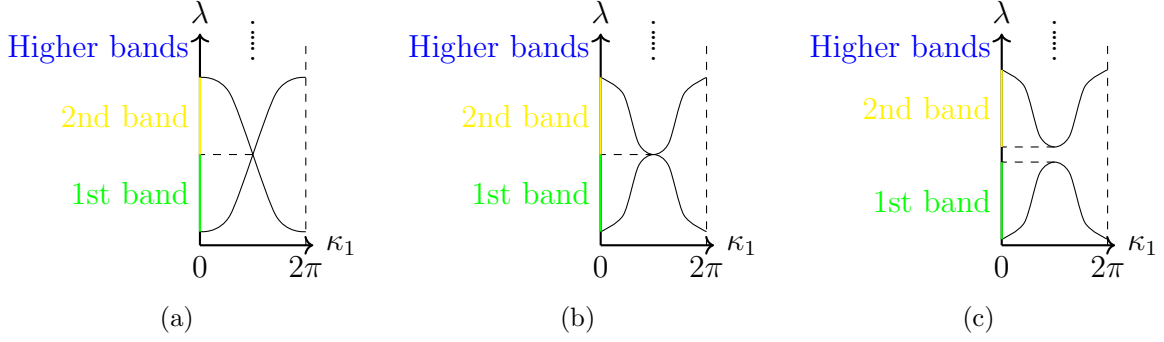
\begin{figure}
\centering
\subfigure[]{
\label{fig_band_structure_dirac_cone}
\begin{tikzpicture}[scale=0.1]

\draw[thick,->] (0,0)--(15,0);
\draw[thick,->] (0,0)--(0,28);
\node[below] at (0,0) {$0$};
\node[below] at (14,0) {$2\pi$};
\node[right] at (15.2,0) {$\kappa_1$};
\node[above] at (0,28) {$\lambda$};
\draw[dashed] (14,0)--(14,27.2);

\node[above] at (7,23) {$\vdots$};
\node[above] at (7,25) {$\vdots$};
\draw plot [smooth] coordinates {(0,22.2)  (4,20)  (10,4)  (14,1.8)};
\draw plot [smooth] coordinates {(0,1.8)  (4,4)  (10,20)  (14,22.2)};
\draw[dashed] (7,12)--(0,12);
\draw[thick,green] (0,1.8)--(0,12);
\draw[thick,yellow] (0,22.2)--(0,12);
\node[left,green] at (0,6.9) {1st band};
\node[left,yellow] at (0,17.1) {2nd band};
\node[left,blue] at (0,26) {Higher bands};
\end{tikzpicture}
}
\subfigure[]{
\label{fig_band_structure_quadratic}
\begin{tikzpicture}[scale=0.1]

\draw[thick,->] (0,0)--(15,0);
\draw[thick,->] (0,0)--(0,28);
\node[below] at (0,0) {$0$};
\node[below] at (14,0) {$2\pi$};
\node[right] at (15.2,0) {$\kappa_1$};
\node[above] at (0,28) {$\lambda$};
\draw[dashed] (14,0)--(14,27.2);

\node[above] at (7,23) {$\vdots$};
\node[above] at (7,25) {$\vdots$};
\draw plot [smooth] coordinates {(0,22.2) (3,20) (5,13.5) (7,12) (9,13.5)(11,20) (14,22.2)};
\draw plot [smooth] coordinates {(0,1.8) (3,4) (5,10.5) (7,12) (9,10.5)(11,4) (14,1.8)};
\draw[dashed] (7,12)--(0,12);
\draw[thick,green] (0,1.8)--(0,12);
\draw[thick,yellow] (0,22.2)--(0,12);
\node[left,green] at (0,6.9) {1st band};
\node[left,yellow] at (0,17.1) {2nd band};
\node[left,blue] at (0,26) {Higher bands};
\end{tikzpicture}
}
\subfigure[]{
\label{fig_band_structure_quadratic_lift}
\begin{tikzpicture}[scale=0.1]

\draw[thick,->] (0,0)--(15,0);
\draw[thick,->] (0,0)--(0,28);
\node[below] at (0,0) {$0$};
\node[below] at (14,0) {$2\pi$};
\node[right] at (15.2,0) {$\kappa_1$};
\node[above] at (0,28) {$\lambda$};
\draw[dashed] (14,0)--(14,27.2);

\node[above] at (7,23) {$\vdots$};
\node[above] at (7,25) {$\vdots$};
\draw plot [smooth] coordinates {(0,23.2) (3,21) (5,14.5) (7,13) (9,14.5)(11,21) (14,23.2)};
\draw plot [smooth] coordinates {(0,0.8) (3,3) (5,9.5) (7,11) (9,9.5)(11,3) (14,0.8)};
\draw[dashed] (7,13)--(0,13);
\draw[dashed] (7,11)--(0,11);
\draw[thick,green] (0,0.8)--(0,11);
\draw[thick,yellow] (0,23.2)--(0,13);
\node[left,green] at (0,5.9) {1st band};
\node[left,yellow] at (0,18.1) {2nd band};
\node[left,blue] at (0,26) {Higher bands};
\end{tikzpicture}
}
\caption{(a) Band structure supposed in Definition \ref{def_linear_degenracy}. A Dirac point exists between the first and second bands; (b) Band structure supposed in Assumption \ref{def_quadratic_degenracy}. The first and second bands touch at $\bm{\kappa}^{(2)}$ quadratically; (c) Band structure of the parity-broken crystal (described by \eqref{eq_perturbed_operator}). The perturbation lifts the quadratic degeneracy and opens a band gap.}
\label{fig_band_structures}
\end{figure}

Note that similar results have been obtained in \cite{weinstein18quadratic} for the Schrödinger operator. As shown in the proof of Theorem \ref{thm_non_exist_dirac_point}, when two bands touch at $\bm{\kappa}^{(i)}$ ($i=1,2,3$), the symmetry of the two Bloch modes therein can be classified by the representation of $C_{4v}$ point group. Following physics literature, the degenerate point is called \textit{accidental} when the two Bloch modes belong to two copies of one-dimensional representations of $C_{4v}$, meaning a perturbation can lift the degeneracy without breaking the $C_{4v}$-symmetry. On the other hand, the degenerate point is called \textit{generic} when the two modes belong to a two-dimensional irreducible representation of $C_{4v}$, where a symmetry-preserving perturbation cannot lift the degeneracy. Following Theorem \ref{thm_non_exist_dirac_point}, we focus on the case where the first and second bands touch at $\bm{\kappa}=\bm{\kappa}^{(2)}$ quadratically, with the two Bloch modes therein classified by the $2D$ irrep of $C_{4v}$. We also assume the slice of $\lambda_1(\bm{\kappa})$ and $\lambda_2(\bm{\kappa})$ along $\kappa_2=\pi$ are opposite up to fourth order near $\bm{\kappa}^{(2)}$ to simplify the analysis. More precisely, we assume that 

\begin{assumption}\label{def_quadratic_degenracy}

(1) $\lambda_{1}(\bm{\kappa}^{(2)})=\lambda_{2}(\bm{\kappa}^{(2)})=\lambda_*$. 

(2) $\lambda_{1}(\bm{\kappa})\neq \lambda_*$, $\lambda_{2}(\bm{\kappa})\neq \lambda_*$ for any $\bm{\kappa}\in Y^*\backslash \{\bm{\kappa}^{(2)}\}$, and $\lambda_{n}(\bm{\kappa})\neq \lambda_*$ for any $n\geq 3$ and $\bm{\kappa}\in Y^*$. 

(3) For $|p|\ll 1$, there exist $\gamma_*,\eta_*>0$ such that
    \begin{equation} \label{eq_quadratic_dispersion}
    \begin{aligned}
    \lambda_{1}(\bm{\kappa}^{(2)}+p\bm{e}_1)-\lambda_{1}(\bm{\kappa}^{(2)})&=-\frac{1}{2}\gamma_*p^2-\eta_*p^4+\mathcal{O}(p^6), \\
    \lambda_{2}(\bm{\kappa}^{(2)}+p\bm{e}_1)-\lambda_{2}(\bm{\kappa}^{(2)})&=\frac{1}{2}\gamma_*p^2+\eta_*p^4+\mathcal{O}(p^6).
    \end{aligned}
    \end{equation}

(4) The Bloch modes at $\bm{\kappa}=\bm{\kappa}^{(2)}$ form a basis of the 2D irreducible representation of $C_{4v}$ in the sense that (see also \eqref{eq_2d_rep})
\begin{equation} \label{eq_assump_mode_symmetry}
\begin{aligned}
&u_{1}(R\bm{x};\bm{\kappa}^{(2)})=i\cdot u_{2}(\bm{x};\bm{\kappa}^{(2)}),\quad
u_{1}(M_2\bm{x};\bm{\kappa}^{(2)})= u_{1}(\bm{x};\bm{\kappa}^{(2)}), \\
&u_{2}(R\bm{x};\bm{\kappa}^{(2)})=i\cdot u_{1}(\bm{x};\bm{\kappa}^{(2)}),\quad
u_{2}(M_2\bm{x};\bm{\kappa}^{(2)})= -u_{2}(\bm{x};\bm{\kappa}^{(2)}).
\end{aligned}
\end{equation}
\end{assumption}
\begin{remark}
Assumption \ref{def_quadratic_degenracy}(2) is essentially the spectral no-fold condition introduced in \cite{fefferman2016honeycomb_edge,fefferman12}. When the no-fold condition fails, either interface modes or resonant modes can bifurcate from the spectral degenerate point, as proved in \cite{qiu2023bifurcationdiracpointphotonic}.
\end{remark}
\begin{remark}
The odd-order terms in the expansion \eqref{eq_quadratic_dispersion} vanish because of the reflection symmetry, specifically, $\lambda_{n}(\bm{\kappa}^{(2)}+p\bm{e}_1)=\lambda_{n}(\bm{\kappa}^{(2)}-p\bm{e}_1)$, which follows from $[\mathcal{L}^{A},\mathcal{M}_2]=0$.
A band structure described in Assumption \ref{def_quadratic_degenracy} is illustrated in Figure \ref{fig_band_structure_quadratic}. For a general characterization of the quadratic degeneracy of the band structure, see \cite{chaban2024instability,weinstein18quadratic}.
\end{remark}
Now we consider the following perturbed operator  
\begin{equation}
\label{eq_perturbed_operator}
\mathcal{L}^{A+\delta\cdot B}:D(\mathcal{L}^{A+\delta\cdot B})\subset L^2(\mathbf{R}^2)\to L^2(\mathbf{R}^2),\quad
u\mapsto \nabla\cdot (A+\delta\cdot B)\nabla u
\end{equation}
where the coefficient matrix $B$ describes the parity-breaking permeability induced by an external magnetic field.  We assume 
\begin{assumption} \label{assum_perturb_coefficient_matrix}
$B(\bm{x})=b(\bm{x})\cdot \sigma_2$, where $b(\bm{x})=\sum_{n_1,n_2\in\mathbf{Z}}\chi_{D_{n_1,n_2}}(\bm{x})$ and $\sigma_2
=\begin{pmatrix}
0 & -i \\ i & 0
\end{pmatrix}
$.
\end{assumption}
The operator $\mathcal{L}^{A+\delta\cdot B}$ is defined in the weak sense by the sesquilinear form \eqref{eq_A_form} with the replacement $A\to A+\delta\cdot B$. 
%We note that $D(\mathcal{L}^{A+\delta\cdot B})\neq D(\mathcal{L}^{A})$.  
%because $\nabla u$ behaves differently near the boundary $\partial D_{n_1,n_2}$ of inclusions, for the function $u$ lying in $D(\mathcal{L}^{A+\delta\cdot B})$ or $D(\mathcal{L}^{A})$.
The coefficient matrix $B$ breaks the reflection symmetry while preserving the rotation symmetry, i.e. the following identities (compared with \eqref{eq_LA_symmetry}) hold
\begin{equation*}
[\mathcal{L}^B,\mathcal{R}]=\{\mathcal{L}^B,\mathcal{M}_2\}=0
\quad \text{on $D(\mathcal{L}^B)$},
\end{equation*}
where $[\cdot, \cdot]$ and $\{\cdot, \cdot\}$ denote the commutator and anticommutator of two operators respectively. 
We note that $B$ also breaks the time-reversal symmetry. 
Such a symmetry-breaking perturbation lifts the quadratic degeneracy at $\bm{\kappa}^{(2)}$. To be more precise, denote 
\[
\Omega=\mathbf{R}\times (-\frac{1}{2},\frac{1}{2}), \quad L_{\pi}^2(\Omega)=
\{u\in L^2_{loc}(\mathbf{R}^2)\cap L^2(\Omega):\, u(\bm{x}+\bm{e}_2)=e^{i\pi}u(\bm{x})\}.
\]
We have
\begin{theorem}[Band gap opening at $\kappa_2=\pi$]
\label{thm_gap_open_corollary}
Let $0<c_0<1$ be a constant close to one and assume that $t_*$ defined in \eqref{eq_perturb_constant} is nonzero. 
Then, for $|\delta|$ being sufficiently small and nonzero, the operator $\mathcal{L}^{A+\delta\cdot B}\Big|_{L_{\pi}^2(\Omega)}$
has a spectral gap $\mathcal{I}_\delta=(\lambda_*-c_0|t_*\delta|,\lambda_*+c_0|t_*\delta|)$ near $\lambda=\lambda_*$.   
\end{theorem}

\begin{remark}
The assumption $t_* \neq 0$ can be verified numerically.    
\end{remark}

Theorem \ref{thm_gap_open_corollary} follows directly from Theorem \ref{thm_gap_open}. Note that in Theorem \ref{thm_gap_open}, it is also proved that the Bloch eigenspace of 
$\mathcal{L}^{A+\delta\cdot B}\Big|_{L_{\pi}^2(\Omega)}$ near $\lambda=\lambda_*$ exchanges its parity for $\delta>0$ and $\delta<0$ (see Remark \ref{rmk_parity_changing}). This is the so-called \textit{band-inversion phenomenon}, a type of topological phase transition \cite{vanderbilt2018berry}, in the physics literature. When such a phase transition occurs, one expects the existence of localized modes at the interface between the two lattices associated with the two operators $\mathcal{L}^{A+\delta\cdot B}$ and $\mathcal{L}^{A-\delta\cdot B}$. This is indeed the case as is proved in this paper. To be more precise, 
%we concern the interface mode when we attach $\mathcal{L}^{A-\delta\cdot B}$ and $\mathcal{L}^{A+\delta\cdot B}$ along the interface $\{0\}\times \mathbf{R}$. In other words, 
we consider the following operator
\begin{equation} \label{eq-interface}
\mathcal{L}^{inter}:
D(\mathcal{L}^{inter})\subset L^2(\mathbf{R}^2)\to L^2(\mathbf{R}^2)
,\quad
u(\bm{x})\mapsto
\left\{
\begin{aligned}
&(\mathcal{L}^{A-\delta\cdot B}u)(\bm{x}),\quad x_1<0, \\
&(\mathcal{L}^{A+\delta\cdot B}u)(\bm{x}),\quad x_1>0.
\end{aligned}
\right.
\end{equation}
The main result of this paper establishes that
\begin{theorem}[Interface modes at $\kappa_2=\pi$]
\label{thm_interface_mode}
Assume that $t_*$ is defined in \eqref{eq_perturb_constant} is nonzero. For $\delta>0$ being sufficiently small, there exist exactly two eigenvalues $\lambda_n^\star(\pi)$ ($n=1,2$) of $\mathcal{L}^{inter}\Big|_{L_{\pi}^2(\Omega)}$ inside the band gap $\mathcal{I}_{\delta}$, with the corresponding eigenmode (or eigenfunction) $u_n^{\star}(\bm{x};\pi)\in L_{\pi}^2(\Omega)$.
\end{theorem}
\begin{remark} \label{rmk-disperse}
With Theorem \ref{thm_interface_mode}, one can immediately prove there exist exactly two eigenvalues $\lambda_n^\star(\kappa_2)$ ($n=1,2$) of $\mathcal{L}^{inter}\Big|_{L_{\kappa_2}^2(\Omega)}$ inside the band gap for $|\kappa_2-\pi|\ll 1$ by a standard perturbation argument. One can also apply the tools developed in this paper to calculate the slope $(\lambda_{n}^\star)^\prime(\pi)$ of the two dispersion curves for the interface modes, as in \cite{li2024interface}.
\end{remark}

\begin{remark}
The IQHE in photonics is featured by robust unidirectional interface or edge modes that propagate along interfaces or boundaries without backscattering. In this paper, we only establish the existence of these interface modes. The unidirectionality can be established by examining the slope of the dispersion curves of interface eigenvalues, see Remark \ref{rmk-disperse}. However, the robustness against structural impurities requires further investigation.

%One way to establish stability is to prove the bulk-interface correspondence principle. Besides the excellent works on this topic as listed in the Introduction, we provide a new proof based on the analysis of Green functions in an upcoming work \cite{qiu24bec}.
\end{remark}

\begin{remark}
We have analyzed the approach of using gyromagnetic materials to generate the IQHE in photonics.
Some proposals rely on different materials including 1) suitably designed temporal modulation of the coupling in optical resonator lattices that creates effective magnetic fields for light \cite{Fang2012} and 2) waveguiding geometries with longitudinal refractive index modulations \cite{Rechtsman2013}. The analysis of these proposals is beyond the scope of this paper. 
    
\end{remark}

\subsection{Outline}
The rest of this paper is organized as follows:

Section 2: We introduce notations and review some basic Floquet theory of periodic operators that are used in this paper.

Section 3: We prove the absence of Dirac points at the three high-symmetry points, $\bm{\kappa}^{(1)}$, $\bm{\kappa}^{(2)}$ and $\bm{\kappa}^{(3)}$,  in the Brillouin zone $Y_*$. This leads to the consideration of a quadratic degenerate point as per Assumption \ref{def_quadratic_degenracy}.

Section 4: We examine the Bloch modes and their momentum derivatives at the quadratic degenerate point. We explicitly construct two analytical branches of the Bloch modes near the quadratic degenerate point and show that the constructed Bloch modes at the quadratic degenerate point have opposite parities to their momentum derivatives. This finding highlights how Bloch mode parity evolves across spectral degenerate points.

%In this section, we derive the asymptotic of the unperturbed Green function in the strip $\Omega$ at the energy level of the quadratic point, by applying the limiting absorption principle(see Theorem \ref{thm_asymp_unperturbed_green}). This reveals the interplay between the localized and extended parts of the wave in the unperturbed structure, providing crucial information for studying localized modes under perturbation. 

Section 5: We derive the asymptotic of the unperturbed Green function in the strip $\Omega$ at the energy level of the quadratic degenerate point, by applying the limiting absorption principle (Theorem \ref{thm_asymp_unperturbed_green}). The asymptotic expansion reveals the distinguished role played by the Bloch modes and their momentum derivatives at the quadratic degenerate point,  providing important information for studying localized modes under perturbation.

Section 6: We examine the interactions between Bloch modes and their momentum derivatives at the quadratic degenerate point by using an energy flux functional. The result shows that the momentum derivatives of Bloch modes act as the ``dual vectors" of Bloch modes. These interactions are essential for analyzing the boundary integral operators associated with interface modes in Section 9.

Section 7: We derive the asymptotic expansion of Bloch eigenpairs near the degenerate point under perturbation (Theorem \ref{thm_gap_open}). We show that a band gap can be opened in the band structure of the perturbed system and that a phase transition occurs during the perturbation process, which is crucial for the existence of interface modes (Remark \ref{rmk_parity_changing}). To accurately capture the phase transition encoded in the momentum derivatives of the Bloch modes, we perform a high-order perturbation analysis.

Section 8: Building on the results from Section 7, we calculate the asymptotics of the perturbed Green function $G^{\delta}(\bm{x},\bm{y};\lambda)$ when $\lambda$ lies within the band gap. This analysis characterizes the impact of perturbations on the Green function, particularly concerning Bloch modes and their momentum derivatives, thereby providing all the necessary information to determine the existence of interface modes.

Section 9: We establish the existence of interface modes, which form the main result of this paper (Theorem \ref{thm_interface_mode}). Our approach is based on the layer-potential techniques. We reformulate the eigenvalue problem of $\mathcal{L}^{inter}$ using a boundary integral equation defined on the interface. Based on the result from Section 8 and the application of an appropriate scaling, we analyze the limiting behavior of the corresponding boundary integral operator. These analyses enable the application of Gohberg-Sigal theory, thereby demonstrating the bifurcation of interface modes from the quadratic degenerate point.

\section{Notations and Prelimiaries}
\subsection{Operators and relations}
Throughout, $R=\begin{pmatrix}
    0 & -1 \\ 1 & 0
    \end{pmatrix}, 
M_{1}=\begin{pmatrix}
    -1 & 0 \\ 0 & 1
    \end{pmatrix},
M_{2}=\begin{pmatrix}
    1 & 0 \\ 0 & -1
\end{pmatrix}$ denote the $\pi/2-$rotation, $x_2$- and $x_1$-axis reflection matrix in $\mathbf{R}^2$. Their counterparts in the function space are denoted by
\begin{equation*}
\mathcal{R}:u(\bm{x})\mapsto u(R\bm{x}),\quad
\mathcal{M}_{1}:u(\bm{x})\mapsto u(M_{1}\bm{x}),\quad
\mathcal{M}_{2}:u(\bm{x})\mapsto u(M_{2}\bm{x}).
\end{equation*}
%The complex conjugate and translation operators are defined as
%\begin{equation*}
%\mathcal{C}:u(\bm{x})\mapsto\overline{u(\bm{x})},\quad 
%\mathcal{T}_{i}:u(\bm{x})\mapsto u(\bm{x}+\bm{e}_i)\,\, (i=1,2).
%\end{equation*}
We denote $u\sim v$ for two functions if there exists a unit $c\in \mathbf{C}$ such that $u=c\cdot v$.

We denote the commutator and anticommutator between two operators $A,B$ by $[A,B]:=AB-BA$ and $\{A,B\}:=AB+BA$, respectively.

\subsection{Geometry}
\noindent $Y=[-\frac{1}{2},\frac{1}{2}]^2$, $Y^*=[0,2\pi]^2$;

\noindent $\Omega=\mathbf{R}\times (-\frac{1}{2},\frac{1}{2}),\quad \Omega^{\text{right}}:=\Omega\cap (\mathbf{R}^{+} \times \mathbf{R})$, \quad $\Omega^{\text{left}}:=\Omega\cap (\mathbf{R}^{-} \times \mathbf{R})$;

\noindent $\Gamma=\{0\}\times (-\frac{1}{2},\frac{1}{2})\subset \Omega,\quad \Gamma^{\text{right}}:=\partial (\Omega^{\text{right}})$, \quad $\Gamma^{\text{left}}:=\partial (\Omega^{\text{left}}),\quad \Gamma^{\pm}=\mathbf{R}\times \{\pm\frac{1}{2}\}$.

\subsection{Funtion spaces and brackets}
\noindent 
$L_{\pi}^2(\Omega):=
\{u\in L^2_{loc}(\mathbf{R}^2)\cap L^2(\Omega):\, u(\bm{x}+\bm{e}_2)=e^{i\pi}u(\bm{x})\}$.

\noindent
$L^2_{\bm{\kappa}}(\mathbf{R}^2):=
\{u\in L^2_{loc}(\mathbf{R}^2):\, u(\bm{x}+n_1\bm{e}_1+n_2\bm{e}_2)=e^{i(n_1\kappa_1+n_2\kappa_2)}u(\bm{x})\}$, equipped with $L^2(Y)-$inner product.

\noindent $H^1_{\bm{\kappa}}(\mathbf{R}^2):=
\{u\in H^1_{loc}(\mathbf{R}^2):\, \partial_{\alpha}u\in L^2_{\bm{\kappa}}(\mathbf{R}^2)\, (|\alpha|\leq 1)\}$, equipped with $H^1(Y)-$inner product;

\noindent $L_{\kappa_2}^2(\Omega):=
\{u\in L^2_{loc}(\mathbf{R}^2): u|_{\Omega}\in L^2(\Omega),\,u(\bm{x}+n\bm{e}_2)=e^{in\kappa_2}u(\bm{x})\}$, equipped with $L^2(\Omega)-$ inner product.

\noindent $H^1_{\kappa_2}(\Omega):=
\{u\in H^1_{loc}(\mathbf{R}^2): u|_{\Omega}\in H^1(\Omega),\, u(\bm{x}+n\bm{e}_2)=e^{in\kappa_2}u(\bm{x}),\partial_{x_2}u(\bm{x}+n\bm{e}_2)=e^{in\kappa_2}\partial_{x_2}u(\bm{x})\}$.

\noindent $(H^1(U))^*$ ($U\subset \mathbf{R}^2$ is an open set): dual of $H^1(U)$ under the dual pairing induced by the $L^2(U)-$inner product.

\noindent $H^{\frac{1}{2}}(\Gamma):=\{u=U|_{\Gamma}:U\in H^{\frac{1}{2}}(\Gamma^{\text{right}})\}$, where $H^{\frac{1}{2}}(\Gamma^{\text{right}})$ is defined in the standard way.

\noindent $\tilde{H}^{-\frac{1}{2}}(\Gamma):=\{u=U|_{\Gamma}:U\in H^{-\frac{1}{2}}(\Gamma^{\text{right}})\text{ and }supp (U)\subset \overline{\Gamma}\}$, where $H^{-\frac{1}{2}}(\Gamma^{\text{right}})$ is the dual of $H^{\frac{1}{2}}(\Gamma^{\text{right}})$ under the dual pair $\langle \varphi,\phi\rangle:=\int_{\Gamma^{\text{right}}}\varphi(\cdot)\phi(\cdot)$.

\medskip

\noindent $(\cdot,\cdot)$: $L^2(Y)-$inner product and dual pairing induced by the $L^2(Y)-$inner product.

\noindent $(\cdot,\cdot)_{\Omega}$: $L^2(\Omega)-$inner product and dual pairing induced by the $L^2(\Omega)-$inner product.

\noindent $\langle\cdot,\cdot\rangle$: the $\tilde{H}^{-\frac{1}{2}}(\Gamma)-H^{\frac{1}{2}}(\Gamma)$ dual pairing.

\subsection{Floquet theory}
Consider the strip domain $\Omega=\mathbf{R}\times (-\frac{1}{2},\frac{1}{2})$. We define 
\begin{equation*}
\mathcal{L}^{A}_{\Omega,\pi}:D(\mathcal{L}^{A}_{\Omega,\pi})\subset L_{\pi}^2(\Omega)\to L_{\pi}^2(\Omega),\quad
u\mapsto -div(A\nabla)u.
\end{equation*}
$\mathcal{L}^{A}_{\Omega,\pi}$ is the section of $\mathcal{L}^{A}(\bm{\kappa})$ along the line $\kappa_2=\pi$. The spectrum of $\mathcal{L}^{A}_{\Omega,\pi}$ can be decomposed by its Floquet components, i.e. 
\[
\sigma(\mathcal{L}^{A}_{\Omega,\pi})=\cup_{0\leq \kappa_1<2\pi}\sigma(\mathcal{L}^{A}_{\Omega,\pi}(\kappa_1)),
\]
where $\mathcal{L}^{A}_{\Omega,\pi}(\kappa_1)=\mathcal{L}^A((\kappa_1,\pi))$. Consequently,
\begin{equation*}
\sigma(\mathcal{L}^{A}_{\Omega,\pi})=\cup_{n\geq 1}\{\lambda_{n}(\kappa_1;\pi),\, 0\leq \kappa_1<2\pi\}
=\cup_{n\geq 1}\{\lambda_{n}((\kappa_1,\pi)),\, 0\leq \kappa_1<2\pi\},
\end{equation*}
where $\{\lambda_{n}(\kappa_1;\pi): n\geq 1\}$ are the Floquet-Bloch eigenvalues of $\mathcal{L}^{A}_{\Omega,\pi}$. On the other hand, since $\mathcal{L}^{A}_{\Omega,\pi}(\kappa_1)$ depends analytically on $\kappa_1$, the Kato-Rellich theorem \cite{kato2013perturbation} indicates that there exist a family of analytic functions $\{\mu_{n}(\kappa_1;\pi)\}_{n\geq 1}$ which forms a rearrangement of $\{\lambda_{n}(\kappa_1;\pi)\}_{n\geq 1}$
\begin{equation*}
\cup_{n\geq 1}\{\mu_{n}(\kappa_1;\pi)\}
=\cup_{n\geq 1}\{\lambda_{n}(\kappa_1;\pi)\},\quad 0\leq \kappa_1<2\pi.
\end{equation*}
We call $\{\mu_{n}(\kappa_1;\pi)\}_{n\geq 1}$ the analytically labeled Floquet-Bloch eigenvalues of $\mathcal{L}^{A}_{\Omega,\pi}$. For simplicity, we abbreviate $\mu_{n}(\kappa_1; \pi)$ as $\mu_{n}(\kappa_1)$ when the context is clear. The normalized Bloch modes associated with $\mu_{n}(\kappa_1)$ can be chosen to depend analytically on $\kappa_1$ and are denoted by $v_n(\mathbf{x}; \kappa_1)$. 
\begin{proposition}[\cite{joly2016solutions}] \label{prop_analytic_label}
For each integer $n\geq 1$, there exists a complex neighborhood $\mathcal{D}_n$ that contains the real line $\mathbf{R}$, and two analytic maps 
\begin{equation*}
\kappa_1\in \mathcal{D}_n \mapsto \mu_{n}(\kappa_1)\in\mathbf{C},\quad
\kappa_1\in \mathcal{D}_n \mapsto v_n(\bm{x};\kappa_1)\in H^1_{(\kappa_1,\pi)}(\mathbf{R}^2),
\end{equation*}
such that $(\mu_{n}(\kappa_1),v_n(\bm{x};\kappa_1))$ $(n\geq 1)$ represent all Floquet-Bloch eigenpairs of $\mathcal{L}^{A}_{\Omega,\pi}(\kappa_1)$ for $0\leq \kappa_1<2\pi$.
\end{proposition}
The following follows directly from Assumption \ref{def_quadratic_degenracy} and the reflection symmetry \eqref{eq_symm_momen_sapce_2}.
\begin{proposition} \label{prop_mu12_asymptotic_even}
Under Assumption \ref{def_quadratic_degenracy}, the first two branches of analytic Floquet-Bloch eigenvalues can be chosen so that they meet at $(\pi, \lambda_*)$ and satisfy $\mu_n(\kappa_1)=\lambda_{n}(\bm{\kappa}^{(2)}+(\kappa_1-\pi)\bm{e}_1)$ for $|\kappa_1-\pi|\ll 1$.  
Moreover,  
\begin{equation*}
\mu_n(\kappa_1)=\mu_n(2\pi-\kappa_1), \quad \mbox{for $n=1,2$.}
\end{equation*}
\end{proposition}

\section{Absence of linear degeneracies}
In this section, we prove Theorem \ref{thm_non_exist_dirac_point} for $\bm{\kappa}=\bm{\kappa}^{(2)}$ by following the argument in \cite{berkolaiko2018symmetry}. The proof for $\bm{\kappa}=\bm{\kappa}^{(1)}$ and $\bm{\kappa}^{(3)}$ is similar.

\subsection{Representation theory of $C_{4v}$ point group and symmetry in momentum space}
We first recall some essential representation theory. Let $\mathcal{L}$ be a self-adjoint operator acting on a separable Hilbert space $\mathcal{X}$. Let $\mathcal{S}$ be a finite group of unitary operators that act on $\mathcal{X}$ and commute with $\mathcal{L}$. Then $\mathcal{X}$ can be decomposed into an orthogonal direct sum of subspaces as follows  
\begin{equation*}
\mathcal{X}=\oplus_{\rho}\mathcal{X}_{\rho},
\end{equation*}
where each $\mathcal{X}_{\rho}$ comprising multiple isomorphic copies of the irreducible representation 
$\rho$ of $\mathcal{S}$. Each copy is referred to as an isotypic component \cite{berkolaiko2018symmetry}.

Since $\mathcal{L}$ commutes with $\mathcal{S}$, each isotypic component $\mathcal{X}_{\rho}$ is invariant under $\mathcal{L}$. Moreover, if $\mathcal{L}$ has a discrete spectrum, then the restriction of $\mathcal{L}$ to $\mathcal{X}_{\rho}$ has eigenvalues with multiplicities divisible by the dimension of $\rho$.

We now restrict to the case $\mathcal{S}=C_{4v}$, where $\mathcal{S}$ is generated by $\mathcal{R}$ ($\pi/2$-rotation) and $\mathcal{M}_2$ ($x_1$-reflection), with the following rules
\begin{equation*}
\mathcal{R}^4=\mathcal{M}_2^2=I,\quad
\mathcal{M}_2\mathcal{R}^{-1}=\mathcal{R}\mathcal{M}_2.
\end{equation*}
Note that $\mathcal{S}$ acts invariantly on the spaces $L^2_{\bm{\kappa}^{(i)}}$ for $i=1,2$, since the quasi-periodicity of $u$ is preserved under $\mathcal{S}$ if $u\in L^2_{\bm{\kappa}^{(i)}}$. Additionally,  
$\mathcal{L}^{A}(\bm{\kappa}^{(i)})$ commutes with $\mathcal{S}$.  Within each space $L^2_{\bm{\kappa}^{(i)}}$, $\mathcal{S}$ admits four one-dimensional representations and a unique two-dimensional representation, as listed below. 
\begin{equation*} 
\rho_1:\quad
\mathcal{R}\mapsto (1),\quad
\mathcal{M}_2\mapsto (1); \quad
\rho_2:\quad
\mathcal{R}\mapsto (1),\quad
\mathcal{M}_2\mapsto (-1);
\end{equation*}
\begin{equation*}
\rho_3:\quad
\mathcal{R}\mapsto (-1),\quad
\mathcal{M}_2\mapsto (1);\quad
\rho_4:\quad
\mathcal{R}\mapsto (-1),\quad
\mathcal{M}_2\mapsto (-1);
\end{equation*}
\begin{equation} \label{eq_2d_rep}
\rho_5:\quad
\mathcal{R}\mapsto i\sigma_3=
\begin{pmatrix}
i & 0 \\ 0 & -i
\end{pmatrix}
,\quad
\mathcal{M}_2\mapsto \sigma_1=
\begin{pmatrix}
0 & 1 \\ 1 & 0
\end{pmatrix}.
\end{equation}
%By \eqref{eq_LA_symmetry}, $\mathcal{L}^{A}$ commutes with the elements of $\mathcal{S}$. 
Particularly, the following holds 
\begin{proposition} \label{prop_classify_degenerate_mode}
Suppose $\{\lambda,u\}\in\mathbf{R}\times L^2_{\bm{\kappa}^{(2)}}$ is an eigenpair of $\mathcal{L}^{A}(\bm{\kappa}^{(2)})$. Then the exists a unique $\rho_k$ ($1\leq k\leq 5$) such that $u\in L^2_{\bm{\kappa}^{(2)},\rho_k}$. Here $L^2_{\bm{\kappa}^{(2)},\rho_k}$ denotes an isotypic component of $L^2_{\bm{\kappa}^{(2)}}$ associated with the representation $\rho_k$. 
\end{proposition}

\subsection{Proof of Theorem \ref{thm_non_exist_dirac_point} }
We prove Theorem \ref{thm_non_exist_dirac_point} in this subsection by contradiction.  Without loss of generality, we assume that the first and second spectral bands of the operator $\mathcal{L}^{A}$ touch at $(\bm{\kappa}^{(2)},\lambda_*)$ and that $(\bm{\kappa}^{(2)},\lambda_*)$ is a linear degenerate point, i.e. the conditions in Definition \ref{def_linear_degenracy} hold for $n_*=1$ and ($\bm{\kappa}_*,\lambda_*)=(\bm{\kappa}^{(2)},\lambda_*)$.  We denote $u_1=u_1(\bm{x};\bm{\kappa}^{(2)})$ and $u_2=u_2(\bm{x};\bm{\kappa}^{(2)})$ for ease of notations. By Proposition \ref{prop_classify_degenerate_mode}, there are two cases:
\begin{itemize}
    \item Case 1. $u_1,u_2\in L^2_{\bm{\kappa}^{(2)},\rho_5}$;
    \item Case 2. $u_1,u_2\in \oplus_{k=1}^{4}L^2_{\bm{\kappa}^{(2)},\rho_k}$.
\end{itemize}
We shall derive contradictions in both cases to prove Theorem \ref{thm_non_exist_dirac_point}.

$\mathbf{Proof \,\,for \,\,Case \,1}$: Without loss of generality, we set $\mathcal{R}u_1=iu_1$, $\mathcal{R}u_2=-iu_2$. In the first-order perturbation argument (i.e. ignoring the 
$\mathcal{O}(|\bm{\kappa}-\bm{\kappa}^{(2)}|^2)$ term in \eqref{eq_conical_dispersion}), the dispersion surface near $(\bm{\kappa}^{(2)},\lambda_*)$ can be obtained by solving the following equation (See \cite{qiu2023mathematical,fefferman2017topologically,fefferman2016honeycomb_edge})
\begin{equation*}
\det(\delta \kappa_1\cdot h_1+\delta \kappa_2\cdot h_2-\delta\lambda)=0,
\end{equation*}
where $\delta \kappa_{i}:=(\bm{\kappa}_{i}-\bm{\kappa}^{(2)}_{i})\cdot \bm{e}_i$, $\delta \lambda:= \lambda-\lambda_*$, and
\begin{equation*}
h_i=
\begin{pmatrix}
\big( u_1,\frac{\partial \mathcal{L}^A}{\partial \kappa_i}(\bm{\kappa}^{(2)})u_1\big)
& \big( u_1,\frac{\partial \mathcal{L}^A}{\partial \kappa_i}(\bm{\kappa}^{(2)})u_2\big) \\
\big( u_2,\frac{\partial \mathcal{L}^A}{\partial \kappa_i}(\bm{\kappa}^{(2)})u_1\big)
& \big( u_2,\frac{\partial \mathcal{L}^A}{\partial \kappa_i}(\bm{\kappa}^{(2)})u_2\big)
\end{pmatrix}
,\quad i=1,2.
\end{equation*}
Here $(\cdot,\cdot)$ denotes the $L^2(Y)-$inner product. We claim that $h_1=h_2=0$. Then the first-order terms in the asymptotic expansion of the dispersion function $\lambda=\lambda(\bm{\kappa})$ near $\bm{\kappa}=\bm{\kappa}^{(2)}$ vanishes, and this contradicts to the definition of Dirac points. We now prove the claim by exploiting the symmetry of $u_1$ and $u_2$. Taking $\bm{\kappa}=\bm{\kappa}^{(2)}+\delta \kappa_{1}\cdot \bm{e}_1+\delta \kappa_{2}\cdot \bm{e}_2$ in \eqref{eq_symm_momen_sapce_1}, we have
\begin{equation*}
\begin{aligned}
\mathcal{R}\mathcal{L}^{A}(\bm{\kappa}^{(2)}+\delta \kappa_{1}\cdot \bm{e}_1+\delta \kappa_{2}\cdot \bm{e}_2)\mathcal{R}^{-1}
&=\mathcal{L}^{A}\Big(R^{-1}\big(\bm{\kappa}^{(2)}+\delta \kappa_{1}\cdot \bm{e}_1+\delta \kappa_{2}\cdot \bm{e}_2\big)\Big) \\
&=\mathcal{L}^{A}\big(\bm{\kappa}^{(2)}-2\pi\bm{e}_2+\delta \kappa_{2}\cdot \bm{e}_1-\delta \kappa_{1}\cdot \bm{e}_2\big).
\end{aligned}
\end{equation*}
The equality \eqref{eq_symm_momen_sapce_3} gives
\begin{equation*}
\begin{aligned}
\mathcal{R}\mathcal{L}^{A}(\bm{\kappa}^{(2)}+\delta \kappa_{1}\cdot \bm{e}_1+\delta \kappa_{2}\cdot \bm{e}_2)\mathcal{R}^{-1}
=\mathcal{L}^{A}\big(\bm{\kappa}^{(2)}+\delta \kappa_{2}\cdot \bm{e}_1-\delta \kappa_{1}\cdot \bm{e}_2\big).
\end{aligned}
\end{equation*}
By expanding the equation to the first order, we arrive at
\begin{equation} \label{eq_sec3_1}
\begin{aligned}
\mathcal{R}\big(
\delta \kappa_{1}\cdot\frac{\partial \mathcal{L}^A}{\partial \kappa_1}(\bm{\kappa}^{(2)})
+
\delta \kappa_{2}\cdot\frac{\partial \mathcal{L}^A}{\partial \kappa_2}(\bm{\kappa}^{(2)})
\big)\mathcal{R}^{-1}
=
-
\delta \kappa_{1}\cdot\frac{\partial \mathcal{L}^A}{\partial \kappa_2}(\bm{\kappa}^{(2)})
+
\delta \kappa_{2}\cdot\frac{\partial \mathcal{L}^A}{\partial \kappa_1}(\bm{\kappa}^{(2)}).
\end{aligned}
\end{equation}
Note that the matrix of $\mathcal{R}$ under the basis $\{u_1,u_2\}$ is $i\sigma_3$ by \eqref{eq_2d_rep}. Thus, by conjugating both sides of \eqref{eq_sec3_1} with $u_i,u_j$ ($i,j\in\{1,2\}$), we obtain
\begin{equation*}
(-i\sigma_3)(\delta \kappa_{1}\cdot h_1+\delta \kappa_{2}\cdot h_2)(i\sigma_3)
=-\delta \kappa_{1}\cdot h_2+\delta \kappa_{2}\cdot h_1.
\end{equation*}
Separating the variables $\delta \kappa_1$ and $\delta \kappa_2$ yields
\begin{equation} \label{eq_sec3_2}
(-i\sigma_3)h_1(i\sigma_3)=-h_2,\quad
(-i\sigma_3)h_2(i\sigma_3)=h_1.
\end{equation}
Hence
\begin{equation*}
h_1=(-i\sigma_3)h_2(i\sigma_3)=-(-i\sigma_3)^2 h_1(i\sigma_3)^2=-h_1.
\end{equation*}
Thus $h_1=0$, and consequently, $h_2=-(-i\sigma_3)h_1(i\sigma_3)=0$.

\medskip

$\mathbf{Proof \,\,for \,\,Case \,2}$: The proof is similar to Case 1. In fact, in all subcases, equation \eqref{eq_sec3_2} holds, which then implies that $h_1=h_2=0$, and hence leads to a contradiction. For instance, 
when $u_1\in L^2_{\bm{\kappa}^{(2)},\rho_1}$ and $u_2\in L^2_{\bm{\kappa}^{(2)},\rho_3}$, the matrix of $\mathcal{R}$ under the basis $\{u_1,u_2\}$ is $\sigma_3$ since $\mathcal{R}u_1=u_1$, $\mathcal{R}u_2=-u_2$. Then $\sigma_3h_1\sigma_3=-h_2, \sigma_3h_2\sigma_3=h_1$, i.e. equation \eqref{eq_sec3_2} holds. The other subcases can be proved similarly.

%Thus, equation \eqref{eq_sec3_2} still holds
%\begin{equation*}
%\sigma_3h_1\sigma_3=-h_2,\quad
%\sigma_3h_2\sigma_3=h_1,
%\end{equation*}
%which implies the triviality of $h_1$ and $h_2$. For other possible cases, one can readily check that the representation of $\mathcal{R}$ in $\text{span}\{u_1,u_2\}$ is given by $I$, $-I$ or $-\sigma_3$. In either case, we can prove $h_1=h_2=0$ by imitating the argument in Section 2.2. Thus we conclude the proof of Theorem \ref{thm_non_exist_dirac_point}.

\section{Bloch modes and their momentum derivatives at the quadratic degenerate point}

In this section, we investigate the Bloch modes and their momentum derivatives at the quadratic degenerate point, which are crucial for understanding the phase transition phenomena in its vicinity and for analyzing the interface modes that bifurcate from it. We explicitly construct two analytical branches of Bloch modes in the vicinity of the quadratic degenerate point using a perturbation argument and demonstrate that the modes at this point possess opposite parities compared to their momentum derivatives. The perturbation argument also paves the way for the more dedicated perturbation arguments in Section 7. 

We first examine the parity of the Bloch modes at the quadratic degenerate point. 
Note that the analytic Floquet-Bloch eigenvalues $\mu_n(\kappa_1)$ coincide with $\lambda_n(\kappa_1;\pi)$ for $n=1,2$ and $|\kappa_1-\pi|\ll 1$, as in Proposition \ref{prop_mu12_asymptotic_even}. Consequently, Assumption \ref{def_quadratic_degenracy}(4) specifies the parity of $v_n(\bm{x}; \pi)$ as follows:
\begin{equation} \label{eq_v12_parity_1}
(\mathcal{M}_1v_1)(\cdot;\pi)=-  v_1(\cdot;\pi),\quad (\mathcal{M}_1v_2)(\cdot;\pi)=  v_2(\cdot;\pi).
\end{equation}

%Equation \eqref{eq_v12_parity_1} is gauge-independent, meaning it remains valide whenever $v_n(\bm{x};\kappa_1)$ is multiplied by a smooth phase function $e^{i\eta(\kappa_1)}$.

We next focus on the momentum derivative of 
$v_n(\bm{x};\kappa_1)$ ($n=1,2$), i.e. $\partial_{\kappa_1} v_n(\cdot;\kappa_1)$, which satisfy the following equation:
\begin{equation} \label{eq_partial_kappa1_vn_interior}
\big(-\nabla\cdot A\nabla-\mu_n(\kappa_1)\big)\partial_{\kappa_1} v_n(\cdot;\kappa_1)
=\mu_{n}^{\prime}(\kappa_1)v_n(\cdot;\kappa_1).
\end{equation}
Since $\mu^{\prime}_{n}(\pi)=0$ for $n=1,2$, it follows that $(\partial_{\kappa_1}v_n)(\bm{x};\pi)$ are generalized eigenmodes of the operator $\mathcal{L}^{A}$. However, these modes are not Bloch modes due to the boundary conditions
\begin{equation*}
\partial_{\kappa_1}v_n(\bm{x}+\bm{e}_1;\pi)
=e^{i\pi}\partial_{\kappa_1}v_n(\bm{x};\pi)
+ie^{i\pi}v_n(\bm{x};\pi).
\end{equation*}
Nevertheless, $(\partial_{\kappa_1}v_n)(\bm{x};\pi)$ are extended modes of the degenerate eigenvalue $\lambda=\lambda_*$. 

\begin{remark}
As established in our previous works \cite{qiu2023mathematical,li2024interface}, the interface modes bifurcated from a Dirac point are determined by the extended modes at the Dirac point that consist solely of Bloch modes. The presence of additional extended modes, $(\partial_{\kappa_1}v_n)(\bm{x};\pi)$, highlights a critical difference between quadratic degenerate points and Dirac points. Consequently, the study of interface modes necessitates a new framework that incorporates these momentum derivatives, as developed in this paper. 
\end{remark}

\subsection{Parity of momentum derivative of Bloch modes at the quadratic degenerate point}
We investigate the parity of the momentum derivative of Bloch modes at the quadratic degenerate point. Specifically, we demonstrate that, with an appropriate gauge choice, $\partial_{\kappa_1}v_n(\bm{x};\pi)$ exhibits opposite parity to $v_n(\bm{x};\pi)$. 
This property elucidates how the parity of Bloch modes evolves across the spectral degenerate point and significantly simplifies our calculations, particularly the energy flux computation in Section 6, facilitating the analysis of interface modes in Section 9. The construction of $v_n(\bm{x};\pi)$, $n=1,2$ that satisfies this gauge condition will be provided in Section 4.2. The main result of this section is:
\begin{proposition} \label{prop_local_gauge}
If
\begin{equation} \label{eq_orthogonal_gauge_criterion}
\text{Im}(\partial_{\kappa_1}v_n(\cdot;\pi),v_n(\cdot;\pi))=0,\quad n=1,2,
\end{equation}
then
\begin{equation} \label{eq_v12_parity_2}
(\mathcal{M}_1\partial_{\kappa_1}v_1)(\cdot;\pi)= (\partial_{\kappa_1}v_1)(\cdot;\pi),\quad
(\mathcal{M}_1\partial_{\kappa_1}v_2)(\cdot;\pi)=- (\partial_{\kappa_1}v_2)(\cdot;\pi).
\end{equation}
\end{proposition}

\begin{proof}
We prove \eqref{eq_v12_parity_2} for $n=1$. The proof of $n=2$ is similar. By Proposition \ref{prop_mu12_asymptotic_even}, $\mu_{1}(\kappa_1)=\mu_{1}(2\pi-\kappa_1)$. By the reflection symmetry 
$[\mathcal{L}^{A},\mathcal{M}_1]=0$, for each $\kappa_1$ with $|\kappa_1-\pi|\ll 1$, $(\mathcal{M}_1 v_1)(\cdot;\kappa_1)$ is a Bloch mode of $\mathcal{L}^{A}_{\Omega,\pi}(2\pi-\kappa_1)$ associated with the eigenvalue $\mu_1(2\pi-\kappa_1)$. Consequently, there exists  $c(\kappa_1)\in\mathbf{C}$ with $|c(\kappa_1)|=1$ such that
\begin{equation} \label{eq_local_parity_vn}
(\mathcal{M}_1 v_1)(\cdot;\kappa_1)=c(\kappa_1)\cdot v_1(\cdot;2\pi-\kappa_1).
\end{equation}
By letting $\kappa_1=\pi$ in \eqref{eq_local_parity_vn}, we see that $v_1(\cdot;\pi)$ is an eigenfunction of the reflection operator $\mathcal{M}_1$ with eigenvalue $c(\pi)$. Hence, by \eqref{eq_v12_parity_1}, we know $c(\pi)=-1$.

By conjugating the equation \eqref{eq_local_parity_vn} with $v_1(\cdot;2\pi-\kappa_1)$, we have 
\begin{equation}
\label{eq_sec3_7}
c(\kappa_1)=\big((\mathcal{M}_1 v_1)(\cdot;\kappa_1),v_1(\cdot;2\pi-\kappa_1)\big).
\end{equation}
The analyticity of $v_1(\cdot;\kappa_1)$ implies $c(\kappa_1)$ is smooth. Differentiating \eqref{eq_sec3_7} gives
\footnotesize
$$
\begin{aligned}
c^{\prime}(\pi)
&=\big(\mathcal{M}_1(\partial_{\kappa_1}v_1)(\cdot;\pi),v_1(\bm{x};\pi)\big)-\big(\mathcal{M}_1v_1(\cdot;\pi),\partial_{\kappa_1}v_1(\cdot;\pi)\big) =\big(\partial_{\kappa_1}v_1(\cdot;\pi),\mathcal{M}_1 v_1(\bm{x};\pi)\big)-\big(\mathcal{M}_1v_1(\cdot;\pi),\partial_{\kappa_1}v_1(\cdot;\pi)\big).
\end{aligned}
$$
\normalsize
Using \eqref{eq_v12_parity_1},
\begin{equation*}
c^{\prime}(\pi)=-\big(\partial_{\kappa_1}v_1(\cdot;\pi),v_1(\bm{x};\pi)\big)+\big(v_1(\cdot;\pi),\partial_{\kappa_1}v_1(\cdot;\pi)\big)
=-2i\text{Im}(\partial_{\kappa_1}v_1(\cdot;\pi),v_1(\cdot;\pi))=0.
\end{equation*}
Hence \eqref{eq_orthogonal_gauge_criterion} implies $c^{\prime}(\pi)=0$. Differentiating \eqref{eq_local_parity_vn} at $\kappa_1=\pi$ gives
\begin{equation*}
\begin{aligned}
(\mathcal{M}_1 \partial_{\kappa_1} v_n)(\cdot;\pi)
&=-c(\pi)\cdot (\partial_{\kappa_1} v_n)(\cdot;\pi)
+c^{\prime}(\pi)\cdot v_n(\cdot;\pi) 
=(\partial_{\kappa_1} v_n)(\cdot;\pi),
\end{aligned}
\end{equation*}
which concludes the proof of \eqref{eq_v12_parity_2} for $n=1$.
\end{proof}

\subsection{Construction of Bloch modes near the quadratic degenerate point}

In this section, we explicitly construct the Bloch modes $v_n(\bm{x};\kappa_1)$ for $n=1, 2$ and $|\kappa_1-\pi|\ll 1$ that satisfy \eqref{eq_orthogonal_gauge_criterion} by using the Bloch modes at $(\bm{\kappa}=\bm{\kappa}^{(2)}, \lambda_*)$ (see \eqref{eq_assump_mode_symmetry}). These constructed Bloch modes will be frequently utilized in the subsequent analysis. 

\begin{proposition} \label{prop_pert_theory_for_unpertrubed}
There exists two analytic branches of Bloch modes $v_1(\bm{x};\kappa_1)$ and $v_1(\bm{x};\kappa_1)$ that satisfy \eqref{eq_orthogonal_gauge_criterion} (constructed in \eqref{eq_perturb_theory_for_unperturbed_proof_14_v1} and \eqref{eq_perturb_theory_for_unperturbed_proof_14_v2}). Moreover, they admit the following asymptotics for $|\kappa_1-\pi|\ll 1$
\footnotesize
\begin{equation} \label{eq_ansatz_unperturb_eigenfunction_1}
v_1(\bm{x};\kappa_1)
=\frac{v_1+(\kappa_1-\pi)\partial_{\kappa_1}v_1(\bm{x};\pi)
+\sum_{k=2}^{5}v^{(k)}_{1}(\bm{x};\pi)
+\mathcal{O}((\kappa_1-\pi)^6)}{\sqrt{1+N_1^{(2)}(\kappa_1-\pi)^2+N_1^{(3)}(\kappa_1-\pi)^3+\mathcal{O}((\kappa_1-\pi)^4)}},
\end{equation}
\begin{equation} \label{eq_ansatz_unperturb_eigenfunction_2}
v_2(\bm{x};\kappa_1)
=\frac{v_1+(\kappa_1-\pi)\partial_{\kappa_1}v_2(\bm{x};\pi)
+\sum_{k=2}^{5}v^{(k)}_2(\bm{x};\pi)
+\mathcal{O}((\kappa_1-\pi)^6)}{\sqrt{1+N_2^{(2)}(\kappa_1-\pi)^2+N_2^{(3)}(\kappa_1-\pi)^3+\mathcal{O}((\kappa_1-\pi)^4)}},
\end{equation}
\normalsize
where $v_{i}^{(j)}$ and the normalization constants $N_{i}^{(j)}$ are introduced in \eqref{eq_perturb_theory_for_unperturbed_final_v1k_def}, \eqref{eq_perturb_theory_for_unperturbed_final_v2k_def}, \eqref{eq_normalization_factor_n1k_def} and \eqref{eq_normalization_factor_n2k_def}.
\end{proposition}

\begin{remark}
In equations \eqref{eq_ansatz_unperturb_eigenfunction_1} and \eqref{eq_ansatz_unperturb_eigenfunction_2}, we provide an expansion of the Bloch modes up to the fifth order, which exceeds the requirements for establishing the first-order condition \eqref{eq_orthogonal_gauge_criterion}. This comprehensive expansion is necessary to characterize the phase transition phenomenon and to capture its impact on the interface modes bifurcating from the quadratic degenerate point (see Remark \ref{rmk_importance_expansion}).
\end{remark}

To prove Proposition \ref{prop_pert_theory_for_unpertrubed}, we first lay the groundwork for the perturbation analysis. Introduce the operator
\begin{equation} \label{eq_transformed_operator_unperturbed}
\begin{aligned}
\tilde{\mathcal{L}}^{A}_{\Omega,\pi}(\kappa_1):
&\quad H\to H^*,\\
&\quad u\mapsto 
e^{-i\kappa_1 x_1}\circ \mathcal{L}^{A}_{\Omega,\pi}(\kappa_1)\circ e^{i\kappa_1 x_1}u 
=-(\nabla+i\kappa_1 \bm{e}_1)\cdot A(\nabla+i\kappa_1 \bm{e}_1)u,
\end{aligned}
\end{equation}
where $H=H_{(0,\pi)}^1(\mathbf{R}^2)$ and $H^*$ denotes the dual of $H$. The $H-H^*$ pairing is the natural extension of the $L^2(Y)-$inner product, still denoted as $(\cdot,\cdot)$. For each $u\in H$, we have for all $v\in H$
\begin{equation*}
( \tilde{\mathcal{L}}^{A}_{\Omega,\pi}(\kappa_1)u,v)=\mathfrak{a}^{A}_{\kappa_1}(u,v)
:=\int_{Y}\big(A(\bm{x})(\nabla+i\kappa_1 \bm{e}_1) u(\bm{x})\big)\cdot \overline{(\nabla+i\kappa_1 \bm{e}_1) v(\bm{x})}d\bm{x}.
\end{equation*}
It is evident that the eigenpairs of the operator $\mathcal{L}^{A+\delta\cdot B}_{\Omega,\pi}(\kappa_1)$ for $|(\kappa_1,\lambda)-(\pi,\lambda_*)|\ll 1$ can be obtained by solving those of $\tilde{\mathcal{L}}^{A+\delta\cdot B}_{\Omega,\pi}(\kappa_1)$. Moreover, the latter operators are more suitable for perturbation arguments as their domain $H$ is independent of both $\kappa_1$ and $\delta$.  We expand $\tilde{\mathcal{L}}^{A}_{\Omega,\pi}(\kappa_1)$ near $\kappa_1=\pi$: 
\begin{equation} \label{eq_expansion_operator_unperturbed}
\begin{aligned}
\tilde{\mathcal{L}}^{A}_{\Omega,\pi}(\kappa_1)
&=-(\nabla+i\pi \bm{e}_1)\cdot A(\nabla+i\pi \bm{e}_1) \\
&\quad +(\kappa_1-\pi)\big(-i(\nabla+i\pi \bm{e}_1)\cdot A\bm{e}_1
-i\bm{e}_1\cdot A(\nabla+i\pi \bm{e}_1) \big) \\
&\quad +(\kappa_1-\pi)^2(\bm{e}_1\cdot A\bm{e}_1) \\
&=:\tilde{\mathcal{L}}^{A,0}
+(\kappa_1-\pi)\tilde{\mathcal{L}}^{A,1}
+(\kappa_1-\pi)^2\tilde{\mathcal{L}}^{A,2}.
\end{aligned}
\end{equation}
Note that all these operators are $\mathcal{M}_2-$invariant:
\begin{equation} \label{eq_perturb_theory_for_unperturbed_proof_1}
[\tilde{\mathcal{L}}^{A,k},\mathcal{M}_2]=0
\quad (k=0,1,2)
\end{equation}
Let 
$$
v_n(\bm{x}):=u_{n}(\bm{x};\bm{\kappa}^{(2)}),\quad
\tilde{v}_n(\bm{x}):=e^{-i\pi x_1}v_n(\bm{x}), \quad n=1, 2,
$$
where $u_{n}(\bm{x};\bm{\kappa}^{(2)})$ are the Bloch modes at $\bm{\kappa}=\bm{\kappa}^{(2)}$. The parity of $u_{n}(\bm{x};\bm{\kappa}^{(2)})$ in \eqref{eq_assump_mode_symmetry} gives
\begin{equation}
\label{eq_perturb_theory_for_unperturbed_proof_7}
\mathcal{M}_2 \tilde{v}_1=\tilde{v}_1,\quad \mathcal{M}_2 \tilde{v}_2=-\tilde{v}_2.
\end{equation}
The symmetry of operators and functions is extensively utilized in subsequent analysis and calculations. The key points are summarized below:
\begin{lemma} \label{lem_symmetry_products}
Suppose $P_1$ commutes with $\mathcal{M}_2$ and $P_2$ anti-commutes, i.e. $[P_1,\mathcal{M}_2]=0$ and $\{P_2,\mathcal{M}_2\}=0$. Then
\begin{equation} \label{eq_symmetry_products_1}
(P_1\tilde{v}_i,\tilde{v}_j)=0 \quad \text{for $i\neq j$},
\end{equation}
\begin{equation} \label{eq_symmetry_products_2}
(P_2\tilde{v}_i,\tilde{v}_j)=0 \quad \text{for $i= j$}.
\end{equation}
\end{lemma}
\begin{proof}
\eqref{eq_symmetry_products_1} follows from the following identity
\begin{equation*}
\begin{aligned}
\big(P\tilde{v}_{i},\tilde{v}_{j}\big)
=\big(\mathcal{M}_2 P\tilde{v}_{i},\mathcal{M}_2 \tilde{v}_{j}\big)
=\big( P\mathcal{M}_2\tilde{v}_{i},\mathcal{M}_2 \tilde{v}_{j}\big)
=-\big(P\tilde{v}_{i},\tilde{v}_{j}\big).
\end{aligned}
\end{equation*}
and \eqref{eq_symmetry_products_2} is proved similarly.
\end{proof}

We also require the following properties of 
the operator $\tilde{\mathcal{L}}^{A}_{\Omega,\pi}(\pi)$, which are derived from the Fredholm alternative for second-order elliptic operators.
\begin{lemma} \label{lem_fred_alternative}
$\tilde{\mathcal{L}}^{A}_{\Omega,\pi}(\pi):H\to H^*$ is a Fredholm operator with zero index. Moreover, $\ker (\tilde{\mathcal{L}}^{A,0})-\lambda_*)=\text{span}\{\tilde{v}_1(\bm{x}),\tilde{v}_2(\bm{x})\}$.
\end{lemma}

We now are ready to prove Proposition \ref{prop_pert_theory_for_unpertrubed}. 

\begin{proof}
{\color{blue}Step 1.} We construct eigenfunctions of $\tilde{\mathcal{L}}^{A}_{\Omega,\pi}(\kappa_1)$ for $|\kappa_1-\pi|\ll 1$ in the following form
\begin{equation} \label{eq_perturb_theory_for_unperturbed_proof_2}
\tilde{v}=\tilde{v}^{(0)}+\tilde{v}^{(1)}
\quad \text{with}\quad
\tilde{v}^{(0)}=a\cdot\tilde{v}_{1}(\bm{x})+b\cdot\tilde{v}_{2}(\bm{x})\in H_1,\quad
\tilde{v}^{(1)}\in H_2,
\end{equation}
where $H_1:=\ker (\tilde{\mathcal{L}}^{A,0}-\lambda_*)$ and $H_2$ is the orthogonal complement of $H_1$ in the Hilbert space $H$. Denote the eigenvalue of $\tilde{v}$ by \[
\mu=\lambda_*+\mu^{(1)}, \quad |\mu^{(1)}|\ll 1, 
\]
and $p=\kappa_1-\pi$.
The equation $(\tilde{\mathcal{L}}^{A}_{\Omega,\pi}(\kappa_1)-\mu)\tilde{v}=0$ can be rewritten using \eqref{eq_expansion_operator_unperturbed} and \eqref{eq_perturb_theory_for_unperturbed_proof_2} as: 
\begin{equation} \label{eq_perturb_theory_for_unperturbed_proof_3}
\begin{aligned}
(\tilde{\mathcal{L}}^{A,0}-\lambda_*)\tilde{v}^{(1)}
&=\Big(\mu^{(1)}-p\tilde{\mathcal{L}}^{A,1}-p^2\tilde{\mathcal{L}}^{A,2}\Big)\tilde{v}^{(0)} 
+\Big(\mu^{(1)}-p\tilde{\mathcal{L}}^{A,1}-p^2\tilde{\mathcal{L}}^{A,2}\Big)\tilde{v}^{(1)}.
\end{aligned}
\end{equation}
Define the projector
\begin{equation} \label{eq_perturb_theory_for_unperturbed_proof_4}
\begin{aligned}
Q_{\perp}:\quad & H^*\to \text{Ran}(\tilde{\mathcal{L}}^{A,0}-\lambda_*),\quad
f\mapsto f-( f,\tilde{v}_1(\bm{x}))\tilde{v}_1(\bm{x})-( f,\tilde{v}_2(\bm{x})) \tilde{v}_2(\bm{x}).
\end{aligned}
\end{equation}
Applying $Q_{\perp}$ to \eqref{eq_perturb_theory_for_unperturbed_proof_3} yields
\begin{equation*}
\begin{aligned}
(\tilde{\mathcal{L}}^{A,0}-\lambda_*)\tilde{v}^{(1)}
&=Q_{\perp}\Big(\mu^{(1)}-p\tilde{\mathcal{L}}^{A,1}-p^2\tilde{\mathcal{L}}^{A,2}\Big)\tilde{v}^{(0)} 
+Q_{\perp}\Big(\mu^{(1)}-p\tilde{\mathcal{L}}^{A,1}-p^2\tilde{\mathcal{L}}^{A,2}\Big)\tilde{v}^{(1)}.
\end{aligned}
\end{equation*}
By Lemma \ref{lem_fred_alternative}, $\big(\tilde{\mathcal{L}}^{A,0}-\lambda_*\big)^{-1}\in\mathcal{B}(\text{Ran}(\tilde{\mathcal{L}}^{A,0}-\lambda_*),H_2)$. Hence the above equation can be rewritten as
\begin{equation} \label{eq_perturb_theory_for_unperturbed_proof_5}
\begin{aligned}
(I-T)\tilde{v}^{(1)}
=T\tilde{v}^{(0)},
\end{aligned}
\end{equation}
where
\begin{equation} \label{eq_T_expansion}
\begin{aligned}
T=T(p,\mu^{(1)})=(\tilde{\mathcal{L}}^{A,0}-\lambda_*)^{-1}Q_{\perp}(\mu^{(1)}-p\tilde{\mathcal{L}}^{A,1}-p^2\tilde{\mathcal{L}}^{A,2}) 
=\mu^{(1)}T_{\mu}+pT_{p}+p^2T_{p^2},
\end{aligned}
\end{equation}
with
\[
T_{\mu}=(\tilde{\mathcal{L}}^{A,0}-\lambda_*)^{-1}Q_{\perp}, \quad
T_{p}= (\tilde{\mathcal{L}}^{A,0}-\lambda_*)^{-1}\tilde{\mathcal{L}}^{A,1}, \quad 
T_{p^2}= (\tilde{\mathcal{L}}^{A,0}-\lambda_*)^{-1}\tilde{\mathcal{L}}^{A,2}). 
\]
For sufficiently small $p$ and $\mu^{(1)}$, $(I-T)^{-1}\in \mathcal{B}(H_2)$ is expanded as a Neumann series. Hence, \eqref{eq_perturb_theory_for_unperturbed_proof_5} implies
\begin{equation} \label{eq_perturb_theory_for_unperturbed_proof_6}
\begin{aligned}
\tilde{v}^{(1)}=(I-T)^{-1}T\tilde{v}^{(0)}
=\sum_{k\geq 1}T^{k}(p,\mu^{(1)})\tilde{v}^{(0)}.
\end{aligned}
\end{equation}
Thus, by solving $\tilde{v}^{(0)}$ and its corresponding eigenvalue, we obtain the expansion of the Bloch modes using \eqref{eq_perturb_theory_for_unperturbed_proof_2} and \eqref{eq_perturb_theory_for_unperturbed_proof_6}. Before proceeding, observe that equations \eqref{eq_perturb_theory_for_unperturbed_proof_1} and \eqref{eq_perturb_theory_for_unperturbed_proof_7} imply
\begin{equation*}
[(\tilde{\mathcal{L}}^{A,0}-\lambda_*)^{-1},\mathcal{M}_2]=[Q_{\perp},\mathcal{M}_2]=0.
\end{equation*}
Hence all operators in \eqref{eq_T_expansion} are $\mathcal{M}_2-$invariant, i.e., 
\begin{equation} \label{eq_perturb_theory_for_unperturbed_proof_9}
[T_{\mu},\mathcal{M}_2]=0,\quad [T_{p},\mathcal{M}_2]=0,\quad
[T_{p^2},\mathcal{M}_2]=0,\quad [T,\mathcal{M}_2]=0.
\end{equation}

{\color{blue}Step 2.} In this step, we solve the following finite-dimensional problem for $\tilde{v}^{(0)}=a\cdot\tilde{v}_{1}+b\cdot\tilde{v}_{2}$, obtained by substituting \eqref{eq_perturb_theory_for_unperturbed_proof_6} to \eqref{eq_perturb_theory_for_unperturbed_proof_3}:
\footnotesize
\begin{equation*}
(\tilde{\mathcal{L}}^{A,0}-\lambda_*)\sum_{k\geq 1}T^{k}(p,\mu^{(1)})\tilde{v}^{(0)}
=\Big(\mu^{(1)}-p\tilde{\mathcal{L}}^{A,1}-p^2\tilde{\mathcal{L}}^{A,2}\Big)\tilde{v}^{(0)} 
+\Big(\mu^{(1)}-p\tilde{\mathcal{L}}^{A,1}-p^2\tilde{\mathcal{L}}^{A,2}\Big)\sum_{k\geq 1}T^{k}(p,\mu^{(1)})\tilde{v}^{(0)}.
\end{equation*}
\normalsize
Taking dual pairs with $\tilde{v}_n$ ($n=1,2$) gives
\footnotesize
\begin{equation} \label{eq_perturb_theory_for_unperturbed_proof_10}
\Bigg(\Big(\mu^{(1)}-p\tilde{\mathcal{L}}^{A,1}-p^2\tilde{\mathcal{L}}^{A,2}\Big)\sum_{k\geq 0}T^{k}(p,\mu^{(1)})(a\cdot\tilde{v}_{1}+b\cdot\tilde{v}_{2})
,\tilde{v}_n\Bigg)=0.
\end{equation}
\normalsize
By \eqref{eq_perturb_theory_for_unperturbed_proof_1} and \eqref{eq_perturb_theory_for_unperturbed_proof_9}, all operators in the above identity commutes with $\mathcal{M}_2$. Hence, by Lemma \ref{lem_symmetry_products}, the cross terms involving both $\tilde{v}_{1}$ and $\tilde{v}_{2}$ in \eqref{eq_perturb_theory_for_unperturbed_proof_10} vanish. Therefore, equation \eqref{eq_perturb_theory_for_unperturbed_proof_10} reduces to a diagonal system:
\footnotesize
\begin{equation}
\label{eq_perturb_theory_for_unperturbed_proof_11}
\begin{aligned}
&\begin{pmatrix}
\Big((\mu^{(1)}-p\tilde{\mathcal{L}}^{A,1}-p^2\tilde{\mathcal{L}}^{A,2})\sum_{k\geq 0}T^{k}(p,\mu^{(1)})\tilde{v}_{1},\tilde{v}_1\Big) & 0 \\
0 & \Big((\mu^{(1)}-p\tilde{\mathcal{L}}^{A,1}-p^2\tilde{\mathcal{L}}^{A,2})\sum_{k\geq 0}T^{k}(p,\mu^{(1)})\tilde{v}_{2},\tilde{v}_2\Big)
\end{pmatrix}
\begin{pmatrix}
a \\ b
\end{pmatrix} \\
&=0.
\end{aligned}
\end{equation}
\normalsize
Equation \eqref{eq_perturb_theory_for_unperturbed_proof_11} admits the following two independent solutions, corresponding to the two dispersion curves intersecting at the quadratic degenerate point $(\pi,\lambda_*)$ as in Proposition \ref{prop_mu12_asymptotic_even}: 
\footnotesize
\begin{equation}
\label{eq_perturb_theory_for_unperturbed_proof_12_eigenfunction}
\tilde{v}^{(0)}=\tilde{v}_{1},
\end{equation}
\begin{equation}
\label{eq_perturb_theory_for_unperturbed_proof_12_eigenvalue}
\Big((\mu_1(\pi+p)-\lambda_*-p\tilde{\mathcal{L}}^{A,1}-p^2\tilde{\mathcal{L}}^{A,2})\sum_{k\geq 0}T^{k}(p,\mu_1(\pi+p)-\lambda_*)\tilde{v}_{1},\tilde{v}_1\Big)=0,
\end{equation}
\normalsize
and
\footnotesize
\begin{equation}
\label{eq_perturb_theory_for_unperturbed_proof_13_eigenfunction}
\tilde{v}^{(0)}=\tilde{v}_{2},
\end{equation}
\begin{equation}
\label{eq_perturb_theory_for_unperturbed_proof_13_eigenvalue}
\Big((\mu_2(\pi+p)-\lambda_*-p\tilde{\mathcal{L}}^{A,1}-p^2\tilde{\mathcal{L}}^{A,2})\sum_{k\geq 0}T^{k}(p,\mu_1(\pi+p)-\lambda_*)\tilde{v}_{2},\tilde{v}_2\Big)=0.
\end{equation}
\normalsize
With \eqref{eq_perturb_theory_for_unperturbed_proof_2}, \eqref{eq_perturb_theory_for_unperturbed_proof_6} and \eqref{eq_perturb_theory_for_unperturbed_proof_12_eigenfunction}, we construct the first branch of normalized Bloch mode that is analytic in $\kappa_1$: 
\footnotesize
\begin{equation} \label{eq_perturb_theory_for_unperturbed_proof_14_v1}
\begin{aligned}
v_1(\bm{x};\kappa_1)
&=\frac{e^{i\kappa_1 x_1}\sum_{k\geq 0}T^{k}(p,\mu_1(\pi+p)-\lambda_*)\tilde{v}_{1}}{\|e^{i\kappa_1 x_1}\sum_{k\geq 0}T^{k}(p,\mu_1(\pi+p)
-\lambda_*)\tilde{v}_{1}\|_{L^2(Y)}} 
=\frac{e^{i\kappa_1 x_1}\sum_{k\geq 0}T^{k}(p,\mu_1(\pi+p)-\lambda_*)\tilde{v}_{1}}{\|\sum_{k\geq 0}T^{k}(p,\mu_1(\pi+p)-\lambda_*)\tilde{v}_{1}\|_{L^2(Y)}}.
\end{aligned}
\end{equation}
\normalsize
The second branch is obtained similarly: 
\footnotesize
\begin{equation} \label{eq_perturb_theory_for_unperturbed_proof_14_v2}
\begin{aligned}
v_2(\bm{x};\kappa_1)
&=\frac{e^{i\kappa_1 x_1}\sum_{k\geq 0}T^{k}(p,\mu_1(\pi+p)-\lambda_*)\tilde{v}_{2}}{\|\sum_{k\geq 0}T^{k}(p,\mu_1(\pi+p)-\lambda_*)\tilde{v}_{2}\|_{L^2(Y)}}.
\end{aligned}
\end{equation}
\normalsize

{\color{blue}Step 3.} We verify that the Bloch mode constructed in \eqref{eq_perturb_theory_for_unperturbed_proof_14_v1} and \eqref{eq_perturb_theory_for_unperturbed_proof_14_v2} satisfy \eqref{eq_orthogonal_gauge_criterion}. We only do this for $v_1(\bm{x};\kappa_1)$. Instead of merely calculating the derivative $\partial_{\kappa_1}v_1(\bm{x};\pi)$, we expand the numerator of \eqref{eq_perturb_theory_for_unperturbed_proof_14_v1} to 
the fifth order and its denominator to the fourth order. These expansions will be used in Section 5-7. 

Using \eqref{eq_T_expansion} and expanding $\mu_n(\kappa_1)$ using \eqref{eq_quadratic_dispersion}, we obtain the following expansion 
\footnotesize
\begin{equation} \label{eq_perturb_theory_for_unperturbed_proof_15}
\begin{aligned}
\sum_{k\geq 0}T^{k}(p,\mu_1(\kappa_1)-\lambda_*)
&=1+pT_p
+p^2\Big(T_{p^2}+\frac{-\gamma_*}{2}T_{\mu}+T_{p}^2 \Big) \\
&\quad +p^3\Big(\frac{-\gamma_*}{2}(T_{\mu}T_{p}+T_{p}T_{\mu})+T_{p}T_{p^2}+T_{p^2}T_{p}+T_{p}^3 \Big) \\
&\quad +p^4\Big(-\eta_*T_{\mu}+\frac{\gamma_*^2}{4}T_{\mu}^2+T_{p^2}^2+
\frac{-\gamma_*}{2}(T_{\mu}T_{p^2}+T_{p^2}T_{\mu}) \\
&\quad\quad\quad\quad  
+\frac{-\gamma_*}{2}(T_{\mu}T_{p}^2+T_{p}T_{\mu}T_{p}+T_{p}^2T_{\mu})
+T_{p^2}T_{p}^2+T_{p}T_{p^2}T_{p}+T_{p}^2 T_{p^2}+T_{p}^4 \Big) \\
&\quad +p^5\Big(
-\eta_*(T_{\mu}T_{p}+T_{p}T_{\mu})
+\frac{\gamma_*^2}{4}(T_{\mu}^2T_{p}+T_{\mu}T_{p}T_{\mu}+T_{p}T_{\mu}^2 )\\
&\quad\quad\quad\quad 
+\frac{-\gamma_*}{2}(T_{\mu}T_{p^2}T_p+T_{\mu}T_pT_{p^2}+T_{p^2}T_{\mu}T_p+T_{p^2}T_pT_{\mu}+T_pT_{\mu}T_{p^2}+T_pT_{p^2}T_{\mu} ) \\
&\quad\quad\quad\quad 
+(T_{p^2}^2 T_p+T_{p^2}T_p T_{p^2}+T_p T_{p^2}^2) +\frac{-\gamma_*}{2}(T_{\mu}T_p^3+T_p T_{\mu} T_p^2+T_p^2T_{\mu}T_p +T_p^3T_{\mu} ) \\
&\quad\quad\quad\quad  
+(T_{p^2}T_p^3+T_p T_{p^2} T_p^2+T_p^2T_{p^2}T_p +T_p^3T_{p^2} ) +T_p^5 \Big) +\mathcal{O}(p^6) \\
&=:\sum_{k=0}^{5}p^k S^{(k)}(\gamma_*,\eta_*) +\mathcal{O}(p^6).
\end{aligned}
\end{equation}
\normalsize
We next expand the normalization factor $N_1(\kappa_1):=\|\sum_{k\geq 0}T^{k}(\kappa_1-\pi,\mu_1(\kappa_1)-\lambda_*)\tilde{v}_{1}\|_{L^2(Y)}$. Note that $\tilde{v}_1\perp \text{Ran }((\tilde{\mathcal{L}}^{A,0}-\lambda_*)^{-1}Q_{\perp})$ implies 
\begin{equation} \label{eq_perturb_theory_for_unperturbed_proof_20}
(\tilde{v}_1,S^{(k)}(\gamma_*,\eta_*)\tilde{v}_1)=0
\quad \forall k\geq 1,
\end{equation}
by the definition of $S^{(k)}$ and \eqref{eq_T_expansion}. Hence 
\begin{equation} \label{eq_perturb_theory_for_unperturbed_proof_16}
\begin{aligned}
N_1(\kappa_1)
=\sqrt{1+N_1^{(2)}(\kappa_1-\pi)^2+N_1^{(3)}(\kappa_1-\pi)^3+\mathcal{O}((\kappa_1-\pi)^4)}
\end{aligned}
\end{equation}
with
\begin{equation} \label{eq_normalization_factor_n1k_def}
\begin{aligned}
N_1^{(2)}=\|S^{(1)}(\gamma_*,\eta_*)\tilde{v}_1\|^2, \quad
N_1^{(3)}=2\text{Re}\big(S^{(1)}(\gamma_*,\eta_*)\tilde{v}_1,S^{(2)}(\gamma_*,\eta_*)\tilde{v}_1\big).
\end{aligned}
\end{equation}
We then expand the numerator in \eqref{eq_perturb_theory_for_unperturbed_proof_14_v1}. Using \eqref{eq_perturb_theory_for_unperturbed_proof_15} and the Taylor expansion of $e^{i\kappa_1x_1}$, 
\footnotesize
\begin{equation} \label{eq_perturb_theory_for_unperturbed_proof_17}
\begin{aligned}
&e^{i\kappa_1 x_1}(\tilde{v}_{1}+\sum_{k\geq 1}T^{k}(\kappa_1-\pi,\mu_1(\kappa_1)-\lambda_*)\tilde{v}_{1}) \\
&=v_1+(\kappa_1-\pi)\big(ix_1v_1+e^{i\pi x_1}S^{(1)}(\gamma_*,\eta_*)\tilde{v}_1\big)
+\sum_{k=2}^{5}\Big(\sum_{j+\ell=k}\frac{\big(i(\kappa_1-\pi)x_1\big)^{j}}{j!}S^{(\ell)}(\gamma_*,\eta_*)\tilde{v}_1 \Big)
+\mathcal{O}((\kappa_1-\pi)^6).
\end{aligned}
\end{equation}
\normalsize
Therefore, combining \eqref{eq_perturb_theory_for_unperturbed_proof_14_v1}, \eqref{eq_perturb_theory_for_unperturbed_proof_16}, and \eqref{eq_perturb_theory_for_unperturbed_proof_17}, we obtain
\footnotesize
\begin{equation} \label{eq_perturb_theory_for_unperturbed_proof_18}
v_1(\bm{x};\kappa_1)
=\frac{v_1+(\kappa_1-\pi)\big(ix_1v_1+e^{i\pi x_1}S^{(1)}(\gamma_*,\eta_*)\tilde{v}_1\big)
+\sum_{k=2}^{5}\big(\sum_{j+\ell=k}\frac{\big(i(\kappa_1-\pi)x_1\big)^{j}}{j!}S^{(\ell)}(\gamma_*,\eta_*)\tilde{v}_1 \big)
+\mathcal{O}((\kappa_1-\pi)^6)}{\sqrt{1+N_1^{(2)}(\kappa_1-\pi)^2+N_1^{(3)}(\kappa_1-\pi)^3+\mathcal{O}((\kappa_1-\pi)^4)}}.
\end{equation}
\normalsize

It follows that
\begin{equation} \label{eq_perturb_theory_for_unperturbed_proof_21}
\partial_{\kappa_1}v_1(\bm{x};\pi)
=ix_1v_1+e^{i\pi x_1}S^{(1)}(\gamma_*,\eta_*)\tilde{v}_1. 
\end{equation}
Note that $ix_1v_1$ is even in $x_1$ while $v_1$ is odd. Therefore $(ix_1v_1,v_1)=0$. On the other hand, \eqref{eq_perturb_theory_for_unperturbed_proof_20} gives
\begin{equation*}
\big(e^{i\pi x_1}S^{(1)}(\gamma_*,\eta_*)\tilde{v}_1,v_1 \big)
=\big(S^{(1)}(\gamma_*,\eta_*)\tilde{v}_1,\tilde{v}_1 \big)
=0.
\end{equation*}
So we have $(\partial_{\kappa_1}v_1(\bm{x};\pi),v_1(\bm{x};\pi))=0$. This proves \eqref{eq_orthogonal_gauge_criterion}.

With \eqref{eq_perturb_theory_for_unperturbed_proof_21}, we can further simplify \eqref{eq_perturb_theory_for_unperturbed_proof_18} as 
\begin{equation} \label{eq_perturb_theory_for_unperturbed_final_1}
v_1(\bm{x};\kappa_1)
=\frac{v_1+(\kappa_1-\pi)\partial_{\kappa_1}v_1(\bm{x};\pi)
+\sum_{k=2}^{5}v^{(k)}_{1}(\bm{x};\pi)
+\mathcal{O}((\kappa_1-\pi)^6)}{\sqrt{1+N_1^{(2)}(\kappa_1-\pi)^2+N_1^{(3)}(\kappa_1-\pi)^3+\mathcal{O}((\kappa_1-\pi)^4)}},
\end{equation}
where
\begin{equation} \label{eq_perturb_theory_for_unperturbed_final_v1k_def}
v^{(k)}_{1}(\bm{x};\pi):=\sum_{j+\ell=k}\frac{\big(i(\kappa_1-\pi)x_1\big)^{j}}{j!}S^{(\ell)}(\gamma_*,\eta_*)\tilde{v}_1. 
\end{equation}

Similarly,
\begin{equation} \label{eq_perturb_theory_for_unperturbed_final_2}
v_2(\bm{x};\kappa_1)
=\frac{v_1+(\kappa_1-\pi)\partial_{\kappa_1}v_2(\bm{x};\pi)
+\sum_{k=2}^{5}v^{(k)}_2(\bm{x};\pi)
+\mathcal{O}((\kappa_1-\pi)^6)}{\sqrt{1+N_2^{(2)}(\kappa_1-\pi)^2+N_2^{(3)}(\kappa_1-\pi)^3+\mathcal{O}((\kappa_1-\pi)^4)}}
\end{equation}
where
\begin{equation} \label{eq_normalization_factor_n2k_def}
N_2^{(2)}=\|S^{(1)}(-\gamma_*,-\eta_*)\tilde{v}_2\|^2, \quad
N_2^{(3)}=2\text{Re}\big(S^{(1)}(-\gamma_*,-\eta_*)\tilde{v}_2,S^{(2)}(-\gamma_*,-\eta_*)\tilde{v}_2\big)
\end{equation}
and
\begin{equation} \label{eq_perturb_theory_for_unperturbed_final_v2k_def}
\partial_{\kappa_1}v_2(\bm{x};\pi)=ix_1v_2+e^{i\pi x_1}S^{(1)}(-\gamma_*,-\eta_*)\tilde{v}_2,\quad
v^{(k)}_2(\bm{x};\pi)=\sum_{j+\ell=k}\frac{\big(i(\kappa_1-\pi)x_1\big)^{j}}{j!}S^{(\ell)}(-\gamma_*,-\eta_*)\tilde{v}_2 .
\end{equation}
Here $S^{k}$ is defined similarly as in \eqref{eq_perturb_theory_for_unperturbed_proof_15} with the change $\gamma_*\to -\gamma_*$ and $\eta_*\to -\eta_*$ (because $\mu_1(\kappa_1)$ and $\mu_2(\kappa_1)$ are locally opposite by \eqref{eq_quadratic_dispersion}).
\end{proof}

As a by-product of the proof of Proposition \ref{prop_pert_theory_for_unpertrubed}, we derive a set of identities that will be utilized in the perturbation analysis in Section 7. These identities are obtained by expanding equations \eqref{eq_perturb_theory_for_unperturbed_proof_12_eigenvalue} and \eqref{eq_perturb_theory_for_unperturbed_proof_13_eigenvalue} into power series in $p$ using equation \eqref{eq_perturb_theory_for_unperturbed_proof_15}, and matching terms up to the fourth order. More precisely, we have: 
\begin{corollary} \label{corol_order_match}
\begin{equation} \label{eq_perturb_theory_for_unperturbed_proof_23}
(\tilde{\mathcal{L}}^{A,1}\tilde{v}_i,\tilde{v}_j)=0,\quad i,j\in\{1,2\},
\end{equation}
\begin{equation} \label{eq_perturb_theory_for_unperturbed_proof_24}
\begin{aligned}
-\frac{\gamma_*}{2}&=(\tilde{\mathcal{L}}^{A,1}T_{p}\tilde{v}_1,\tilde{v}_1)
+(\tilde{\mathcal{L}}^{A,2}\tilde{v}_1,\tilde{v}_1), \\
\frac{\gamma_*}{2}&=(\tilde{\mathcal{L}}^{A,1}T_{p}\tilde{v}_2,\tilde{v}_2)
+(\tilde{\mathcal{L}}^{A,2}\tilde{v}_2,\tilde{v}_2),
\end{aligned}
\end{equation}
\begin{equation} \label{eq_perturb_theory_for_unperturbed_proof_25}
\begin{aligned}
(\tilde{\mathcal{L}}^{A,2}T_{p}\tilde{v}_1,\tilde{v}_1)+(\tilde{\mathcal{L}}^{A,1}(T_{p^2}+T_{p}^2)\tilde{v}_1,\tilde{v}_1)&=0, \\
(\tilde{\mathcal{L}}^{A,2}T_{p}\tilde{v}_2,\tilde{v}_2)+(\tilde{\mathcal{L}}^{A,1}(T_{p^2}+T_{p}^2)\tilde{v}_2,\tilde{v}_2)&=0,
\end{aligned}
\end{equation}
\begin{equation} \label{eq_perturb_theory_for_unperturbed_proof_26}
\begin{aligned}
-\eta_*&=(\tilde{\mathcal{L}}^{A,2}S^{(2)}(\gamma_*)\tilde{v}_1,\tilde{v}_1)
+(\tilde{\mathcal{L}}^{A,1}S^{(3)}(\gamma_*)\tilde{v}_1,\tilde{v}_1), \\
\eta_*&=(\tilde{\mathcal{L}}^{A,2}S^{(2)}(-\gamma_*)\tilde{v}_2,\tilde{v}_2)
+(\tilde{\mathcal{L}}^{A,1}S^{(3)}(-\gamma_*)\tilde{v}_2,\tilde{v}_2).
\end{aligned}
\end{equation}
Here for $k=2,3$ we write $S^{(k)}(\gamma_*,\eta_*)=S^{(k)}(\gamma_*)$ since these terms are independent of $\eta_*$.
\end{corollary}

\section{Asymptotics of the unperturbed Green function at the quadratic degenerate point}

%In this section, we derive the asymptotic of the unperturbed Green function in the strip $\Omega$ at the energy level of the quadratic point, by applying the limiting absorption principle(see Theorem \ref{thm_asymp_unperturbed_green}). This reveals the interplay between the localized and extended parts of the wave in the unperturbed structure, providing crucial information for studying localized modes under perturbation. 

In this section, we investigate the Green function $G(\bm{x},\bm{y};\lambda_*+i\epsilon)$, which satisfies the following equations
\begin{equation*} 
    \left\{
    \begin{aligned}
        &(\nabla\cdot A\nabla+\lambda_*+i\epsilon)G(\bm{x},\bm{y};\lambda_*+i\epsilon)=\delta(\bm{x}-\bm{y}),\quad x,y \in \Omega, \\
&G(\bm{x},\bm{y};\lambda_*+i\epsilon)\big|_{\Gamma^{+}}=e^{i\pi}G(\bm{x},\bm{y};\lambda_*+i\epsilon)\big|_{\Gamma^{-}}, \\
        &\frac{\partial G}{\partial x_2}(\bm{x},\bm{y};\lambda_*+i\epsilon)\big|_{\Gamma^{+}}=e^{i\pi}\frac{\partial G}{\partial x_2}(\bm{x},\bm{y};\lambda_*+i\epsilon)\big|_{\Gamma^{-}}.
    \end{aligned}
    \right.
\end{equation*}

Applying the Floquet transform, for each $\epsilon>0$, the Green function can be expressed as
\begin{equation*}
G(\bm{x},\bm{y};\lambda_*+i\epsilon)=
\frac{1}{2\pi}\int_{0}^{2\pi}
\sum_{n\geq 1}\frac{v_{n}(\bm{x};\kappa_1)\overline{v_{n}(\bm{y};\kappa_1)}}{\lambda_*+i\epsilon-\mu_{n}(\kappa_1)}d\kappa_1.
\end{equation*}
We define the integral operator
\begin{equation*}
\mathcal{G}(\lambda_*+i\epsilon)f:=\int_{\Omega}G(\bm{x},\bm{y};\lambda_*+i\epsilon)f(\bm{y})d\bm{y}.
\end{equation*}
The main result of this section concerns the asymptotic behavior of $\mathcal{G}(\lambda_*+i\epsilon)$ as $\epsilon\to 0^+$. Such asymptotic behavior has been established for the situation where $\mu_1(\kappa_1)$ and $\mu_2(\kappa_1)$ intersect linearly at $(\pi,\lambda_*)$ (forming a Dirac point) in \cite{joly2016solutions,qiu2023mathematical,li2024interface}. However, in our case, where $\mu_1(\kappa_1)$ and $\mu_2(\kappa_1)$ touch quadratically, $\mathcal{G}(\lambda_*+i\epsilon)$ diverges as $\epsilon\to 0^+$. To characterize the asymptotic behavior of $\mathcal{G}(\lambda_*+i\epsilon)$, especially  the divergent term and the remaining finite part, we leverage the analyticity of the Bloch eigenpairs, following the approach presented in \cite{joly2016solutions}. 

By Proposition \ref{prop_analytic_label}, the Floquet-Bloch eigenpairs $(\mu_{n}(\kappa_1),v_n(\bm{x};\kappa_1))$ $(n=1,2)$ are analytic in $\mathcal{D}=\mathcal{D}_1\cap \mathcal{D}_2$, which encompasses the closed rectangle (see Figure \ref{fig_integral_contour})
$$
\mathcal{R}_{\nu}=\{z\in\mathbf{C}:0\leq \text{Re}z\leq 2\pi,-\nu\leq \text{Im}z\leq \nu\}
$$ 
for some $\nu>0$.  
By Assumption \ref{def_quadratic_degenracy}, $\kappa_1=\pi$ is the unique solution to $\mu_{n}(\kappa_1)=\lambda_*$ for $\kappa_1\in [0,2\pi]$ and $n=1,2$. Thus, by selecting $\nu$ sufficiently small, we can ensure that $\kappa_1=\pi$ remains the unique solution to $\mu_{n}(\kappa_1)=\lambda_*$ within the domain $\mathcal{D}$ for $n=1,2$. 
With these preparations, we have the following detailed asymptotic characterization of $\mathcal{G}(\lambda_*+i\epsilon)$:

\begin{figure}
\centering
\begin{tikzpicture}[scale=0.4]
%lambda real-> interface mode step 1:pv
\tikzset{->-/.style={decoration={markings,mark=at position #1 with 
{\arrow{latex}}},postaction={decorate}}};

%right-decaying contour
\node[left,scale=1] at (-11.2,2.3) {(a)};
\draw[->] (-11,0)--(11,0);
\node[right] at (11.2,0) {$Re(\kappa_1)$};
\draw[thick] (-10,-0.1)--(-10,0.1);
\node[below] at (-10,-0.29) {$0$};
\draw[thick] (0,-0.1)--(0,0.1);
\draw[thick] (10,-0.1)--(10,0.1);
\node[below] at (10,-0.29) {$2\pi$};
\node[below,scale=0.9] at (-1.8,-0.2) {$\pi-\epsilon^{\frac{1}{9}}$};
\node[below,scale=0.9] at (1.8,-0.2) {$\pi+\epsilon^{\frac{1}{9}}$};

\draw[->-=0.5,very thick,red,domain=0:180] plot ({1.5*cos(\x)},{1.5*sin(\x)});
\node[above,scale=1] at (0,1.6) {$C_{\epsilon}$};

\draw[->-=0.5,very thick,red] (-10,0)--(-10,5);
\draw[->-=0.5,very thick,blue] (-10,0)--(-1.5,0);
\draw[->-=0.5,very thick,blue] (1.5,0)--(10,0);
\draw[->-=0.5,very thick,red] (10,5)--(10,0);
\draw[->-=0.5,very thick,red] (-10,5)--(10,5);
\node[above] at (-10,5.2) {$i\nu$};
\node[above] at (10,5.2) {$2\pi+i\nu$};

%left-decaying contour
\node[left,scale=1] at (-11.2,-5.3) {(b)};
\draw[->] (-11,-3)--(11,-3);
\node[right] at (11.2,-3) {$Re(\kappa_1)$};
\draw[thick] (-10,-3.1)--(-10,-2.9);
\node[above] at (-10,-2.81) {$0$};
\draw[thick] (0,-3.1)--(0,-2.9);
\draw[thick] (10,-3.1)--(10,-2.9);
\node[above] at (10,-2.81) {$2\pi$};
\node[above,scale=0.9] at (-1.8,-2.8) {$\pi-\epsilon^{\frac{1}{9}}$};
\node[above,scale=0.9] at (1.8,-2.8) {$\pi+\epsilon^{\frac{1}{9}}$};
\draw[->-=0.5,very thick,red,domain=360:180] plot ({1.5*cos(\x)},{-3+1.5*sin(\x)});

\draw[->-=0.5,very thick,blue] (-10,-3)--(-1.5,-3);
\draw[->-=0.5,very thick,blue] (1.5,-3)--(10,-3);
\draw[->-=0.5,very thick,red] (10,-8)--(10,-3);
\draw[->-=0.5,very thick,red] (-10,-8)--(10,-8);
\draw[->-=0.5,very thick,red] (-10,-3)--(-10,-8);
\node[below] at (-10,-8.2) {$-i\nu$};
\node[below] at (10,-8.2) {$2\pi-i\nu$};
\end{tikzpicture}
\caption{Analytic domain and integral contour used in the proof.}
\label{fig_integral_contour}
\end{figure}
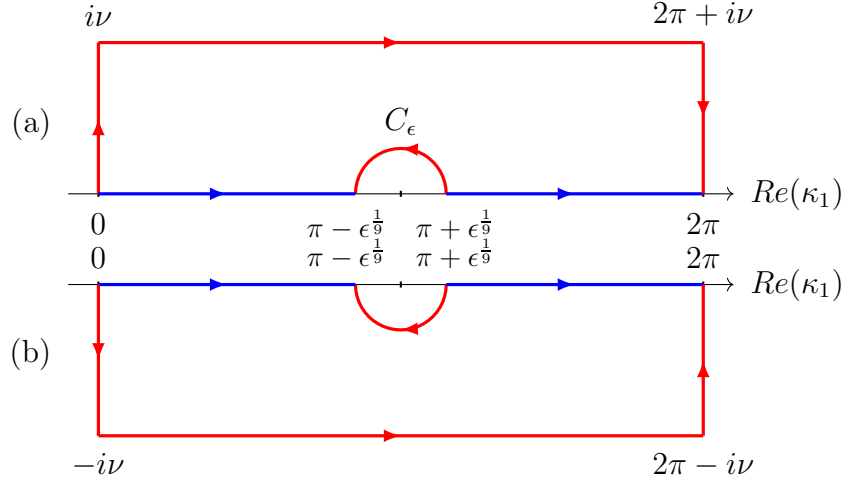

\begin{theorem} \label{thm_asymp_unperturbed_green}
For each $f\in (H^{1}(\Omega))^*$ with compact support, we have in $H_{loc}^1(\Omega)$
\begin{equation} \label{eq_asymp_unperturbed_green_1}
\begin{aligned}
\mathcal{G}(\lambda_*+i\epsilon)f
=&
\big(\epsilon^{-\frac{1}{2}}+o(\epsilon^{-\frac{1}{2}})\big)\cdot
\Big(\frac{1-i}{2\gamma_*^{\frac{1}{2}}}v_1(\bm{x};\pi)\big( f(\cdot),v_1(\cdot;\pi)\big)_{\Omega}-\frac{1+i}{2\gamma_*^{\frac{1}{2}}}v_2(\bm{x};\pi)\big( f(\cdot),v_2(\cdot;\pi)\big)_{\Omega}\Big) \\
&+\mathcal{G}_0(\lambda_*)f +\mathcal{O}(\epsilon^{\frac{2}{9}}),
\end{aligned}
\end{equation}
where $\mathcal{G}_0(\lambda_*)$ is the integral operator associated with the kernel function
\footnotesize
\begin{equation}
\label{eq_asymp_unperturbed_green_2}
\begin{aligned}
G_0(\bm{x},\bm{y};\lambda_*):=
&\frac{1}{2\pi}\int_{0}^{2\pi}\sum_{n\geq 3}\frac{v_n(\bm{x};\kappa_1)\overline{v_n(\bm{y};\kappa_1)}}{\lambda_*-\mu_{n}(\kappa_1)}d\kappa_1 \\
&+\frac{1}{2\pi}\int_{0}^{2\pi}
\frac{v_1(\bm{x};\kappa_1+i\nu)\overline{v_1(\bm{y};\overline{\kappa_1+i\nu})}}{\lambda_*-\mu_{1}(\kappa_1+i\nu)}d\kappa_1
+\frac{i}{\gamma_*}\Big(\partial_{\kappa_1}v_{1}(\bm{x};\pi)\overline{v_{1}(\bm{y};\pi)}
+v_{1}(\bm{x};\pi)\overline{(\partial_{\kappa_1}v_{1})(\bm{y};\pi)}\Big) \\
&+\frac{1}{2\pi}\int_{0}^{2\pi}
\frac{v_2(\bm{x};\kappa_1+i\nu)\overline{v_2(\bm{y};\overline{\kappa_1+i\nu})}}{\lambda_*-\mu_{2}(\kappa_1+i\nu)}d\kappa_1 -\frac{i}{\gamma_*}\Big(\partial_{\kappa_1}v_{2}(\bm{x};\pi)\overline{v_{2}(\bm{y};\pi)}
+v_{2}(\bm{x};\pi)\overline{(\partial_{\kappa_1}v_{2})(\bm{y};\pi)}\Big),
\end{aligned}
\end{equation}
\normalsize
Moreover,  $G_0(\bm{x},\bm{y};\lambda_*)$ has the following equivalent form: 
\footnotesize
\begin{equation}
\label{eq_asymp_unperturbed_green_limit_form}
\begin{aligned}
G_0(\bm{x},\bm{y};\lambda_*)=
&\frac{1}{2\pi}\int_{0}^{2\pi}\sum_{n\geq 3}\frac{v_n(\bm{x};\kappa_1)\overline{v_n(\bm{y};\kappa_1)}}{\lambda_*-\mu_{n}(\kappa_1)}d\kappa_1 \\
&+\lim_{\epsilon\to 0^+}
\Big(
\frac{1}{2\pi}\int_{[0,\pi-\epsilon^\frac{1}{9})\cup (\pi+\epsilon^\frac{1}{9},2\pi]}\frac{v_1(\bm{x};\kappa_1)\overline{v_1(\bm{y};\kappa_1)}}{\lambda_*-\mu_{1}(\kappa_1)}d\kappa_1
-\epsilon^{-\frac{1}{9}}\cdot\frac{2}{\pi\gamma_*}v_1(\bm{x};\pi)\overline{v_1(\bm{y};\pi)} \\
&\qquad +\frac{1}{2\pi}\int_{[0,\pi-\epsilon^\frac{1}{9})\cup (\pi+\epsilon^\frac{1}{9},2\pi]}\frac{v_2(\bm{x};\kappa_1)\overline{v_2(\bm{y};\kappa_1)}}{\lambda_*-\mu_{2}(\kappa_1)}d\kappa_1
+\epsilon^{-\frac{1}{9}}\cdot\frac{2}{\pi\gamma_*}v_2(\bm{x};\pi)\overline{v_2(\bm{y};\pi)} \Big).
\end{aligned}
\end{equation}
\normalsize
The complex integrals in \eqref{eq_asymp_unperturbed_green_2} are illustrated in Figure \ref{fig_integral_contour}(a).
\end{theorem}
 %Note that the integral in \eqref{eq_asymp_unperturbed_green_2} on the interval $[0,\pi-\epsilon^\frac{1}{9})\cup (\pi+\epsilon^\frac{1}{9},2\pi]$ diverges as $\epsilon\to 0^+$ due to the quadratic shape of $\mu_{n}(\kappa_1)$. This is manifested by the emergence of a blowing-up term of the order $\epsilon^{-\frac{1}{9}}$, as indicated in \eqref{eq_asymp_unperturbed_green_2}. We note that the choice of scaling $\epsilon^{-\frac{1}{9}}$ is not unique (indeed, one can replace it by any $\epsilon^{p}$ ($p>0$) in \eqref{eq_asymp_unperturbed_green_2}). This choice of scaling is designed to conduct the asymptotic analysis to study the interface modes (see the discussion under Lemma \ref{lem_app_C_1}).
We further demonstrate that the kernel function $G_0$ introduced in the above theorem possesses the following properties,  and is therefore referred to as the physical Green function: 
\begin{proposition}
\label{prop_G0_fundamental_solution}
The kernel function $G_0(\bm{x},\bm{y};\lambda_*)$ defined in \eqref{eq_asymp_unperturbed_green_2} is the fundamental solution to the following system of equations
\begin{equation} \label{eq_G0_fund}
    \left\{
    \begin{aligned}
        &(\nabla \cdot A\nabla +\lambda_*)G_0(\bm{x},\bm{y};\lambda_*)=\delta(\bm{x}-\bm{y}),\quad x,y \in \Omega,\\
        &G_0(\bm{x},\bm{y};\lambda_*)\big|_{\Gamma^{+}}=e^{i\pi}G_0(\bm{x},\bm{y};\lambda_*)\big|_{\Gamma^{-}}, \\
        &\frac{\partial G_0}{\partial x_2}(\bm{x},\bm{y};\lambda_*)\big|_{\Gamma^{+}}=e^{i\pi}\frac{\partial G_0}{\partial x_2}(\bm{x},\bm{y};\lambda_*)\big|_{\Gamma^{-}}.
    \end{aligned}
    \right.
\end{equation}
Additionally, for $\bm{x},\bm{y}\in\Omega$, 
\begin{equation} \label{eq_G0_parity}
G_0(\bm{x},\bm{y};\lambda_*)=G_0(\mathcal{M}_1\bm{x},\mathcal{M}_1\bm{y};\lambda_*),
\end{equation}
\begin{equation} \label{eq_G0_argument_symmetry}
G_0(\bm{x},\bm{y};\lambda_*)=\overline{G_0(\bm{y},\bm{x};\lambda_*)}. 
\end{equation}
Furthermore, for fixed $y\in \Omega$, $G_0(\bm{x},\bm{y};\lambda_*)$ admits the following decomposition:
\begin{equation} \label{eq_G0_decay_1}
\begin{aligned}
G_0(\bm{x},\bm{y};\lambda_*)=G_0^{+}(\bm{x},\bm{y};\lambda_*)
&+\frac{i}{\gamma_*}v_1(\bm{x};\pi)\overline{(\partial_{\kappa_1}v_1)(\bm{y};\pi)}
+\frac{i}{\gamma_*}(\partial_{\kappa_1}v_1)(\bm{x};\pi)\overline{v_1(\bm{y};\pi)} \\
&-\frac{i}{\gamma_*}v_2(\bm{x};\pi)\overline{(\partial_{\kappa_1}v_2)(\bm{y};\pi)}-\frac{i}{\gamma_*}(\partial_{\kappa_1}v_2)(\bm{x};\pi)\overline{v_2(\bm{y};\pi)},
\end{aligned}
\end{equation}
and
\begin{equation} \label{eq_G0_decay_2}
\begin{aligned}
G_0(\bm{x},\bm{y};\lambda_*)=G_0^{-}(\bm{x},\bm{y};\lambda_*)
&-\frac{i}{\gamma_*}v_1(\bm{x};\pi)\overline{(\partial_{\kappa_1}v_1)(\bm{y};\pi)}
-\frac{i}{\gamma_*}(\partial_{\kappa_1}v_1)(\bm{x};\pi)\overline{v_1(\bm{y};\pi)} \\
&+\frac{i}{\gamma_*}v_2(\bm{x};\pi)\overline{(\partial_{\kappa_1}v_2)(\bm{y};\pi)}+\frac{i}{\gamma_*}(\partial_{\kappa_1}v_2)(\bm{x};\pi)\overline{v_2(\bm{y};\pi)},
\end{aligned}
\end{equation}
where $G_0^{+}(\bm{x},\bm{y};\lambda_*)$ ($G_0^{-}(\bm{x},\bm{y};\lambda_*)$, resp.) decays exponentially as $x_1\to \infty$ ($x_1\to -\infty$, resp.).
\end{proposition}
The proof of Theorem \ref{thm_asymp_unperturbed_green} and Proposition \ref{prop_G0_fundamental_solution} are provided in Section 5.1 and 5.2, respectively.
We remark that Theorem \ref{thm_asymp_unperturbed_green} essentially provides an asymptotic expansion for the solution to the equation 
$$\big(-\nabla\cdot A\nabla -(\lambda_*+i\epsilon)\big)u_{\epsilon}=f 
$$ 
in $\Omega$. The results show that $u_{\epsilon}$ comprises two distinct parts: an extended part that involves the Bloch functions $v_n(\cdot;\pi)$ and their derivatives $(\partial_{\kappa_1}v_n)(\cdot;\pi)$, and an evanescent part represented by $G_0^{\pm}$ in \eqref{eq_G0_decay_1}-\eqref{eq_G0_decay_2}. This decomposition also characterizes the interplay between the extended modes at the energy level $\lambda=\lambda_*$ in the unperturbed structure. Such an understanding is crucial for determining the interface modes that bifurcate from the quadratic degenerate point. 
We also note that Theorem \ref{thm_asymp_unperturbed_green} extends previous studies on the asymptotic of the Green function in periodic media.  Earlier works \cite{MURATA203greenfunction, kuchment12greenfunction, kuchment17greenfunction} focused on the whole space, while \cite{joly2016solutions} examined a waveguide at energies away from degeneracies. Our theorem, however, addresses the situation when the energy level is at a quadratic degenerate point. 

\medskip

We now introduce the following layer potential operators associated with the Green function $G_0(\bm{x},\bm{y};\lambda_*)$:
\begin{equation*}
\begin{aligned}
&\mathcal{S}(\lambda_*;G_0):\tilde{H}^{-\frac{1}{2}}(\Gamma)
\to H_{loc}^{1}(\Omega),\quad
\varphi\mapsto \int_{\Gamma}G_0(\bm{x},\bm{y};\lambda_*)\varphi(\bm{y})ds(\bm{y}), \\
&\mathcal{D}(\lambda_*;G_0):H^{\frac{1}{2}}(\Gamma)
\to H_{loc}^{1}(\Omega),\quad
\phi\mapsto \int_{\Gamma}\frac{\partial G_0}{\partial y_1}(\bm{x},\bm{y};\lambda_*)\phi(\bm{y})ds(\bm{y}),\\
&\mathcal{N}(\lambda_*;G_0):H^{\frac{1}{2}}(\Gamma)
\to \tilde{H}^{-\frac{1}{2}}(\Gamma),\quad
\phi\mapsto \int_{\Gamma}\frac{\partial^2 G_0}{\partial x_1\partial y_1}(\bm{x},\bm{y};\lambda_*)\phi(\bm{y})ds(\bm{y}).
\end{aligned}
\end{equation*}

In the above notations, $\mathcal{S}(\lambda;G)$ denotes the single-layer potential operator associated with the Green function $G$ at energy $\lambda$, and $\mathcal{D}(\lambda;G)$ stands for the double-layer potential operator. Although the Green function $G_0$ differs from the one for the empty waveguide $\Omega$, its associated layer potential operators satisfy similar jump relations as shown below. These jump relations will be used to analyze the boundary integral equations related to the interface modes.

\begin{proposition}
\label{prop_G0_jump_formula}
For $\bm{x}\in\Gamma$, we have
\begin{equation} \label{eq_G0_jump_1}
\lim_{t\to 0}\Big(\mathcal{S}(\lambda_*;G_0)[\varphi]\Big)(\bm{x}+t\bm{e}_1)=\mathcal{S}(\lambda_*;G_0)[\varphi](\bm{x})
\end{equation}
\begin{equation} \label{eq_G0_jump_2}
\lim_{t\to 0}\frac{\partial}{\partial x_1}\Big(\mathcal{D}(\lambda_*;G_0)[\phi]\Big)(\bm{x}+t\bm{e}_1)=\mathcal{N}(\lambda_*;G_0)[\phi](\bm{x})
\end{equation}
\begin{equation} \label{eq_G0_jump_3}
\lim_{t\to 0^{\pm}}\frac{\partial}{\partial x_1}\Big(\mathcal{S}(\lambda_*;G_0)[\varphi]\Big)(\bm{x}+t\bm{e}_1)=\pm\frac{1}{2} \varphi(\bm{x}),
\end{equation}
\begin{equation} \label{eq_G0_jump_4}
\lim_{t\to 0^{\pm}}\Big(\mathcal{D}(\lambda_*;G_0)[\phi]\Big)(\bm{x}+t\bm{e}_1)=\mp\frac{1}{2} \phi(\bm{x}).
\end{equation}
\end{proposition}
Proof of Proposition \ref{prop_G0_jump_formula} is given in Section 5.3. 

\subsection{Proof of Theorem \ref{thm_asymp_unperturbed_green}}
%The idea of the proof is to separate $\mathcal{G}(\lambda_*+i\epsilon)f$ into the sum of an extended part (or the `near-energy' part) and an evanescent part (or the `far-energy' part) based on different integral domains, then to analyze the asymptotics of those two parts, respectively.

We decompose the operator $\mathcal{G}(\lambda_*+i\epsilon)$ as
\begin{equation*}
\mathcal{G}
=\mathcal{G}_{1}^{\epsilon,prop}+\mathcal{G}_{2}^{\epsilon,prop}
+\mathcal{G}_{1}^{\epsilon,evan}+\mathcal{G}_{2,1}^{\epsilon,evan}+\mathcal{G}_{2,2}^{\epsilon,evan},
\end{equation*}
where
\begin{equation*}
\begin{aligned}
\mathcal{G}_{n}^{\epsilon,prop}(\lambda_*+i\epsilon)f:=&
\frac{1}{2\pi}\int_{\pi-\epsilon^{\frac{1}{9}}}^{\pi+\epsilon^{\frac{1}{9}}}
\frac{v_{n}(\bm{x};\kappa_1)(f(\cdot),v_{n}(\cdot;\kappa_1))_{\Omega}}{\lambda_*+i\epsilon-\mu_{n}(\kappa_1)}d\kappa_1, \quad n=1, 2, \\
\mathcal{G}_{1}^{\epsilon,evan}(\lambda_*+i\epsilon)f:=
&\frac{1}{2\pi}\int_{0}^{2\pi}\sum_{n\geq 3}\frac{v_n(\bm{x};\kappa_1)(f(\cdot),v_{n}(\cdot;\kappa_1))_{\Omega}}{\lambda_*+i\epsilon-\mu_{n}(\kappa_1)}d\kappa_1, \\
\mathcal{G}_{n}^{\epsilon,evan}(\lambda_*+i\epsilon)f:=
&\frac{1}{2\pi}\int_{[0,\pi-\epsilon^\frac{1}{9})\cup (\pi+\epsilon^\frac{1}{9},2\pi]}\frac{v_n(\bm{x};\kappa_1)(f(\cdot),v_{n}(\cdot;\kappa_1))_{\Omega}}{\lambda_*+i\epsilon-\mu_{n}(\kappa_1)}d\kappa_1, \quad n=1, 2.
\end{aligned}
\end{equation*}
Here $v_{1}(\bm{x};\kappa_1)$ and $v_{2}(\bm{x};\kappa_1)$ are chosen in accordance with Proposition \ref{prop_pert_theory_for_unpertrubed}.  The scaling factor $\epsilon^{-\frac{1}{9}}$ in the integrals above 
are meticulously chosen to balance various reminder terms in the asymptotic analysis for interface modes in later sections (See Remark \ref{rmk_1/9_scaling}). Theorem \ref{thm_asymp_unperturbed_green} then follows from the three lemmas below. 

\begin{lemma}[Asymptotics of $\mathcal{G}_{n}^{\epsilon,prop}(\lambda_*+i\epsilon), n=1,2$] \label{lem_app_A_1}
\footnotesize
\begin{equation} \label{eq_app_A_1}
\begin{aligned}
\mathcal{G}_{1}^{\epsilon,prop}(\lambda_*+i\epsilon)f
&=\big(\epsilon^{-\frac{1}{2}}+o(-\epsilon^{\frac{1}{2}})\big)\cdot \frac{1-i}{2\gamma_*^{\frac{1}{2}}}v_1(\bm{x};\pi)\big( f(\cdot),v_1(\cdot;\pi)\big)_{\Omega} \\
&\quad +\frac{2}{\pi\gamma_*}\epsilon^{\frac{1}{9}}
\sum_{j+\ell =2}v_1^{(j)}(\bm{x};\pi)(f(\cdot),v_1^{(\ell)}(\bm{x};\pi))_{\Omega}+\mathcal{O}(\epsilon^{\frac{2}{9}}).
\end{aligned}
\end{equation}

\begin{equation} \label{eq_app_A_2}
\begin{aligned}
\mathcal{G}_{2}^{\epsilon,prop}(\lambda_*+i\epsilon)f
&=-\big(\epsilon^{-\frac{1}{2}}+o(\epsilon^{-\frac{1}{2}})\big)\cdot \frac{1+i}{2\gamma_*^{\frac{1}{2}}}v_2(\bm{x};\pi)\big( f(\cdot),v_2(\cdot;\pi)\big)_{\Omega} \\
&\quad -\frac{2}{\pi\gamma_*}\epsilon^{\frac{1}{9}}
\sum_{j+\ell =2}v_2^{(j)}(\bm{x};\pi)(f(\cdot),v_2^{(\ell)}(\bm{x};\pi))_{\Omega}  +\mathcal{O}(\epsilon^{\frac{2}{9}}).
\end{aligned}
\end{equation}
\normalsize
where $v_{i}^{(j)}$ is introduced in Proposition \ref{prop_pert_theory_for_unpertrubed} and we denote $v_{i}^{(0)}=v_i(\bm{x};\pi)$, $v_{i}^{(1)}=(\partial_{\kappa_1}v_i)(\bm{x};\pi)$ for $i=1,2$.
\end{lemma}

\begin{lemma}[Asymptotics of $\mathcal{G}_{1}^{\epsilon,evan}(\lambda_*+i\epsilon)$: $n\geq 3$ part] \label{lem_app_A_2_ngeq3}
\footnotesize
\begin{equation} \label{eq_app_A_3}
\begin{aligned}
\mathcal{G}_{1}^{\epsilon,evan}(\lambda_*+i\epsilon)f
&=\frac{1}{2\pi}\int_{0}^{2\pi}
\sum_{n\geq 3}\frac{v_{n}(\bm{x};\kappa_1)(f(\cdot),v_{n}(\cdot;\kappa_1))_{\Omega}}{\lambda_*-\mu_{n}(\kappa_1)}d\kappa_1
+\mathcal{O}(\epsilon).
\end{aligned}
\end{equation}
\normalsize
\end{lemma}

\begin{lemma}[Asymptotics of $\mathcal{G}_{2,n}^{\epsilon,evan}(\lambda_*+i\epsilon)$: $n=1,2$] \label{lem_app_A_2_n=12}
\footnotesize
\begin{equation} \label{eq_app_A_4}
\begin{aligned}
\mathcal{G}_{2, 1}^{\epsilon,evan}(\lambda_*+i\epsilon)f  
&=(\epsilon^{-\frac{1}{9}}\cdot \frac{2}{\pi\gamma_*}+\mathcal{O}(\epsilon^{\frac{1}{9}}))\cdot v_1(\bm{x};\pi)(f(\cdot),v_1(\cdot;\pi))_{\Omega} \\
&\quad +\mathcal{G}_0^{(1)}(\lambda_*)f  -\frac{2\epsilon^{\frac{1}{9}}}{\pi\gamma_*}\sum_{j+\ell =2}v_1^{(j)}(\bm{x};\pi)(f(\cdot),v_1^{(\ell)}(\bm{x};\pi))_{\Omega}  +\mathcal{O}(\epsilon^{\frac{2}{9}})
\end{aligned}
\end{equation}
\normalsize
where the integral operator $\mathcal{G}_0^{(1)} (\lambda_*)$ attains the following kernel function
\footnotesize
\begin{equation*}
G_0^{(1)}(\bm{x},\bm{y};\lambda_*)=
\frac{1}{2\pi}\int_{0}^{2\pi}
\frac{v_1(\bm{x};\kappa_1+i\nu)\overline{v_1(\bm{y};\overline{\kappa_1+i\nu})}}{\lambda_*-\mu_{1}(\kappa_1+i\nu)}d\kappa_1
+\frac{i}{\gamma_*}\Big(\partial_{\kappa_1}v_{1}(\bm{x};\pi)\overline{v_{1}(\bm{y};\pi)}
+v_{1}(\bm{x};\pi)\overline{(\partial_{\kappa_1}v_{1})(\bm{y};\pi)}\Big).
\end{equation*}
\normalsize
Moreover, $G_0^{(1)}(\bm{x},\bm{y};\lambda_*)$ has the following alternative expression
\footnotesize
\begin{equation*}
G_0^{(1)}(\bm{x},\bm{y};\lambda_*)=\lim_{\epsilon\to 0^+}
\Big(
\frac{1}{2\pi}\int_{[0,\pi-\epsilon^\frac{1}{9})\cup (\pi+\epsilon^\frac{1}{9},2\pi]}\frac{v_1(\bm{x};\kappa_1)\overline{v_1(\bm{y};\kappa_1)}}{\lambda_*-\mu_{1}(\kappa_1)}d\kappa_1
-\epsilon^{-\frac{1}{9}}\cdot\frac{2}{\pi\gamma_*}v_1(\bm{x};\pi)\overline{v_1(\bm{y};\pi)} \Big).
\end{equation*}
\normalsize
Similarly,
\footnotesize
\begin{equation} \label{eq_app_A_5}
\begin{aligned}
\mathcal{G}_{2,2}^{\epsilon,evan}(\lambda_*+i\epsilon)f 
&=(-\epsilon^{-\frac{1}{9}}\cdot \frac{2}{\pi\gamma_*}+\mathcal{O}(\epsilon^{\frac{1}{9}}))\cdot v_2(\bm{x};\pi)(f(\cdot),v_2(\cdot;\pi))_{\Omega} \\
&\quad +\mathcal{G}_0^{(2)}(\lambda_*)f   +\frac{2\epsilon^{\frac{1}{9}}}{\pi\gamma_*}\sum_{j+\ell =2}v_2^{(j)}(\bm{x};\pi)(f(\cdot),v_2^{(\ell)}(\bm{x};\pi))_{\Omega} +\mathcal{O}(\epsilon^{\frac{2}{9}}),
\end{aligned}
\end{equation}
\normalsize
where the integral operator $\mathcal{G}_0^{(2)} (\lambda_*)$ attains the following kernel function
\footnotesize
\begin{equation*}
\begin{aligned}
G_0^{(2)}(\bm{x},\bm{y};\lambda_*)&=\frac{1}{2\pi}\int_{0}^{2\pi}
\frac{v_2(\bm{x};\kappa_1+i\nu)\overline{v_2(\bm{y};\overline{\kappa_1+i\nu})}}{\lambda_*-\mu_{2}(\kappa_1+i\nu)}d\kappa_1 -\frac{i}{\gamma_*}\Big(\partial_{\kappa_1}v_{2}(\bm{x};\pi)\overline{v_{2}(\bm{y};\pi)}
+v_{2}(\bm{x};\pi)\overline{(\partial_{\kappa_1}v_{2})(\bm{y};\pi)}\Big).
\end{aligned}
\end{equation*}
\normalsize
Moreover, $G_0^{(2)}(\bm{x},\bm{y};\lambda_*)$ has the following alternative expression
\footnotesize
\begin{equation*}
\begin{aligned}
G_0^{(2)}(\bm{x},\bm{y};\lambda_*)&=\lim_{\epsilon\to 0^+}
\Big(
\frac{1}{2\pi}\int_{[0,\pi-\epsilon^\frac{1}{9})\cup (\pi+\epsilon^\frac{1}{9},2\pi]}\frac{v_2(\bm{x};\kappa_1)\overline{v_2(\bm{y};\kappa_1)}}{\lambda_*-\mu_{2}(\kappa_1)}d\kappa_1
+\epsilon^{-\frac{1}{9}}\cdot\frac{2}{\pi\gamma_*}v_2(\bm{x};\pi)\overline{v_2(\bm{y};\pi)} \Big).
\end{aligned}
\end{equation*}
\normalsize
\end{lemma}

\begin{proof} [Proof of Lemma \ref{lem_app_A_1}]
We prove \eqref{eq_app_A_1} by direct calculation and elementary estimate of integrals. \eqref{eq_app_A_2} is proved similarly. By the expansion of Bloch mode \eqref{eq_ansatz_unperturb_eigenfunction_1} and eigenvalue \eqref{eq_quadratic_dispersion}
\footnotesize
\begin{equation} \label{eq_app_A_6}
\begin{aligned}
&v_1(\bm{x};\kappa_1)=\Big(\sum_{j=0}^{5}(\kappa_1-\pi)^{j}v_1^{(j)}(\bm{x};\pi)+\mathcal{O}(|\kappa_1-\pi|^6)\Big)\Big(1-\frac{N_1^{(2)}}{2}(\kappa_1-\pi)^2-\frac{N_1^{(3)}}{2}(\kappa_1-\pi)^3+\mathcal{O}(|\kappa_1-\pi|^4)\Big), \\
&\quad\quad\quad\quad =(1+\mathcal{O}(|\kappa_1-\pi|^2))v_1^{(0)}(\bm{x};\pi)+\sum_{j=1}^{2}(\kappa_1-\pi)^{j}v_1^{(j)}(\bm{x};\pi)+\mathcal{O}(|\kappa_1-\pi|^3) \\
&\mu_1(\kappa_1)=\lambda_*-\frac{1}{2}\gamma_*(\kappa_1-\pi)^2-\eta_*(\kappa_1-\pi)^4+\mathcal{O}(|\kappa_1-\pi|^6).
\end{aligned}
\end{equation}
\normalsize
It follows that
\begin{equation} \label{eq_app_A_7}
\begin{aligned}
\frac{1}{\lambda_*+i\epsilon-\mu_{1}(\kappa_1)}
=\frac{1+\mathcal{O}(|\kappa_1-\pi|^2)}{i\epsilon+\frac{1}{2}\gamma_*(\kappa_1-\pi)^2}. 
\end{aligned}
\end{equation}
Therefore
\footnotesize
\begin{equation} \label{eq_app_A_8}
\begin{aligned}
\mathcal{G}_{1}^{\epsilon,prop}(\lambda_*+i\epsilon)f
&=\Big(\frac{1}{2\pi}\int_{\pi-\epsilon^{\frac{1}{9}}}^{\pi+\epsilon^{\frac{1}{9}}}
\frac{1+\mathcal{O}(|\kappa_1-\pi|^2)}{i\epsilon+\frac{1}{2}\gamma_*(\kappa_1-\pi)^2}d\kappa_1 \Big)\cdot v_{1}(\bm{x};\pi)(f(\cdot),v_{1}(\cdot;\pi))_{\Omega} \\
&\quad+ \Big(\frac{1}{2\pi}\int_{\pi-\epsilon^{\frac{1}{9}}}^{\pi+\epsilon^{\frac{1}{9}}}
\frac{\kappa_1-\pi+\mathcal{O}(|\kappa_1-\pi|^3)}{i\epsilon+\frac{1}{2}\gamma_*(\kappa_1-\pi)^2}d\kappa_1 \Big)\cdot \sum_{j+\ell =1}v_1^{(j)}(\bm{x};\pi)(f(\cdot),v_1^{(\ell)}(\bm{x};\pi))_{\Omega} \\
&\quad +\Big(\frac{1}{2\pi}\int_{\pi-\epsilon^{\frac{1}{9}}}^{\pi+\epsilon^{\frac{1}{9}}}
\frac{(\kappa_1-\pi)^2+\mathcal{O}(|\kappa_1-\pi|^4)}{i\epsilon+\frac{1}{2}\gamma_*(\kappa_1-\pi)^2}d\kappa_1 \Big)\cdot \sum_{j+\ell =2}v_1^{(j)}(\bm{x};\pi)(f(\cdot),v_1^{(\ell)}(\bm{x};\pi))_{\Omega} \\
&\quad +\mathcal{O}\Big(\int_{\pi-\epsilon^{\frac{1}{9}}}^{\pi+\epsilon^{\frac{1}{9}}}
\frac{|\kappa_1-\pi|^3}{i\epsilon+\frac{1}{2}\gamma_*(\kappa_1-\pi)^2}d\kappa_1 \Big)\\
&=\big(\epsilon^{-\frac{1}{2}}+o(-\epsilon^{\frac{1}{2}})\big)\cdot \frac{1-i}{2\gamma_*^{\frac{1}{2}}}v_1(\bm{x};\pi)\big( f(\cdot),v_1(\cdot;\pi)\big)_{\Omega} \\
&\quad +\frac{2}{\pi\gamma_*}\epsilon^{\frac{1}{9}}
\sum_{j+\ell =2}v_1^{(j)}(\bm{x};\pi)(f(\cdot),v_1^{(\ell)}(\bm{x};\pi))_{\Omega}+\mathcal{O}(\epsilon^{\frac{2}{9}}).
\end{aligned}
\end{equation}
\normalsize
This gives \eqref{eq_app_A_1}.
\end{proof}

\begin{proof}[Proof of Lemma \ref{lem_app_A_2_ngeq3}]
\medskip
To prove \eqref{eq_app_A_3}, we introduce the projectors
\begin{equation*}
\mathbb{P}_{n}(\kappa_1)g=\big(g(\cdot),v_n(\cdot;\overline{\kappa_1}) \big)v_n(\cdot;\kappa_1)\,\, (n=1,2), \quad
\mathbb{Q}(\kappa_1)=I-\mathbb{P}_{1}(\kappa_1)-\mathbb{P}_{2}(\kappa_1).
\end{equation*}
Note that $\mathbb{P}_{n}$ and $\mathbb{Q}$ are periodic by the quasi-periodicity of Bloch mode
\begin{equation} \label{eq_app_A_11_periodicity}
\mathbb{P}_{n}(\kappa_1+2\pi)=\mathbb{P}_{n}(\kappa_1),\quad
\mathbb{Q}(\kappa_1+2\pi)=\mathbb{Q}(\kappa_1).
\end{equation}
Then, by Floquet transform,
\footnotesize
\begin{equation} \label{eq_app_A_12}
\begin{aligned}
\frac{1}{2\pi}\int_{0}^{2\pi}\sum_{n\geq 3}
\frac{v_{n}(\bm{x};\kappa_1)(f(\cdot),v_{n}(\cdot;\kappa_1))_{\Omega}}{\lambda_*+i\epsilon-\mu_{n}(\kappa_1)}d\kappa_1
&=\frac{1}{2\pi}\int_{0}^{2\pi}
\big(\lambda_*+i\epsilon-\mathcal{L}_{\Omega,\pi}^A(\kappa_1)\mathbb{Q}(\kappa_1)\big)^{-1}\mathbb{Q}(\kappa_1)\hat{f}(\kappa_1)d\kappa_1,\quad \\
\frac{1}{2\pi}\int_{0}^{2\pi}\sum_{n\geq 3}
\frac{v_{n}(\bm{x};\kappa_1)(f(\cdot),v_{n}(\cdot;\kappa_1))_{\Omega}}{\lambda_*-\mu_{n}(\kappa_1)}d\kappa_1
&=\frac{1}{2\pi}\int_{0}^{2\pi}
\big(\lambda_*-\mathcal{L}_{\Omega,\pi}^A(\kappa_1)\mathbb{Q}(\kappa_1)\big)^{-1}\mathbb{Q}(\kappa_1)\hat{f}(\kappa_1)d\kappa_1,
\end{aligned}
\end{equation}
\normalsize
where $\hat{f}(\kappa_1)$ denotes the Floquet transform of $f$. Notice that
\begin{equation*}
\sigma(\mathcal{L}^A_{\Omega,\pi}(\kappa_1)\mathbb{Q}(\kappa_1))=\{0\}\cup\{\mu_n(\kappa_1)\}_{n\geq 3}.
\end{equation*}
Consequently, $\lambda_*\notin \sigma(\mathcal{L}^A_{\Omega,\pi}(\kappa_1)\mathbb{Q}(\kappa_1))$ for any $\kappa_1\in [0,2\pi)$ by Assumption \ref{def_quadratic_degenracy}(2). Thus, for $\epsilon$ being sufficiently small, the resolvent $(\lambda_*+i\epsilon-\mathcal{L}_{\Omega,\pi}^A(\kappa_1)\mathbb{Q}(\kappa_1))^{-1}:(H^{1}(Y))^*\to H^{1}(Y)$ is uniformly bounded. Hence \eqref{eq_app_A_12} and the resolvent identity shows that
\footnotesize
\begin{equation*}
\begin{aligned}
&\frac{1}{2\pi}\int_{0}^{2\pi}
\sum_{n\geq 3}\frac{v_{n}(\bm{x};\kappa_1)(f(\cdot),v_{n}(\cdot;\kappa_1))_{\Omega}}{\lambda_*+i\epsilon-\mu_{n}(\kappa_1)}d\kappa_1
-
\frac{1}{2\pi}\int_{0}^{2\pi}
\sum_{n\geq 3}\frac{v_{n}(\bm{x};\kappa_1)(f(\cdot),v_{n}(\cdot;\kappa_1))_{\Omega}}{\lambda_*-\mu_{n}(\kappa_1)}d\kappa_1 \\
%&=\frac{1}{2\pi}\int_{0}^{2\pi}
%\Big[(\lambda_*+i\epsilon-\mathcal{L}^A_{\Omega,\pi}(\kappa_1)\mathbb{Q}_{+}(\kappa_1))^{-1}-(\lambda_*-\mathcal{L}^A_{\Omega,\pi}(\kappa_1)\mathbb{Q}_{+}(\kappa_1))^{-1}]
%\mathbb{Q}_{+}(\kappa_1)\hat{f}(\kappa_1)d\kappa_1 \\
&=i\epsilon\cdot \frac{1}{2\pi}\int_{0}^{2\pi}
\Big[(\lambda_*+i\epsilon-\mathcal{L}^A_{\Omega,\pi}(\kappa_1)\mathbb{Q}(\kappa_1))^{-1}\cdot (\lambda_*-\mathcal{L}^A_{\Omega,\pi}(\kappa_1)\mathbb{Q}(\kappa_1))^{-1}\Big]
\mathbb{Q}(\kappa_1)\hat{f}(\kappa_1)d\kappa_1 \\
&=\mathcal{O}(\epsilon),
\end{aligned}
\end{equation*}
\normalsize
which concludes the proof of \eqref{eq_app_A_3}.
\end{proof}

\begin{proof}[Proof of Lemma \ref{lem_app_A_2_n=12}]
We prove \eqref{eq_app_A_4}. The proof of \eqref{eq_app_A_5} is similar. First, by Assumption \ref{def_quadratic_degenracy}, 
\begin{equation*}
|\lambda_*-\mu_1(\kappa_1)|\gtrsim \epsilon^{\frac{2}{9}}\quad \text{for}\,
|\kappa_1-\pi|>\epsilon^{\frac{1}{9}}.
\end{equation*}
We can derive that
\footnotesize
\begin{equation} \label{eq_app_A_13}
\begin{aligned}
\mathcal{G}_{1}^{\epsilon,evan}(\lambda_*+i\epsilon)f 
&=I +\mathcal{O}(\epsilon^{\frac{5}{9}}),
\end{aligned}
\end{equation}
\normalsize
where
\footnotesize
\begin{equation*}
I=\frac{1}{2\pi}\int_{[0,\pi-\epsilon^\frac{1}{9})\cup (\pi+\epsilon^\frac{1}{9},2\pi]}\frac{v_1(\bm{x};\kappa_1)(f(\cdot),v_{1}(\cdot;\kappa_1))_{\Omega}}{\lambda_*-\mu_{1}(\kappa_1)}d\kappa_1
=\frac{1}{2\pi}\int_{[0,\pi-\epsilon^\frac{1}{9})\cup (\pi+\epsilon^\frac{1}{9},2\pi]}
\frac{\mathbb{P}_{1}(\kappa_1)\hat{f}(\kappa_1)}{\lambda_*-\mu_{1}(\kappa_1)}d\kappa_1.
\end{equation*}
\normalsize
Note that $\mathbb{P}_{1}(\kappa_1)$ and $\mu_1(\kappa_1)$ are both analytic in $\kappa_1\in \mathcal{D}\supset\mathcal{R}_{\nu}=\{z\in\mathbf{C}:0\leq \text{Re}z\leq 2\pi,-\nu\leq \text{Im}z\leq \nu\}$. The Cauchy theorem gives
\footnotesize
\begin{equation} \label{eq_app_A_14}
\begin{aligned}
I&=\frac{1}{2\pi}\int_{0}^{\nu}
\frac{\mathbb{P}_{1}(it)\hat{f}(it)}{\lambda_*-\mu_{1}(it)}idt
-\frac{1}{2\pi}\int_{0}^{\nu}
\frac{\mathbb{P}_{1}(2\pi+it)\hat{f}(2\pi+it)}{\lambda_*-\mu_{1}(2\pi+it)}idt  +\frac{1}{2\pi}\int_{0}^{2\pi}
\frac{\mathbb{P}_{1}(\kappa_1+i\nu)\hat{f}(\kappa_1+i\nu)}{\lambda_*-\mu_{1}(\kappa_1+i\nu)}d\kappa_1 \\
&\quad +\frac{1}{2\pi}\int_{C_{\epsilon}}
\frac{\mathbb{P}_{1}(\kappa_1)\hat{f}(\kappa_1)}{\lambda_*-\mu_{1}(\kappa_1)}d\kappa_1,
\end{aligned}
\end{equation}
\normalsize
where $C_{\epsilon}:=\{\pi+\epsilon^{\frac{1}{9}}e^{i\theta}:0\leq \theta\leq \pi\}$. The contour is illustrated in Figure \ref{fig_integral_contour}(a). Note that $\nu$ is independent of $\epsilon$. Since the periodicity \eqref{eq_app_A_11_periodicity} holds within the rectangle $\kappa_1\in \mathcal{R}_{\nu}$ by the analytic continuation, the first and second integral in \eqref{eq_app_A_14} cancel. Hence
\footnotesize
\begin{equation} \label{eq_app_A_14_cancelled}
\begin{aligned}
I&=\frac{1}{2\pi}\int_{0}^{2\pi}
\frac{\mathbb{P}_{1}(\kappa_1+i\nu)\hat{f}(\kappa_1+i\nu)}{\lambda_*-\mu_{1}(\kappa_1+i\nu)}d\kappa_1 +\frac{1}{2\pi}\int_{C_{\epsilon}}
\frac{\mathbb{P}_{1}(\kappa_1)\hat{f}(\kappa_1)}{\lambda_*-\mu_{1}(\kappa_1)}d\kappa_1.
\end{aligned}
\end{equation}
\normalsize

Using the Taylor expansion \eqref{eq_app_A_6} and a similar calculation as in \eqref{eq_app_A_8}, 
\footnotesize
\begin{equation} \label{eq_app_A_15}
\begin{aligned}
\frac{1}{2\pi}\int_{C_{\epsilon}}
\frac{\mathbb{P}_{1}(\kappa_1)\hat{f}(\kappa_1)}{\lambda_*-\mu_{1}(\kappa_1)}d\kappa_1 
&=\big(v_1^{(0)}(\bm{x};\pi)(f(\cdot),v_1^{(0)}(\bm{x};\pi))_{\Omega} \big)\cdot\Big(\frac{2\epsilon^{-\frac{1}{9}}}{\pi\gamma_*}+\mathcal{O}(\epsilon^{\frac{1}{9}}) \Big) \\
&\quad +\big(\sum_{j+\ell =1}v_1^{(j)}(\bm{x};\pi)(f(\cdot),v_1^{(\ell)}(\bm{x};\pi))_{\Omega} \big)\cdot\Big(\frac{i}{\gamma_*}+\mathcal{O}(\epsilon^{\frac{1}{3}}) \Big) \\
&\quad +\big(\sum_{j+\ell =2}v_1^{(j)}(\bm{x};\pi)(f(\cdot),v_1^{(\ell)}(\bm{x};\pi))_{\Omega} \big)\cdot\Big(\frac{-2\epsilon^{\frac{1}{9}}}{\pi\gamma_*}+\frac{2}{3\pi\gamma_*}\big(\frac{2\eta_*}{\gamma_*}+N_1^{(2)}\big)\epsilon^{\frac{1}{3}} \Big) \\
&\quad +\big(\sum_{j+\ell =4}v_1^{(j)}(\bm{x};\pi)(f(\cdot),v_1^{(\ell)}(\bm{x};\pi))_{\Omega} \big)\cdot\big(\frac{-2}{3\pi\gamma_*}\epsilon^{\frac{1}{3}} \big)
 +\mathcal{O}(\epsilon^{\frac{5}{9}}).
\end{aligned}
\end{equation}
\normalsize
In summary, \eqref{eq_app_A_14_cancelled} and \eqref{eq_app_A_15} give
\footnotesize
\begin{equation} \label{eq_first_branch_evanescent_summary}
\begin{aligned}
I&=\frac{1}{2\pi}\int_{[0,\pi-\epsilon^\frac{1}{9})\cup (\pi+\epsilon^\frac{1}{9},2\pi]}\frac{v_1(\bm{x};\kappa_1)(f(\cdot),v_{1}(\cdot;\kappa_1))_{\Omega}}{\lambda_*-\mu_{1}(\kappa_1)}d\kappa_1 \\
&=\big(v_1^{(0)}(\bm{x};\pi)(f(\cdot),v_1^{(0)}(\bm{x};\pi))_{\Omega} \big)\cdot\Big(\frac{2\epsilon^{-\frac{1}{9}}}{\pi\gamma_*}+\mathcal{O}(\epsilon^{\frac{1}{9}}) \Big)   \\
&\quad +\frac{1}{2\pi}\int_{0}^{2\pi}
\frac{\mathbb{P}_{1}(\kappa_1+i\nu)\hat{f}(\kappa_1+i\nu)}{\lambda_*-\mu_{1}(\kappa_1+i\nu)}d\kappa_1 +\big(\sum_{j+\ell =1}v_1^{(j)}(\bm{x};\pi)(f(\cdot),v_1^{(\ell)}(\bm{x};\pi))_{\Omega} \big)\cdot\Big(\frac{i}{\gamma_*}+\mathcal{O}(\epsilon^{\frac{1}{3}}) \Big) \\
&\quad +\big(\sum_{j+\ell =2}v_1^{(j)}(\bm{x};\pi)(f(\cdot),v_1^{(\ell)}(\bm{x};\pi))_{\Omega} \big)\cdot\Big(\frac{-2\epsilon^{\frac{1}{9}}}{\pi\gamma_*}+\frac{2}{3\pi\gamma_*}\big(\frac{2\eta_*}{\gamma_*}+N_1^{(2)}\big)\epsilon^{\frac{1}{3}} \Big) \\
&\quad +\big(\sum_{j+\ell =4}v_1^{(j)}(\bm{x};\pi)(f(\cdot),v_1^{(\ell)}(\bm{x};\pi))_{\Omega} \big)\cdot\big(\frac{-2}{3\pi\gamma_*}\epsilon^{\frac{1}{3}} \big)
 +\mathcal{O}(\epsilon^{\frac{5}{9}}). 
\end{aligned}
\end{equation}
\normalsize
This, together with \eqref{eq_app_A_13}, conclude the proof of \eqref{eq_app_A_4}. 

%shows that $\frac{1}{2\pi}\int_{[0,\pi-\epsilon^\frac{1}{9})\cup (\pi+\epsilon^\frac{1}{9},2\pi]}\frac{v_1(\bm{x};\kappa_1)(f(\cdot),v_{1}(\cdot;\kappa_1))_{\Omega}}{\lambda_*+i\epsilon-\mu_{1}(\kappa_1)}d\kappa_1$ is the sum of 1) a leading term of order $\mathcal{O}(\epsilon^{-\frac{1}{9}})$, 2) zeroth order terms involving $\sum_{j+\ell =1}v_1^{(j)}(\bm{x};\pi)(f(\cdot),v_1^{(\ell)}(\bm{x};\pi))_{\Omega}$ and contour integrals on the upper portion of $\mathcal{R}_{\nu}$, 3) remainders of the order $\mathcal{O}(\epsilon^{\frac{1}{9}})$. Denote the zeroth order term as $\mathcal{G}_0^{(1)}(\lambda_*)f$, i.e.
%\footnotesize
%\begin{equation} \label{eq_G01_expression}
%\begin{aligned}
%\mathcal{G}_0^{(1)}(\lambda_*)f 
%&:= \frac{1}{2\pi}\int_{0}^{2\pi}
%\frac{\mathbb{P}_{1}(\kappa_1+i\nu)\hat{f}(\kappa_1+i\nu)}{\lambda_*-\mu_{1}(\kappa_1+i\nu)}d\kappa_1+\frac{i}{\gamma_*}\sum_{j+\ell =1}v_1^{(j)}(\bm{x};\pi)(f(\cdot),v_1^{(\ell)}(\bm{x};\pi))_{\Omega}
%\end{aligned}
%\end{equation}
%\%normalsize
%Then \eqref{eq_app_A_4} and \eqref{eq_G01_definition} follow from \eqref{eq_app_A_13} and \eqref{eq_first_branch_evanescent_summary}.
\end{proof}

\subsection{Proof of Proposition \ref{prop_G0_fundamental_solution}}
{\color{blue}Step 1:} We first prove $(\nabla\cdot A\nabla+\lambda_*)G_0(\bm{x},\bm{y};\lambda_*)=\delta(\bm{x}-\bm{y})$. It's equivalent to showing that
\begin{equation} \label{eq_app_A_18}
\mathfrak{a}^{A}_{\Omega}(\mathcal{G}_0(\lambda_*)f,g)-\lambda_*\cdot(\mathcal{G}_0(\lambda_*)f,g)_{\Omega}
=(f,g)_{\Omega},\quad \text{for any}\, f,g\in C_c^{\infty}(\Omega)
\end{equation}
where the form $\mathfrak{a}^{A}_{\Omega}(\cdot,\cdot)$ on $H^1(\Omega)$ is defined similarly to \eqref{eq_A_form}. From \eqref{eq_asymp_unperturbed_green_1}, we obtain
\footnotesize
\begin{equation} \label{eq_app_A_34}
\begin{aligned}
&\mathfrak{a}^{A}_{\Omega}(\mathcal{G}_0(\lambda_*)f,g)-\lambda_*\cdot(\mathcal{G}_0(\lambda_*)f,g)_{\Omega} \\
&= \frac{(1-i)\cdot \big(\epsilon^{-\frac{1}{2}}+o(\epsilon^{-\frac{1}{2}})\big)}{2\gamma_*^{\frac{1}{2}}}\big( f(\cdot),v_1(\cdot;\pi)\big)\cdot 
\Big(\mathfrak{a}^{A}(v_1(\cdot;\pi),g(\cdot))-\lambda_*\cdot(v_1(\cdot;\pi),g(\cdot))_{\Omega}\Big) \\
&\quad -\frac{(1+i)\cdot \big(\epsilon^{-\frac{1}{2}}+o(\epsilon^{-\frac{1}{2}})\big)}{2\gamma_*^{\frac{1}{2}}}\big( f(\cdot),v_2(\cdot;\pi)\big)\cdot 
\Big(\mathfrak{a}^{A}(v_2(\cdot;\pi),g(\cdot))-\lambda_*\cdot(v_2(\cdot;\pi),g(\cdot))_{\Omega}\Big) \\
&\quad +\mathfrak{a}^{A}_{\Omega}(\mathcal{G}(\lambda_*+i\epsilon)f,g)-\lambda_*\cdot(\mathcal{G}(\lambda_*+i\epsilon)f,g)_{\Omega}+\mathcal{O}(\epsilon^{\frac{2}{9}}).
\end{aligned}
\end{equation}
\normalsize
In addition, the weak formulation of $(\nabla\cdot A\nabla+\lambda_*)v_n(\bm{x};\pi)=0$ ($n=1,2$) and $(\nabla\cdot A\nabla+\lambda_*+i\epsilon)G(\bm{x},\bm{y};\lambda_*+i\epsilon)=\delta(\bm{x}-\bm{y})$ yield that for any $f,g\in C_c^{\infty}(\Omega)$,
\begin{equation} \label{eq_app_A_35}
\mathfrak{a}^{A}_{\Omega}(v_n(\cdot;\pi),g(\cdot))-\lambda_*\cdot(v_n(\cdot;\pi),g(\cdot))_{\Omega}=0,\quad n=1,2, 
\end{equation}
\begin{equation} \label{eq_unperturb_Green_function_definition_weak}
\mathfrak{a}^{A}_{\Omega}(\mathcal{G}(\lambda_*+i\epsilon)f,g)-(\lambda_*+i\epsilon)\cdot(\mathcal{G}(\lambda_*+i\epsilon)f,g)_{\Omega}
=(f,g)_{\Omega}.
\end{equation}
Substituting \eqref{eq_app_A_35} and \eqref{eq_unperturb_Green_function_definition_weak} into \eqref{eq_app_A_34}, we derive that
\begin{equation} \label{eq_app_A_20}
\begin{aligned}
\mathfrak{a}^{A}_{\Omega}(\mathcal{G}_0(\lambda_*)f,g)-\lambda_*\cdot(\mathcal{G}_0(\lambda_*)f,g)_{\Omega}
=(f,g)_{\Omega}+i\epsilon\cdot(\mathcal{G}(\lambda_*+i\epsilon)f,g)_{\Omega}+\mathcal{O}(\epsilon^{\frac{2}{9}}).
\end{aligned}
\end{equation}
Note that by \eqref{eq_asymp_unperturbed_green_1},
\begin{equation*}
\epsilon\cdot(\mathcal{G}(\lambda_*+i\epsilon)f,g)_{\Omega}
=\mathcal{O}(\epsilon^{\frac{1}{2}}).
\end{equation*}
Hence, by letting $\epsilon\to 0^+$ in \eqref{eq_app_A_20}, we obtain \eqref{eq_app_A_18}. The boundary condition of $G_0(\bm{x},\bm{y};\lambda_*)$ as stated in \eqref{eq_G0_fund} can be checked directly.

{\color{blue}Step 2:} We prove \eqref{eq_G0_parity} by exploiting the reflection symmetry $[\mathcal{L}^A,\mathcal{M}_1]=0$. Indeed, $[\mathcal{L}^A,\mathcal{M}_1]=0$ implies the Bloch eigenpairs of $\mathcal{L}^A$ are reflection symmetric in the sense that
\begin{equation*}
\mu_n(\kappa_1)=\mu_{n}(2\pi-\kappa_1),\quad
v_n(\bm{x};\kappa_1)\sim v_{n}(\mathcal{M}_1\bm{x};2\pi-\kappa_1), \quad n=1,2,
\end{equation*}
and for each $n\geq 3$, there exists $n^{\prime}\geq 3$ such that
\begin{equation*}
\mu_n(\kappa_1)=\mu_{n^{\prime}}(2\pi-\kappa_1),\quad
v_n(\bm{x};\kappa_1)\sim v_{n^{\prime}}(\mathcal{M}_1\bm{x};2\pi-\kappa_1).
\end{equation*}
Therefore, the quantity inside the limit of formula \eqref{eq_asymp_unperturbed_green_limit_form} is invariant under the substitution $\bm{x}\to \mathcal{M}_1\bm{x}$, $\bm{y}\to \mathcal{M}_1\bm{y}$. Taking $\epsilon\to 0$, this implies $G_0(\bm{x},\bm{y};\lambda_*)=G_0(\mathcal{M}_1\bm{x},\mathcal{M}_1\bm{y};\lambda_*)$. 

Equation \eqref{eq_G0_argument_symmetry} follows from a similar argument using the time-reversal symmetry of the operator $\mathcal{L}^A$.

{\color{blue}Step 3:} We prove \eqref{eq_G0_decay_1}. Note that $G_0(\bm{x},\bm{y};\lambda_*)$ equals the following by Theorem \ref{thm_asymp_unperturbed_green}
\begin{equation} \label{eq_app_A_16}
\begin{aligned}
G_0(\bm{x},\bm{y};\lambda_*)
&=G_0^{+}(\bm{x},\bm{y};\lambda_*)  +\frac{i}{\gamma_*}\Big(\partial_{\kappa_1}v_{1}(\bm{x};\pi)\overline{v_{1}(\bm{y};\pi)}
+v_{1}(\bm{x};\pi)\overline{(\partial_{\kappa_1}v_{1})(\bm{y};\pi)}\Big) \\
&\quad -\frac{i}{\gamma_*}\Big(\partial_{\kappa_1}v_{2}(\bm{x};\pi)\overline{v_{2}(\bm{y};\pi)}
+v_{2}(\bm{x};\pi)\overline{(\partial_{\kappa_1}v_{2})(\bm{y};\pi)}\Big),
\end{aligned}
\end{equation}
where
\footnotesize
\begin{equation*}
\begin{aligned}
G_0^{+}(\bm{x},\bm{y};\lambda_*)
&:=\frac{1}{2\pi}\int_{0}^{2\pi}\sum_{n\geq 3}\frac{v_n(\bm{x};\kappa_1)\overline{v_n(\bm{y};\kappa_1)}}{\lambda_*-\mu_{n}(\kappa_1)}d\kappa_1 +\frac{1}{2\pi}\int_{0}^{2\pi}\sum_{n=1,2}
\frac{v_n(\bm{x};\kappa_1+i\nu)\overline{v_n(\bm{y};\overline{\kappa_1+i\nu})}}{\lambda_*-\mu_{n}(\kappa_1+i\nu)}d\kappa_1.
\end{aligned}
\end{equation*}
\normalsize
Since $G_0^{+}(\bm{x},\bm{y};\lambda_*)$ decays exponentially as $x_1\to \infty$ (see Theorem 7 in \cite{joly2016solutions}),   \eqref{eq_G0_decay_1} follows from \eqref{eq_app_A_16}. \eqref{eq_G0_decay_2} is proved similarly by using a similar expression of $G_0$ to \eqref{eq_asymp_unperturbed_green_2}, obtained using the contour integral as in Figure \ref{fig_integral_contour}(b).

\subsection{Proof of Proposition \ref{prop_G0_jump_formula}}
Here we only prove \eqref{eq_G0_jump_1} and \eqref{eq_G0_jump_3}. The proof of \eqref{eq_G0_jump_2} and \eqref{eq_G0_jump_4} is similar. The strategy is to first demonstrate that the jump formula \eqref{eq_G0_jump_1} and \eqref{eq_G0_jump_3} hold for the Green function $G(\bm{x},\bm{y};\lambda_*+i\epsilon)$. We then derive the jump formula for $G_0(\bm{x},\bm{y};\lambda_*)$ by letting $\epsilon\to 0^+$. 

To show the jump formula \eqref{eq_G0_jump_1} and \eqref{eq_G0_jump_3} holds $G(\bm{x},\bm{y};\lambda_*+i\epsilon)$, we first fix a neighborhood of $\Gamma$, say, $U:=(-\frac{1-2r_0}{4},\frac{1-2r_0}{4})\times (-\frac{1}{2},\frac{1}{2})$, where $r_0:=\text{diam}(V)$. Let 
\begin{equation*}
G^{empty}(\bm{x},\bm{y};\lambda)
=\frac{1}{2\pi}\int_{0}^{2\pi}\sum_{\bm{n}\in\mathbf{Z}^2}\frac{e^{i(2\pi\bm{n}+(\kappa_1,\pi))\cdot (\bm{x}-\bm{y})}}{\lambda-|2\pi\bm{n}+(\kappa_1,\pi)|}d\kappa_1
\end{equation*}
be the Green function for the Laplacian operator $-\Delta$ in $\Omega$. Since $U\cap (\cup_{n_1,n_2\in\mathbf{Z}}V_{n_1,n_2})=\emptyset$, $\mathcal{L}^A=-\Delta$ in $U$,
%We see that $\mathcal{S}(\lambda_*+i\epsilon;G)[\varphi](\bm{x})$ differs from $\mathcal{S}(\lambda_*+i\epsilon;G^{empty})[\varphi](\bm{x})$ by a function $f\in H^1(U)$ for $\bm{x}\in U$. Indeed, we have
\begin{equation} \label{eq_app_A_23}
\mathcal{S}(\lambda_*+i\epsilon;G)[\varphi](\bm{x})
-\mathcal{S}(\lambda_*+i\epsilon;G^{empty})[\varphi](\bm{x})=f(\bm{x}),\quad \bm{x}\in U 
\end{equation}
for some function $f\in H^1(U)$. By \cite{colton2013integral}, the following jump relations hold 
\begin{equation} \label{eq_app_A_24}
\lim_{t\to 0}\Big(\mathcal{S}(\lambda_*+i\epsilon;G^{empty})[\varphi]\Big)(\bm{x}+t\bm{e}_1)=\mathcal{S}(\lambda_*+i\epsilon;G^{empty})[\varphi](\bm{x}),
\end{equation}
\begin{equation} \label{eq_app_A_25}
\lim_{t\to 0^{\pm}}\frac{\partial}{\partial x_1}\Big(\mathcal{S}(\lambda_*+i\epsilon;G^{empty})[\varphi]\Big)(\bm{x}+t\bm{e}_1)=\pm\frac{1}{2} \varphi(\bm{x})+\frac{1}{2}\mathcal{K}^{empty,*}[\varphi](\bm{x}),
\end{equation}
where 
\begin{equation*}
\mathcal{K}^{empty,*}(\lambda):\tilde{H}^{-\frac{1}{2}}(\Gamma)
\to \tilde{H}^{-\frac{1}{2}}(\Gamma),\,\,
\varphi\mapsto p.v.\int_{\Gamma}\frac{\partial G^{empty}}{\partial n_{\bm{x}}}(\bm{x},\bm{y};\lambda)\varphi(\bm{y})ds(\bm{y}).
\end{equation*}
Thus, by \eqref{eq_app_A_23} and \eqref{eq_app_A_24}, 
\begin{equation} \label{eq_app_A_26}
\begin{aligned}
\lim_{t\to 0^{\pm}}\mathcal{S}(\lambda_*+i\epsilon;G)[\varphi](\bm{x}+t\bm{e}_1)
&=\lim_{t\to 0^{\pm}}\mathcal{S}(\lambda_*+i\epsilon;G^{empty})[\varphi](\bm{x}+t\bm{e}_1)+f(\bm{x}) \\
&=\mathcal{S}(\lambda_*+i\epsilon;G^{empty})[\varphi](\bm{x})+f(\bm{x}) \\
&= \mathcal{S}(\lambda_*+i\epsilon;G)[\varphi](\bm{x}).
\end{aligned}
\end{equation}
Therefore \eqref{eq_G0_jump_1} holds for $G(\bm{x},\bm{y};\lambda_*+i\epsilon)$. On the other hand, the reflection symmetry implies the following identities (analogous to \eqref{eq_G0_parity})
\begin{equation*}
G(\bm{x},\bm{y};\lambda_*+i\epsilon)=G(\mathcal{M}_1\bm{x},\mathcal{M}_1\bm{y};\lambda_*+i\epsilon),\,\,
G^{empty}(\bm{x},\bm{y};\lambda_*+i\epsilon)=G^{empty}(\mathcal{M}_1\bm{x},\mathcal{M}_1\bm{y};\lambda_*+i\epsilon).
\end{equation*}
It follows that $f(\bm{x})=f(\mathcal{M}_1\bm{x})$ by \eqref{eq_app_A_23}, and that $\mathcal{K}^{empty,*}(\lambda_*+i\epsilon)=0$. Thus \eqref{eq_app_A_23} and \eqref{eq_app_A_25} yield
\begin{equation} \label{eq_app_A_28}
\begin{aligned}
\lim_{t\to 0^{\pm}}\frac{\partial}{\partial x_1}\Big(\mathcal{S}(\lambda_*+i\epsilon;G)[\varphi]\Big)(\bm{x}+t\bm{e}_1)
%&=\lim_{t\to 0^{\pm}}\frac{\partial}{\partial x_1}\Big(\mathcal{S}(\lambda_*+i\epsilon;G^{empty})[\varphi]\Big)(\bm{x}+t\bm{e}_1)
%+\frac{\partial f}{\partial x_1}(\bm{x}) \\
=\pm\frac{1}{2} \varphi(\bm{x})+\frac{\partial f}{\partial x_1}(\bm{x})
=\pm\frac{1}{2} \varphi(\bm{x}),
\end{aligned}
\end{equation}
where we've applied $\frac{\partial f}{\partial x_1}|_{\Gamma}=0$ (since $f(\bm{x})=f(\mathcal{M}_1\bm{x})$). Hence \eqref{eq_G0_jump_3} holds for $G(\bm{x},\bm{y};\lambda_*+i\epsilon)$.

We now prove \eqref{eq_G0_jump_1} and \eqref{eq_G0_jump_3} using \eqref{eq_app_A_26} and \eqref{eq_app_A_28}. Note that \eqref{eq_asymp_unperturbed_green_1} implies
\begin{equation} \label{eq_app_A_29}
\mathcal{S}(\lambda_*;G_0)[\varphi](\bm{x})-\mathcal{S}(\lambda_*+i\epsilon;G)[\varphi](\bm{x})
=f_{\epsilon}(\bm{x})\in H^1(U),
\end{equation}
with
\footnotesize
\begin{equation} \label{eq_app_A_30}
\begin{aligned}
f_{\epsilon}(\bm{x})
&=-\big(\epsilon^{-\frac{1}{2}}+o(\epsilon^{-\frac{1}{2}})\big)\cdot
\Big(\frac{(1-i)\cdot \int_{\Gamma}\varphi(\cdot)\overline{v_1(\cdot;\pi)}}{2\gamma_*^{\frac{1}{2}}}v_1(\bm{x};\pi)-\frac{(1+i)\cdot \int_{\Gamma}\varphi(\cdot)\overline{v_2(\cdot;\pi)}}{2\gamma_*^{\frac{1}{2}}}v_2(\bm{x};\pi)\Big)
+\mathcal{O}(\epsilon^{\frac{2}{9}}).
\end{aligned}
\end{equation}
\normalsize
Then, by \eqref{eq_app_A_26}, 
\begin{equation*}
\begin{aligned}
\lim_{t\to 0^{\pm}}\mathcal{S}(\lambda_*;G_0)[\varphi](\bm{x}+t\bm{e}_1)
&=\lim_{t\to 0^{\pm}}\Big(\mathcal{S}(\lambda_*+i\epsilon;G)[\varphi](\bm{x}+t\bm{e}_1)+f_{\epsilon}(\bm{x}+t\bm{e}_1)\Big) \\
&=\mathcal{S}(\lambda_*+i\epsilon;G)[\varphi](\bm{x})+f_{\epsilon}(\bm{x}) \\
&= \mathcal{S}(\lambda_*;G_0)[\varphi](\bm{x}),
\end{aligned}
\end{equation*}
which gives \eqref{eq_G0_jump_1}. For \eqref{eq_G0_jump_3}, note that the parity of $v_n(\bm{x};\pi)$ ($n=1,2$) in \eqref{eq_v12_parity_1} implies that
\begin{equation*}
v_1(\bm{x};\pi)=0,\quad \frac{\partial v_2}{\partial x_2}(\bm{x};\pi)=0,\quad \bm{x}\in\Gamma .
\end{equation*}
Thus \eqref{eq_app_A_30} indicates
\begin{equation*}
\frac{\partial f_{\epsilon}}{\partial x_1}(\bm{x})=\mathcal{O}(\epsilon^{\frac{2}{9}}),\quad \bm{x}\in\Gamma .
\end{equation*}
By \eqref{eq_app_A_28} and \eqref{eq_app_A_29},
\begin{equation*}
\lim_{t\to 0^{\pm}}\frac{\partial}{\partial x_1}\Big(\mathcal{S}(\lambda_*;G_0)[\varphi]\Big)(\bm{x}+t\bm{e}_1)
=\pm\frac{1}{2} \varphi(\bm{x})+\frac{\partial f_{\epsilon}}{\partial x_1}(\bm{x})
=\pm\frac{1}{2} \varphi(\bm{x})+\mathcal{O}(\epsilon^{\frac{2}{9}}).
\end{equation*}
Hence, by letting $\epsilon\to 0^+$, we conclude the proof of \eqref{eq_G0_jump_3}.

\section{Energy flux of Bloch modes at the quadratic degenerate point}
In this section, we characterize the interactions/couplings between Bloch modes and their momentum derivatives at the quadratic degenerate point along the interface $\Gamma$ using an energy flux functional. We show that $(\partial_{\kappa_1}v_n)(\bm{x};\pi)$ act as ``dual vectors'' to $v_n(\bm{x};\pi)$ for $n=1, 2$. This characterization is crucial for analyzing the boundary integral operators related to the interface modes. 

We first introduce the energy flux functional that takes the following sesquilinear form:
\begin{equation*}
\mathfrak{q}(u,v;\Gamma_s)=\int_{\Gamma_s}\Big(\frac{\partial u}{\partial x_1}\overline{v}-u\frac{\partial \overline{v}}{\partial x_1}\Big)dx_2,\quad
\Gamma_s:=\{s\}\times (-\frac{1}{2},\frac{1}{2})\subset \Omega.
\end{equation*}
When $s=0$, we write $\mathfrak{q}(u,v;\Gamma)=\mathfrak{q}(u,v;\Gamma_s)$. Theorem 3 in \cite{joly2016solutions} gives that
\begin{proposition} \label{prop_energy_flux_1}
For all $0\leq \kappa_1<2\pi$, and $n\geq 1$,
\begin{equation} \label{eq_energy_flux_1_1}
\mathfrak{q}(v_n(\cdot;\kappa_1),v_n(\cdot;\kappa_1);\Gamma)=i\mu^{\prime}(\kappa_1).
\end{equation}
If $\mu_n(\kappa_1)=\mu_{n^{\prime}}(\kappa_1^{\prime})$, then
\begin{equation} \label{eq_energy_flux_1_2}
\mathfrak{q}(v_n(\cdot;\kappa_1),v_{n^{\prime}}(\cdot;\kappa_1^{\prime});\Gamma)=0,\quad \text{if }\,
n\neq n^{\prime}\,\text{ or }\, \kappa_1\neq \kappa_1^{\prime}.
\end{equation}
If $\mu_{n}(\kappa_1)\neq \mu_{m}(\kappa_1)$, then
\begin{equation} \label{eq_energy_flux_1_3}
\mathfrak{q}(v_n(\cdot;\kappa_1),v_m(\cdot;\kappa_1);\Gamma)=0.
\end{equation}
\end{proposition}
Proposition \ref{prop_energy_flux_1} states that the energy flux of a Bloch mode equals its group velocity. Bloch modes of different eigenvalues or the same eigenvalue but different momentum or branches are decoupled in terms of energy flux. 
%The following complements Theorem 3 of \cite{joly2016solutions} with proof is similar.
%\begin{proposition} \label{prop_energy_flux_2}
%If $\mu_{n}(\kappa_1)\neq \mu_{m}(\kappa_1)$,
%\begin{equation} \label{eq_energy_flux_1_3}
%\mathfrak{q}(v_n(\cdot;\kappa_1),v_m(\cdot;\kappa_1);\Gamma)=0.
%\end{equation}
%\end{proposition}

Note that the following identities hold on the interface $\Gamma$ by \eqref{eq_v12_parity_1}-\eqref{eq_v12_parity_2}:
\begin{equation} \label{eq_sec3_10}
v_1(\bm{x};\pi)=\partial_{\kappa_1}v_2(\bm{x};\pi)=0,\quad \frac{\partial v_2}{\partial x_1}(\bm{x};\pi)=\frac{\partial }{\partial x_1}(\partial_{\kappa_1}v_1)(\bm{x};\pi)=0,\quad \bm{x}\in\Gamma.
\end{equation}

We have the following result on the coupling between the Bloch modes $v_n(\bm{x};\pi)$ and their momentum-derivative $(\partial_{\kappa_1}v_n)(\bm{x};\pi)$ in terms of energy flux across $\Gamma$.
Specifically,

\begin{proposition} \label{prop_12flux_alternative}
Assume \eqref{eq_v12_parity_2} holds. Then we have
\begin{equation}
\label{eq_12flux_1}
\int_{\Gamma}\frac{\partial v_1}{\partial x_1}(\cdot;\pi)\cdot \overline{(\partial_{\kappa_1}v_1)(\cdot;\pi)}=-\frac{i}{2}\gamma_* ,
\end{equation}
\begin{equation}
\label{eq_12flux_2}
\int_{\Gamma}\frac{\partial }{\partial x_1}(\partial_{\kappa_1}v_2)(\cdot;\pi)\cdot \overline{v_2(\cdot;\pi)}=\frac{i}{2}\gamma_*.
\end{equation}
Moreover,
\begin{equation}
\label{eq_12flux_3}
\int_{\Gamma}\frac{\partial v_1}{\partial x_1}(\cdot;\pi)\cdot \overline{v_2(\cdot;\pi)}=0,
\end{equation}
\begin{equation}
\label{eq_12flux_4}
\int_{\Gamma}\frac{\partial }{\partial x_1}(\partial_{\kappa_1}v_2)(\cdot;\pi)\cdot \overline{(\partial_{\kappa_1}v_1)(\cdot;\pi)}=0.
\end{equation}
\end{proposition}

Or equivalently, 
\begin{corollary} \label{prop_12flux}
Assume \eqref{eq_v12_parity_2} holds. Then
\begin{equation} \label{eq_energy_flux_with_derivative_full}
\mathfrak{q}(v_1(\cdot;\pi),(\partial_{\kappa_1}v_1)(\cdot;\pi);\Gamma)=-\frac{i}{2}\gamma_*,\quad
\mathfrak{q}(v_2(\cdot;\pi),(\partial_{\kappa_1}v_2)(\cdot;\pi);\Gamma)=\frac{i}{2}\gamma_*.
\end{equation}
Moreover,
\begin{equation*}
\mathfrak{q}(v_1(\cdot;\pi),v_2(\cdot;\pi);\Gamma)=0,
\end{equation*}
\begin{equation*}
\mathfrak{q}(v_1(\cdot;\pi),(\partial_{\kappa_1}v_2)(\cdot;\pi);\Gamma)=0,\quad
\mathfrak{q}(v_2(\cdot;\pi),(\partial_{\kappa_1}v_1)(\cdot;\pi);\Gamma)=0,
\end{equation*}
\begin{equation*}
\mathfrak{q}((\partial_{\kappa_1}v_n)(\cdot;\pi),(\partial_{\kappa_1}v_m)(\cdot;\pi);\Gamma)=0,\quad n,m\in\{1,2\}.
\end{equation*}
\end{corollary}

Before showing the proof, we emphasize the significance of these identities. As indicated by Proposition \ref{prop_energy_flux_1}, the energy flux carried by the Bloch modes at the quadratic degenerate point equals zero, i.e., $\mathfrak{q}(v_n(\cdot;\pi),v_n(\cdot;\pi);\Gamma)=0$, because $\mu_n^{\prime}(\pi)=0$ for $n=1,2$. Hence, \eqref{eq_energy_flux_with_derivative_full} states that the ``dual vector'' of $v_n(\cdot;\pi)$ in terms of energy flux is 
not itself (as it should be at a Dirac point \cite{li2024interface,joly2016solutions}) but its momentum-derivative $(\partial_{\kappa_1}v_n)(\cdot;\pi)$. This illustrates the critical interplay between Bloch modes and their momentum derivatives at the quadratic degenerate point.
%Physically, Proposition \ref{prop_12flux_alternative} and Corollary \ref{prop_12flux} indicate the `acceleration' of energy flux of standing modes at the quadratic degenerate point is nonzero (and equals to the curvature of the dispersion relations).

\begin{proof}[Proof of Proposition \ref{prop_12flux_alternative}]
We first prove \eqref{eq_12flux_1}. The proof of \eqref{eq_12flux_2} is similar. By \eqref{eq_energy_flux_1_1}, 
\begin{equation} \label{eq_sec3_9}
\int_{\Gamma}\frac{\partial v_1}{\partial x_1}(\cdot;\kappa_1)\overline{v_1(\cdot;\kappa_1)}-v_1(\cdot;\kappa_1)\overline{\frac{\partial v_1}{\partial x_1}(\cdot;\kappa_1)}=i\mu_{1}^{\prime}(\kappa_1).
\end{equation}
Using the time-reversal symmetry, there exists $c(\kappa_1)\in\mathbf{C}$ such that $|c(\kappa_1)|=1$ and
\begin{equation*}
v_1(\bm{x};\kappa_1)=c(\kappa_1)\cdot \overline{v_1(\bm{x};2\pi-\kappa_1)}. 
\end{equation*}
%Since $|c(\kappa_1)|=1$, we have
%\begin{equation*}
%v_1(\bm{x};\kappa_1)\overline{\frac{\partial v_1}{\partial x_1}(\bm{x};\kappa_1)}
%=\overline{v_1(\bm{x};2\pi-%\kappa_1)}\frac{\partial v_1}{\partial x_1}%(\bm{x};2\pi-\kappa_1).
%\end{equation*}
Then \eqref{eq_sec3_9} implies that
\begin{equation*}
\int_{\Gamma}\frac{\partial v_1}{\partial x_1}(\cdot;\kappa_1)\overline{v_1(\cdot;\kappa_1)}-\overline{v_1(\cdot;2\pi-\kappa_1)}\frac{\partial v_1}{\partial x_1}(\cdot;2\pi-\kappa_1)=i\mu_{1}^{\prime}(\kappa_1).
\end{equation*}
Taking derivative at $\kappa_1=\pi$ and applying \eqref{eq_sec3_10}, we obtain
\begin{equation*}
\int_{\Gamma}\frac{\partial v_1}{\partial x_1}(\cdot;\pi)\overline{(\partial_{\kappa_1}v_1)(\cdot;\pi)}+\overline{(\partial_{\kappa_1}v_1)(\cdot;\pi)}\frac{\partial v_1}{\partial x_1}(\cdot;\pi)=i\mu_{1}^{\prime\prime}(\pi)=-i\gamma_*
\end{equation*}
This gives \eqref{eq_12flux_1}.

We next prove \eqref{eq_12flux_3}. By Proposition \ref{prop_energy_flux_1}, we have
\begin{equation*}
\int_{\Gamma}\frac{\partial v_1}{\partial x_1}(\cdot;\kappa_1)\overline{v_2(\cdot;\kappa_1)}-v_1(\cdot;\kappa_1)\overline{\frac{\partial v_2}{\partial x_1}(\cdot;\kappa_1)}=0,\quad \kappa_1\neq \pi.
\end{equation*}
Since the functions inside the integral are smooth in $\kappa_1$, we can let $\kappa_1\to\pi$ to obtain
\begin{equation*}
\int_{\Gamma}\frac{\partial v_1}{\partial x_1}(\cdot;\pi)\overline{v_2(\cdot;\pi)}-v_1(\cdot;\pi)\overline{\frac{\partial v_2}{\partial x_1}(\cdot;\pi)}=0.
\end{equation*}
By \eqref{eq_sec3_10}, we have $\int_{\Gamma}v_1(\cdot;\pi)\overline{\frac{\partial v_2}{\partial x_1}(\cdot;\pi)}=0$. Thus $\int_{\Gamma}\frac{\partial v_1}{\partial x_1}(\cdot;\pi)\overline{v_2(\cdot;\pi)}=0$, which is \eqref{eq_12flux_3}.

Finally, we prove \eqref{eq_12flux_4}. 
%which is more complicated and requires the result in Section 3.3. 
By \eqref{eq_partial_kappa1_vn_interior}
$$
(\nabla\cdot A\nabla+\lambda_*)\partial_{\kappa_1} v_1(\cdot;\pi)=0.
$$ 
Applying the Green's formula to $\partial_{\kappa_1} v_1(\cdot;\pi)$ and the fundamental solution $G_0$ in \eqref{eq_G0_fund} inside the rectangle $(0,N)\times (-\frac{1}{2},\frac{1}{2})\subset \Omega$, we obtain
\footnotesize
\begin{equation} \label{eq_sec3_11}
\begin{aligned}
\overline{\partial_{\kappa_1} v_1(\bm{y};\pi)}
&=\int_{(0,N)\times (-\frac{1}{2},\frac{1}{2})}\Big(
(\nabla\cdot A\nabla+\lambda_*)G_0(\bm{x},\bm{y};\lambda_*)\cdot \overline{\partial_{\kappa_1} v_1(\bm{x};\pi)} 
-(\nabla\cdot A\nabla+\lambda_*)\overline{\partial_{\kappa_1} v_1(\bm{x};\pi)}\cdot G_0(\bm{x},\bm{y};\lambda_*) \Big) \\
&=\int_{\Gamma_N}\frac{\partial G_0}{\partial x_1}(\cdot,\bm{y};\lambda_*)\overline{\partial_{\kappa_1} v_1(\cdot;\pi)}-G_0(\cdot,\bm{y};\lambda_*)\overline{\frac{\partial}{\partial x_1}(\partial_{\kappa_1} v_1)(\bm{x};\pi)} \\
&\quad -\int_{\Gamma}\frac{\partial G_0}{\partial x_1}(\cdot,\bm{y};\lambda_*)\overline{\partial_{\kappa_1} v_1(\cdot;\pi)}-G_0(\cdot,\bm{y};\lambda_*)\overline{\frac{\partial}{\partial x_1}(\partial_{\kappa_1} v_1)(\bm{x};\pi)},
\end{aligned}
\end{equation}
\normalsize
where in the last line above we used $A(\bm{x})=I$ for $\bm{x}\in \Gamma\cup\Gamma_N$. By \eqref{eq_G0_decay_1}, 
\footnotesize
\begin{equation} \label{eq_sec3_12}
\begin{aligned}
&\lim_{N\to \infty} \int_{\Gamma_N}\frac{\partial G_0}{\partial x_1}(\cdot,\bm{y};\lambda_*)\overline{\partial_{\kappa_1} v_1(\cdot;\pi)}-G_0(\cdot,\bm{y};\lambda_*)\overline{\frac{\partial}{\partial x_1}(\partial_{\kappa_1} v_1)(\bm{x};\pi)} \\
&=\frac{i}{\gamma_*}\big(\lim_{N\to \infty}\mathfrak{q}(v_1(\cdot;\pi),\partial_{\kappa_1} v_1(\cdot;\pi);\Gamma_N) \big)\cdot \overline{\partial_{\kappa_1} v_1(\bm{y};\pi)}  +\frac{i}{\gamma_*}\big(\lim_{N\to \infty}\mathfrak{q}(\partial_{\kappa_1} v_1(\cdot;\pi),\partial_{\kappa_1} v_1(\cdot;\pi);\Gamma_N) \big)\cdot \overline{v_1(\bm{y};\pi)} \\
&\quad -\frac{i}{\gamma_*}\big(\lim_{N\to \infty}\mathfrak{q}(v_2(\cdot;\pi),\partial_{\kappa_1} v_1(\cdot;\pi);\Gamma_N) \big)\cdot \overline{\partial_{\kappa_1} v_2(\bm{y};\pi)} 
-\frac{i}{\gamma_*}\big(\lim_{N\to \infty}\mathfrak{q}(\partial_{\kappa_1} v_2(\cdot;\pi),\partial_{\kappa_1} v_1(\cdot;\pi);\Gamma_N) \big)\cdot \overline{v_2(\bm{y};\pi)}.
\end{aligned}
\end{equation}
\normalsize
Since all the energy fluxes on the right side are independent of $N$ (a consequence of the fact that
$(\nabla\cdot A\nabla+\lambda_*)\partial_{\kappa_1} v_n(\cdot;\pi)=(\nabla\cdot A\nabla+\lambda_*)v_n(\cdot;\pi)=0$  for $n=1,2$),
\footnotesize
\begin{equation} \label{eq_sec3_13}
\begin{aligned}
&\lim_{N\to \infty} \int_{\Gamma_N}\frac{\partial G_0}{\partial x_1}(\cdot,\bm{y};\lambda_*)\overline{\partial_{\kappa_1} v_1(\cdot;\pi)}-G_0(\cdot,\bm{y};\lambda_*)\overline{\frac{\partial}{\partial x_1}(\partial_{\kappa_1} v_1)(\bm{x};\pi)} \\
&=\frac{i}{\gamma_*}\mathfrak{q}(v_1(\cdot;\pi),\partial_{\kappa_1} v_1(\cdot;\pi);\Gamma) \cdot \overline{\partial_{\kappa_1} v_1(\bm{y};\pi)}  +\frac{i}{\gamma_*}\mathfrak{q}(\partial_{\kappa_1} v_1(\cdot;\pi),\partial_{\kappa_1} v_1(\cdot;\pi);\Gamma) \cdot \overline{v_1(\bm{y};\pi)} \\
&\quad -\frac{i}{\gamma_*}\mathfrak{q}(v_2(\cdot;\pi),\partial_{\kappa_1} v_1(\cdot;\pi);\Gamma) \cdot \overline{\partial_{\kappa_1} v_2(\bm{y};\pi)}
-\frac{i}{\gamma_*}\mathfrak{q}(\partial_{\kappa_1} v_2(\cdot;\pi),\partial_{\kappa_1} v_1(\cdot;\pi);\Gamma) \cdot \overline{v_2(\bm{y};\pi)}.
\end{aligned}
\end{equation}
\normalsize
On the other hand, \eqref{eq_G0_jump_1}-\eqref{eq_G0_jump_4} and  \eqref{eq_G0_argument_symmetry} imply that
\footnotesize
\begin{equation*}
\begin{aligned}
\lim_{\bm{y}\to \Gamma}\int_{\Gamma}\frac{\partial G_0}{\partial x_1}(\cdot,\bm{y};\lambda_*)\overline{\partial_{\kappa_1} v_1(\cdot;\pi)}-G_0(\cdot,\bm{y};\lambda_*)\overline{\frac{\partial}{\partial x_1}(\partial_{\kappa_1} v_1)(\cdot;\pi)}
&=-\frac{1}{2}\overline{\partial_{\kappa_1} v_1(\bm{y};\pi)}-\overline{\int_{\Gamma}G_0(\bm{y},\cdot;\lambda_*)\frac{\partial}{\partial x_1}(\partial_{\kappa_1} v_1)(\cdot;\pi)}\\
&=-\frac{1}{2}\overline{\partial_{\kappa_1} v_1(\bm{y};\pi)},
\end{aligned}
\end{equation*}
\normalsize
where we used \eqref{eq_sec3_10} in the last equality. By letting $N\to\infty$ and $\bm{y}\to \Gamma$ in \eqref{eq_sec3_11}, we obtain
\begin{equation*}
\begin{aligned}
\overline{\partial_{\kappa_1} v_1(\bm{y};\pi)}
&=\frac{i}{\gamma_*}\mathfrak{q}(v_1(\cdot;\pi),\partial_{\kappa_1} v_1(\cdot;\pi);\Gamma) \cdot \overline{\partial_{\kappa_1} v_1(\bm{y};\pi)}  
+\frac{i}{\gamma_*}\mathfrak{q}(\partial_{\kappa_1} v_1(\cdot;\pi),\partial_{\kappa_1} v_1(\cdot;\pi);\Gamma) \cdot \overline{v_1(\bm{y};\pi)} \\
&\quad -\frac{i}{\gamma_*}\mathfrak{q}(v_2(\cdot;\pi),\partial_{\kappa_1} v_1(\cdot;\pi);\Gamma) \cdot \overline{\partial_{\kappa_1} v_2(\bm{y};\pi)} 
-\frac{i}{\gamma_*}\mathfrak{q}(\partial_{\kappa_1} v_2(\cdot;\pi),\partial_{\kappa_1} v_1(\cdot;\pi);\Gamma) \cdot \overline{v_2(\bm{y};\pi)} \\
&\quad +\frac{1}{2}\overline{\partial_{\kappa_1} v_1(\bm{y};\pi)}.
\end{aligned}
\end{equation*}
Using \eqref{eq_sec3_10} and \eqref{eq_12flux_1}, 
%indicate that
%\begin{equation*}
%\mathfrak{q}(v_1(\cdot;\pi),\partial_{\kappa_1} v_1(\cdot;\pi);\Gamma)=-\frac{i}{2}\gamma_*,\quad
%\mathfrak{q}(\partial_{\kappa_1} v_1(\cdot;\pi),\partial_{\kappa_1} v_1(\cdot;\pi);\Gamma)
%=\mathfrak{q}(v_2(\cdot;\pi),\partial_{\kappa_1} v_1(\cdot;\pi);\Gamma)=0 .
%\end{equation*}
we further obtain
\begin{equation*}
\overline{\partial_{\kappa_1} v_1(\bm{y};\pi)}=\frac{1}{2}\overline{\partial_{\kappa_1} v_1(\bm{y};\pi)}+\frac{1}{2}\overline{\partial_{\kappa_1} v_1(\bm{y};\pi)}
-\frac{i}{\gamma_*}\mathfrak{q}(\partial_{\kappa_1} v_2(\cdot;\pi),\partial_{\kappa_1} v_1(\cdot;\pi);\Gamma) \cdot \overline{v_2(\bm{y};\pi)}.
\end{equation*}
It follows that
\begin{equation*}
\mathfrak{q}(\partial_{\kappa_1} v_2(\cdot;\pi),\partial_{\kappa_1} v_1(\cdot;\pi);\Gamma) \cdot \overline{v_2(\bm{y};\pi)}=0,\quad
\bm{y}\in\Gamma .
\end{equation*}
Since $v_2(\cdot;\pi)|_{\Gamma}\neq 0$ (otherwise we have an contradiction to \eqref{eq_12flux_2}), we conclude that 
$$
\mathfrak{q}(\partial_{\kappa_1} v_2(\cdot;\pi),\partial_{\kappa_1} v_1(\cdot;\pi);\Gamma)=0.
$$
Equation \eqref{eq_12flux_4} follows by using \eqref{eq_sec3_10} again. 
\end{proof}

\section{Band-gap opening at the quadratic degenerate point}
In this section, we derive asymptotic expansion of the Bloch eigen-pairs $(\mu_{n,\delta}(\kappa_1), v_{n,\delta}(\bm{x};\kappa_1))$ of the operator $\mathcal{L}^{A+ \delta\cdot B}_{\Omega,\pi}(\kappa_1)$ for $n=1,2$ and $0< |\delta|\ll 1$. %Consequently, it's clear that the perturbation lifts quadratic degeneracy at $\kappa_1=\pi$ is lifted by the perturbation, and a band gap is opened as described in Theorem \ref{thm_gap_open_corollary}. 
The main result is
\begin{theorem}
\label{thm_gap_open}
Suppose $|\delta|\ll 1$ and $|\kappa_1|\ll 1$. Then
\footnotesize
\begin{equation}
\label{eq_asymp_eigenvalue_positive}
\begin{aligned}
&\mu_{1,\delta}(\kappa_1)
=\lambda_*-\sqrt{\big(\frac{\gamma_*}{2}(\kappa_1-\pi)^2+\eta_* (\kappa_1-\pi)^4\big)^2
+\delta^2|t_*+r_*(\kappa_1-\pi)|^2}\cdot \left(1+\mathcal{O}(|\kappa_1-\pi|^4+|\delta|)\right), \\
&\mu_{2,\delta}(\kappa_1)
=\lambda_*+\sqrt{\big(\frac{\gamma_*}{2}(\kappa_1-\pi)^2+\eta_* (\kappa_1-\pi)^4\big)^2
+\delta^2|t_*+r_*(\kappa_1-\pi)|^2}\cdot \left(1+\mathcal{O}(|\kappa_1-\pi|^4+|\delta|)\right),
\end{aligned}
\end{equation}
\normalsize
where
\begin{equation}
\label{eq_perturb_constant}
t_*:=-(\mathcal{L}^B_{\Omega,\pi}(\pi)v_1(\cdot;\pi),v_2(\cdot;\pi))\in\mathbf{R},
\end{equation}
and the constant $r_*\in\mathbf{C}$ is introduced in \eqref{eq_sec4_9}. In addition, 
\footnotesize
\begin{equation}
\label{eq_asymp_eigenfunction_1}
\begin{aligned}
v_{1,\delta}(\bm{x};\kappa_1)
&=v_1(\bm{x};\pi)+(\kappa_1-\pi)\partial_{\kappa_1}v_1(\bm{x};\pi)
+\sum_{k=2}^{5}(\kappa_1-\pi)^{k}v^{(k)}_{1}(\bm{x};\pi) \\
&\quad+ f_*(\kappa_1;\delta)\big(1+ r(\kappa_1-\pi;\delta)\cdot (\kappa_1-\pi)\big) v_2(\bm{x};\pi)+
f_*(\kappa_1;\delta)(\kappa_1-\pi)\partial_{\kappa_1}v_2(\bm{x};\pi) \\
&\quad +R_1(\bm{x};\kappa_1,\delta), 
\end{aligned}
\end{equation}
\begin{equation}
\label{eq_asymp_eigenfunction_2}
\begin{aligned}
v_{2,\delta}(\bm{x};\kappa_1)
&=-f_*(\kappa_1;\delta)\big(1+ \overline{r(\kappa_1-\pi;\delta)}\cdot (\kappa_1-\pi)\big) v_1(\bm{x};\pi)-
f_*(\kappa_1;\delta)(\kappa_1-\pi)\partial_{\kappa_1}v_1(\bm{x};\pi) \\
&\quad +v_2(\bm{x};\pi)+(\kappa_1-\pi)\partial_{\kappa_1}v_2(\bm{x};\pi)
+\sum_{k=2}^{5}(\kappa_1-\pi)^{k}v^{(k)}_{2}(\bm{x};\pi) \\
&\quad +R_2(\bm{x};\kappa_1,\delta).
\end{aligned}
\end{equation}
\normalsize
Here $v_k^{(i)}$ are defined in  \eqref{eq_ansatz_unperturb_eigenfunction_1}-\eqref{eq_ansatz_unperturb_eigenfunction_2}, $r(p;\delta)$ is introduced in \eqref{eq_sec4_28}, and
\begin{equation*}
f_*(\kappa_1;\delta)=\frac{t_*\delta}{\frac{\gamma_*}{2}(\kappa_1-\pi)^2+\sqrt{\frac{\gamma_*^2}{4}(\kappa_1-\pi)^4+t_*^2\delta^2}}.
\end{equation*}
Moreover, 
%the remainders $R_i(\bm{x};\kappa_1,\delta)\in H^1(Y)$, $i=1,2$, have the estimate 
\begin{equation*}
\|R_i(\bm{x};\kappa_1,\delta)\|_{H^1(Y)}=\mathcal{O}(|\delta|+|\kappa_1-\pi|^6), \,\,i=1,2, 
\end{equation*}
and 
%the norm of $v_{n,\delta}(\bm{x};\kappa_1)$ satisfy
\footnotesize
\begin{equation} \label{eq_perturbed_normalization_factor_1}
\begin{aligned}
N^2_{1,\delta}(\kappa_1)
&=\|v_{1,\delta}(\bm{x};\kappa_1)\|^2 \\
&=1+f_*^2(\kappa_1;\delta)
+N_1^{(2)}(\kappa_1-\pi)^2+N_1^{(3)}(\kappa_1-\pi)^3
+(\kappa_1-\pi)f_*^2(\kappa_1;\delta)\cdot 2\text{Re}\big(r(\kappa_1-\pi)\big)
+\mathcal{O}(\delta+(\kappa_1-\pi)^4),
\end{aligned}
\end{equation}
\begin{equation} \label{eq_perturbed_normalization_factor_2}
\begin{aligned}
N^2_{2,\delta}(\kappa_1)
&=\|v_{2,\delta}(\bm{x};\kappa_1)\|^2 \\
&=1+f_*^2(\kappa_1;\delta)
+N_2^{(2)}(\kappa_1-\pi)^2+N_2^{(3)}(\kappa_1-\pi)^3
+(\kappa_1-\pi)f_*^2(\kappa_1;\delta)\cdot 2\text{Re}\big(r(\kappa_1-\pi)\big)
+\mathcal{O}(\delta+(\kappa_1-\pi)^4),
\end{aligned}
\end{equation}
\normalsize
where the constants $N_i^{(j)}$ are same as in Proposition \ref{prop_pert_theory_for_unpertrubed}.
\end{theorem}

\begin{remark}
By \eqref{eq_asymp_eigenvalue_positive}, a band gap of the size $|t_* \delta|$ is opened near $\lambda=\lambda_*$, thereby proving Theorem \ref{thm_gap_open_corollary}. Moreover, we emphasize that the constant $t_*$ characterizes the leading-order behavior of the perturbation.  
\end{remark}

\begin{remark}
\label{rmk_parity_changing}
Suppose $\delta>0$ and $t_*\neq 0$ and $|\kappa-\pi|\ll 1$. Compared to the expansion of Bloch modes in Proposition \ref{prop_pert_theory_for_unpertrubed}, the additional leading order term with the $f_*(\kappa_1;\delta)$-coefficient manifests the perturbation effect that mixes the two Bloch modes at the quadratic degenerate point. By \eqref{eq_asymp_eigenfunction_1} and \eqref{eq_asymp_eigenfunction_2}, 
\begin{equation*}
\begin{aligned}
&v_{1,\delta}(\bm{x};\pi)=v_{1}(\bm{x};\pi)+\text{sgn}(t_*)v_{1}(\bm{x};\pi)+\mathcal{O}(\delta),\quad
v_{2,\delta}(\bm{x};\pi)=-\text{sgn}(t_*)v_{1}(\bm{x};\pi)+v_{1}(\bm{x};\pi)+\mathcal{O}(\delta), \\
&v_{1,-\delta}(\bm{x};\pi)=v_{1}(\bm{x};\pi)-\text{sgn}(t_*)v_{1}(\bm{x};\pi)+\mathcal{O}(\delta),\quad
v_{2,-\delta}(\bm{x};\pi)=\text{sgn}(t_*)v_{1}(\bm{x};\pi)+v_{1}(\bm{x};\pi)+\mathcal{O}(\delta).
\end{aligned}
\end{equation*}
Therefore 
\begin{equation*}
v_{1,\delta}(\bm{x};\pi)\sim v_{2,-\delta}(\bm{x};\pi),\quad
v_{2,\delta}(\bm{x};\pi)\sim v_{1,-\delta}(\bm{x};\pi)
\end{equation*}
if the $\mathcal{O}(\delta)$ terms are ignored. 
This implies that $\mathcal{L}^{A+ \delta\cdot B}_{\Omega,\pi}(\pi)$ and $\mathcal{L}^{A- \delta\cdot B}_{\Omega,\pi}(\pi)$ approximately exchange their eigenspaces near the energy $\lambda_*$. One can view this band-inversion phenomenon as a topological phase transition. This observation plays an important role in our proof of Theorem \ref{thm_interface_mode}, especially, in the calculation of $\mathbb{E}(h)$ and $\mathbb{F}(h)$ operator in Proposition \ref{prop_matrix_convergence}.
\end{remark}

\begin{remark} \label{rmk_importance_expansion}
In \eqref{eq_asymp_eigenfunction_1} and \eqref{eq_asymp_eigenfunction_2}, the phase transition phenomenon is observed not only between the Bloch modes $v_{1}(\bm{x};\pi)$ and $v_{2}(\bm{x};\pi)$, but also between their momentum derivatives $(\partial_{\kappa_1}v_{1})(\bm{x};\pi)$ and $(\partial_{\kappa_1}v_{2})(\bm{x};\pi)$. 
However, the effect of the derivatives on the bifurcation of interface modes is relatively small (of order $\delta^{\frac{1}{2}}$) compared with the effect of the Bloch modes (of order $\delta^{-\frac{1}{2}}$). This is detailed in Theorem \ref{thm_asymptotics_perturbed_Green_function} of Section 8, and is evident from the differing scales of the second and third diagonals of the operator $\mathbb{M}^{\delta}$ in \eqref{eq_matrix_operator_def} of Section 9. 

To capture this subtle yet important effect, we must eliminate interference from larger perturbative terms (greater than $\delta^{\frac{1}{2}}$). This necessitates detailed expansions, as shown in Theorem \ref{thm_gap_open}, and ensuring that all auxiliary terms larger than $\delta^{\frac{1}{2}}$ cancel out in the analysis of interface modes, as proven in Proposition \ref{prop_matrix_convergence}.

It's worth noting that such high-order perturbation analysis is unnecessary when considering interface modes bifurcating from Dirac points \cite{qiu2023mathematical,li2024interface}. This highlights a key distinction between Dirac points and quadratic degenerate points.
\end{remark}

\subsection{Preliminaries for the proof of Theorem \ref{thm_gap_open}}
Before proving Theorem \ref{thm_gap_open}, 
we establish the perturbative framework for solving the concerned eigenvalue problem  
\begin{equation} \label{eq_sec4_1}
(\tilde{\mathcal{L}}^{A+\delta\cdot B}_{\Omega,\pi}(\kappa_1)-\mu)\tilde{v}=0.
\end{equation}
As in Section 4.2, we introduce the following operator, which shares the same eigenpairs as $\mathcal{L}^{A+\delta\cdot B}_{\Omega,\pi}(\kappa_1)$ and is particularly suited for perturbative analysis:
\footnotesize
\begin{equation*}
\begin{aligned}
\tilde{\mathcal{L}}^{A+\delta\cdot B}_{\Omega,\pi}(\kappa_1):
&\quad H\to H^*,\\
&\quad u\mapsto 
e^{-i\kappa_1 x_1}\circ \mathcal{L}^{A+\delta\cdot B}_{\Omega,\pi}(\kappa_1)\circ e^{i\kappa_1 x_1}u =-(\nabla+i\kappa_1 \bm{e}_1)\cdot(A+\delta\cdot B)(\nabla+i\kappa_1 \bm{e}_1)u.
\end{aligned}
\end{equation*}
\normalsize
Here $H$ and $H^*$ are same as in \eqref{eq_transformed_operator_unperturbed}. For each $u\in H$, the following holds for all $v\in H$
\footnotesize
\begin{equation*}
( \tilde{\mathcal{L}}^{A+\delta\cdot B}_{\Omega,\pi}(\kappa_1)u,v)=\mathfrak{a}^{A+\delta B}_{\kappa_1}(u,v)
:=\int_{Y}\big((A(\bm{x})+\delta B(\bm{x}))(\nabla+i\kappa_1 \bm{e}_1) u(\bm{x})\big)\cdot \overline{(\nabla+i\kappa_1 \bm{e}_1) v(\bm{x})}d\bm{x}.
\end{equation*}
\normalsize
Since $A$ and $B$ are Hermitian, it follows that
\begin{equation*}
( \tilde{\mathcal{L}}^{A+\delta\cdot B}_{\Omega,\pi}(\kappa_1)u,v)=( u,\tilde{\mathcal{L}}^{A+\delta\cdot B}_{\Omega,\pi}(\kappa_1)v).
\end{equation*}
Similar to \eqref{eq_expansion_operator_unperturbed}, we expand $\tilde{\mathcal{L}}^{A+\delta\cdot B}_{\Omega,\pi}(\kappa_1)$ near $\kappa_1=\pi$ and $\delta=0$ to obtain
\footnotesize
\begin{equation} \label{eq_expansion_operator_perturbed}
\begin{aligned}
\tilde{\mathcal{L}}^{A+\delta\cdot B}_{\Omega,\pi}(\kappa_1)
=\tilde{\mathcal{L}}^{A,0}
+(\kappa_1-\pi)\tilde{\mathcal{L}}^{A,1}
+\delta \tilde{\mathcal{L}}^{B,0}
+(\kappa_1-\pi)^2\tilde{\mathcal{L}}^{A,2}  +\delta(\kappa_1-\pi)\tilde{\mathcal{L}}^{B,1}
+\delta(\kappa_1-\pi)^2\tilde{\mathcal{L}}^{B,2}.
\end{aligned}
\end{equation}
\normalsize
The following properties hold by a straightforward calculation:
\begin{lemma}
 \begin{equation}
( \tilde{\mathcal{L}}^{A,k}u,v)=( u,\tilde{\mathcal{L}}^{A,k}v), \quad ( \tilde{\mathcal{L}}^{B,k}u,v)=( u,\tilde{\mathcal{L}}^{B,k}v),\quad \forall \,\, k=0,1,2. 
\end{equation}
Moreover, the operators $\tilde{\mathcal{L}}^{B,k}$ anti-commute with $\mathcal{M}_2$:
\begin{equation} \label{eq_perturbation_operators_anti_commute}
\{\tilde{\mathcal{L}}^{B,k},\mathcal{M}_2\}=0
\quad k=0,1,2,
\end{equation} 
contrasting to the operators $\tilde{\mathcal{L}}^{A,k}$ which commute with
$\mathcal{M}_2$. 
\end{lemma}

For ease of notation, we denote
$$
v_n(\bm{x}):=u_{n}(\bm{x};\bm{\kappa}^{(2)}),\quad
\tilde{v}_n(\bm{x}):=e^{-i\pi x_1}v_n(\bm{x}), \quad n=1, 2.
$$
The following identity characterizes the parity-breaking effect of the perturbation operator $\tilde{\mathcal{L}}^{B,0}$ by introducing off-diagonal terms into equation \eqref{eq_perturb_theory_for_unperturbed_proof_11}. This mechanism is pivotal for the phase transition that occurs at the quadratic degenerate point. 
\begin{lemma}
\label{lem_non-degeneracy}
\begin{equation*}
\Big((\tilde{\mathcal{L}}^{B,0}\tilde{v}_n,\tilde{v}_m) \Big)_{1\leq n,m\leq 2}
=\begin{pmatrix}
0 & -t_* \\ -t_* & 0
\end{pmatrix},
\end{equation*}
where $t_*$ is defined in \eqref{eq_perturb_constant}. Moreover, $t_*\in\mathbf{R}$. 
\end{lemma}
\begin{proof}
First note that the diagonals $(\tilde{\mathcal{L}}^{B,0}\tilde{v}_n,\tilde{v}_n)$ vanish by \eqref{eq_perturbation_operators_anti_commute} and Lemma \ref{lem_symmetry_products}. We next consider the off-diagonal terms. By definition, 
\[
\big(\tilde{\mathcal{L}}^{B,0}\tilde{v}_{i},\tilde{v}_{j}\big) = (\mathcal{L}^B_{\Omega,\pi}(\pi)v_i(\cdot;\pi),v_j(\cdot;\pi)), \quad 1\leq i, j \leq 2, \,\, i\neq j.
\]
To complete the proof, it suffices to show that 
\begin{equation} \label{eq_app_B_1}
\int_{Y}(B\nabla v_1(\bm{x}))\cdot \overline{\nabla v_2(\bm{x})}d\bm{x}
=\int_{Y}(B\nabla v_2(\bm{x}))\cdot \overline{\nabla v_1(\bm{x})}d\bm{x},
\end{equation}
\begin{equation} \label{eq_app_B_2}
\int_{Y}(B\nabla v_1(\bm{x}))\cdot \overline{\nabla v_2(\bm{x})}d\bm{x}\in\mathbf{R}.
\end{equation}
Note that \eqref{eq_app_B_2} follows from \eqref{eq_app_B_1} and the fact that $B$ is Hermitian. Thus we only need to prove \eqref{eq_app_B_1}. It is proved by exploiting the rotation symmetry.
Recall that $v_n(\bm{x}):=u_{n}(\bm{x};\bm{\kappa}^{(2)})$.  \eqref{eq_assump_mode_symmetry} in Assumption \ref{def_quadratic_degenracy} implies that
\begin{equation*} 
\mathcal{R}v_1=iv_2,\quad
\mathcal{R}v_2=iv_1.
\end{equation*}
Thus,
\begin{equation} \label{eq_app_B_17}
\int_{Y}(B\nabla v_1(\bm{x}))\cdot \overline{\nabla v_2(\bm{x})}d\bm{x}
=\int_{Y}(B\nabla (v_2(R\bm{x})))\cdot \overline{\nabla (v_1(R\bm{x}))}d\bm{x}.
\end{equation}
Since for any functions $f(\bm{x})$, it holds that
\begin{equation*} 
\nabla(f(R\bm{x}))=R^T(\nabla f)(R\bm{x}).
\end{equation*}
Therefore, \eqref{eq_app_B_17} yields
\begin{equation*}
\begin{aligned}
\int_{Y}(B\nabla v_1(\bm{x}))\cdot \overline{\nabla v_2(\bm{x})}d\bm{x}
&=\int_{Y}(BR^T(\nabla v_2)(R\bm{x}))\cdot \overline{R^T (\nabla v_1)(R\bm{x})}d\bm{x}.
\end{aligned}
\end{equation*}
By Assumption \ref{assum_perturb_coefficient_matrix}, $B$ commutes with $R^T$. Thus
\begin{equation*}
\begin{aligned}
\int_{Y}(B\nabla v_1(\bm{x}))\cdot \overline{\nabla v_2(\bm{x})}d\bm{x}
&=\int_{Y}B(\bm{x})(\nabla v_2)(R\bm{x})\cdot \overline{(\nabla v_1)(R\bm{x})}d\bm{x}.
\end{aligned}
\end{equation*}
Then \eqref{eq_app_B_1} follows by a change of variable $R\bm{x}\to\bm{x}$ and the identity $B(\bm{x})=B(R\bm{x})$.
\end{proof}

\subsection{Proof of Theorem \ref{thm_gap_open}}

{\color{blue}Step 0}. The solvability of \eqref{eq_sec4_1} for $|\kappa_1-\pi|\ll 1$, $|\mu-\lambda_*|\ll 1$ and $|\delta|\ll 1$ follows from the fact that \eqref{eq_sec4_1} is solvable when $\kappa_1=\pi,\mu=\lambda_*$ and $\delta=0$. Indeed, the expansion \eqref{eq_expansion_operator_perturbed} shows
\begin{equation*}
\|\tilde{\mathcal{L}}^{A+\delta\cdot B}_{\Omega,\pi}(\kappa_1)-\tilde{\mathcal{L}}^{A}_{\Omega,\pi}(\pi)\|_{\mathcal{B}(H,H^*)}\ll 1,\quad\text{for }\,
|\kappa_1-\pi|,|\delta|\ll 1.
\end{equation*}
Since $\lambda_*$ is an isolated eigenvalue of $\tilde{\mathcal{L}}^{A}_{\Omega,\pi}(\pi)$ with multiplicity equal to two, 
the stability theorem of eigenvalues (see Chapter VI of \cite{kato2013perturbation}) guarantees that $\tilde{\mathcal{L}}^{A+\delta\cdot B}_{\Omega,\pi}(\kappa_1)$ has two isolated eigenvalues near $\lambda_*$ for $|\kappa_1-\pi|\ll 1$ and $|\delta|\ll 1$.

{\color{blue}Step 1}. We construct solutions of \eqref{eq_sec4_1} in the following form for $|\kappa_1-\pi|\ll 1$ and $|\delta|\ll1$
\begin{equation} \label{eq_sec4_2}
\mu=\lambda_*+\mu^{(1)}, \quad
\kappa_1=\pi+p,\quad
\tilde{v}=\tilde{v}^{(0)}+\tilde{v}^{(1)},
\end{equation}
with
\begin{equation*}
|\mu^{(1)}|,|p|\ll 1 ,\quad
\tilde{v}^{(0)}=a\cdot\tilde{v}_{1}(\bm{x})+b\cdot\tilde{v}_{2}(\bm{x})\in H_1,\quad
\tilde{v}^{(1)}\in H_2.
\end{equation*}
Here $H_1:=\ker (\tilde{\mathcal{L}}^{A}_{\Omega,\pi}(\pi)-\lambda_*)$ and $H_2$ is the orthogonal complement of $H_1$ in $H$. Plugging \eqref{eq_sec4_2} and \eqref{eq_expansion_operator_perturbed} into \eqref{eq_sec4_1}, we get
\footnotesize
\begin{equation} \label{eq_sec4_4}
\begin{aligned}
(\tilde{\mathcal{L}}^{A}_{\Omega,\pi}(\pi)-\lambda_*)\tilde{v}^{(1)}
&=\Big(\mu^{(1)}-p\tilde{\mathcal{L}}^{A,1}-p^2\tilde{\mathcal{L}}^{A,2}-\delta\tilde{\mathcal{L}}^{B,0}-p\delta \tilde{\mathcal{L}}^{B,1}-p^2\delta\tilde{\mathcal{L}}^{B,2}\Big)\tilde{v}^{(0)} \\
&\quad
+\Big(\mu^{(1)}-p\tilde{\mathcal{L}}^{A,1}-p^2\tilde{\mathcal{L}}^{A,2}-\delta\tilde{\mathcal{L}}^{B,0}-p\delta \tilde{\mathcal{L}}^{B,1}-p^2\delta\tilde{\mathcal{L}}^{B,2}\Big)\tilde{v}^{(1)}.
\end{aligned}
\end{equation}
\normalsize
{\color{blue}Step 2}. We solve \eqref{eq_sec4_4} by following a Lyapunov-Schmidt reduction argument as in Section 4.2. Using the projector $Q_{\perp}$ defined in \eqref{eq_perturb_theory_for_unperturbed_proof_4}, we obtain
\footnotesize
\begin{equation*}
\begin{aligned}
(\tilde{\mathcal{L}}^{A}_{\Omega,\pi}(\pi)-\lambda_*)\tilde{v}^{(1)}
&=Q_{\perp}\Big(\mu^{(1)}-p\tilde{\mathcal{L}}^{A,1}-p^2\tilde{\mathcal{L}}^{A,2}-\delta\tilde{\mathcal{L}}^{B,0}-p\delta \tilde{\mathcal{L}}^{B,1}-p^2\delta\tilde{\mathcal{L}}^{B,2}\Big)\tilde{v}^{(0)} \\
&\quad
+Q_{\perp}\Big(\mu^{(1)}-p\tilde{\mathcal{L}}^{A,1}-p^2\tilde{\mathcal{L}}^{A,2}-\delta\tilde{\mathcal{L}}^{B,0}-p\delta \tilde{\mathcal{L}}^{B,1}-p^2\delta\tilde{\mathcal{L}}^{B,2}\Big)\tilde{v}^{(1)}.
\end{aligned}
\end{equation*}
\normalsize
As in \eqref{eq_perturb_theory_for_unperturbed_proof_5}, the above equation can be rewritten as
\begin{equation*}
\begin{aligned}
(I-T)\tilde{v}^{(1)}
=T\tilde{v}^{(0)},
\end{aligned}
\end{equation*}
where the operator $T$ differs from \eqref{eq_T_expansion} by the terms involving the $\delta$-perturbation
\footnotesize
\begin{equation} \label{eq_T_expansion_perturbed}
\begin{aligned}
T&=T(\delta,p,\mu^{(1)})=(\tilde{\mathcal{L}}^{A,0}-\lambda_*)^{-1}Q_{\perp}(\mu^{(1)}-p\tilde{\mathcal{L}}^{A,1}-p^2\tilde{\mathcal{L}}^{A,2}-\delta \tilde{\mathcal{L}}^{B,0}-p\delta \tilde{\mathcal{L}}^{B,1}-p^2\delta \tilde{\mathcal{L}}^{B,2}) \\
&=:\mu^{(1)}T_{\mu}+pT_{p}+p^2T_{p^2}
+\delta T_{\delta}+p\delta T_{p\delta}+p^2\delta T_{p^2\delta}
\end{aligned}
\end{equation}
\normalsize
%Compared with \eqref{eq_perturb_theory_for_unperturbed_proof_9}, the operator $T$ above contains components that anti-commute with $\mathcal{M}_2$
Note that
\begin{equation} \label{eq_T_delta_operator_symmetry_summary}
[P,\mathcal{M}_2]=0\quad \text{for }P\in\{T_{\mu},T_{p},T_{p^2}\}, \quad \text{and}
\{P,\mathcal{M}_2\}=0\quad \text{for }P\in\{T_{\delta},T_{p\delta},T_{p^2\delta}\}.
\end{equation}
For $\delta,p,\mu^{(1)}$ sufficiently small, $(I-T)^{-1}\in \mathcal{B}(H_2)$ can be expanded as a Neumann series, leading to
\begin{equation} \label{eq_sec4_5}
\begin{aligned}
\tilde{v}^{(1)}
=(I-T)^{-1}T\tilde{v}^{(0)}
=\sum_{k\geq 1}T^{k}(\delta,p,\mu^{(1)})\tilde{v}^{(0)}.
\end{aligned}
\end{equation}
By substituting equation \eqref{eq_sec4_5} into \eqref{eq_sec4_4} and taking the dual pairs with $\tilde{v}_n(\bm{x})$ for $n=1,2$, we obtain a two-by-two linear system of equations for $\tilde{v}^{(0)}=a\cdot\tilde{v}_{1}+b\cdot\tilde{v}_{2}$. The corresponding matrix elements are multivariate polynomials in $\mu^{(1)}$, $p$ and $\delta$,  derived by expanding each term in the Neumann series \eqref{eq_sec4_5} using equation \eqref{eq_T_expansion_perturbed}. We decompose this matrix into two parts: one contains the lower-order terms involving $\mu^{(1)},p,\delta,p^2,p\delta,p^3,\mu^{(1)}p, \mu^{(1)}p^2$ and $p^4$, and the other the remaining higher-order terms. Specifically, 
\begin{equation} \label{eq_perturbation_matrix_prep}
\Big(\mathcal{M}^{(0)}(\mu^{(1)},p,\delta)
+\mathcal{M}^{(1)}(\mu^{(1)},p,\delta)\Big)
\begin{pmatrix}
a \\ b
\end{pmatrix}=0,
\end{equation}
where
\footnotesize
\begin{equation*}
\begin{aligned}
\mathcal{M}^{(0)}_{ij}
&= \mu^{(1)}\big(\tilde{v}_{i},\tilde{v}_{j}\big)
-p\big(\tilde{\mathcal{L}}^{A,1}\tilde{v}_{i},\tilde{v}_{j}\big)
-\delta\big(\tilde{\mathcal{L}}^{B,0}\tilde{v}_{i},\tilde{v}_{j}\big)
-p^2\Big[\big(\tilde{\mathcal{L}}^{A,2}\tilde{v}_{i},\tilde{v}_{j}\big)+\big(\tilde{\mathcal{L}}^{A,1}T_p\tilde{v}_{i},\tilde{v}_{j}\big) \Big] \\
&\quad -p\delta \Big[\big(\tilde{\mathcal{L}}^{B,1}\tilde{v}_{i},\tilde{v}_{j}\big)+\big(\tilde{\mathcal{L}}^{A,1}T_{\delta}\tilde{v}_{i},\tilde{v}_{j}\big) +\big(\tilde{\mathcal{L}}^{B,0}T_{p}\tilde{v}_{i},\tilde{v}_{j}\big) \Big]  -p^3 \Big[\big(\tilde{\mathcal{L}}^{A,1}(T_{p^2}+T_{p}^2)\tilde{v}_{i},\tilde{v}_{j}\big) +\big(\tilde{\mathcal{L}}^{A,2}T_{p}\tilde{v}_{i},\tilde{v}_{j}\big) \Big] \\
&\quad -\mu^{(1)}p^2 \big(\tilde{\mathcal{L}}^{A,1}T_{\mu}T_p\tilde{v}_{i},\tilde{v}_{j}\big) -p^4 \Big[\big(\tilde{\mathcal{L}}^{A,1}(T_{p^2}T_{p}+T_{p}T_{p^2}+T_{p}^3)\tilde{v}_{i},\tilde{v}_{j}\big)+\big(\tilde{\mathcal{L}}^{A,2}(T_{p^2}+T_{p}^2)\tilde{v}_{i},\tilde{v}_{j}\big) \Big].
\end{aligned}
\end{equation*}
\normalsize
Here the terms 
$\mu^{(1)}p^2 ((T_{p^2}+T_{p}^2)\tilde{v}_{i},\tilde{v}_{j})$, $\mu^{(1)}p^2 ((\tilde{\mathcal{L}}^{A,1}T_{p}+\tilde{\mathcal{L}}^{A,2})T_{\mu}\tilde{v}_{i},\tilde{v}_{j})$ and all $\mu^{(1)}p$-terms (such as $\mu^{(1)}p ((T_{p}+\tilde{\mathcal{L}}^{A,1}T_{\mu})\tilde{v}_{i},\tilde{v}_{j})$) do not occur since they equal to zero, a consequence of the fact $\tilde{v}_i\perp \text{Ran }Q_{\perp}$ and $\tilde{v}_i\perp \text{Ran }((\tilde{\mathcal{L}}^{A,0}-\lambda_*)^{-1}Q_{\perp})$.
%We design this decomposition to solve the Bloch eigenvalues in the form of \eqref{eq_quadratic_mu_dispersion} with an appropriate $\delta -$modification, as being clear later.

\medskip

{\color{blue}Step 3.} 
%In Step 4, We will solve \eqref{eq_perturbation_matrix_prep} for $\mu^{(1)}=\mu^{(1)}(p,\delta)$ to obtain the Bloch eigenvalue as shown in \eqref{eq_asymp_eigenvalue_positive}.
In this step, we simplify the matrix $\mathcal{M}^{(0)}$ and estimate the remainder matrix $\mathcal{M}^{(1)}$.

{\color{blue}Step 3.1: Estimation of $\mathcal{M}^{(1)}$.} Instead of examining each multivariate monomial within the matrix elements of $\mathcal{M}^{(1)}$, we first utilize the symmetry of operators to simplify the analysis. 

First we consider the off-diagonal elements $\mathcal{M}^{(1)}_{ij}$ for $i\neq j$. We claim that each monomial in $\mathcal{M}^{(1)}_{ij}$ is at least first order in $\delta$. 
Indeed, for those monomial without $\delta$-factor, their coefficients are of the form $(P\tilde{v}_{i},\tilde{v}_{j})$, where $P$ is a composition of operators in the set $\{\mathcal{L}^{A,k}, T_{\mu}, T_{p}, T_{p^2}\}$ and is thus $\mathcal{M}_2$-invariant. Lemma \ref{lem_symmetry_products} implies that $(P\tilde{v}_{i},\tilde{v}_{j})=0$ and hence the claim follows. In addition, note that the $\delta$-term is incorporated in $\mathcal{M}^{(0)}$ and hence excluded from $\mathcal{M}^{(1)}$. 
Consequently, we have the estimate
\begin{equation*} 
\mathcal{M}^{(1)}_{ij}=\mathcal{O}\Big(\delta
(\delta+p+\mu^{(1)})\Big)\quad \text{if $i\neq j$}. 
\end{equation*}

We further observe that: (1) No $p\delta-$terms appears in $\mathcal{M}^{(1)}_{ij}$; (2) 
No $\delta\mu^{(1)}-$terms appears in $\mathcal{M}^{(1)}_{ij}$.
Observation (1) holds because all such terms are included in $\mathcal{M}^{(0)}$.  
Observation (2) holds because all such terms have coefficient $(\mathcal{L}^{B,0}T_{\mu}\tilde{v}_{i},\tilde{v}_{j})=(\mathcal{L}^{B,0}(\tilde{\mathcal{L}}^{A,0}-\lambda_*)^{-1}Q_{\perp}\tilde{v}_{i},\tilde{v}_{j})$ or $(T_{\delta}\tilde{v}_{i},\tilde{v}_{j})=-((\tilde{\mathcal{L}}^{A,0}-\lambda_*)^{-1}Q_{\perp}\mathcal{L}^{B,0}\tilde{v}_{i},\tilde{v}_{j})$.  Both are equal to zero since $\tilde{v}_i\perp \text{Ran }Q_{\perp}$ and $\tilde{v}_i\perp \text{Ran }((\tilde{\mathcal{L}}^{A,0}-\lambda_*)^{-1}Q_{\perp})$. With these two observations, we further obtain   
\begin{equation} \label{eq_perturb_matrix_remainder_estimate_1}
\begin{aligned}
\mathcal{M}^{(1)}_{ij}&=\mathcal{O}\Big(\delta
\big(\delta+(p^2+p\mu^{(1)}+p\delta)+((\mu^{(1)})^2+p\mu^{(1)}+ \mu^{(1)}\delta)\big)\Big) \\
&=\mathcal{O}\Big(\delta
\big(\delta+p^2+(\mu^{(1)})^2\big)\Big)
\quad\quad \text{for $i\neq j$}.
\end{aligned}
\end{equation}
The diagonal elements of $\mathcal{M}^{(1)}$ are estimated similarly. We observe that
\begin{itemize}
    \item $\mathcal{M}^{(1)}_{ii}$ does not contain any first-order terms in $\delta$ (such as $p\mu^{(1)}\delta$). In fact, such terms have coefficients in the form $(P\tilde{v}_{i},\tilde{v}_{i})$ where $P$ is the composition of one operator from $\{\mathcal{L}^{B,k},T_{\delta},T_{p\delta},T_{p^2\delta}\}$, which anti-commutes with $\mathcal{M}_2$ by \eqref{eq_perturbation_operators_anti_commute} and \eqref{eq_T_delta_operator_symmetry_summary}, and another from $\{\mathcal{L}^{A,k},T_{\mu},T_{p},T_{p^2}\}$,  which commute with $\mathcal{M}_2$. Consequently, $P$ anti-commutes with $\mathcal{M}_2$, implying that $(P\tilde{v}_{i},\tilde{v}_{i})=0$ by Lemma \ref{lem_symmetry_products}. 
    \item The monomials in $\mathcal{M}^{(1)}_{ii}$ that depend only on $\mu^{(1)}$ and $p$ 
    are of at least third order, specifically of the form  $(\mu^{(1)})^{i}p^{j}$ with $i+j=3$ for reasons analogous to those in Observation (2). In addition, the $\mu^{(1)}p^2$- and $p^3-$ terms are incorporated within $\mathcal{M}^{(0)}$ and is therefore excluded from $\mathcal{M}^{(1)}_{ii}$. 
    \item The monomials in $\mathcal{M}^{(1)}_{ii}$ that depend solely on $p$ are at least of the order $\mathcal{O}(p^5)$ since all lower orders are included in $\mathcal{M}^{(0)}$.
    \end{itemize}
These observations lead to the following estimate
\begin{equation} \label{eq_perturb_matrix_remainder_estimate_2}
\begin{aligned}
\mathcal{M}^{(1)}_{ii}=\mathcal{O}\Big(\delta^2(p+\mu^{(1)})+(\mu^{(1)})^3 +(\mu^{(1)})^2 p+\mu^{(1)}p^3+p^5\Big).
\end{aligned}
\end{equation}

{\color{blue}Step 3.2: Simplify $\mathcal{M}^{(0)}$.} First, by \eqref{eq_perturb_theory_for_unperturbed_proof_23}, the $p$-term in $\mathcal{M}^{(0)}_{ij}$ vanishes. Moreover, the $p\delta$-term vanishes on the diagonals and the $p^k$-terms ($k=2,3,4$) vanish on the off-diagonals, as demonstrated in the analysis of $\mathcal{M}^{(1)}_{ij}$. Additionally, the $p^3$-term also vanishes on the diagonals by \eqref{eq_perturb_theory_for_unperturbed_proof_25}. In conclusion, with \eqref{eq_perturb_theory_for_unperturbed_proof_24}, \eqref{eq_perturb_theory_for_unperturbed_proof_26} and Lemma \ref{lem_non-degeneracy}, 
\footnotesize
\begin{equation} \label{eq_perturb_matrix_major_prep}
\begin{aligned}
\mathcal{M}^{(0)}
&=\begin{pmatrix}
\mu^{(1)}+\frac{\gamma_*}{2}p^2+\eta_* p^4 & \delta (t_*+r_*p) \\
\delta (t_*+\overline{r_*}p) & \mu^{(1)}-\frac{\gamma_*}{2}p^2-\eta_* p^4
\end{pmatrix} \\
&\quad +\begin{pmatrix}
-p^2(\mu^{(1)}+\frac{\gamma_*}{2}p^2)\big(\tilde{\mathcal{L}}^{A,1}T_{\mu}T_p\tilde{v}_{1},\tilde{v}_{1}\big) & 0 \\
0 & -p^2(\mu^{(1)}-\frac{\gamma_*}{2}p^2)\big(\tilde{\mathcal{L}}^{A,1}T_{\mu}T_p\tilde{v}_{2},\tilde{v}_{2}\big)
\end{pmatrix},
\end{aligned}
\end{equation}
\normalsize
where
\begin{equation} \label{eq_sec4_9}
r_{*}=-\Big[\big(\tilde{\mathcal{L}}^{B,1}\tilde{v}_{2},\tilde{v}_{1}\big)+\big(\tilde{\mathcal{L}}^{A,1}T_{\delta}\tilde{v}_{2},\tilde{v}_{1}\big) +\big(\tilde{\mathcal{L}}^{B,0}T_{p}\tilde{v}_{2},\tilde{v}_{1}\big) \Big].
\end{equation}
Combining \eqref{eq_perturb_matrix_remainder_estimate_1}, \eqref{eq_perturb_matrix_remainder_estimate_2} and \eqref{eq_perturb_matrix_major_prep}, we can rewrite \eqref{eq_perturbation_matrix_prep} as
\begin{equation} \label{eq_perturbation_matrix}
\Big(\tilde{\mathcal{M}}^{(0)}(\mu^{(1)},p,\delta)
+\tilde{\mathcal{M}}^{(1)}(\mu^{(1)},p,\delta)\Big)
\begin{pmatrix}
a \\ b
\end{pmatrix}=0,
\end{equation}
with
\begin{equation*}
\tilde{\mathcal{M}}^{(0)}(\mu^{(1)},p,\delta)
=\begin{pmatrix}
\mu^{(1)}+\frac{\gamma_*}{2}p^2+\eta_* p^4 & \delta (t_*+r_*p) \\
\delta (t_*+\overline{r_*}p) & \mu^{(1)}-\frac{\gamma_*}{2}p^2-\eta_* p^4
\end{pmatrix},
\end{equation*}
\footnotesize
\begin{equation} \label{eq_perturbation_matrix_minor}
\tilde{\mathcal{M}}^{(1)}(\mu^{(1)},p,\delta)
=\begin{pmatrix}
\substack{-p^2(\mu^{(1)}+\frac{\gamma_*}{2}p^2)\big(\tilde{\mathcal{L}}^{A,1}T_{\mu}T_p\tilde{v}_{1},\tilde{v}_{1}\big) \\ +\mathcal{O}\Big(\delta^2(p+\mu^{(1)})+(\mu^{(1)})^3 +(\mu^{(1)})^2 p+\mu^{(1)}p^3+p^5\Big)} & \mathcal{O}\Big(\delta
\big(\delta+p^2+(\mu^{(1)})^2\big)\Big) \\
\mathcal{O}\Big(\delta
\big(\delta+p^2+(\mu^{(1)})^2\big)\Big) & \substack{-p^2(\mu^{(1)}-\frac{\gamma_*}{2}p^2)\big(\tilde{\mathcal{L}}^{A,1}T_{\mu}T_p\tilde{v}_{2},\tilde{v}_{2}\big) \\+\mathcal{O}\Big(\delta^2(p+\mu^{(1)})+(\mu^{(1)})^3 +(\mu^{(1)})^2 p+\mu^{(1)}p^3+p^5\Big)}
\end{pmatrix}.
\end{equation}
\normalsize

{\color{blue}Step 4}. We prove \eqref{eq_asymp_eigenvalue_positive} by solving $\mu^{(1)}=\mu^{(1)}(p,\delta)$ from \eqref{eq_perturbation_matrix}.  
%(then $\mu=\lambda_*+\mu^{(1)}$ gives the desired Bloch eigenvalue).
Specifically, we solve
\footnotesize
\begin{equation} \label{eq_sec4_10}
F(\mu^{(1)},p,\delta):=\det \Big(\tilde{\mathcal{M}}^{(0)}
+\tilde{\mathcal{M}}^{(1)}\Big)
=(\mu^{(1)})^2-\Big(\big(\frac{\gamma_*}{2}p^2+\eta_* p^4\big)^2
+\delta^2|t_*+r_*p|^2\Big)+\rho(\mu^{(1)},p,\delta)=0, 
\end{equation}
\normalsize
where
%Note that $\rho(\mu^{(1)},p,\delta)$ is smooth in $\mu^{(1)},p,\delta$ for $|\mu^{(1)}|+|p|+|\delta|$ small enough. Moreover, 
\footnotesize
\begin{equation} \label{eq_sec4_11}
\begin{aligned}
\rho(\mu^{(1)},p,\delta) 
&=\tilde{\mathcal{M}}^{(0)}_{11}\tilde{\mathcal{M}}^{(1)}_{22}+\tilde{\mathcal{M}}^{(1)}_{11}\tilde{\mathcal{M}}^{(0)}_{22}-\tilde{\mathcal{M}}^{(0)}_{12}\tilde{\mathcal{M}}^{(1)}_{21}-\tilde{\mathcal{M}}^{(1)}_{12}\tilde{\mathcal{M}}^{(0)}_{21}+\det \tilde{\mathcal{M}}^{(1)} \\
&=-2p^2\big((\mu^{(1)})^2-\frac{\gamma_*^2}{4}p^4\big)+\mathcal{O}\Big(\mu^{(1)}(p\delta^2+p^5)+(\mu^{(1)})^2(p^3+\delta^2)+(\mu^{(1)})^3 p+(\mu^{(1)})^4+\delta^3+\delta^2p^3+p^7\Big).
\end{aligned}
\end{equation}
\normalsize
We have
\begin{equation} \label{eq_sec4_16}
\frac{\partial \rho}{\partial \mu^{(1)}}(\mu^{(1)},p,\delta) =\mathcal{O}(\mu^{(1)}(p^2+\delta^2)+(\mu^{(1)})^2 p+(\mu^{(1)})^3+p\delta^2+p^5).
\end{equation}
We then solve $\mu^{(1)}=\mu^{(1)}(p,\delta)$ from \eqref{eq_sec4_10} for each $|p|\ll 1$ and $|\delta|\ll 1$. First, note that without the remainder $\rho$, 
two branches of solutions are given by
\begin{equation*}
\pm \sqrt{\big(\frac{\gamma_*}{2}p^2+\eta_* p^4\big)^2
+\delta^2|t_*+r_*p|^2}. 
\end{equation*}
Thus, we seek a solution to \eqref{eq_sec4_10} in the following form
\begin{equation} \label{eq_sec4_12}
\mu^{(1)}(p,\delta)=x\cdot \sqrt{\big(\frac{\gamma_*}{2}p^2+\eta_* p^4\big)^2
+\delta^2|t_*+r_*p|^2}
\end{equation}
with $|x|$ close to $1$. By substituting \eqref{eq_sec4_12} into \eqref{eq_sec4_10}, we obtain the following equation of $x$ 
\begin{equation} \label{eq_sec4_13}
\begin{aligned}
H(x;p,\delta)&:=\frac{1}{\big(\frac{\gamma_*}{2}p^2+\eta_* p^4\big)^2
+\delta^2|t_*+r_*p|^2}F(x\cdot\sqrt{\big(\frac{\gamma_*}{2}p^2+\eta_* p^4\big)^2
+\delta^2|t_*+r_*p|^2},p,\delta) \\
&=x^2-1+\rho_1(x;p,\delta)=0,
\end{aligned}
\end{equation}
where 
\footnotesize
\begin{equation*}
\rho_1(x;p,\delta)
=\frac{\rho(x\cdot\sqrt{\big(\frac{\gamma_*}{2}p^2+\eta_* p^4\big)^2
+\delta^2|t_*+r_*p|^2},p,\delta)}{\big(\frac{\gamma_*}{2}p^2+\eta_* p^4\big)^2
+\delta^2|t_*+r_*p|^2}. 
\end{equation*}
\normalsize
The estimate \eqref{eq_sec4_11} and \eqref{eq_sec4_16} indicate that  $\rho_1(x;p,\delta) \ll 1, \,\, \frac{\partial \rho_1}{\partial x}(x;p,\delta) \ll 1$
for $p,\delta$ sufficiently small and $x$ close to $1$. Hence \eqref{eq_sec4_13} can be solved by the implicit function theorem. Denote the solution close to $x=1$ by $x_s(p,\delta)$. Using the estimate \eqref{eq_sec4_11} and \eqref{eq_sec4_16}, 
$$x_s(p,\delta)=1+\mathcal{O}(|p|^3+|\delta|).$$ 
It follows that
\footnotesize
\begin{equation*}
\begin{aligned}
\mu^{(1)}_{+}(p,\delta)=x_s(p,\delta)\cdot \sqrt{\big(\frac{\gamma_*}{2}p^2+\eta_* p^4\big)^2
+\delta^2|t_*+r_*p|^2}
=\sqrt{\big(\frac{\gamma_*}{2}p^2+\eta_* p^4\big)^2
+\delta^2|t_*+r_*p|^2}\cdot \left(1+\mathcal{O}(|p|^3+|\delta|)\right).
\end{aligned}
\end{equation*}
\normalsize
The estimate of the remainder above can be further refined by noting that the above expansion should coincide with the one in Assumption \ref{def_quadratic_degenracy}(4) when we set $\delta=0$. Consequently, the remainder is of the order $\mathcal{O}(|p|^4)$ when $\delta=0$ and $\mathcal{O}(|p|^4+|\delta|)$
for $\delta \neq 0$. 
%(one can also see this by directly yet tediously calculating that all $p^{3}-$terms in the remainder vanish using Assumption \ref{def_quadratic_degenracy}(4) and the order-matching method in Cororllary \ref{corol_order_match}). 
Therefore
\begin{equation*}
\begin{aligned}
\mu^{(1)}_{+}(p,\delta)
=\sqrt{\big(\frac{\gamma_*}{2}p^2+\eta_* p^4\big)^2
+\delta^2|t_*+r_*p|^2}\cdot \left(1+\mathcal{O}(|p|^4+|\delta|)\right).
\end{aligned}
\end{equation*}
Compared with Assumption \ref{def_quadratic_degenracy}(4), the additional $\delta$-terms in the above expansion illustrate the perturbation effect on the Bloch eigenvalue. Similarly, the other solution to \eqref{eq_sec4_10} is
\begin{equation} \label{eq_sec4_23}
\begin{aligned}
\mu^{(1)}_{-}(p,\delta)
=-\sqrt{\big(\frac{\gamma_*}{2}p^2+\eta_* p^4\big)^2
+\delta^2|t_*+r_*p|^2}\cdot \left(1+\mathcal{O}(|p|^4+|\delta|)\right).
\end{aligned}
\end{equation}
Note that $\mu_{1,\delta}(\kappa_1)=\lambda_*+\mu^{(1)}_{-}(\kappa_1-\pi,\delta)$ and $\mu_{2,\delta}(\kappa_1)=\lambda_*+\mu^{(1)}_{+}(\kappa_1-\pi,\delta)$. This proves \eqref{eq_asymp_eigenvalue_positive}.

%Remark that the detailed perturbation analysis presented in this part is designed to derive the elaborate asymptotics \eqref{eq_asymp_eigenvalue_positive} of the Bloch eigenvalues. In fact, a coarse estimate is obtained by the standard perturbation theory \cite{kato2013perturbation}
%\%begin{equation} %\label{eq_mu12_coarse_perturbation_theory}
%|\mu_{n,\delta}(\kappa_1)-\mu_{n}%(\kappa_1)|=\mathcal{O}(\delta), \quad
%(n=1,2,\, \kappa_1\in [0,2\pi])
%\end{equation}

{\color{blue}Step 5}. We prove \eqref{eq_asymp_eigenfunction_1}. The proof of \eqref{eq_asymp_eigenfunction_2} is similar. By substituting $\mu^{(1)}=\mu^{(1)}_{-}(\kappa_1-\pi,\delta)$ into \eqref{eq_perturbation_matrix}, we obtain the following eigenvector
\begin{equation*}
\begin{pmatrix}
a(\kappa_1,\delta) \\ b(\kappa_1,\delta)
\end{pmatrix}
=
\begin{pmatrix}
1 \\ \frac{\delta(t_*+\overline{r_*}p)+\tilde{\mathcal{M}}^{(1)}_{21}(\mu^{(1)}_{-}(p,\delta),p,\delta)}{-\mu^{(1)}_{-}(p,\delta)+\frac{\gamma_*}{2}p^2+\eta_* p^4-\tilde{\mathcal{M}}^{(1)}_{22}(\mu^{(1)}_{-}(p,\delta),p,\delta)}
\end{pmatrix}.
\end{equation*}
By \eqref{eq_sec4_23}, the leading term order of $\mu^{(1)}_{-}$ is $-\sqrt{\frac{\gamma_*^2}{4}p^4+t_*^2\delta^2}$. This combined with the estimate of $\tilde{\mathcal{M}}^{(1)}$ in \eqref{eq_perturbation_matrix_minor} yield
\footnotesize
\begin{equation*}
b(\kappa_1,\delta)
=\frac{t_*\delta}{\frac{\gamma_*}{2}p^2+\sqrt{\frac{\gamma_*^2}{4}p^4+t_*^2\delta^2}}
\Big(1+ r(p;\delta)\cdot p+\mathcal{O}(\delta+p^2)\Big),
\end{equation*}
\normalsize
where 
\footnotesize
\begin{equation} \label{eq_sec4_28}
r(p;\delta):=\overline{r_*}-\frac{t_*\text{Re}(r_*)\delta^2}{(\frac{\gamma_*}{2}p^2+\sqrt{\frac{\gamma_*^2}{4}p^4+t_*^2\delta^2})\sqrt{\frac{\gamma_*^2}{4}p^4+t_*^2\delta^2}}
\end{equation}
\normalsize
is uniformly bounded. 
Using the expression $\tilde{v}^{(0)}=a(\kappa_1,\delta)\tilde{v}_1+b(\kappa_1,\delta)\tilde{v}_2$ and the expansion of $\tilde{v}^{(1)}$ in \eqref{eq_sec4_5}, we have 
%the full expansion of the (periodic part of) Bloch eigenfunction
\begin{equation*}
\begin{aligned}
\tilde{v}=\sum_{k\geq 0}T^{k}(\delta,p,\mu^{(1)}_{-})\tilde{v}_1
+\frac{t_*\delta}{\frac{\gamma_*}{2}p^2+\sqrt{\frac{\gamma_*^2}{4}p^4+t_*^2\delta^2}}
\Big(1+ r(p;\delta)\cdot p+\mathcal{O}(\delta+p^2)\Big)\sum_{k\geq 0}T^{k}(\delta,p,\mu^{(1)}_{-})\tilde{v}_2. 
\end{aligned}
\end{equation*}
Multiplying both sides by $e^{i\kappa_1 x_1}$ yields 
%, this gives the first branch of (quasi-periodic) Bloch eigenfunction near $\kappa_1=\pi$ when $\delta\neq 0$
\footnotesize
\begin{equation} \label{eq_sec4_24}
\begin{aligned}
v_{1,\delta}(\bm{x};\kappa_1)&=e^{i\kappa_1 x_1}\sum_{k\geq 0}T^{k}(\delta,\kappa_1-\pi,\mu^{(1)}_{-}(\kappa_1-\pi,\delta))\tilde{v}_1 \\
&\quad +f_*(\kappa_1;\delta)\Big(1+ r(\kappa_1-\pi;\delta)\cdot (\kappa_1-\pi)+\mathcal{O}(\delta+(\kappa_1-\pi)^2)\Big)
\sum_{k\geq 0}T^{k}(\delta,\kappa_1-\pi,\mu^{(1)}_{-}(\kappa_1-\pi,\delta))\tilde{v}_2 \\
&=:v_{1,\delta}^{I}(\bm{x};\kappa_1)+v_{1,\delta}^{II}(\bm{x};\kappa_1).
\end{aligned}
\end{equation}
\normalsize
We note that when $\delta=0$
the first term $v_{1,\delta}^{I}(\bm{x};\kappa_1)$ coincides with the unperturbed Bloch mode $v_{1}(\bm{x};\kappa_1)$ in \eqref{eq_perturb_theory_for_unperturbed_proof_14_v1}. The leading-order perturbation effect is manifested by the second part $v_{1,\delta}^{II}(\bm{x};\kappa_1)$.  

We prove \eqref{eq_asymp_eigenfunction_1} by estimating \eqref{eq_sec4_24}. First we consider $v_{1,\delta}^{I}$. By the uniform convergence of the Neumann series \eqref{eq_sec4_5}, 
\footnotesize
\begin{equation*}
\begin{aligned}
v_{1,\delta}^{I}(\bm{x};\kappa_1) 
&=e^{i\kappa_1 x_1}\sum_{k\geq 0}T^{k}
\left(\delta,\kappa_1-\pi,\mu_{1}(\kappa_1)+(\mu_{1,\delta}(\kappa_1)-\mu_{1}(\kappa_1))\right)\tilde{v}_1 \\
&=e^{i\kappa_1 x_1}\sum_{k\geq 0}T^{k}
\left(0,\kappa_1-\pi,\mu_{1}(\kappa_1)\right)\tilde{v}_1
+\mathcal{O}(\delta+|\mu_{1,\delta}(\kappa_1)-\mu_{1}(\kappa_1)|). 
\end{aligned}
\end{equation*}
\normalsize
The first term is exactly the numerator of \eqref{eq_ansatz_unperturb_eigenfunction_1}. In addition,  $|\mu_{1,\delta}(\kappa_1)-\mu_{1}(\kappa_1)|=\mathcal{O}(\delta)$ by \eqref{eq_asymp_eigenvalue_positive}. It follows that
\begin{equation} \label{eq_sec4_25}
v_{1,\delta}^{I}(\bm{x};\kappa_1)
=v_1+(\kappa_1-\pi)\partial_{\kappa_1}v_1(\bm{x};\pi)
+\sum_{k=2}^{5}v^{(k)}_{1}(\bm{x};\pi)
+\mathcal{O}(\delta+(\kappa_1-\pi)^6). 
\end{equation}
Next, using the estimates
\[
\sum_{k\geq 0}T^{k}(\delta,\kappa_1-\pi,\mu^{(1)}_{-}(\kappa_1-\pi,\delta))\tilde{v}_2= v_2(\bm{x};\pi)+\partial_{\kappa_1}v_2(\bm{x};\pi)(\kappa_1-\pi)+ \mathcal{O}\big((\kappa_1-\pi)^2+ \delta
\big),\]
and 
\begin{equation} \label{eq_sec4_26}
(\kappa_1-\pi)^{k}f_*(\kappa_1;\delta)=\frac{ (\kappa_1-\pi)^{k}t_*\delta}{\frac{\gamma_*}{2}(\kappa_1-\pi)^2+\sqrt{\frac{\gamma_*^2}{4}(\kappa_1-\pi)^4+t_*^2\delta^2}}
=\mathcal{O}(\delta),\quad \forall k\geq 2, 
\end{equation}
we have
\footnotesize
\begin{equation} \label{eq_sec4_27}
v_{1,\delta}^{II}(\bm{x};\kappa_1)
=f_*(\kappa_1;\delta)\big(1+ r(\kappa_1-\pi;\delta)\cdot (\kappa_1-\pi)\big) v_2(\bm{x};\pi)+
f_*(\kappa_1;\delta)(\kappa_1-\pi)\partial_{\kappa_1}v_2(\bm{x};\pi)
+\mathcal{O}(\delta).
\end{equation}
\normalsize
%where
%\begin{equation*}
%f_*(\kappa_1;\delta)=\frac{t_*\delta}{\frac{\gamma_*}{2}(\kappa_1-\pi)^2+\sqrt{\frac{\gamma_*^2}{4}(\kappa_1-\pi)^4+t_*^2\delta^2}}. 
%\end{equation*}
The proof of \eqref{eq_asymp_eigenfunction_1} is completed by combining \eqref{eq_sec4_24}, \eqref{eq_sec4_25} and \eqref{eq_sec4_27}. Finally, the estimation of the norm $N_{1,\delta}(\kappa_1)=\|v_{1,\delta}(\bm{x};\kappa_1)\|=\sqrt{(v_{1,\delta},v_{1,\delta})}$ follows from the expansion of the normalization factor in the unperturbed case \eqref{eq_perturb_theory_for_unperturbed_proof_16} and the estimate \eqref{eq_sec4_26}. 
\normalsize
This completes the proof of \eqref{eq_perturbed_normalization_factor_1}.

\section{Asymtotics of the perturbed Green function}
In this section, we study the Green function $G^{\delta}$ for the perturbed operator $\mathcal{L}^{A+\delta\cdot B}$: 
\begin{equation*} 
    \left\{
    \begin{aligned}
        &\Big(\nabla\cdot(A+\delta \cdot B)\nabla+\lambda\Big)G^{  \delta}(\bm{x},\bm{y};\lambda)=\delta(\bm{x}-\bm{y}),\quad x,y \in \Omega,\\
        &G^{  \delta}(\bm{x},\bm{y};\lambda)\big|_{\Gamma^{+}}=e^{i\pi}G^{  \delta}(\bm{x},\bm{y};\lambda)\big|_{\Gamma^{-}}, \\
        &\frac{\partial G^{  \delta}}{\partial x_2}(\bm{x},\bm{y};\lambda_*)\big|_{\Gamma^{+}}=e^{i\pi}\frac{\partial G^{  \delta}}{\partial x_2}(\bm{x},\bm{y};\lambda)\big|_{\Gamma^{-}}.
    \end{aligned}
    \right.
\end{equation*}
By Floquet transform, $G^{\delta}$ has the spectral representation
\begin{equation} \label{eq_perturbed_Green_function_floquet_expansion}
G^{  \delta}(\bm{x},\bm{y};\lambda)
=\frac{1}{2\pi}\int_{0}^{2\pi}
\sum_{n\geq 1}\frac{v_{n,  \delta}(\bm{x};\kappa_1)\overline{v_{n,  \delta}(\bm{y};\kappa_1)}}{\lambda-\mu_{n,  \delta}(\kappa_1)}d\kappa_1.
\end{equation}

We are concerned with the asymptotics of $G^{\delta}(\bm{x},\bm{y};\lambda)$ when $\lambda$ lies in the band gap, which reveals the perturbation effect. For this purpose, we set $\delta>0$ and reparametrize $\lambda$ by
\begin{equation*}
\lambda=\lambda_*+\delta\cdot h,\quad 
h\in\mathcal{J}:=\{z\in\mathbf{C}:\, |z|<c_0|t_*|\}.
\end{equation*}
As in Section 5, we recast the asymptotics of the Green function in terms of its associated integral operator
\begin{equation*}
\mathcal{G}^{\delta}(\lambda_*+\delta\cdot h)f:=\int_{\Omega}G^{  \delta}(\bm{x},\bm{y};\lambda_*+\delta\cdot h)f(\bm{y})d\bm{y}.
\end{equation*}
The main result of this section is:
\begin{theorem}
\label{thm_asymptotics_perturbed_Green_function}
There exist $f^{prop}_{n,\delta},g^{prop}_{n,\delta}\in H^1_{loc}(\Omega)$ for $n=1,2$ with the estimate $\|f^{prop}_{n,\delta}\|_{H^1(U)},\|g^{prop}_{n,\delta}\|_{H^1(U)}=\mathcal{O}(\delta^{\frac{1}{9}})$ for any compact set $U\subset \Omega$, such that the following asymptotic behavior holds uniformly for each $f\in (H^{1}(\Omega))^*$ with compact support with unit norm:
\footnotesize
\begin{equation*} 
\begin{aligned}
\mathcal{G}^{\pm\delta}(\lambda_*+\delta\cdot h)f
&=\mathcal{G}_0(\lambda_*)f \\
&\quad +\sum_{n=1,2}\Bigg[\delta^{-\frac{1}{2}}k_{n,*}(h)
+o(\delta^{-\frac{1}{2}})\Bigg]v_{n}(\bm{x};\pi) \big(f(\bm{y}),v_{n}(\bm{y};\pi)\big)_{\Omega}  \\
&\quad \pm \Bigg[\delta^{-\frac{1}{2}}k_{cross,*}(h)+o(\delta^{-\frac{1}{2}}) \Bigg]\Big(v_{1}(\bm{x};\pi) \big(f(\bm{y}),v_{2}(\bm{y};\pi)\big)_{\Omega}+v_{2}(\bm{x};\pi) \big(f(\bm{y}),v_{1}(\bm{x};\pi)\big)_{\Omega} \Big) \\
&\quad+ \sum_{n=1,2}\Bigg[\delta^{\frac{1}{2}}e_{n,*}(h)+o(\delta^{\frac{1}{2}}) \Bigg]\partial_{\kappa_1}v_{n}(\bm{x};\pi) \big(f(\bm{y}),\partial_{\kappa_1}v_{n}(\bm{y};\pi)\big)_{\Omega}  \\
&\quad \pm\Bigg[\delta^{\frac{1}{2}}e_{cross,*}(h)+o(\delta^{\frac{1}{2}}) \Bigg]\Big(\partial_{\kappa_1}v_{1}(\bm{x};\pi) \big(f(\bm{y}),\partial_{\kappa_1}v_{2}(\bm{y};\pi)\big)_{\Omega}
+\partial_{\kappa_1}v_{2}(\bm{x};\pi) \big(f(\bm{y}),\partial_{\kappa_1}v_{1}(\bm{y};\pi)\big)_{\Omega} \Big) \\
&\quad +\sum_{n=1,2}v_n(\bm{x};\pi)(f(\bm{y}),f^{prop}_{n,\delta}(\bm{y}))_{\Omega}+g^{prop}_{n,\delta}(\bm{x})(f(\bm{y}),v_n(\bm{y};\pi))_{\Omega} +\mathcal{O}(\delta^{\frac{5}{9}}).
\end{aligned}
\end{equation*}
\normalsize
Here $\mathcal{G}_0(\lambda_*)$ is the integral operator introduced in Theorem \ref{thm_asymp_unperturbed_green}, and
\footnotesize
\begin{equation*}
\begin{aligned}
k_{1}(h)=-\frac{1}{2\gamma_*}
\frac{\frac{h}{\gamma_*}-\sqrt{(\frac{t_*}{\gamma_*})^2-(\frac{h}{\gamma_*})^2}}{((\frac{t_*}{\gamma_*})^2-(\frac{h}{\gamma_*})^2)^{\frac{3}{4}}},\quad
k_{2}(h)=-\frac{1}{2\gamma_*}
\frac{\frac{h}{\gamma_*}+\sqrt{(\frac{t_*}{\gamma_*})^2-(\frac{h}{\gamma_*})^2}}{((\frac{t_*}{\gamma_*})^2-(\frac{h}{\gamma_*})^2)^{\frac{3}{4}}},\quad
k_{3}(h)=\frac{1}{2\gamma_*}
\frac{\frac{t_*}{\gamma_*}}{((\frac{t_*}{\gamma_*})^2-(\frac{h}{\gamma_*})^2)^{\frac{3}{4}}},
\end{aligned}
\end{equation*}
\normalsize
\footnotesize
\begin{equation*}
\begin{aligned}
e_{1}(h)=-\frac{1}{\gamma_*}
\frac{\frac{h}{\gamma_*}+\sqrt{(\frac{t_*}{\gamma_*})^2-(\frac{h}{\gamma_*})^2}}{((\frac{t_*}{\gamma_*})^2-(\frac{h}{\gamma_*})^2)^{\frac{1}{4}}},\quad
e_{2}(h)=-\frac{1}{\gamma_*}
\frac{\frac{h}{\gamma_*}-\sqrt{(\frac{t_*}{\gamma_*})^2-(\frac{h}{\gamma_*})^2}}{((\frac{t_*}{\gamma_*})^2-(\frac{h}{\gamma_*})^2)^{\frac{1}{4}}},\quad
e_{3}(h)=\frac{1}{\gamma_*}
\frac{\frac{t_*}{\gamma_*}}{((\frac{t_*}{\gamma_*})^2-(\frac{h}{\gamma_*})^2)^{\frac{1}{4}}}.
\end{aligned}
\end{equation*}
\normalsize
\end{theorem}
The proof of Theorem \ref{thm_asymptotics_perturbed_Green_function} requires a thorough asymptotic analysis of integrals. Although these integrals are analogous to those encountered in the proof of Theorem \ref{thm_asymp_unperturbed_green}, they are significantly more complex due to the additional terms introduced by the perturbation. The detailed proof is provided in the next section.

\subsection{Proof of Theorem \ref{thm_asymptotics_perturbed_Green_function}}
We only calculate the asymptotics of $\mathcal{G}^{\delta}(\lambda_*+\delta\cdot h)f$. The asymptotics of $\mathcal{G}^{-\delta}(\lambda_*+\delta\cdot h)f$ is obtained similarly.

Similar to the analysis of the unperturbed Green function in Section 5, we decompose $\mathcal{G}^{\delta}$ into four parts $\mathcal{G}^{\delta}=\mathcal{G}_1^{\delta,evan}+\mathcal{G}_{2,1}^{\delta,evan}+ \mathcal{G}_{2,2}^{\delta,evan}+\mathcal{G}^{\delta,prop}$, where
\footnotesize
\begin{equation} \label{eq_G_delta_evan_def}
\begin{aligned}
\mathcal{G}_1^{\delta,evan}(\lambda_*+\delta\cdot h)f
&=\frac{1}{2\pi}\int_{0}^{2\pi}
\sum_{n\geq 3}\frac{v_{n,  \delta}(\bm{x};\kappa_1)(f(\cdot),v_{n,  \delta}(\cdot;\kappa_1))_{\Omega}}{\lambda_*+\delta\cdot h-\mu_{n,  \delta}(\kappa_1)}d\kappa_1 \\
\mathcal{G}_{2,n}^{\delta,evan}(\lambda_*+\delta\cdot h)f=
&\frac{1}{2\pi}\int_{[0,\pi-\delta^{\frac{1}{9}}]\cup [\pi+\delta^{\frac{1}{9}},2\pi]}
\frac{v_{n,  \delta}(\bm{x};\kappa_1)f(\cdot),v_{n,  \delta}(\cdot;\kappa_1))_{\Omega}}{\lambda_*+\delta\cdot h-\mu_{n,\delta}(\kappa_1)}d\kappa_1, \,\, n=1, 2,
\end{aligned}
\end{equation}
\begin{equation*}
\begin{aligned}
\mathcal{G}^{\delta,prop}(\lambda_*+\delta\cdot h)f
=\frac{1}{2\pi}\int_{\pi-\delta^{\frac{1}{9}}}^{\pi+\delta^{\frac{1}{9}}}
\frac{v_{1,  \delta}(\bm{x};\kappa_1)f(\cdot),v_{1,  \delta}(\cdot;\kappa_1))_{\Omega}}{\lambda_*+\delta\cdot h-\mu_{1,\delta}(\kappa_1)}d\kappa_1  +\frac{1}{2\pi}\int_{\pi-\delta^{\frac{1}{9}}}^{\pi+\delta^{\frac{1}{9}}}
\frac{v_{2,  \delta}(\bm{x};\kappa_1)f(\cdot),v_{2,  \delta}(\cdot;\kappa_1))_{\Omega}}{\lambda_*+\delta\cdot h-\mu_{2,\delta}(\kappa_1)}d\kappa_1 .
\end{aligned}
\end{equation*}
\normalsize
Then Theorem \ref{thm_asymptotics_perturbed_Green_function} follows by combining the following three lemmas, whose proofs are given in Sections 8.2, 8.3, and 8.4 respectively. 

\begin{lemma} [Asymptotics of $\mathcal{G}_{1}^{\delta,even}$]\label{lem_app_C_1_evan_ngeq3}
\footnotesize
\begin{equation} \label{eq_app_C_1_ngeq3}
\begin{aligned}
\mathcal{G}_1^{\delta,evan}(\lambda_*+\delta\cdot h)f
=\frac{1}{2\pi}\int_{0}^{2\pi}
\sum_{n\geq 3}\frac{v_{n}(\bm{x};\kappa_1)(f(\cdot),v_{n}(\cdot;\kappa_1))_{\Omega}}{\lambda_*-\mu_{n}(\kappa_1)}d\kappa_1
+\mathcal{O}(\delta).
\end{aligned}
\end{equation}
\normalsize
\end{lemma}

\begin{lemma} [Asymptotics of $\mathcal{G}_{2, 1}^{\delta,even}$ and $\mathcal{G}_{2,2}^{\delta,even}$]\label{lem_app_C_1_evan_n=12}
\footnotesize
\begin{equation} \label{eq_app_C_1_n=1}
\begin{aligned}
&\mathcal{G}_{2,1}^{\delta,evan}(\lambda_*+\delta\cdot h)f \\
&=(\delta^{-\frac{1}{9}}\cdot \frac{2}{\pi\gamma_*}+\mathcal{O}(\delta^{\frac{1}{9}}))\cdot v_1(\bm{x};\pi)(f(\cdot),v_1(\cdot;\pi))_{\Omega}  +\mathcal{G}_0^{(1)}(\lambda_*)f 
+\big(\sum_{j+\ell =1}v_1^{(j)}(\bm{x};\pi)(f(\cdot),v_1^{(\ell)}(\cdot;\pi))_{\Omega} \big)\cdot \mathcal{O}(\delta^{\frac{1}{3}}) \\
&\quad +\big(\sum_{j+\ell =2}v_1^{(j)}(\bm{x};\pi)(f(\cdot),v_1^{(\ell)}(\cdot;\pi))_{\Omega} \big)\cdot\Big(\frac{-2\delta^{\frac{1}{9}}}{\pi\gamma_*}+\frac{2}{3\pi\gamma_*}\big(\frac{2\eta_*}{\gamma_*}+N_1^{(2)}\big)\delta^{\frac{1}{3}} \Big) \\
&\quad +\big(\sum_{j+\ell =4}v_1^{(j)}(\bm{x};\pi)(f(\cdot),v_1^{(\ell)}(\cdot;\pi))_{\Omega} \big)\cdot\big(\frac{-2}{3\pi\gamma_*}\delta^{\frac{1}{3}} \big)
 +\mathcal{O}(\delta^{\frac{5}{9}}),
\end{aligned}
\end{equation}
\normalsize
\footnotesize
\begin{equation} \label{eq_app_C_1_n=2}
\begin{aligned}
&\mathcal{G}_{2,2}^{\delta,evan}(\lambda_*+\delta\cdot h)f  \\
&=(-\delta^{-\frac{1}{9}}\cdot \frac{2}{\pi\gamma_*}+\mathcal{O}(\delta^{\frac{1}{9}}))\cdot v_2(\bm{x};\pi)(f(\cdot),v_2(\cdot;\pi))_{\Omega}  +\mathcal{G}_0^{(2)}(\lambda_*)f 
+\big(\sum_{j+\ell =1}v_2^{(j)}(\bm{x};\pi)(f(\cdot),v_2^{(\ell)}(\cdot;\pi))_{\Omega} \big)\cdot \mathcal{O}(\delta^{\frac{1}{3}}) \\
&\quad +\big(\sum_{j+\ell =2}v_2^{(j)}(\bm{x};\pi)(f(\cdot),v_2^{(\ell)}(\cdot;\pi))_{\Omega} \big)\cdot\Big(\frac{2\delta^{\frac{1}{9}}}{\pi\gamma_*}-\frac{2}{3\pi\gamma_*}\big(\frac{2\eta_*}{\gamma_*}+N_2^{(2)}\big)\delta^{\frac{1}{3}} \Big) \\
&\quad +\big(\sum_{j+\ell =4}v_2^{(j)}(\bm{x};\pi)(f(\cdot),v_2^{(\ell)}(\cdot;\pi))_{\Omega} \big)\cdot\big(\frac{2}{3\pi\gamma_*}\delta^{\frac{1}{3}} \big)
 +\mathcal{O}(\delta^{\frac{5}{9}}).
\end{aligned}
\end{equation}
\normalsize
Here $\mathcal{G}_0^{(1)}$ and $\mathcal{G}_0^{(2)}$ are introduced in Lemma \ref{lem_app_A_2_n=12}. We have included higher-order terms in \eqref{eq_app_C_1_n=1} and \eqref{eq_app_C_1_n=2} compared with Lemma \ref{lem_app_A_2_n=12}. Here $v_i^{(k)}$ are the same as in \eqref{eq_ansatz_unperturb_eigenfunction_1} and \eqref{eq_ansatz_unperturb_eigenfunction_2}, $v_i^{(0)}=v_i(\bm{x};\pi)$, $v_i^{(1)}=(\partial_{\kappa_1}v_i)(\bm{x};\pi)$.
\end{lemma}

\begin{lemma}[Asymptotics of $\mathcal{G}^{\delta,prop}$] \label{lem_app_C_1}
There exist $f^{prop}_{n,\delta},g^{prop}_{n,\delta}\in H_{loc}^1(\Omega)$ ($n=1,2$) such that
\footnotesize
\begin{equation} \label{eq_app_C_2}
\begin{aligned}
\mathcal{G}^{\delta,prop}(\lambda_*+\delta\cdot h)f
&=\sum_{n=1,2}\Big[\delta^{-\frac{1}{2}}k_{n}(h)
+o(\delta^{-\frac{1}{2}})\Big]v^{(0)}_{n}(\bm{x};\pi) \big(f(\bm{y}),v^{(0)}_{n}(\bm{y};\pi)\big)_{\Omega} \\
&\quad +\Big[\delta^{-\frac{1}{2}}k_{3}(h)+o(\delta^{-\frac{1}{2}}) \Big]\Big(v^{(0)}_{1}(\bm{x};\pi) \big(f(\bm{y}),v^{(0)}_{2}(\bm{y};\pi)\big)_{\Omega}+v^{(0)}_{2}(\bm{x};\pi) \big(f(\bm{y}),v^{(0)}_{1}(\bm{y};\pi)\big)_{\Omega} \Big) \\
&\quad +\big(\sum_{j+\ell =2}v_1^{(j)}(\bm{x};\pi)(f(\cdot),v_1^{(\ell)}(\cdot;\pi))_{\Omega} \big)\cdot\Big(\frac{2\delta^{\frac{1}{9}}}{\pi\gamma_*}-\frac{2}{3\pi\gamma_*}\big(\frac{2\eta_*}{\gamma_*}+N_1^{(2)}\big)\delta^{\frac{1}{3}} \Big) \\
&\quad +\big(\sum_{j+\ell =2}v_2^{(j)}(\bm{x};\pi)(f(\cdot),v_2^{(\ell)}(\cdot;\pi))_{\Omega} \big)\cdot\Big(\frac{-2\delta^{\frac{1}{9}}}{\pi\gamma_*}+\frac{2}{3\pi\gamma_*}\big(\frac{2\eta_*}{\gamma_*}+N_2^{(2)}\big)\delta^{\frac{1}{3}} \Big) \\
&\quad+ \sum_{n=1,2}\Big[\delta^{\frac{1}{2}}e_{n}(h)+o(\delta^{\frac{1}{2}}) \Big]v^{(1)}_{n}(\bm{x};\pi)\cdot \big(f(\bm{y}),v^{(1)}_{n}(\bm{y};\pi)\big)_{\Omega} \\
&\quad +\Big[\delta^{\frac{1}{2}}e_{3}(h)+o(\delta^{\frac{1}{2}}) \Big]\Big(v_{1}^{(1)}(\bm{x};\pi)\cdot \big(f(\bm{y}),v_{2}^{(1)}(\bm{y};\pi)\big)_{\Omega}
+v_{2}^{(1)}(\bm{x};\pi)\cdot \big(f(\bm{y}),v_{1}^{(1)}(\bm{y};\pi)\big)_{\Omega} \Big) \\
&\quad +\frac{2}{3\pi\gamma_*}\delta^{\frac{1}{3}} \sum_{j+\ell=4}v^{(j)}_{1}(\bm{x};\pi)\cdot \big(f(\cdot),v^{(\ell)}_{1}(\cdot;\pi)\big)_{\Omega}  -\frac{2}{3\pi\gamma_*}\delta^{\frac{1}{3}}\sum_{j+\ell=4}v^{(j)}_{2}(\bm{x};\pi)\cdot \big(f(\cdot),v^{(\ell)}_{2}(\cdot;\pi)\big)_{\Omega} \\
&\quad +\sum_{n=1,2}v^{(0)}_n(\bm{x};\pi)(f(\cdot),f^{prop}_{n,\delta}(\cdot))_{\Omega}+g^{prop}_{n,\delta}(\bm{x})(f(\cdot),v_n^{(0)}(\cdot;\pi))_{\Omega} +\mathcal{O}(\delta^{\frac{5}{9}}),
\end{aligned}
\end{equation}
\normalsize
with $\|f^{prop}_{n,\delta}\|_{H^1(U)},\|g^{prop}_{n,\delta}\|_{H^1(U)}=\mathcal{O}(\delta^{\frac{1}{9}})$ for any compact set $U\subset \Omega$. The functions $k_{n},e_{n}$ are introduced in Theorem \ref{thm_asymptotics_perturbed_Green_function}.
\end{lemma}

\begin{remark} \label{rmk_1/9_scaling}
The terms in \eqref{eq_app_C_2} involving the summations $\sum_{j+\ell =2}$ and $\sum_{j+\ell =4}$ cancel with the corresponding terms in \eqref{eq_app_C_1_n=1} and \eqref{eq_app_C_1_n=2} when summed together. These terms are larger in amplitude than $\mathcal{O}(\delta^{\frac{1}{2}})$, and hence overshadow the critical terms $v^{(1)}_{n}(\bm{x};\pi)\cdot \big(f(\bm{y}),v^{(1)}_{n}(\bm{y};\pi)\big)_{\Omega}$, which encapsulate the phase transition information and are essential for determining the interface modes (as discussed in Remark \ref{rmk_importance_expansion}). Therefore, it is necessary to compute each of these miscellaneous terms in detail and demonstrate their mutual cancellation. This accounts for the tedious calculations performed in this section.

We include all terms from equations \eqref{eq_app_C_1_ngeq3} to \eqref{eq_app_C_2} up to a remainder of order $\mathcal{O}(\delta^{\frac{5}{9}})$, ensuring that this remainder is smaller than the terms involving $v^{(1)}_{n}(\bm{x};\pi)\cdot \big(f(\bm{y}),v^{(1)}_{n}(\bm{y};\pi)\big)_{\Omega}$ and does not obscure any significant information about the phase transition. In particular, the contribution of $\mathcal{G}^{\delta,evan}$ in this remainder is determined by the domain of integrals for $n=1,2$ in equation \eqref{eq_G_delta_evan_def}. We require the domain to be sufficiently distant from $\kappa_1 = \pi$ to control the remainder. This necessity justifies our choice of the $\delta^{\frac{1}{9}}$ scaling in the integral domains.
\end{remark}

%We prove Lemma \ref{lem_app_C_1_evan_ngeq3}, Lemma \ref{lem_app_C_1_evan_n=12}, and Lemma \ref{lem_app_C_1} separately in the following sections.

\subsection{Proof of Lemma \ref{lem_app_C_1_evan_ngeq3}}
The proof is similar to Lemma \ref{lem_app_A_2_ngeq3}. We define the following projections
\begin{equation*}
\mathbb{P}_{n,\delta}(\kappa_1)g=\big(g(\cdot),v_{n,\delta}(\cdot;\overline{\kappa_1}) \big)v_{n,\delta}(\cdot;\kappa_1)\,\, (n=1,2)\quad
\mathbb{Q}_{\delta}(\kappa_1)=I-\mathbb{P}_{1,\delta}(\kappa_1)-\mathbb{P}_{2,\delta}(\kappa_1).
\end{equation*}
By Floquet transform, 
\footnotesize
\begin{equation*}
\begin{aligned}
\mathcal{G}_1^{\delta,evan}(\lambda_*+\delta\cdot h)f
&=\frac{1}{2\pi}\int_{0}^{2\pi}
\big(\lambda_*+\delta\cdot h-\mathcal{L}^{A+\delta\cdot B}(\kappa_1,\pi)\mathbb{Q}_{\delta}(\kappa_1)\big)^{-1}\mathbb{Q}_{\delta}(\kappa_1)\hat{f}(\kappa_1)d\kappa_1, \\
\frac{1}{2\pi}\int_{0}^{2\pi}
\sum_{n\geq 3}\frac{v_{n}(\bm{x};\kappa_1)(f(\bm{y}),v_{n}(\bm{y};\kappa_1))_{\Omega}}{\lambda_*-\mu_{n}(\kappa_1)}d\kappa_1 
&=\frac{1}{2\pi}\int_{0}^{2\pi}
\big(\lambda_*-\mathcal{L}^A(\kappa_1,\pi)\mathbb{Q}(\kappa_1)\big)^{-1}\mathbb{Q}(\kappa_1)\hat{f}(\kappa_1)d\kappa_1.
\end{aligned}
\end{equation*}
\normalsize
Then the estimate \eqref{eq_app_C_1_ngeq3} follows in the same way as in Lemma \ref{lem_app_A_2_ngeq3} with the help of the following estimate
\begin{equation*}
\begin{aligned}
%\big\|\mathbb{Q}_{\delta}(\kappa_1)-\mathbb{Q}(\kappa_1)\big\|_{\mathcal{B}((H^{1}(Y))^*, H^{1}(Y))^*}&=\mathcal{O}(\delta), \\
\Big\|\Big(\lambda_*-\mathcal{L}^{A+\delta\cdot B}_{\Omega,\pi}(\kappa_1)\mathbb{Q}_{\delta}(\kappa_1)\Big)^{-1} 
-\Big(\lambda_*-\mathcal{L}^{A}_{\Omega,\pi}(\kappa_1)\mathbb{Q}(\kappa_1)\Big)^{-1}\Big\|_{\mathcal{B}((H^{1}(Y))^*, H^{1}(Y))}
=\mathcal{O}(\delta).
\end{aligned}
\end{equation*}
%Details are omitted here.

\subsection{Proof of Lemma \ref{lem_app_C_1_evan_n=12}}
We only prove \eqref{eq_app_C_1_n=1}; the proof of \eqref{eq_app_C_1_n=2} is similar. The approach involves first estimating the error introduced by taking the limit $\delta\to 0$ in the integrand, and then applying the argument from the proof of Lemma \ref{lem_app_A_2_n=12}. We claim that
\footnotesize
\begin{equation}
\label{eq_app_C_10}
\begin{aligned}
&\mathcal{G}_{2,1}^{\delta,evan}(\lambda_*+\delta\cdot h)f
=\frac{1}{2\pi}\int_{[0,\pi-\delta^{\frac{1}{9}}]\cup [\pi+\delta^{\frac{1}{9}},2\pi]}\frac{v_{1}(\bm{x};\kappa_1)(f(\bm{y}),v_{1}(\bm{y};\kappa_1))_{\Omega}}{\lambda_*-\mu_{1}(\kappa_1)}d\kappa_1
+\mathcal{O}(\delta^{\frac{5}{9}}).
\end{aligned}
\end{equation}
\normalsize
Using \eqref{eq_app_C_10} and estimating the right side using \eqref{eq_first_branch_evanescent_summary}, we obtain the full expansion \eqref{eq_app_C_1_n=1}.

We now prove \eqref{eq_app_C_10}. Note that Assumption \ref{def_quadratic_degenracy} indicates that for $\delta^{\frac{1}{9}}\leq|\kappa_1-\pi|\leq \pi$
\begin{equation} \label{eq_app_C_12}
|\mu_1(\kappa_1)-\lambda_*|\geq |\mu_1(\pi+\delta^{\frac{1}{9}})-\lambda_*|
\gtrsim \delta^{\frac{2}{9}}.
\end{equation}
By a similar perturbation argument as in Theorem \ref{thm_gap_open}, the following estimates hold for $\delta^{\frac{1}{9}}\leq|\kappa_1-\pi|\leq \pi$: 
\begin{equation} \label{eq_app_C_13}
|\mu_{1,\delta}(\kappa_1)-\mu_1(\kappa_1)|=\mathcal{O}(\delta),\quad
\|v_{1,\delta}(\bm{x};\kappa_1)-v_{1}(\bm{x};\kappa_1)\|_{H^{1}(Y)}=\mathcal{O}(\delta^{\frac{7}{9}}).
\end{equation}
Combining \eqref{eq_app_C_12} and \eqref{eq_app_C_13} yield that 
\begin{equation} \label{eq_app_C_15}
\frac{1}{\lambda_*+\delta\cdot h-\mu_{1,\delta}(\kappa_1)}=\mathcal{O}(\delta^{-\frac{2}{9}}),\quad \forall h\in\mathcal{J}.
\end{equation}
By \eqref{eq_app_C_15} and \eqref{eq_app_C_13}, one can check that for $\delta^{\frac{1}{9}}\leq|\kappa_1-\pi|\leq \pi$
\begin{equation*}
\begin{aligned}
\frac{v_{1,  \delta}(\bm{x};\kappa_1)(f(\bm{y}),v_{1,\delta}(\bm{y};\kappa_1))_{\Omega}}{\lambda_*+\delta\cdot h-\mu_{1,  \delta}(\kappa_1)} 
-\frac{v_{1}(\bm{x};\kappa_1)(f(\bm{y}),v_{1}(\bm{y};\kappa_1))_{\Omega}}{\lambda_*-\mu_{1}(\kappa_1)}
=\mathcal{O}(\delta^{\frac{5}{9}}),
\end{aligned}
\end{equation*}
whence \eqref{eq_app_C_10} follows.

\subsection{Proof of Lemma \ref{lem_app_C_1}}

The proof requires dedicated estimates of integrals using the asymptotics in Theorem \ref{thm_gap_open}. Note that
\begin{equation*}
\begin{aligned}
\mathcal{G}^{\delta,prop}(\lambda_*+\delta\cdot h)f
&=\frac{1}{2\pi}\int_{\pi-\delta^{\frac{1}{9}}}^{\pi+\delta^{\frac{1}{9}}}
\frac{v_{1,  \delta}(\bm{x};\kappa_1)(f(\bm{y}),v_{1,\delta}(\bm{y};\kappa_1))_{\Omega})}{\lambda_*+\delta\cdot h-\mu_{1,\delta}(\kappa_1)}\frac{1}{N^2_{1,\delta}(\kappa_1)}d\kappa_1 \\
&\quad +\frac{1}{2\pi}\int_{\pi-\delta^{\frac{1}{9}}}^{\pi+\delta^{\frac{1}{9}}}
\frac{v_{2,  \delta}(\bm{x};\kappa_1)(f(\bm{y}),v_{2,\delta}(\bm{y};\kappa_1))_{\Omega}}{\lambda_*+\delta\cdot h-\mu_{2,\delta}(\kappa_1)}\frac{1}{N^2_{2,\delta}(\kappa_1)}d\kappa_1 ,
\end{aligned}
\end{equation*}
where $v_{n,\delta}(\bm{x};\kappa_1)$ and their normalization factors $N_{n,\delta}(\kappa_1)$ admit the asymptotic expansions as given in Theorem \ref{thm_gap_open}. A direct calculation yields 
\footnotesize
\begin{equation*}
\begin{aligned}
&(\lambda_*+\delta\cdot h-\mu_{1,\delta}(\kappa_1))N^{2}_{1,\delta}(\kappa_1) 
=\Big(\delta\cdot h+\sqrt{\frac{1}{4}\gamma_*^2(\kappa_1-\pi)^4+t_*^2\delta^2}\Big)\Big[1+\sum_{n=1}^{3}w_1^{(n)}(\kappa_1;\delta)\cdot (\kappa_1-\pi)^n+\mathcal{O}((\kappa_1-\pi)^{4}+|\delta|)\Big],
\end{aligned}
\end{equation*}
\begin{equation*}
\begin{aligned}
&(\lambda_*+\delta\cdot h-\mu_{2,\delta}(\kappa_1))N^{2}_{2,\delta}(\kappa_1) 
=\Big(\delta\cdot h-\sqrt{\frac{1}{4}\gamma_*^2(\kappa_1-\pi)^4+t_*^2\delta^2}\Big)\Big[1+\sum_{n=1}^{3}w_2^{(n)}(\kappa_1;\delta)\cdot (\kappa_1-\pi)^n+\mathcal{O}((\kappa_1-\pi)^{4}+|\delta|)\Big],
\end{aligned}
\end{equation*}
where
\footnotesize
\begin{equation*}
\begin{aligned}
&w_k^{(1)}(\kappa_1;\delta)=\frac{f_*^2(\kappa_1;\delta)\text{Re}(r(\kappa_1-\pi))}{2(1+f_*^2(\kappa_1;\delta))}+ \\
&\frac{2(-1)^{k-1}t_*\delta^2 \text{Re}(r_*)}{\big(\delta\cdot h+(-1)^{k-1}\sqrt{\frac{1}{4}\gamma_*^2(\kappa_1-\pi)^4+t_*^2\delta^2}\big) \big(\sqrt{\big(\frac{\gamma_*}{2}(\kappa_1-\pi)^2+\eta_* (\kappa_1-\pi)^4\big)^2+\delta^2|t_*+r_*(\kappa_1-\pi)|^2} +\sqrt{\frac{1}{4}\gamma_*^2(\kappa_1-\pi)^4+t_*^2\delta^2} \big)},
\end{aligned}
\end{equation*}

\begin{equation*}
\begin{aligned}
&w_k^{(2)}(\kappa_1;\delta)=\frac{N_{k}^{(2)}}{1+f_*^2(\kappa_1;\delta)}+ \\
&\frac{(-1)^{k-1}\gamma_*\eta_* (\kappa_1-\pi)^4}{\big(\delta\cdot h+(-1)^{k-1}\sqrt{\frac{1}{4}\gamma_*^2(\kappa_1-\pi)^4+t_*^2\delta^2}\big) \big(\sqrt{\big(\frac{\gamma_*}{2}(\kappa_1-\pi)^2+\eta_* (\kappa_1-\pi)^4\big)^2+\delta^2|t_*+r_*(\kappa_1-\pi)|^2} +\sqrt{\frac{1}{4}\gamma_*^2(\kappa_1-\pi)^4+t_*^2\delta^2} \big)},\\
&w_k^{(3)}(\kappa_1;\delta) =\frac{N_{k}^{(3)}}{1+f_*^2(\kappa_1;\delta)}.
\end{aligned}
\end{equation*}
\normalsize
Hence $\mathcal{G}^{\delta,prop}(\lambda_*+\delta\cdot h)f$ can be decomposed as
\begin{equation*}
\mathcal{G}^{\delta,prop}(\lambda_*+\delta\cdot h)f
=\sum_{k=0}^{4}I_{k},
\end{equation*}
where
\footnotesize
\begin{equation}
\label{eq_app_C_16}
\begin{aligned}
&I_{0}=\frac{1}{2\pi}\int_{\pi-\delta^{\frac{1}{9}}}^{\pi+\delta^{\frac{1}{9}}}
\frac{v_{1,  \delta}(\bm{x};\kappa_1)(f(\bm{y}),v_{1,\delta}(\bm{y};\kappa_1))_{\Omega}}{\delta\cdot h+\sqrt{\frac{1}{4}\gamma_*^2(\kappa_1-\pi)^4+t_*^2\delta^2}}\frac{1}{1+f_*^2(\kappa_1;\delta)}d\kappa_1  \\
&\quad +\frac{1}{2\pi}\int_{\pi-\delta^{\frac{1}{9}}}^{\pi+\delta^{\frac{1}{9}}}
\frac{v_{2,  \delta}(\bm{x};\kappa_1)(f(\bm{y}),v_{2,\delta}(\bm{y};\kappa_1))_{\Omega}}{\delta\cdot h-\sqrt{\frac{1}{4}\gamma_*^2(\kappa_1-\pi)^4+t_*^2\delta^2}}\frac{1}{1+f_*^2(\kappa_1;\delta)}d\kappa_1,
\end{aligned}
\end{equation}

\begin{equation*}
\begin{aligned}
&I_{k}=\frac{1}{2\pi}\int_{\pi-\delta^{\frac{1}{9}}}^{\pi+\delta^{\frac{1}{9}}}
\frac{v_{1,  \delta}(\bm{x};\kappa_1)(f(\bm{y}),v_{1,\delta}(\bm{y};\kappa_1))_{\Omega}}{\delta\cdot h+\sqrt{\frac{1}{4}\gamma_*^2(\kappa_1-\pi)^4+t_*^2\delta^2}}\frac{w_1^{(k)}(\kappa_1;\delta)(\kappa_1-\pi)^{k}}{1+f_*^2(\kappa_1;\delta)}d\kappa_1  \\
&\quad +\frac{1}{2\pi}\int_{\pi-\delta^{\frac{1}{9}}}^{\pi+\delta^{\frac{1}{9}}}
\frac{v_{2,  \delta}(\bm{x};\kappa_1)(f(\bm{y}),v_{2,\delta}(\bm{y};\kappa_1))_{\Omega}}{\delta\cdot h-\sqrt{\frac{1}{4}\gamma_*^2(\kappa_1-\pi)^4+t_*^2\delta^2}}\frac{w_2^{(k)}(\kappa_1;\delta)(\kappa_1-\pi)^{k}}{1+f_*^2(\kappa_1;\delta)}d\kappa_1,
\end{aligned}
\end{equation*}
\normalsize
for $k=1,2,3$, and
\footnotesize
\begin{equation*}
\begin{aligned}
&I_{4}=\frac{1}{2\pi}\int_{\pi-\delta^{\frac{1}{9}}}^{\pi+\delta^{\frac{1}{9}}}
\frac{v_{1,  \delta}(\bm{x};\kappa_1)(f(\bm{y}),v_{1,\delta}(\bm{y};\kappa_1))_{\Omega}}{\delta\cdot h+\sqrt{\frac{1}{4}\gamma_*^2(\kappa_1-\pi)^4+t_*^2\delta^2}}\frac{\mathcal{O}(|\kappa_1-\pi|^4+\delta)}{1+f_*^2(\kappa_1;\delta)}d\kappa_1  \\
&\quad +\frac{1}{2\pi}\int_{\pi-\delta^{\frac{1}{9}}}^{\pi+\delta^{\frac{1}{9}}}
\frac{v_{2,  \delta}(\bm{x};\kappa_1)(f(\bm{y}),v_{2,\delta}(\bm{y};\kappa_1))_{\Omega}}{\delta\cdot h-\sqrt{\frac{1}{4}\gamma_*^2(\kappa_1-\pi)^4+t_*^2\delta^2}}\frac{\mathcal{O}(|\kappa_1-\pi|^4+\delta)}{1+f_*^2(\kappa_1;\delta)}d\kappa_1 .
\end{aligned}
\end{equation*}
\normalsize
In Section 8.4.1, we estimate $I_0$, the main part of $\mathcal{G}^{\delta,prop}(\lambda_*+\delta\cdot h)f$. We show that it contributes to most of the major terms in \eqref{eq_app_C_2}; see \eqref{eq_app_C_I0_summary}. In Section 8.4.2, we estimate $I_1, I_2, I_3$ and $I_4$. We show that they only contribute to two major terms, as indicated in \eqref{eq_app_C_I134_summary} and \eqref{eq_app_C_I2_summary}. Then the proof of \eqref{eq_app_C_2} is completed by combining \eqref{eq_app_C_I0_summary}, \eqref{eq_app_C_I134_summary} and \eqref{eq_app_C_I2_summary}.

\subsubsection{Estimate of $I_0$}
In this section, we prove that
\footnotesize
\begin{equation} \label{eq_app_C_I0_summary}
\begin{aligned}
I_0&=
\sum_{n=1,2}\Big[\delta^{-\frac{1}{2}}k_{n}(h)
+o(\delta^{-\frac{1}{2}})\Big]v^{(0)}_{n}(\bm{x};\pi) \big(f(\bm{y}),v^{(0)}_{n}(\bm{y};\pi)\big)_{\Omega} \\
&\quad +\Big[\delta^{-\frac{1}{2}}k_{3}(h)+o(\delta^{-\frac{1}{2}}) \Big]\Big(v^{(0)}_{1}(\bm{x};\pi) \big(f(\bm{y}),v^{(0)}_{2}(\bm{y};\pi)\big)_{\Omega}+v^{(0)}_{2}(\bm{x};\pi) \big(f(\bm{y}),v^{(0)}_{1}(\bm{y};\pi)\big)_{\Omega} \Big) \\
&\quad+ \frac{2\delta^{\frac{1}{9}}}{\pi\gamma_*}\sum_{j+\ell =2}v_1^{(j)}(\bm{x};\pi)(f(\cdot),v_1^{(\ell)}(\bm{x};\pi))_{\Omega}
+\frac{-2\delta^{\frac{1}{9}}}{\pi\gamma_*}\sum_{j+\ell =2}v_2^{(j)}(\bm{x};\pi)(f(\cdot),v_2^{(\ell)}(\bm{x};\pi))_{\Omega} \\
&\quad+ \sum_{n=1,2}\Big[\delta^{\frac{1}{2}}e_{n}(h)+o(\delta^{\frac{1}{2}}) \Big]v^{(1)}_{n}(\bm{x};\pi)\cdot \big(f(\bm{y}),v^{(1)}_{n}(\bm{y};\pi)\big)_{\Omega} \\
&\quad +\Big[\delta^{\frac{1}{2}}e_{3}(h)+o(\delta^{\frac{1}{2}}) \Big]\Big(v_{1}^{(1)}(\bm{x};\pi)\cdot \big(f(\bm{y}),v_{2}^{(1)}(\bm{y};\pi)\big)_{\Omega}
+v_{2}^{(1)}(\bm{x};\pi)\cdot \big(f(\bm{y}),v_{1}^{(1)}(\bm{y};\pi)\big)_{\Omega} \Big) \\
&\quad +\frac{2}{3\pi\gamma_*}\delta^{\frac{1}{3}}\sum_{j+\ell=4}v^{(j)}_{1}(\bm{x};\pi)\cdot \big(f(\bm{y}),v^{(\ell)}_{1}(\bm{y};\pi)\big)_{\Omega}  -\frac{2}{3\pi\gamma_*}\delta^{\frac{1}{3}}\sum_{j+\ell=4}v^{(j)}_{2}(\bm{x};\pi)\cdot \big(f(\bm{y}),v^{(\ell)}_{2}(\bm{y};\pi)\big)_{\Omega} \\
&\quad +\sum_{\substack{n=1,2 \\ k=1,2,3,4}}v^{(0)}_n(\bm{x};\pi)(f(\bm{y}),f^{prop,(k)}_{n,\delta}(\bm{y}))_{\Omega}+g^{prop,(k)}_{n,\delta}(\bm{x})(f(\bm{y}),v_n^{(0)}(\bm{y};\pi))_{\Omega} +\mathcal{O}(\delta^{\frac{5}{9}})
\end{aligned}
\end{equation}
\normalsize
with $\|f^{prop,(k)}_{n,\delta}\|,\|g^{prop,(k)}_{n,\delta}\|=\mathcal{O}(\delta^{\frac{1}{9}})$ for $1\leq k\leq 4$. Here $v_i^{(k)}$ are the same as in \eqref{eq_ansatz_unperturb_eigenfunction_1} and \eqref{eq_ansatz_unperturb_eigenfunction_2}, $v_i^{(0)}=v_i(\bm{x};\pi)$, $v_i^{(1)}=(\partial_{\kappa_1}v_i)(\bm{x};\pi)$. This is achieved by further decomposing $I_0$ into several parts and estimating each. 

By the asymptotic expansion \eqref{eq_asymp_eigenfunction_1}, 
\begin{equation} \label{eq_app_C_split_numerator}
\begin{aligned}
v_{1,\delta}(\bm{x};\kappa_1)&=\sum_{k=0}^{5}(\kappa_1-\pi)^{k}v_{1}^{(k)}(\bm{x};\pi)+f_*(\kappa_1;\delta)\sum_{k=0}^{1}(\kappa_1-\pi)^{k}v_{1,t}^{(k)}(\bm{x};\kappa_1)+R_1(\bm{x};\kappa_1,\delta), \\
v_{2,\delta}(\bm{x};\kappa_1)&=\sum_{k=0}^{5}(\kappa_1-\pi)^{k}v_{2}^{(k)}(\bm{x};\pi)+f_*(\kappa_1;\delta)\sum_{k=0}^{1}(\kappa_1-\pi)^{k}v_{2,t}^{(k)}(\bm{x};\kappa_1)+R_2(\bm{x};\kappa_1,\delta),
\end{aligned}
\end{equation}
where 
\begin{equation*}
\begin{aligned}
v_{1,t}^{(0)}(\bm{x};\kappa_1)&=\big(1+ r(\kappa_1-\pi;\delta)\cdot (\kappa_1-\pi)\big) v_2(\bm{x};\pi),\quad v_{1,t}^{(1)}(\bm{x};\kappa_1)=\partial_{\kappa_1}v_2(\bm{x};\pi), \\
v_{2,t}^{(0)}(\bm{x};\kappa_1)&=-\big(1+ \overline{r(\kappa_1-\pi;\delta)}\cdot (\kappa_1-\pi)\big) v_1(\bm{x};\pi),\quad v_{2,t}^{(1)}(\bm{x};\kappa_1)=-\partial_{\kappa_1}v_1(\bm{x};\pi),
\end{aligned}
\end{equation*}
are the phase transition terms induced by the perturbation. Note that the leading terms in $v_{i,t}^{(0)}$ are $\kappa_1-$independent, and $v_{i,t}^{(1)}$ are all $\kappa_1-$independent. We denote $v_{i,t}^{(k)}=0$ for $k\geq 2$ and $i=1,2$ to align the notations. With \eqref{eq_app_C_split_numerator}, we further decompose $I_0$ by arranging the numerator in \eqref{eq_app_C_16} in ascending order. Specifically,
\begin{equation*}
    I_0=\sum_{k=0}^{7}I_{0}^{(k)},
\end{equation*}
with
\footnotesize
\begin{equation*}
\begin{aligned}
I_{0}^{(k)}&:=\frac{1}{2\pi}\int_{\pi-\delta^{\frac{1}{9}}}^{\pi+\delta^{\frac{1}{9}}}
\frac{\sum_{j+\ell=k}\Big(v^{(j)}_{1}(\bm{x};\pi)\cdot \big(f(\bm{y}),v^{(\ell)}_{1}(\bm{y};\pi)\big)_{\Omega}
+f^2_{*}(\kappa_1)v^{(j)}_{1,t}(\bm{x};\kappa_1)\cdot \big(f(\bm{y}),v^{(\ell)}_{1,t}(\bm{y};\kappa_1)\big)_{\Omega}\Big)}{\delta\cdot h+\sqrt{\frac{1}{4}\gamma_*^2(\kappa_1-\pi)^4+t_*^2\delta^2}}\frac{(\kappa_1-\pi)^{k}d\kappa_1}{1+f^2_{*}(\kappa_1)}, \\
&\quad +\frac{1}{2\pi}\int_{\pi-\delta^{\frac{1}{9}}}^{\pi+\delta^{\frac{1}{9}}}
\frac{\sum_{j+\ell=k}\Big(v^{(j)}_{2}(\bm{x};\pi)\cdot \big(f(\bm{y}),v^{(\ell)}_{2}(\bm{y};\pi)\big)_{\Omega}
+f^2_{*}(\kappa_1)v^{(j)}_{2,t}(\bm{x};\kappa_1)\cdot \big(f(\bm{y}),v^{(\ell)}_{2,t}(\bm{y};\kappa_1)\big)_{\Omega}\Big)}{\delta\cdot h-\sqrt{\frac{1}{4}\gamma_*^2(\kappa_1-\pi)^4+t_*^2\delta^2}}\frac{(\kappa_1-\pi)^{k}d\kappa_1}{1+f^2_{*}(\kappa_1)} \\
&\quad +\frac{1}{2\pi}\int_{\pi-\delta^{\frac{1}{9}}}^{\pi+\delta^{\frac{1}{9}}}
\frac{\sum_{j+\ell=k}\Big(v_{1}^{(j)}(\bm{x};\pi)\cdot \big(f(\bm{y}),v_{1,t}^{(\ell)}(\bm{y};\kappa_1)\big)_{\Omega}
+v_{1,t}^{(j)}(\bm{x};\kappa_1)\cdot \big(f(\bm{y}),v_{1}^{(\ell)}(\bm{y};\pi)\big)_{\Omega}\Big)}{\delta\cdot h+\sqrt{\frac{1}{4}\gamma_*^2(\kappa_1-\pi)^4+t_*^2\delta^2}}\frac{(\kappa_1-\pi)^{k}f_{*}(\kappa_1)d\kappa_1}{1+f^2_{*}(\kappa_1)}  \\
&\quad +\frac{1}{2\pi}\int_{\pi-\delta^{\frac{1}{9}}}^{\pi+\delta^{\frac{1}{9}}}
\frac{\sum_{j+\ell=k}\Big(v_{2}^{(j)}(\bm{x};\pi)\cdot \big(f(\bm{y}),v_{2,t}^{(\ell)}(\bm{y};\kappa_1)\big)_{\Omega}
+v_{2,t}^{(j)}(\bm{x};\kappa_1)\cdot \big(f(\bm{y}),v_{2}^{(\ell)}(\bm{y};\kappa_1)\big)_{\Omega}\Big)}{\delta\cdot h-\sqrt{\frac{1}{4}\gamma_*^2(\kappa_1-\pi)^4+t_*^2\delta^2}}\frac{(\kappa_1-\pi)^{k}f_{*}(\kappa_1)d\kappa_1}{1+f^2_{*}(\kappa_1)}
\end{aligned}
\end{equation*}
\normalsize
for $k=0,1,2,3,4,5,6$, and
\footnotesize
\begin{equation*}
\begin{aligned}
I_0^{(7)}&:=\frac{1}{2\pi}\int_{\pi-\delta^{\frac{1}{9}}}^{\pi+\delta^{\frac{1}{9}}}
\frac{v_{1,\delta}(\bm{x};\kappa_1)\cdot \big(f(\bm{y}),R_{1}(\bm{y};\kappa_1,\delta)\big)_{\Omega}
+R_{1}(\bm{x};\kappa_1,\delta)\cdot \big(f(\bm{y}),v_{1,\delta}(\bm{y};\kappa_1)\big)_{\Omega}}{\delta\cdot h+\sqrt{\frac{1}{4}\gamma_*^2(\kappa_1-\pi)^4+t_*^2\delta^2}}\frac{d\kappa_1}{1+f^2_{*}(\kappa_1)} \\
&\quad +\frac{1}{2\pi}\int_{\pi-\delta^{\frac{1}{9}}}^{\pi+\delta^{\frac{1}{9}}}
\frac{v_{2,\delta}(\bm{x};\kappa_1)\cdot \big(f(\bm{y}),R_{2}(\bm{y};\kappa_1,\delta)\big)_{\Omega}
+R_{2}(\bm{x};\kappa_1,\delta)\cdot \big(f(\bm{y}),v_{2,\delta}(\bm{y};\kappa_1)\big)_{\Omega}}{\delta\cdot h-\sqrt{\frac{1}{4}\gamma_*^2(\kappa_1-\pi)^4+t_*^2\delta^2}}\frac{d\kappa_1}{1+f^2_{*}(\kappa_1)}.
\end{aligned}
\end{equation*}
\normalsize

In the following, we calculate the asymptotics of $I_0^{(k)}$ in Step 1 to 4, respectively.

Before we calculate, we observe that the expressions of $I_0^{(k)}$ have a common structure. Specifically, each integral within $I_0^{(k)}$ can be expressed as the product of a vector component and a scalar integral in $\kappa_1$. The vector component, such as 
$v^{(0)}_{1}(\bm{x};\pi)\cdot \big(f(\bm{y}),v^{(0)}_{1}(\bm{x};\pi)\big)_{\Omega}$ is $\kappa_1$-independent. Terms involving  $v^{(k)}_{i,t}(\cdot;\kappa_1)$ maintain this form, modulo smaller terms.  The scalar integral in $\kappa_1$ of the following form
\begin{equation} \label{eq_center_scalar_integral}
\frac{1}{2\pi}\int_{\pi-\delta^{\frac{1}{9}}}^{\pi+\delta^{\frac{1}{9}}}
\frac{(\kappa_1-\pi)^{k}f^{\ell}_{*}(\kappa_1;\delta)}{\delta\cdot h\pm\sqrt{\frac{1}{4}\gamma_*^2(\kappa_1-\pi)^4+t_*^2\delta^2}}\frac{d\kappa_1}{1+f^2_{*}(\kappa_1;\delta)},
\end{equation}
where $k=0,2,4,6$ and $\ell=0,1,2$. It's important to estimate these elementary integrals to obtain the asymptotics of $I_{0}^{(k)}$. We have the following results whose proofs are supplemented in Appendix B.

%\begin{lemma} \label{lem_app_C_f_*_property}
%Let $h\in \overline{\mathcal{J}}$. The function $f_{*}(\kappa_1;\delta)$ and $r(\kappa_1-\pi;\delta)$ (defined in \eqref{eq_sec4_28}) are uniformly bounded for all $\delta> 0$ and $\kappa_1\in\mathbf{R}$. Moreover, $f_{*}(\sqrt{\delta}p+\pi;\delta)=\mathcal{O}(p^{-2})$, $r(\sqrt{\delta}p;\delta)=\overline{r}_*+\mathcal{O}(p^{-4})$ as $p\to\infty$.
%\end{lemma}
\begin{lemma} \label{lem_app_C_scalar_integrals_property}
Let $h\in \overline{\mathcal{J}}$. 
\begin{enumerate}
\item 
The functions $f_{*}(\kappa_1;\delta)$ and $r(\kappa_1-\pi;\delta)$ (see \eqref{eq_sec4_28}) are uniformly bounded for all $\delta> 0$ and $\kappa_1\in\mathbf{R}$. Moreover, $f_{*}(\sqrt{\delta}p+\pi;\delta)=\mathcal{O}(p^{-2})$, $r(\sqrt{\delta}p;\delta)=\overline{r}_*+\mathcal{O}(p^{-4})$ as $p\to\infty$.

\item
When $k,\ell \in \mathbf{N}$, \eqref{eq_center_scalar_integral} has the asymptotics
\footnotesize
\begin{equation} \label{eq_app_C_scalar_integrals_property_1}
\frac{1}{2\pi}\int_{\pi-\delta^{\frac{1}{9}}}^{\pi+\delta^{\frac{1}{9}}}
\frac{(\kappa_1-\pi)^{k}f^{\ell}_{*}(\kappa_1;\delta)}{\delta\cdot h\pm\sqrt{\frac{1}{4}\gamma_*^2(\kappa_1-\pi)^4+t_*^2\delta^2}}\frac{d\kappa_1}{1+f^2_{*}(\kappa_1;\delta)}
=\left\{
\begin{aligned}
&\mathcal{O}(\delta^{\frac{k-1}{2}}),\quad \text{for $k-2\ell\leq 0$}, \\
&\mathcal{O}(\delta^{\frac{k-1}{2}}\log\delta),\quad \text{for $k-2\ell=1$}, \\
&\mathcal{O}(\delta^{\frac{k-1}{9}+\frac{7\ell}{9}}),\quad \text{for $k-2\ell\geq 2$}. \\
\end{aligned}
\right.
\end{equation}
\normalsize
\item
\footnotesize
\begin{equation} \label{eq_app_C_scalar_integrals_property_2}
\begin{aligned}
&\frac{1}{2\pi}\int_{\pi-\delta^{\frac{1}{9}}}^{\pi+\delta^{\frac{1}{9}}}
\Big(\frac{1}{\delta\cdot h\pm\sqrt{\frac{1}{4}\gamma_*^2(\kappa_1-\pi)^4+t_*^2\delta^2}}+\frac{f^2_{*}(\kappa_1;\delta)}{\delta\cdot h\mp\sqrt{\frac{1}{4}\gamma_*^2(\kappa_1-\pi)^4+t_*^2\delta^2}} \Big)\frac{1}{1+f^2_{*}(\kappa_1;\delta)}d\kappa_1 \\
&=-\frac{\delta^{-\frac{1}{2}}}{2\gamma_*}\frac{\frac{h}{\gamma_*}\mp\sqrt{(\frac{t_*}{\gamma_*})^2-(\frac{h}{\gamma_*})^2}}{((\frac{t_*}{\gamma_*})^2-(\frac{h}{\gamma_*})^2)^{\frac{3}{4}}}+o(\delta^{-\frac{1}{2}}).
\end{aligned}
\end{equation}
\item 
\begin{equation} \label{eq_app_C_scalar_integrals_property_3}
\begin{aligned}
&\frac{1}{2\pi}\int_{\pi-\delta^{\frac{1}{9}}}^{\pi+\delta^{\frac{1}{9}}}
\Big(\frac{f_{*}(\kappa_1;\delta)}{\delta\cdot h+\sqrt{\frac{1}{4}\gamma_*^2(\kappa_1-\pi)^4+t_*^2\delta^2}}+\frac{-f_{*}(\kappa_1;\delta)}{\delta\cdot h-\sqrt{\frac{1}{4}\gamma_*^2(\kappa_1-\pi)^4+t_*^2\delta^2}} \Big)\frac{1}{1+f^2_{*}(\kappa_1;\delta)}d\kappa_1 \\
&=\frac{\delta^{-\frac{1}{2}}}{2\gamma_*}
\frac{\frac{t_*}{\gamma_*}}{((\frac{t_*}{\gamma_*})^2-(\frac{h}{\gamma_*})^2)^{\frac{3}{4}}}+o(\delta^{-\frac{1}{2}}).
\end{aligned}
\end{equation}
\normalsize
\item
\footnotesize
\begin{equation} \label{eq_app_C_scalar_integrals_property_6}
\begin{aligned}
\frac{1}{2\pi}\int_{\pi-\delta^{\frac{1}{9}}}^{\pi+\delta^{\frac{1}{9}}}
\frac{1}{\delta\cdot h\pm\sqrt{\frac{1}{4}\gamma_*^2(\kappa_1-\pi)^4+t_*^2\delta^2}}\frac{(\kappa_1-\pi)^2}{1+f^2_{*}(\kappa_1;\delta)}d\kappa_1 
=\pm\frac{2}{\pi\gamma_*}\delta^{\frac{1}{9}}+\mathcal{O}(\delta^{\frac{1}{2}}).
\end{aligned}
\end{equation}
\item
\begin{equation} \label{eq_app_C_scalar_integrals_property_7}
\begin{aligned}
\frac{1}{2\pi}\int_{\pi-\delta^{\frac{1}{9}}}^{\pi+\delta^{\frac{1}{9}}}
\frac{1}{\delta\cdot h\pm\sqrt{\frac{1}{4}\gamma_*^2(\kappa_1-\pi)^4+t_*^2\delta^2}}\frac{(\kappa_1-\pi)^4}{1+f^2_{*}(\kappa_1;\delta)}d\kappa_1 
=\pm\frac{2}{3\pi\gamma_*}\delta^{\frac{1}{3}}+\mathcal{O}(\delta^{\frac{10}{9}}).
\end{aligned}
\end{equation}
\item
\begin{equation} \label{eq_app_C_scalar_integrals_property_4}
\begin{aligned}
&\frac{1}{2\pi}\int_{\pi-\delta^{\frac{1}{9}}}^{\pi+\delta^{\frac{1}{9}}}
\Big(\frac{1}{\delta\cdot h\pm\sqrt{\frac{1}{4}\gamma_*^2(\kappa_1-\pi)^4+t_*^2\delta^2}}+\frac{f^2_{*}(\kappa_1;\delta)}{\delta\cdot h\mp\sqrt{\frac{1}{4}\gamma_*^2(\kappa_1-\pi)^4+t_*^2\delta^2}} \Big)\frac{(\kappa_1-\pi)^2}{1+f^2_{*}(\kappa_1;\delta)}d\kappa_1 \\
&=\pm\frac{2}{\pi\gamma_*}\delta^{\frac{1}{9}}
-\frac{\delta^{\frac{1}{2}}}{\gamma_*}\frac{\frac{h}{\gamma_*}\pm\sqrt{(\frac{t_*}{\gamma_*})^2-(\frac{h}{\gamma_*})^2}}{((\frac{t_*}{\gamma_*})^2-(\frac{h}{\gamma_*})^2)^{\frac{1}{4}}}+o(\delta^{\frac{1}{2}}).
\end{aligned}
\end{equation}
\item
\begin{equation} \label{eq_app_C_scalar_integrals_property_5}
\begin{aligned}
&\frac{1}{2\pi}\int_{\pi-\delta^{\frac{1}{9}}}^{\pi+\delta^{\frac{1}{9}}}
\Big(\frac{f_{*}(\kappa_1;\delta)}{\delta\cdot h+\sqrt{\frac{1}{4}\gamma_*^2(\kappa_1-\pi)^4+t_*^2\delta^2}}+\frac{-f_{*}(\kappa_1;\delta)}{\delta\cdot h-\sqrt{\frac{1}{4}\gamma_*^2(\kappa_1-\pi)^4+t_*^2\delta^2}} \Big)\frac{(\kappa_1-\pi)^2}{1+f^2_{*}(\kappa_1;\delta)}d\kappa_1 \\
&=\frac{\delta^{\frac{1}{2}}}{\gamma_*}
\frac{\frac{t_*}{\gamma_*}}{((\frac{t_*}{\gamma_*})^2-(\frac{h}{\gamma_*})^2)^{\frac{1}{4}}}+o(\delta^{\frac{1}{2}}).
\end{aligned}
\end{equation}
\end{enumerate}
\normalsize
\end{lemma}
With these preparations, we are ready to estimate $I_0^{(k)}$ for $0\leq k \leq 7$. 

{\color{blue}Step 1: Estimate of $I_{0}^{(0)}$}. Note that in the contributions to $I_{0}^{(0)}$ from the $v_{i,t}^{(0)}$-terms, the integrals that involve $r(\kappa_1-\pi;\delta)\cdot(\kappa_1-\pi)$ are integrals of odd functions in $\kappa_1-\pi$, and hence equal to zero. Therefore
\footnotesize
\begin{equation*}
\begin{aligned}
I_0^{(0)}&=\frac{1}{2\pi}\int_{\pi-\delta^{\frac{1}{9}}}^{\pi+\delta^{\frac{1}{9}}}
\frac{v^{(0)}_{1}(\bm{x};\pi)\cdot \big(f(\bm{y}),v^{(0)}_{1}(\bm{x};\pi)\big)_{\Omega}
+f^2_{*}(\kappa_1)v^{(0)}_{2}(\bm{x};\pi)\cdot \big(f(\bm{y}),v^{(0)}_{2}(\bm{y};\pi)\big)_{\Omega}}{\delta\cdot h+\sqrt{\frac{1}{4}\gamma_*^2(\kappa_1-\pi)^4+t_*^2\delta^2}}\frac{d\kappa_1}{1+f^2_{*}(\kappa_1)}  \\
&\quad +\frac{1}{2\pi}\int_{\pi-\delta^{\frac{1}{9}}}^{\pi+\delta^{\frac{1}{9}}}
\frac{v^{(0)}_{2}(\bm{x};\pi)\cdot \big(f(\bm{y}),v^{(0)}_{2}(\bm{y};\pi)\big)_{\Omega}
+f^2_{*}(\kappa_1)v^{(0)}_{1}(\bm{x};\pi)\cdot \big(f(\bm{y}),v^{(0)}_{1}(\bm{y};\pi)\big)_{\Omega}}{\delta\cdot h-\sqrt{\frac{1}{4}\gamma_*^2(\kappa_1-\pi)^4+t_*^2\delta^2}}\frac{d\kappa_1}{1+f^2_{*}(\kappa_1)} \\
&\quad +\frac{1}{2\pi}\int_{\pi-\delta^{\frac{1}{9}}}^{\pi+\delta^{\frac{1}{9}}}
\frac{v_{1}^{(0)}(\bm{x};\pi)\cdot \big(f(\bm{y}),v_{2}^{(0)}(\bm{y};\pi)\big)_{\Omega}
+v_{2}^{(0)}(\bm{x};\pi)\cdot \big(f(\bm{y}),v_{1}^{(0)}(\bm{y};\pi)\big)_{\Omega}}{\delta\cdot h+\sqrt{\frac{1}{4}\gamma_*^2(\kappa_1-\pi)^4+t_*^2\delta^2}}\frac{f_{*}(\kappa_1)d\kappa_1}{1+f^2_{*}(\kappa_1)}  \\
&\quad -\frac{1}{2\pi}\int_{\pi-\delta^{\frac{1}{9}}}^{\pi+\delta^{\frac{1}{9}}}
\frac{v_{2}^{(0)}(\bm{x};\pi)\cdot \big(f(\bm{y}),v_{1}^{(0)}(\bm{y};\pi)\big)_{\Omega}
+v_{1}^{(0)}(\bm{x};\pi)\cdot \big(f(\bm{y}),v_{2}^{(0)}(\bm{y};\pi)\big)_{\Omega}}{\delta\cdot h-\sqrt{\frac{1}{4}\gamma_*^2(\kappa_1-\pi)^4+t_*^2\delta^2}}\frac{f_{*}(\kappa_1)d\kappa_1}{1+f^2_{*}(\kappa_1)}.
\end{aligned}
\end{equation*}
\normalsize
By \eqref{eq_app_C_scalar_integrals_property_2}, we obtain
\footnotesize
\begin{equation} \label{eq_app_C_I00_summary}
\begin{aligned}
I_0^{(0)}&=
\sum_{n=1,2}\Big[\delta^{-\frac{1}{2}}k_{n}(h)
+o(\delta^{-\frac{1}{2}})\Big]v^{(0)}_{n}(\bm{x};\pi) \big(f(\bm{y}),v^{(0)}_{n}(\bm{y};\pi)\big)_{\Omega} \\
&\quad +\Big[\delta^{-\frac{1}{2}}k_{3}(h)+o(\delta^{-\frac{1}{2}}) \Big]\Big(v^{(0)}_{1}(\bm{x};\pi) \big(f(\bm{y}),v^{(0)}_{2}(\bm{x};\pi)\big)_{\Omega}+v^{(0)}_{2}(\bm{x};\pi) \big(f(\bm{y}),v^{(0)}_{1}(\bm{x};\pi)\big)_{\Omega} \Big).
\end{aligned}
\end{equation}
\normalsize
Hence $I_{0}^{(0)}$ gives the first three major terms in \eqref{eq_app_C_I0_summary}.

{\color{blue}Step 2: Estimate of $I_{0}^{(1)}$.} Compared with $I_{0}^{(0)}$, most terms in $I_{0}^{(1)}$ vanish since the integrand is odd, except for the ones contributed by the term in $v_{i,t}^{(0)}$ with coefficient $r(\kappa_1-\pi;\delta)\cdot(\kappa_1-\pi)$. Those non-trivial terms take the following form
\footnotesize
\begin{equation*}
\Bigg[\frac{1}{2\pi}\int_{\pi-\delta^{\frac{1}{9}}}^{\pi+\delta^{\frac{1}{9}}}
\frac{(\kappa_1-\pi)^{2}f_{*}^{k}(\kappa_1;\delta)r(\kappa_1-\pi;\delta)}{\delta\cdot h\pm\sqrt{\frac{1}{4}\gamma_*^2(\kappa_1-\pi)^4+t_*^2\delta^2}}\frac{d\kappa_1}{1+f^2_{*}(\kappa_1)}\Bigg]\sum_{j+\ell =1}v_{n}^{(j)}(\bm{x};\pi)(f(\cdot),v_{m}^{(\ell)}(\bm{x};\pi))_{\Omega},
\end{equation*}
\normalsize
or with $r(\kappa_1-\pi;\delta)$ replaced by $\overline{r(\kappa_1-\pi;\delta)}$. Here $k\in\{0,1,2\}$, $n,m\in\{1,2\}$. Since the functions $f_*$ and $r$ are uniformly bounded as shown in Lemma \ref{lem_app_C_scalar_integrals_property},
these scalar integrals are dominated by the one in \eqref{eq_app_C_scalar_integrals_property_1} for $k=2$ and $\ell=0$, and hence are of the order $\mathcal{O}(\delta^{\frac{1}{9}})$. Collecting all terms, $I_{0}^{(1)}$ can be written as
\begin{equation} \label{eq_app_C_I01_summary}
I_{0}^{(1)}=\sum_{n=1,2}v^{(0)}_n(\bm{x};\pi)(f(\bm{y}),f^{prop,(1)}_{n,\delta}(\bm{y}))_{\Omega}+g^{prop,(1)}_{n,\delta}(\bm{x})(f(\bm{y}),v_n^{(0)}(\bm{y};\pi))_{\Omega},
\end{equation}
with $\|f^{prop,(2)}_{n,\delta}\|,\|g^{prop,(2)}_{n,\delta}\|=\mathcal{O}(\delta^{\frac{1}{9}})$. Hence $I_{0}^{(1)}$ is incorporated in the minor terms in \eqref{eq_app_C_I0_summary}.

{\color{blue}Step 3: Estimate of $I_{0}^{(2)}$.} Similar as in $I_0^{(0)}$, since the integrals involving the $(\kappa_1-\pi)$-term in $v_{i,t}^{(0)}$ vanish and $v^{(k)}_{i,t}=0$ for $k\geq 2$, 
\[
I_0^{(2)}= \sum_{1\leq j\leq 6}I_{0,j}^{(2)},
\]
where 
\footnotesize
\begin{equation*} \label{eq_app_C_I02_six_lines}
\begin{aligned}
I_{0,1}^{(2)}&=\frac{1}{2\pi}\int_{\pi-\delta^{\frac{1}{9}}}^{\pi+\delta^{\frac{1}{9}}}
\frac{f^2_{*}(\kappa_1)v^{(1)}_{2}(\bm{x};\pi)\cdot \big(f(\bm{y}),v^{(1)}_{2}(\bm{y};\pi)\big)_{\Omega}+\sum_{j+\ell=2}v^{(j)}_{1}(\bm{x};\pi)\cdot \big(f(\bm{y}),v^{(\ell)}_{1}(\bm{y};\pi)\big)_{\Omega}}{\delta\cdot h+\sqrt{\frac{1}{4}\gamma_*^2(\kappa_1-\pi)^4+t_*^2\delta^2}}\frac{(\kappa_1-\pi)^{2}d\kappa_1}{1+f^2_{*}(\kappa_1)},  \\
I_{0,2}^{(2)}&=
\frac{1}{2\pi}\int_{\pi-\delta^{\frac{1}{9}}}^{\pi+\delta^{\frac{1}{9}}}
\frac{f^2_{*}(\kappa_1)v^{(1)}_{1}(\bm{x};\pi)\cdot \big(f(\bm{y}),v^{(1)}_{1}(\bm{y};\pi)\big)_{\Omega}+\sum_{j+\ell=2}v^{(j)}_{2}(\bm{x};\pi)\cdot \big(f(\bm{y}),v^{(\ell)}_{2}(\bm{y};\pi)\big)_{\Omega}}{\delta\cdot h-\sqrt{\frac{1}{4}\gamma_*^2(\kappa_1-\pi)^4+t_*^2\delta^2}}\frac{(\kappa_1-\pi)^{2}d\kappa_1}{1+f^2_{*}(\kappa_1)}, \\
I_{0,3}^{(2)}&=\frac{1}{2\pi}\int_{\pi-\delta^{\frac{1}{9}}}^{\pi+\delta^{\frac{1}{9}}}
\frac{v_{1}^{(2)}(\bm{x};\pi)\cdot \big(f(\bm{y}),v_{2}^{(0)}(\bm{y};\pi)\big)_{\Omega}
+v_{2}^{(0)}(\bm{x};\pi)\cdot \big(f(\bm{y}),v_{1}^{(2)}(\bm{y};\pi)\big)_{\Omega}}{\delta\cdot h+\sqrt{\frac{1}{4}\gamma_*^2(\kappa_1-\pi)^4+t_*^2\delta^2}}\frac{(\kappa_1-\pi)^{2}f_{*}(\kappa_1)d\kappa_1}{1+f^2_{*}(\kappa_1)},  \\
I_{0,4}^{(2)}&=\frac{1}{2\pi}\int_{\pi-\delta^{\frac{1}{9}}}^{\pi+\delta^{\frac{1}{9}}}
\frac{v_{1}^{(1)}(\bm{x};\pi)\cdot \big(f(\bm{y}),v_{2}^{(1)}(\bm{y};\pi)\big)_{\Omega}
+v_{2}^{(1)}(\bm{x};\pi)\cdot \big(f(\bm{y}),v_{1}^{(1)}(\bm{y};\pi)\big)_{\Omega}}{\delta\cdot h+\sqrt{\frac{1}{4}\gamma_*^2(\kappa_1-\pi)^4+t_*^2\delta^2}}\frac{(\kappa_1-\pi)^{2}f_{*}(\kappa_1)d\kappa_1}{1+f^2_{*}(\kappa_1)},  \\
I_{0,5}^{(2)}&= -\frac{1}{2\pi}\int_{\pi-\delta^{\frac{1}{9}}}^{\pi+\delta^{\frac{1}{9}}}
\frac{v_{2}^{(2)}(\bm{x};\pi)\cdot \big(f(\bm{y}),v_{1}^{(0)}(\bm{y};\pi)\big)_{\Omega}
+v_{1}^{(0)}(\bm{x};\kappa_1)\cdot \big(f(\bm{y}),v_{2}^{(2)}(\bm{y};\kappa_1)\big)_{\Omega}}{\delta\cdot h-\sqrt{\frac{1}{4}\gamma_*^2(\kappa_1-\pi)^4+t_*^2\delta^2}}\frac{(\kappa_1-\pi)^{2}f_{*}(\kappa_1)d\kappa_1}{1+f^2_{*}(\kappa_1)}, \\
I_{0,6}^{(2)}&= -\frac{1}{2\pi}\int_{\pi-\delta^{\frac{1}{9}}}^{\pi+\delta^{\frac{1}{9}}}
\frac{v_{2}^{(1)}(\bm{x};\pi)\cdot \big(f(\bm{y}),v_{1}^{(1)}(\bm{y};\pi)\big)_{\Omega}
+v_{1}^{(1)}(\bm{x};\kappa_1)\cdot \big(f(\bm{y}),v_{2}^{(1)}(\bm{y};\kappa_1)\big)_{\Omega}}{\delta\cdot h-\sqrt{\frac{1}{4}\gamma_*^2(\kappa_1-\pi)^4+t_*^2\delta^2}}\frac{(\kappa_1-\pi)^{2}f_{*}(\kappa_1)d\kappa_1}{1+f^2_{*}(\kappa_1)}.
\end{aligned}
\end{equation*}
\normalsize
By \eqref{eq_app_C_scalar_integrals_property_6} and \eqref{eq_app_C_scalar_integrals_property_4}, 
\footnotesize
\begin{equation*}
\begin{aligned}
%&\frac{1}{2\pi}\int_{\pi-\delta^{\frac{1}{9}}}^{\pi+\delta^{\frac{1}{9}}}
%\frac{f^2_{*}(\kappa_1)v^{(1)}_{2}(\bm{x};\pi)\cdot \big(f(\bm{y}),v^{(1)}_{2}(\bm{y};\pi)\big)_{\Omega}+\sum_{j+\ell=2}v^{(j)}_{1}(\bm{x};\pi)\cdot \big(f(\bm{y}),v^{(\ell)}_{1}(\bm{y};\pi)\big)_{\Omega}}{\delta\cdot h+\sqrt{\frac{1}{4}\gamma_*^2(\kappa_1-\pi)^4+t_*^2\delta^2}}\frac{(\kappa_1-\pi)^{2}d\kappa_1}{1+f^2_{*}(\kappa_1)}  \\
%&\quad +\frac{1}{2\pi}\int_{\pi-\delta^{\frac{1}{9}}}^{\pi+\delta^{\frac{1}{9}}}
%\frac{f^2_{*}(\kappa_1)v^{(1)}_{1}(\bm{x};\pi)\cdot \big(f(\bm{y}),v^{(1)}_{1}(\bm{y};\pi)\big)_{\Omega}+\sum_{j+\ell=k}v^{(j)}_{2}(\bm{x};\pi)\cdot \big(f(\bm{y}),v^{(\ell)}_{2}(\bm{y};\pi)\big)_{\Omega}}{\delta\cdot h-\sqrt{\frac{1}{4}\gamma_*^2(\kappa_1-\pi)^4+t_*^2\delta^2}}\frac{(\kappa_1-\pi)^{2}d\kappa_1}{1+f^2_{*}(\kappa_1)}  \\
I_{0,1}^{(2)}+ I_{0,2}^{(2)}
&=\sum_{n=1,2}\Big[\delta^{\frac{1}{2}}e_{n}(h)+o(\delta^{\frac{1}{2}}) \Big]v^{(1)}_{n}(\bm{x};\pi)\cdot \big(f(\bm{y}),v^{(1)}_{n}(\bm{y};\pi)\big)_{\Omega} \\
&\quad +\frac{2\delta^{\frac{1}{9}}}{\pi\gamma_*}\sum_{j+\ell =2}v_1^{(j)}(\bm{x};\pi)(f(\cdot),v_1^{(\ell)}(\bm{x};\pi))_{\Omega}
+\frac{-2\delta^{\frac{1}{9}}}{\pi\gamma_*}\sum_{j+\ell =2}v_2^{(j)}(\bm{x};\pi)(f(\cdot),v_2^{(\ell)}(\bm{x};\pi))_{\Omega} \\
&\quad +\sum_{n=1,2}v^{(0)}_n(\bm{x};\pi)(f(\bm{y}),f^{prop,(2)}_{n,\delta}(\bm{y}))_{\Omega}+g^{prop,(2)}_{n,\delta}(\bm{x})(f(\bm{y}),v_n^{(0)}(\bm{y};\pi))_{\Omega},
\end{aligned}
\end{equation*}
\normalsize
with $\|f^{prop,(2)}_{n,\delta}\|,\|g^{prop,(2)}_{n,\delta}\|=\mathcal{O}(\delta^{\frac{1}{2}})$. Here each $f^{prop,(1)}_{n,\delta},g^{prop,(1)}_{n,\delta}$ is the product of a vector part and an $\mathcal{O}(\delta^{\frac{1}{2}})$ scalar part, such as $f^{prop,(1)}_{1,\delta}=v_{1}^{(2)}(\bm{x};\pi)\times (\text{remainder of \eqref{eq_app_C_scalar_integrals_property_6}})$, as in Step 2. 

Similarly, by \eqref{eq_app_C_scalar_integrals_property_6}, 
\footnotesize
\begin{equation*}
\begin{aligned}
%&\frac{1}{2\pi}\int_{\pi-\delta^{\frac{1}{9}}}^{\pi+\delta^{\frac{1}{9}}}
%\frac{v_{1}^{(1)}(\bm{x};\pi)\cdot \big(f(\bm{y}),v_{2}^{(1)}(\bm{y};\pi)\big)_{\Omega}
%+v_{2}^{(1)}(\bm{x};\pi)\cdot \big(f(\bm{y}),v_{1}^{(1)}(\bm{y};\pi)\big)_{\Omega}}{\delta\cdot h+\sqrt{\frac{1}{4}\gamma_*^2(\kappa_1-\pi)^4+t_*^2\delta^2}}\frac{(\kappa_1-\pi)^{2}f_{*}(\kappa_1)d\kappa_1}{1+f^2_{*}(\kappa_1)}  \\
%&\quad -\frac{1}{2\pi}\int_{\pi-\delta^{\frac{1}{9}}}^{\pi+\delta^{\frac{1}{9}}}
%\frac{v_{2}^{(1)}(\bm{x};\pi)\cdot \big(f(\bm{y}),v_{1}^{(1)}(\bm{y};\pi)\big)_{\Omega}
%+v_{1}^{(1)}(\bm{x};\kappa_1)\cdot \big(f(\bm{y}),v_{2}^{(1)}(\bm{y};\kappa_1)\big)_{\Omega}}{\delta\cdot h-\sqrt{\frac{1}{4}\gamma_*^2(\kappa_1-\pi)^4+t_*^2\delta^2}}\frac{(\kappa_1-\pi)^{2}f_{*}(\kappa_1)d\kappa_1}{1+f^2_{*}(\kappa_1)} \\
I_{0,4}^{(2)}+ I_{0,6}^{(2)}&=\Big[\delta^{\frac{1}{2}}e_{3}(h)+o(\delta^{\frac{1}{2}}) \Big]\Big(v_{1}^{(1)}(\bm{x};\pi)\cdot \big(f(\bm{y}),v_{2}^{(1)}(\bm{y};\pi)\big)_{\Omega}
+v_{2}^{(1)}(\bm{x};\pi)\cdot \big(f(\bm{y}),v_{1}^{(1)}(\bm{y};\pi)\big)_{\Omega} \Big), \\
I_{0,3}^{(2)}+I_{0,5}^{(2)}&=\sum_{n=1,2}v^{(0)}_n(\bm{x};\pi)(f(\bm{y}),f^{prop,(3)}_{n,\delta}(\bm{y}))_{\Omega}+g^{prop,(3)}_{n,\delta}(\bm{x})(f(\bm{y}),v_n^{(0)}(\bm{y};\pi))_{\Omega},
\end{aligned}
\end{equation*}
\normalsize
where $\|f^{prop,(3)}_{n,\delta}\|,\|g^{prop,(3)}_{n,\delta}\|=\mathcal{O}(\delta^{\frac{1}{2}})$. In conclusion,
\footnotesize
\begin{equation} \label{eq_app_C_I02_summary}
\begin{aligned}
I_0^{(2)}
&=\frac{2\delta^{\frac{1}{9}}}{\pi\gamma_*}\sum_{j+\ell =2}v_1^{(j)}(\bm{x};\pi)(f(\cdot),v_1^{(\ell)}(\bm{x};\pi))_{\Omega}
+\frac{-2\delta^{\frac{1}{9}}}{\pi\gamma_*}\sum_{j+\ell =2}v_2^{(j)}(\bm{x};\pi)(f(\cdot),v_2^{(\ell)}(\bm{x};\pi))_{\Omega} \\
&\quad+ \sum_{n=1,2}\Big[\delta^{\frac{1}{2}}e_{n}(h)+o(\delta^{\frac{1}{2}}) \Big]v^{(1)}_{n}(\bm{x};\pi)\cdot \big(f(\bm{y}),v^{(1)}_{n}(\bm{y};\pi)\big)_{\Omega} \\
&\quad +\Big[\delta^{\frac{1}{2}}e_{3}(h)+o(\delta^{\frac{1}{2}}) \Big]\Big(v_{1}^{(1)}(\bm{x};\pi)\cdot \big(f(\bm{y}),v_{2}^{(1)}(\bm{y};\pi)\big)_{\Omega}
+v_{2}^{(1)}(\bm{x};\pi)\cdot \big(f(\bm{y}),v_{1}^{(1)}(\bm{y};\pi)\big)_{\Omega} \Big) \\
&\quad +\sum_{n=1,2}v^{(0)}_n(\bm{x};\pi)(f(\bm{y}),f^{prop,(2)}_{n,\delta}(\bm{y})+f^{prop,(3)}_{n,\delta}(\bm{y}))_{\Omega}+\big(g^{prop,(2)}_{n,\delta}(\bm{x})+g^{prop,(3)}_{n,\delta}(\bm{x})\big)(f(\bm{y}),v_n^{(0)}(\bm{y};\pi))_{\Omega}.
\end{aligned}
\end{equation}
\normalsize
We note that all major terms in \eqref{eq_app_C_I0_summary} are accounted for in equations  \eqref{eq_app_C_I00_summary}, \eqref{eq_app_C_I01_summary} and \eqref{eq_app_C_I02_summary}. The exception comprises the terms involving $\sum_{j+\ell =4}v_n^{(j)}(\bm{x};\pi)(f(\cdot),v_n^{(\ell)}(\bm{x};\pi))_{\Omega}$, which is due to $I_{0}^{(4)}$, as shown in the next step.

{\color{blue}Step 4: Estimate of $I_{0}^{(3)}$ to $I_{0}^{(7)}$.} Following the same lines as the previous three steps, %$I_{0}^{(3)}+I_{0}^{(5)}+I_{0}^{(6)}$
\footnotesize
\begin{equation} \label{eq_app_C_I0356_summary}
I_{0}^{(3)}+I_{0}^{(5)}+I_{0}^{(6)}+I_{0}^{(7)}
=\sum_{n=1,2}v^{(0)}_n(\bm{x};\pi)(f(\bm{y}),f^{prop,(4)}_{n,\delta}(\bm{y}))_{\Omega}+g^{prop,(4)}_{n,\delta}(\bm{x})(f(\bm{y}),v_n^{(0)}(\bm{y};\pi))_{\Omega} +\mathcal{O}(\delta^{\frac{5}{9}})
\end{equation}
\normalsize
with $\|f^{prop,(4)}_{n,\delta}\|,\|g^{prop,(4)}_{n,\delta}\|=\mathcal{O}(\delta^{\frac{1}{3}})$. On the other hand,  since $v^{(k)}_{i,t}=0$ vanishes for $k\geq 2$, 
\[
I_0^{(4)}= I_{0,1}^{(4)}+I_{0,2}^{(4)}+I_{0,3}^{(4)},
\]
where
\footnotesize
\begin{equation*}
\begin{aligned}
I_{0,1}^{(4)}&=\frac{1}{2\pi}\int_{\pi-\delta^{\frac{1}{9}}}^{\pi+\delta^{\frac{1}{9}}}
\Big[\frac{\sum_{j+\ell=4}v^{(j)}_{1}(\bm{x};\pi)\cdot \big(f(\bm{y}),v^{(\ell)}_{1}(\bm{y};\pi)\big)_{\Omega}}{\delta\cdot h+\sqrt{\frac{1}{4}\gamma_*^2(\kappa_1-\pi)^4+t_*^2\delta^2}} 
+\frac{\sum_{j+\ell=4}v^{(j)}_{2}(\bm{x};\pi)\cdot \big(f(\bm{y}),v^{(\ell)}_{2}(\bm{y};\pi)\big)_{\Omega}}{\delta\cdot h-\sqrt{\frac{1}{4}\gamma_*^2(\kappa_1-\pi)^4+t_*^2\delta^2}}
\Big]\frac{(\kappa_1-\pi)^{4}d\kappa_1}{1+f^2_{*}(\kappa_1)},  \\
I_{0,2}^{(4)}&=\frac{1}{2\pi}\int_{\pi-\delta^{\frac{1}{9}}}^{\pi+\delta^{\frac{1}{9}}}
\frac{\sum_{j+\ell=4}\Big(v_{1}^{(j)}(\bm{x};\pi)\cdot \big(f(\bm{y}),v_{1,t}^{(\ell)}(\bm{y};\kappa_1)\big)_{\Omega}
+v_{1,t}^{(j)}(\bm{x};\kappa_1)\cdot \big(f(\bm{y}),v_{1}^{(\ell)}(\bm{y};\pi)\big)_{\Omega}\Big)}{\delta\cdot h+\sqrt{\frac{1}{4}\gamma_*^2(\kappa_1-\pi)^4+t_*^2\delta^2}}\frac{(\kappa_1-\pi)^{4}f_{*}(\kappa_1)d\kappa_1}{1+f^2_{*}(\kappa_1)},  \\
I_{0,3}^{(4)}&=\frac{1}{2\pi}\int_{\pi-\delta^{\frac{1}{9}}}^{\pi+\delta^{\frac{1}{9}}}
\frac{\sum_{j+\ell=4}\Big(v_{2}^{(j)}(\bm{x};\pi)\cdot \big(f(\bm{y}),v_{2,t}^{(\ell)}(\bm{y};\kappa_1)\big)_{\Omega}
+v_{2,t}^{(j)}(\bm{x};\kappa_1)\cdot \big(f(\bm{y}),v_{2}^{(\ell)}(\bm{y};\kappa_1)\big)_{\Omega}\Big)}{\delta\cdot h-\sqrt{\frac{1}{4}\gamma_*^2(\kappa_1-\pi)^4+t_*^2\delta^2}}\frac{(\kappa_1-\pi)^{4}f_{*}(\kappa_1)d\kappa_1}{1+f^2_{*}(\kappa_1)}.
\end{aligned}
\end{equation*}
\normalsize
By \eqref{eq_app_C_scalar_integrals_property_1},  $I_{0,2}^{(4)} + I_{0,3}^{(4)} =\mathcal{O}(\delta^{\frac{10}{9}})$. In addition, $I_{0,2}^{(4)}$ can be estimated by \eqref{eq_app_C_scalar_integrals_property_7}. Therefore 
\footnotesize
\begin{equation} \label{eq_app_C_I04_summary}
\begin{aligned}
I_0^{(4)}=\frac{2}{3\pi\gamma_*}\delta^{\frac{1}{3}}\sum_{j+\ell=4}v^{(j)}_{1}(\bm{x};\pi)\cdot \big(f(\bm{y}),v^{(\ell)}_{1}(\bm{y};\pi)\big)_{\Omega}-\frac{2}{3\pi\gamma_*}\delta^{\frac{1}{3}}\sum_{j+\ell=4}v^{(j)}_{1}(\bm{x};\pi)\cdot \big(f(\bm{y}),v^{(\ell)}_{1}(\bm{y};\pi)\big)_{\Omega} +\mathcal{O}(\delta^{\frac{10}{9}}).
\end{aligned}
\end{equation}
\normalsize
With \eqref{eq_app_C_I00_summary} to \eqref{eq_app_C_I04_summary}, the proof of \eqref{eq_app_C_I0_summary} is complete.

\subsubsection{Estimation of $I_{k}$ for $k\geq 1$}
We proceed to estimate $I_{k}$ for $k\geq 1$ following the methodology outlined in Section 8.4.1, which involves calculating the corresponding scalar integrals. In this section, we present the results without delving into the detailed calculations.

For $k=1,3,4$, $I_{k}$ only contribute minor terms to \eqref{eq_app_C_2}. Specifically, the estimate of $I_1$ relies on the rapid decay of $w_i^{(1)}(\kappa_1;\delta)$, similar to Lemma \ref{lem_app_C_scalar_integrals_property}), and the estimates $I_3$ and $I_4$ are obtained by using \eqref{eq_app_C_scalar_integrals_property_1}. The cumulative result is
\footnotesize
\begin{equation} \label{eq_app_C_I134_summary}
I_{1}+I_{3}+I_{4}
=\sum_{n=1,2}v^{(0)}_n(\bm{x};\pi)(f(\bm{y}),f^{prop,(5)}_{n,\delta}(\bm{y}))_{\Omega}+g^{prop,(5)}_{n,\delta}(\bm{x})(f(\bm{y}),v_n^{(0)}(\bm{y};\pi))_{\Omega} +\mathcal{O}(\delta^{\frac{5}{9}}),
\end{equation}
\normalsize
with $\|f^{prop,(4)}_{n,\delta}\|,\|g^{prop,(4)}_{n,\delta}\|=\mathcal{O}(\delta^{\frac{1}{3}})$. On the other hand, $w_i^{(2)}(\kappa_1;\delta)(\kappa_1-\pi)^{2}$ are asymptotically constant in the following sense: 
\begin{equation*}
w_i^{(2)}(\sqrt{\delta}p;\delta)=\frac{2\eta_*}{\gamma_*}+N_i^{(2)}+\mathcal{O}(p^{-2})
\end{equation*}
as $p\to\infty$. This asymptotic behavior is the reason why $I_2$ contributes two major terms to \eqref{eq_app_C_2}. By Lemma \ref{lem_app_C_scalar_integrals_property}, particularly equation \eqref{eq_app_C_scalar_integrals_property_7}, we obtain
\footnotesize
\begin{equation} \label{eq_app_C_I2_summary}
\begin{aligned}
I_{2}
&=-\frac{2}{3\pi\gamma_*}\big(\frac{2\eta_*}{\gamma_*}+N_1^{(2)}\big)\delta^{\frac{1}{3}}\sum_{j+\ell =2}v_1^{(j)}(\bm{x};\pi)(f(\cdot),v_1^{(\ell)}(\bm{x};\pi))_{\Omega}   +\frac{2}{3\pi\gamma_*}\big(\frac{2\eta_*}{\gamma_*}+N_2^{(2)}\big)\delta^{\frac{1}{3}}\sum_{j+\ell =2}v_2^{(j)}(\bm{x};\pi)(f(\cdot),v_2^{(\ell)}(\bm{x};\pi))_{\Omega} \\
&\quad +\sum_{n=1,2}v^{(0)}_n(\bm{x};\pi)(f(\bm{y}),f^{prop,(6)}_{n,\delta}(\bm{y}))_{\Omega}+g^{prop,(6)}_{n,\delta}(\bm{x})(f(\bm{y}),v_n^{(0)}(\bm{y};\pi))_{\Omega} +\mathcal{O}(\delta^{\frac{5}{9}}),
\end{aligned}
\end{equation}
\normalsize
with $\|f^{prop,(6)}_{n,\delta}\|,\|g^{prop,(6)}_{n,\delta}\|=\mathcal{O}(\delta^{\frac{1}{9}})$.

\section{Proof of Theorem \ref{thm_interface_mode}}

In this section, we prove Theorem \ref{thm_interface_mode} on the existence of the interface modes at $\kappa_2=\pi$. 

\subsection{Boundary integral equation formulation of interface modes}
We start with formulating the interface eigenvalue problem in Theorem \ref{thm_interface_mode} using boundary integral equations. To this end, we introduce the following layer potential operators associated with the perturbed Green function $G^{\delta}(\bm{x},\bm{y};\lambda)$ in \eqref{eq_perturbed_Green_function_floquet_expansion}
\footnotesize
\begin{equation} \label{eq_layer_potential_operators}
\begin{aligned}
&\mathcal{S}(\lambda;G^{\delta}):\tilde{H}^{-\frac{1}{2}}(\Gamma)
\to H_{loc}^{1}(\Omega),\quad
\varphi\mapsto \int_{\Gamma}G^{  \delta}(\bm{x},\bm{y};\lambda)\varphi(\bm{y})ds(\bm{y}), \\
&\mathcal{D}(\lambda;G^{\delta}):H^{\frac{1}{2}}(\Gamma)
\to H_{loc}^{1}(\Omega),\quad
\phi\mapsto \int_{\Gamma}\frac{\partial G^{  \delta}}{\partial y_1}(\bm{x},\bm{y};\lambda)\phi(\bm{y})ds(\bm{y}), \\
&\mathcal{K}^*(\lambda;G^{\delta}):\tilde{H}^{-\frac{1}{2}}(\Gamma)
\to \tilde{H}^{-\frac{1}{2}}(\Gamma),\quad
\varphi\mapsto p.v.\int_{\Gamma}\frac{\partial G^{  \delta}}{\partial x_1}(\bm{x},\bm{y};\lambda)\varphi(\bm{y})ds(\bm{y}), \\
&\mathcal{K}(\lambda;G^{\delta}):H^{\frac{1}{2}}(\Gamma)
\to H^{\frac{1}{2}}(\Gamma),\quad
\phi\mapsto p.v.\int_{\Gamma}\frac{\partial G^{  \delta}}{\partial y_1}(\bm{x},\bm{y};\lambda)\phi(\bm{y})ds(\bm{y}), \\
&\mathcal{N}(\lambda;G^{\delta}):H^{\frac{1}{2}}(\Gamma)
\to \tilde{H}^{-\frac{1}{2}}(\Gamma),\quad
\phi\mapsto \int_{\Gamma}\frac{\partial^2 G^{  \delta}}{\partial x_1\partial y_1}(\bm{x},\bm{y};\lambda)\phi(\bm{y})ds(\bm{y}).
\end{aligned}
\end{equation}
\normalsize
We have the following jump formulas, which are analogous to those in Proposition \ref{prop_G0_jump_formula}.
\footnotesize
\begin{equation} \label{eq_layer_potential_operators_jump_1}
\lim_{t\to 0}\Big(\mathcal{S}(\lambda;G^{\delta})[\varphi]\Big)(\bm{x}+t\bm{e}_1)=\mathcal{S}(\lambda;G^{\delta})[\varphi](\bm{x}),
\end{equation}
\begin{equation} \label{eq_layer_potential_operators_jump_2}
\lim_{t\to 0}\frac{\partial}{\partial x_1}\Big(\mathcal{D}(\lambda;G^{\delta})[\phi]\Big)(\bm{x}+t\bm{e}_1)=\mathcal{N}(\lambda;G^{\delta})[\phi](\bm{x}),
\end{equation}
\begin{equation} \label{eq_layer_potential_operators_jump_3}
\lim_{t\to 0^{\pm}}\frac{\partial}{\partial x_1}\Big(\mathcal{S}(\lambda;G^{\delta})[\varphi]\Big)(\bm{x}+t\bm{e}_1)=\Big(\pm\frac{1}{2}+\mathcal{K}^*(\lambda;G^{\delta}) \Big)\varphi(\bm{x}),
\end{equation}
\begin{equation} \label{eq_layer_potential_operators_jump_4}
\lim_{t\to 0^{\pm}}\Big(\mathcal{D}(\lambda;G^{\delta})[\phi]\Big)(\bm{x}+t\bm{e}_1)=\Big(\mp\frac{1}{2}+\mathcal{K}(\lambda;G^{\delta}) \Big) \phi(\bm{x}),
\end{equation}
\normalsize
where $\bm{x}\in\Gamma$. Now we fix $\delta>0$. We construct an interface mode of \eqref{eq-interface} in the following form
\begin{equation} \label{eq_joint_solution}
u(\bm{x};\lambda)=
\left\{
\begin{aligned}
&-\mathcal{D}(\lambda;G^{\delta})[\phi](\bm{x})+\mathcal{S}(\lambda;G^{\delta})[\varphi](\bm{x}),\quad x_1>0, \\
&\mathcal{D}(\lambda;G^{-\delta})[\phi](\bm{x})-\mathcal{S}(\lambda;G^{-\delta})[\varphi](\bm{x}),\quad x_1<0.
\end{aligned}
\right.
\end{equation}
By the jump formulas, for $\bm{x}\in\Gamma$, we have
\begin{equation*}
\lim_{t\to 0^+}
\begin{pmatrix}
u(\bm{x}+t\bm{e}_1) \\
\frac{\partial u}{\partial x_1}(\bm{x}+t\bm{e}_1)
\end{pmatrix}
=
\begin{pmatrix}
\frac{1}{2}-\mathcal{K}(\lambda;G^{\delta})
& \mathcal{S}(\lambda;G^{\delta})\\
-\mathcal{N}(\lambda;G^{\delta}) &
\frac{1}{2}+\mathcal{K}^*(\lambda;G^{\delta})
\end{pmatrix}
\begin{pmatrix}
\phi \\ \varphi
\end{pmatrix},
\end{equation*}
and
\begin{equation*}
\lim_{t\to 0^-}
\begin{pmatrix}
u(\bm{x}+t\bm{e}_1) \\
\frac{\partial u}{\partial x_1}(\bm{x}+t\bm{e}_1)
\end{pmatrix}
=
\begin{pmatrix}
\frac{1}{2}+\mathcal{K}(\lambda;G^{-\delta})
& -\mathcal{S}(\lambda;G^{-\delta}) \\
\mathcal{N}(\lambda;G^{-\delta}) &
\frac{1}{2}-\mathcal{K}^*(\lambda;G^{-\delta})
\end{pmatrix}
\begin{pmatrix}
\phi \\ \varphi
\end{pmatrix}.
\end{equation*}
Therefore, \eqref{eq_joint_solution} is an interface mode if and only if
\begin{equation*}
\lim_{t\to 0^+}
\begin{pmatrix}
u(\bm{x}+t\bm{e}_1) \\
\frac{\partial u}{\partial x_1}(\bm{x}+t\bm{e}_1)
\end{pmatrix}
=\lim_{t\to 0^-}
\begin{pmatrix}
u(\bm{x}+t\bm{e}_1) \\
\frac{\partial u}{\partial x_1}(\bm{x}+t\bm{e}_1)
\end{pmatrix}.
\end{equation*}
This leads to the following system of equations
\begin{equation} \label{eq_existence_condition_1}
\Big(\mathbb{T}^{\delta}(\lambda)+\mathbb{T}^{-\delta}(\lambda)\Big)
\begin{pmatrix}
\phi \\ \varphi
\end{pmatrix}
=0,
\end{equation}
where
\begin{equation*}
\mathbb{T}^{\delta}\in
\mathcal{B}(H^{\frac{1}{2}}(\Gamma)\times
\tilde{H}^{-\frac{1}{2}}(\Gamma)),\quad
\mathbb{T}^{\delta}
\begin{pmatrix}
\phi \\ \varphi
\end{pmatrix}
:=
\begin{pmatrix}
-\mathcal{K}(\lambda;G^{\delta})
& \mathcal{S}(\lambda;G^{\delta}) \\
-\mathcal{N}(\lambda;G^{\delta}) &
\mathcal{K}^*(\lambda;G^{\delta})
\end{pmatrix}
\begin{pmatrix}
\phi \\ \varphi
\end{pmatrix}.
\end{equation*}
We point out that any interface mode can be expressed in the form \eqref{eq_joint_solution}. Hence \eqref{eq_existence_condition_1} gives an equivalent formulation of interface modes. To determine the number of interface modes once their existence is established, it is important to recognize that different potentials $(\phi,\varphi)^{T}$ can yield the same interface mode through equation \eqref{eq_joint_solution}. This ambiguity is resolved by noting that each interface mode is uniquely determined by its interface data $(u|_{\Gamma},\frac{\partial u}{\partial n}\big|_{\Gamma})^T$. By setting $(\phi,\varphi)^{T}=(u|_{\Gamma},\frac{\partial u}{\partial n}\big|_{\Gamma})^T$ and applying the jump formulas, we have the following equations which complement \eqref{eq_existence_condition_1}: 
\begin{equation} \label{eq_existence_condition_3}
\Big(\frac{1}{2}-\mathbb{T}^{\delta}(\lambda)\Big)
\begin{pmatrix}
\phi \\ \varphi
\end{pmatrix}
= 0,\quad
\Big(\frac{1}{2}+\mathbb{T}^{-\delta}(\lambda)\Big)
\begin{pmatrix}
\phi \\ \varphi
\end{pmatrix}
= 0.
\end{equation}
We will utilize these equations to precisely determine the number of interface modes after establishing their existence.

We next introduce the following projections on $H^{\frac{1}{2}}(\Gamma)\times
\tilde{H}^{-\frac{1}{2}}(\Gamma)$
\begin{equation*}
\begin{aligned}
\Pi_1:
\begin{pmatrix}
\phi \\ \varphi
\end{pmatrix}
\mapsto\,\,
\frac{\int_{\Gamma}\phi(\cdot)\overline{\frac{\partial v_1}{\partial x_1}(\cdot;\pi)}}{i\gamma_*/2}
\begin{pmatrix}
\partial_{\kappa_1}v_1(\bm{x};\pi) \\ 0
\end{pmatrix}
+\frac{\int_{\Gamma}\varphi(\cdot)\overline{v_2(\cdot;\pi)}}{i\gamma_*/2}
\begin{pmatrix}
0 \\ \frac{\partial}{\partial x_1}(\partial_{\kappa_1}v_2)(\bm{x};\pi)
\end{pmatrix},
\end{aligned}
\end{equation*}
\begin{equation*}
\begin{aligned}
\Pi_2:
\begin{pmatrix}
\phi \\ \varphi
\end{pmatrix}
\mapsto\,\,
&\frac{\int_{\Gamma}\varphi(\cdot)\overline{(\partial_{\kappa_1}v_1)(\cdot;\pi)}}{-i\gamma_*/2}
\begin{pmatrix}
0 \\ \frac{\partial v_1}{\partial x_1}(\bm{x};\pi)
\end{pmatrix}
+\frac{\int_{\Gamma}\phi(\cdot)\overline{\frac{\partial}{\partial x_1}(\partial_{\kappa_1}v_2)(\cdot;\pi)}}{-i\gamma_*/2}
\begin{pmatrix}
v_2(\bm{x};\pi) \\ 0
\end{pmatrix},
\end{aligned}
\end{equation*}
and
\[
\mathbb{Q}=\mathbb{I}_{2\times 2}-\Pi_1-\Pi_2.
\]
Here the Bloch modes $v_n(\bm{x};\kappa_1)$ are constructed in Theorem \ref{thm_gap_open}. By Proposition \ref{prop_12flux_alternative}, $\Pi_1$ and $\Pi_2$ are orthogonal, i.e. 
\begin{equation*}
\Pi_1\cdot \Pi_2
=\Pi_2\cdot \Pi_1=0.
\end{equation*}
This results in the following orthogonal decomposition
\begin{equation}
\label{eq_interface_decomposition}
H^{\frac{1}{2}}(\Gamma)\times
\tilde{H}^{-\frac{1}{2}}(\Gamma)=
X\oplus Y_1\oplus Y_2,\quad
X=\text{Ran} (\mathbb{Q}),\,
Y_1=\text{Ran}(\Pi_1) ,\,
Y_2=\text{Ran}(\Pi_2),
\end{equation}
which will be used for analyzing the behavior of the integral operator $\mathbb{T}^{\delta}$ on each subspace of $H^{\frac{1}{2}}(\Gamma)\times
\tilde{H}^{-\frac{1}{2}}(\Gamma)$. This orthogonal decomposition provides a perfect framework to study the behavior of the integral operator $\mathbb{T}^{\delta}$ on each subspace of $H^{\frac{1}{2}}(\Gamma)\times
\tilde{H}^{-\frac{1}{2}}(\Gamma)$. That's the most important technical reason we devote great efforts in Section 6 to study the energy flux carried by the Bloch modes and their momentum-derivatives.

Using the decomposition \eqref{eq_interface_decomposition}, $(\phi,\varphi)$ solves \eqref{eq_existence_condition_1} if and only if
\footnotesize
\begin{equation} \label{eq_matrix_equation_1}
\begin{pmatrix}
\mathbb{Q}\Big(\mathbb{T}^{\delta}(\lambda)+\mathbb{T}^{-\delta}(\lambda)\Big)\mathbb{Q} &
\mathbb{Q}\Big(\mathbb{T}^{\delta}(\lambda)+\mathbb{T}^{-\delta}(\lambda)\Big)\Pi_1 &
\mathbb{Q}\Big(\mathbb{T}^{\delta}(\lambda)+\mathbb{T}^{-\delta}(\lambda)\Big)\Pi_2 \\
\Pi_2\Big(\mathbb{T}^{\delta}(\lambda)+\mathbb{T}^{-\delta}(\lambda)\Big)\mathbb{Q} &
\Pi_2\Big(\mathbb{T}^{\delta}(\lambda)+\mathbb{T}^{-\delta}(\lambda)\Big)\Pi_1 &
\Pi_2\Big(\mathbb{T}^{\delta}(\lambda)+\mathbb{T}^{-\delta}(\lambda)\Big)\Pi_2 \\
\Pi_1\Big(\mathbb{T}^{\delta}(\lambda)+\mathbb{T}^{-\delta}(\lambda)\Big)\mathbb{Q} &
\Pi_1\Big(\mathbb{T}^{\delta}(\lambda)+\mathbb{T}^{-\delta}(\lambda)\Big)\Pi_1 &
\Pi_1\Big(\mathbb{T}^{\delta}(\lambda)+\mathbb{T}^{-\delta}(\lambda)\Big)\Pi_2
\end{pmatrix}
\begin{pmatrix}
{\Psi} \\ {\Phi}^{(1)} \\ {\Phi}^{(1)}
\end{pmatrix}
=0,
\end{equation}
\normalsize
where ${\Psi}=\mathbb{Q}\begin{pmatrix}
    \phi \\ \varphi
\end{pmatrix}$, ${\Phi}^{(1)}=\Pi_1\begin{pmatrix}
    \phi \\ \varphi
\end{pmatrix}$, ${\Phi}^{(2)}=\Pi_2\begin{pmatrix}
    \phi \\ \varphi
\end{pmatrix}$. 
We observe that the matrix in \eqref{eq_matrix_equation_1} diverges as $\delta\to 0$. To obtain a uniform limit that facilitates solving \eqref{eq_matrix_equation_1}, we normalize the equation by multiplying both sides by the diagonal matrix $\begin{pmatrix}
1 & 0 & 0 \\ 0 & \delta^{\frac{1}{4}} & 0 \\
0 & 0 & \delta^{-\frac{1}{4}}
\end{pmatrix}$. This normalization yields
\begin{equation} \label{eq_matrix_equation_2}
\begin{aligned}
\Big(\mathbb{M}^{\delta}(\lambda)
+\mathbb{M}^{-\delta}(\lambda)\Big)
\begin{pmatrix}
{\Psi} \\ {\Phi}^{(1)} \\ {\Phi}^{(2)}
\end{pmatrix}
=0,
\end{aligned}
\end{equation}
where the operator $\mathbb{M}^{\delta}(\lambda)$ is defined as
\begin{equation} \label{eq_matrix_operator_def}
\mathbb{M}^{\delta}(\lambda)
:=\begin{pmatrix}
\mathbb{Q}\mathbb{T}^{\delta}(\lambda)\mathbb{Q} &
\delta^{\frac{1}{4}}\mathbb{Q}\mathbb{T}^{\delta}(\lambda)\Pi_1 &
\delta^{-\frac{1}{4}}\mathbb{Q}\mathbb{T}^{\delta}(\lambda)\Pi_2 \\
\delta^{\frac{1}{4}}\Pi_2\mathbb{T}^{\delta}(\lambda)\mathbb{Q} &
\delta^{\frac{1}{2}}\Pi_2\mathbb{T}^{\delta}(\lambda)\Pi_1 &
\Pi_2\mathbb{T}^{\delta}(\lambda)\Pi_2 \\
\delta^{-\frac{1}{4}}\Pi_1\mathbb{T}^{\delta}(\lambda)\mathbb{Q} &
\Pi_1\mathbb{T}^{\delta}(\lambda)\Pi_1 &
\delta^{-\frac{1}{2}}\Pi_1\mathbb{T}^{\delta}(\lambda)\Pi_2
\end{pmatrix}.
\end{equation}
The introduced scaling in $\mathbb{M}^{\delta}(\lambda)$ effectively balances the orders of the operators across different subspaces, ensuring a uniform limit as $\delta \to 0$, which is crucial for solving \eqref{eq_matrix_equation_1}. We will address the normalized equation \eqref{eq_matrix_equation_2} in Section 9.3. After that, we use the following normalized version of \eqref{eq_existence_condition_3} to study the precise number of interface modes
\begin{equation}
\label{eq_existence_criterion_3_scaled}
\Big(
\begin{pmatrix}
\frac{1}{2} & 0 & 0 \\
0 & 0 & \frac{1}{2} \\
0 & \frac{1}{2} & 0
\end{pmatrix}
-\mathbb{M}^{\delta}(\lambda)
\Big)
\begin{pmatrix}
{\Psi} \\ {\Phi}^{(1)} \\ {\Phi}^{(2)}
\end{pmatrix}=0,\quad
\Big(\begin{pmatrix}
\frac{1}{2} & 0 & 0 \\
0 & 0 & \frac{1}{2} \\
0 & \frac{1}{2} & 0
\end{pmatrix}
+\mathbb{M}^{-\delta}(\lambda)
\Big)
\begin{pmatrix}
{\Psi} \\ {\Phi}^{(1)} \\ {\Phi}^{(2)}
\end{pmatrix}=0.
\end{equation}

\subsection{Properties of $\mathbb{M}^{\delta}(\lambda)$}
In this subsection, we show that $\mathbb{M}^{\delta}(\lambda_*+\delta\cdot h)$ converges uniformly to a diagonal operator as $\delta\to 0$, which is crucial for solving \eqref{eq_matrix_equation_2}. By employing the orthogonal decomposition \eqref{eq_interface_decomposition} and applying distinct scalings to each subspace, we fully characterize the behavior of $\mathbb{M}^{\delta}(\lambda_*+\delta\cdot h)$ in the perturbative regime. In particular, the scaling applied to the subspaces $Y_1$ and $Y_2$ reveals the phase transition phenomena in the Bloch modes $v_1(\bm{x};\pi), v_2(\bm{x};\pi)$ and their momentum derivatives $\partial_{\kappa_1}v_1(\bm{x};\pi),\partial_{\kappa_1}v_2(\bm{x};\pi)$. This scaling effectively distinguishes these modes from other components, highlighting their unique roles.  
As a result, this comprehensive analysis provides all the necessary information to determine the interface modes that bifurcate from the quadratic degenerate point.

\begin{proposition} \label{prop_matrix_convergence}
The following convergence holds uniformly for $h\in \overline{\mathcal{J}}$ as $\delta\to 0^+$:
\begin{equation} \label{eq_matrix_con_1}
\mathbb{M}^{\pm\delta}(\lambda_*+\delta \cdot h)\overset{\|\cdot\|_{\mathcal{B}(X\oplus Y_1\oplus Y_2,X\oplus Y_2\oplus Y_1)}}{\longrightarrow}
\begin{pmatrix}
\mathbb{Q}\mathbb{T}^{0}(\lambda_*)\mathbb{Q} & 0 & 0 \\
 0 & \mathbb{E}(h) & 0 \\
0 & 0 & \mathbb{F}(h)
\end{pmatrix}
\pm\begin{pmatrix}
0 & 0 & 0 \\
 0 & \mathbb{E}^{\times}(h) & 0 \\
0 & 0 & \mathbb{F}^{\times}(h)
\end{pmatrix},
\end{equation}
where
\footnotesize
\begin{equation*}
\mathbb{T}^{0}(\lambda_*)\in
\mathcal{B}(H^{\frac{1}{2}}(\Gamma)\times
\tilde{H}^{-\frac{1}{2}}(\Gamma)),\quad
\mathbb{T}^{0}(\lambda_*)
\begin{pmatrix}
\phi \\ \varphi
\end{pmatrix}
:=
\begin{pmatrix}
0 & \mathcal{S}(\lambda_*;G_0) \\
-\mathcal{N}(\lambda_*;G_0) & 0
\end{pmatrix}
\begin{pmatrix}
\phi \\ \varphi
\end{pmatrix},
\end{equation*}
\begin{equation*}
\begin{aligned}
\mathbb{E}(h)
\begin{pmatrix}
\phi \\ \varphi
\end{pmatrix}
:=&k_{2}(h)\big(\int_{\Gamma}\varphi(\cdot)\overline{v_2(\cdot;\pi)}\big)
\begin{pmatrix}
v_2(\bm{x};\pi) \\ 0
\end{pmatrix}
-k_{1}(h)\big(\int_{\Gamma}\phi(\cdot)\overline{\frac{\partial v_1}{\partial x_1}(\cdot;\pi)}\big)
\begin{pmatrix}
0 \\ \frac{\partial v_1}{\partial x_1}(\cdot;\pi)
\end{pmatrix},
\end{aligned}
\end{equation*}

\begin{equation*}
\begin{aligned}
\mathbb{F}(h)
\begin{pmatrix}
\phi \\ \varphi
\end{pmatrix}
:=&e_{1}(h)
\big(\int_{\Gamma}\varphi(\cdot)\overline{(\partial_{\kappa_1}v_1)(\cdot;\pi)}\big)
\begin{pmatrix}
(\partial_{\kappa_1}v_1)(\bm{x};\pi) \\ 0
\end{pmatrix}
-e_{2}(h)\big(\int_{\Gamma}\phi(\cdot)\overline{\frac{\partial }{\partial x_1}(\partial_{\kappa_1}v_2)(\cdot;\pi)}\big)
\begin{pmatrix}
0 \\ \frac{\partial}{\partial x_1}(\partial_{\kappa_1}v_2)(\cdot;\pi)
\end{pmatrix}.
\end{aligned}
\end{equation*}

\begin{equation*}
\begin{aligned}
\mathbb{E}^{\times}(h)
\begin{pmatrix}
\phi \\ \varphi
\end{pmatrix}
:=&-k_{3}(h)
\Bigg[\big(\int_{\Gamma}\phi(\cdot)\overline{\frac{\partial v_1}{\partial x_1}(\cdot;\pi)}\big)
\begin{pmatrix}
v_2(\bm{x};\pi) \\ 0
\end{pmatrix} 
-\big(\int_{\Gamma}\varphi(\cdot)\overline{v_2(\cdot;\pi)}\big)
\begin{pmatrix}
0 \\ \frac{\partial v_1}{\partial x_1}(\cdot;\pi)
\end{pmatrix}\Bigg],
\end{aligned}
\end{equation*}

\begin{equation*}
\begin{aligned}
\mathbb{F}^{\times}(h)
\begin{pmatrix}
\phi \\ \varphi
\end{pmatrix}
:=-e_{3}(h)
\Bigg[\big(\int_{\Gamma}\phi(\cdot)\overline{\frac{\partial }{\partial x_1}(\partial_{\kappa_1}v_2)(\cdot;\pi)}\big)
\begin{pmatrix}
(\partial_{\kappa_1}v_1)(\bm{x};\pi) \\ 0
\end{pmatrix} 
-\big(\int_{\Gamma}\varphi(\cdot)\overline{(\partial_{\kappa_1}v_1)(\cdot;\pi)}\big)
\begin{pmatrix}
0 \\ \frac{\partial }{\partial x_1}(\partial_{\kappa_1}v_2)(\cdot;\pi)
\end{pmatrix}\Bigg],
\end{aligned}
\end{equation*}
\normalsize
\end{proposition}
\begin{proof}
%The convergence of the integral operator $\mathbb{M}(\lambda_* + \delta h)$ follows directly from Theorem \ref{thm_asymp_unperturbed_green} and Lemma \ref{lem_boundary_convergence_to_domain} in the Appendix. Notably, the scaling applied to each subspace of $H^{\frac{1}{2}}(\Gamma) \times \tilde{H}^{-\frac{1}{2}}(\Gamma)$ within $\mathbb{M}(\lambda_* + \delta h)$ is critical for ensuring convergence. As demonstrated in Theorem \ref{thm_asymptotics_perturbed_Green_function}, this scaling effectively balances the $\delta$-order of the corresponding operators, thereby guaranteeing the uniform convergence of every component of $\mathbb{M}(\lambda_* + \delta h)$.
%To examine this uniform convergence, consider each element of $\mathbb{M}(\lambda_* + \delta h)$. The major terms—whose detailed expressions are provided in Theorem \ref{thm_asymptotics_perturbed_Green_function}—correspond to the nonzero entries in the limiting operator $\lim_{\delta \to 0^+} \mathbb{M}(\lambda_* + \delta h)$. In contrast, the minor terms, whose orders are estimated rather than explicitly detailed, vanish in the limit, resulting in zero entries in the limiting operator.

The convergence of $\mathbb{M}^{\pm\delta}(\lambda_*+\delta\cdot h)$ follows from Theorem \ref{thm_asymptotics_perturbed_Green_function} and Lemma \ref{lem_boundary_convergence_to_domain} in the Appendix.  Note that the scaling in $\mathbb{M}^{\pm\delta}(\lambda_*+\delta\cdot h)$ for each subspace of $H^{\frac{1}{2}}(\Gamma)\times \tilde{H}^{-\frac{1}{2}}(\Gamma)$ is critical for the convergence. As demonstrated in Theorem \ref{thm_asymptotics_perturbed_Green_function}, this scaling effectively balances the $\delta$-order terms in the corresponding operators, thereby ensuring the uniform convergence of each element of $\mathbb{M}^{\pm\delta}(\lambda_*+\delta\cdot h)$. One can examine this convergence by noting that each major term, which we give detailed expressions in Theorem \ref{thm_asymptotics_perturbed_Green_function}, corresponds to a nonzero term in $\lim_{\delta\to 0^+}\mathbb{M}(\lambda_*+\delta\cdot h)$, while the minor terms, whose orders are estimated rather than explicitly detailed, vanish in the limit.
\end{proof}

The limiting operator $\mathbb{T}^{0}$
in \eqref{eq_matrix_con_1} has the following properties that are crucial to solve the integral equation \eqref{eq_matrix_equation_2} and subsequently to determine the interface modes.

%consist of a Fredholm operator (with infinite-dimensional range) $\mathbb{T}^{0}$ and finite-rank projections. We summarize as follows the properties of $\mathbb{T}^{0}$, especially a characterization of its kernel space (in short, the kernel of $\mathbb{T}^{0}$ consists of the interface data of extended modes at the quadratic degenerate point). Those properties, with Proposition \ref{prop_matrix_convergence}, pave the way to solve the integral equation \eqref{eq_matrix_equation_2} and to determine the interface modes.

\begin{proposition}
\label{prop_T0_operator_kernel}
$\mathbb{T}^{0}(\lambda_*)\in
\mathcal{B}(H^{\frac{1}{2}}(\Gamma)\times
\tilde{H}^{-\frac{1}{2}}(\Gamma))$ is a Fredholm operator with zero index. Moreover,
\begin{equation} \label{eq_T0_operator_kernel}
\ker\mathbb{T}^0(\lambda_*)=\text{span}
\big\{
\begin{pmatrix}
v_2(\bm{x};\pi) \\ 0
\end{pmatrix},\,
\begin{pmatrix}
\partial_{\kappa_1}v_1(\bm{x};\pi) \\ 0
\end{pmatrix},\,
\begin{pmatrix}
0 \\ \frac{\partial v_1}{\partial x_1}(\bm{x};\pi)
\end{pmatrix},\,
\begin{pmatrix}
0 \\ \frac{\partial}{\partial x_1}(\partial_{\kappa_1}v_2)(\bm{x};\pi)
\end{pmatrix}
\big\}.
\end{equation}
\end{proposition}
\begin{proof}
Since the integral kernel of $\mathcal{S}(\lambda_*;G_0)$ and $\mathcal{N}(\lambda_*;G_0)$ is the fundamental solution $G_0(\bm{x},\bm{y};\lambda_*)$ of the elliptic operator $\mathcal{L}^{A}-\lambda_*$ in $\Omega$ by Proposition \ref{prop_G0_fundamental_solution}, the Fredholmness of $\mathbb{T}^0(\lambda_*)$ follows from Theorem 7.6 of \cite{mclean2000strongly}. See also Proposition 4.4 of \cite{qiu2023mathematical} for another proof.

We continue to prove \eqref{eq_T0_operator_kernel}. We first show that
\begin{equation}
\label{eq_app_D_1}
\begin{pmatrix}
v_2(\bm{x};\pi) \\ 0
\end{pmatrix}
\in\ker\mathbb{T}^0(\lambda_*).
\end{equation}
By applying the Green's formula between $v_2(\bm{x})$ and $G_0(\bm{x},\bm{y};\lambda_*)$ inside $(0,N)\times (-\frac{1}{2},\frac{1}{2})$, as we did in the proof of Proposition \ref{prop_12flux_alternative}, we have
\begin{equation}
\label{eq_app_D_2}
\begin{aligned}
\overline{v_2(\bm{y})}
&=\frac{i}{\gamma_*}\mathfrak{q}(v_1(\cdot;\pi),v_2(\cdot;\pi);\Gamma) \cdot \overline{\partial_{\kappa_1} v_1(\bm{y};\pi)}  +\frac{i}{\gamma_*}\mathfrak{q}(\partial_{\kappa_1} v_1(\cdot;\pi),v_2(\cdot;\pi);\Gamma) \cdot \overline{v_1(\bm{y};\pi)} \\
&\quad -\frac{i}{\gamma_*}\mathfrak{q}(v_2(\cdot;\pi),v_2(\cdot;\pi);\Gamma) \cdot \overline{\partial_{\kappa_1} v_2(\bm{y};\pi)}
-\frac{i}{\gamma_*}\mathfrak{q}(\partial_{\kappa_1} v_2(\cdot;\pi),v_2(\cdot;\pi);\Gamma) \cdot \overline{v_2(\bm{y};\pi)} \\
&\quad -\int_{\Gamma}\frac{\partial G_0}{\partial x_1}(\cdot,\bm{y};\lambda_*)\overline{v_2(\cdot;\pi)}-G_0(\cdot,\bm{y};\lambda_*)\overline{\frac{\partial v_2}{\partial x_1}(\bm{x};\pi)}.
\end{aligned}
\end{equation}
By Corollary \ref{prop_12flux}, \eqref{eq_app_D_2} becomes
\begin{equation*}
\begin{aligned}
\overline{v_2(\bm{y})}
=\frac{1}{2} \overline{v_2(\bm{y};\pi)}
-\int_{\Gamma}\frac{\partial G_0}{\partial x_1}(\cdot,\bm{y};\lambda_*)\overline{v_2(\cdot;\pi)}-G_0(\cdot,\bm{y};\lambda_*)\overline{\frac{\partial v_2}{\partial x_1}(\bm{x};\pi)}.
\end{aligned}
\end{equation*}
Taking the normal derivative for $\bm{y}\in\Gamma$, the jump formula \eqref{eq_G0_jump_2}, \eqref{eq_G0_jump_3} and the symmetry \eqref{eq_G0_argument_symmetry} show that
\begin{equation*}
\begin{aligned}
\overline{\frac{\partial v_2}{\partial y_1}(\bm{y})}
=\frac{1}{2} \overline{\frac{\partial v_2}{\partial y_1}(\bm{y})}
-\overline{\mathcal{N}(\lambda_*;G_0)[v_2(\cdot;\pi)](\bm{y})}
+\frac{1}{2} \overline{\frac{\partial v_2}{\partial y_1}(\bm{y})},
\end{aligned}
\end{equation*}
which implies $\mathcal{N}(\lambda_*;G_0)[v_2(\cdot;\pi)]=0$, and \eqref{eq_app_D_1} follows. Similarly, we can prove
\begin{equation}
\label{eq_app_D_3}
\begin{pmatrix}
\partial_{\kappa_1}v_1(\bm{x};\pi) \\ 0
\end{pmatrix}
\in\ker\mathbb{T}^0(\lambda_*).
\end{equation}
Now we show that
\begin{equation}
\label{eq_app_D_4}
\begin{pmatrix}
0 \\ \frac{\partial v_1}{\partial x_1}(\bm{x})
\end{pmatrix}
\in\ker\mathbb{T}^0(\lambda_*).
\end{equation}
By applying the Green's formula to $v_1(\bm{x})$ and $G_0(\bm{x},\bm{y};\lambda_*)$, \eqref{eq_app_D_2} becomes
\begin{equation*}
\begin{aligned}
\overline{v_1(\bm{y})}
&=\frac{i}{\gamma_*}\mathfrak{q}(v_1(\cdot;\pi),v_1(\cdot;\pi);\Gamma) \cdot \overline{\partial_{\kappa_1} v_1(\bm{y};\pi)}  +\frac{i}{\gamma_*}\mathfrak{q}(\partial_{\kappa_1} v_1(\cdot;\pi),v_1(\cdot;\pi);\Gamma) \cdot \overline{v_1(\bm{y};\pi)} \\
&\quad -\frac{i}{\gamma_*}\mathfrak{q}(v_2(\cdot;\pi),v_1(\cdot;\pi);\Gamma) \cdot \overline{\partial_{\kappa_1} v_2(\bm{y};\pi)}
-\frac{i}{\gamma_*}\mathfrak{q}(\partial_{\kappa_1} v_2(\cdot;\pi),v_1(\cdot;\pi);\Gamma) \cdot \overline{v_2(\bm{y};\pi)} \\
&\quad -\int_{\Gamma}\frac{\partial G_0}{\partial x_1}(\cdot,\bm{y};\lambda_*)\overline{v_1(\cdot;\pi)}-G_0(\cdot,\bm{y};\lambda_*)\overline{\frac{\partial v_1}{\partial x_1}(\bm{x};\pi)}.
\end{aligned}
\end{equation*}
With \eqref{prop_12flux}, we obtain
\begin{equation*}
\begin{aligned}
\overline{v_1(\bm{y})}
=\frac{1}{2} \overline{v_1(\bm{y};\pi)}
-\int_{\Gamma}\frac{\partial G_0}{\partial x_1}(\cdot,\bm{y};\lambda_*)\overline{v_1(\cdot;\pi)}-G_0(\cdot,\bm{y};\lambda_*)\overline{\frac{\partial v_1}{\partial x_1}(\bm{x};\pi)}.
\end{aligned}
\end{equation*}
By letting $\bm{y}\to\Gamma$, the jump formula \eqref{eq_G0_jump_1}, \eqref{eq_G0_jump_4} and the symmetry \eqref{eq_G0_argument_symmetry} imply that
\begin{equation*}
\begin{aligned}
\overline{v_1(\bm{y})}
=\frac{1}{2} \overline{v_1(\bm{y};\pi)}
-\big(-\frac{1}{2} \overline{v_1(\bm{y};\pi)}\big)
+\overline{\mathcal{S}(\lambda_*;G_0)[\frac{\partial v_1}{\partial x_1}(\cdot;\pi)](\bm{y})},
\end{aligned}
\end{equation*}
whence $\mathcal{S}(\lambda_*;G_0)[\frac{\partial v_1}{\partial x_1}(\cdot;\pi)]=0$, and \eqref{eq_app_D_4} follows. Similarly, we can prove
\begin{equation}
\label{eq_app_D_5}
\begin{pmatrix}
0 \\ \frac{\partial }{\partial x_1}(\partial_{\kappa_1} v_2)(\bm{x})
\end{pmatrix}
\in\ker\mathbb{T}^0(\lambda_*).
\end{equation}
Thus, from \eqref{eq_app_D_1} and \eqref{eq_app_D_3}-\eqref{eq_app_D_4}, we conclude that
\begin{equation}
\label{eq_app_D_6}
\ker\mathbb{T}^0(\lambda_*)\supset\text{span}
\big\{
\begin{pmatrix}
v_2(\bm{x};\pi) \\ 0
\end{pmatrix},\,
\begin{pmatrix}
\partial_{\kappa_1}v_1(\bm{x};\pi) \\ 0
\end{pmatrix},\,
\begin{pmatrix}
0 \\ \frac{\partial v_1}{\partial x_1}(\bm{x};\pi)
\end{pmatrix},\,
\begin{pmatrix}
0 \\ \frac{\partial}{\partial x_1}(\partial_{\kappa_1}v_2)(\bm{x};\pi)
\end{pmatrix}
\big\}.
\end{equation}
We next prove
\begin{equation}
\label{eq_app_D_7}
\ker\mathbb{T}^0(\lambda_*)\subset\text{span}
\big\{
\begin{pmatrix}
v_2(\bm{x};\pi) \\ 0
\end{pmatrix},\,
\begin{pmatrix}
\partial_{\kappa_1}v_1(\bm{x};\pi) \\ 0
\end{pmatrix},\,
\begin{pmatrix}
0 \\ \frac{\partial v_1}{\partial x_1}(\bm{x};\pi)
\end{pmatrix},\,
\begin{pmatrix}
0 \\ \frac{\partial}{\partial x_1}(\partial_{\kappa_1}v_2)(\bm{x};\pi)
\end{pmatrix}
\big\}.
\end{equation}
By the decomposition \eqref{eq_interface_decomposition}, in order to prove \eqref{eq_app_D_7}, it's sufficient to show that $\begin{pmatrix}
    \phi \\ \varphi
\end{pmatrix}=0$ for any $\begin{pmatrix}
    \phi \\ \varphi
\end{pmatrix}\in X$ which satisfies $\mathbb{T}^0(\lambda_*)\begin{pmatrix}
    \phi \\ \varphi
\end{pmatrix}=0$. Suppose the contrary that $\begin{pmatrix}
    \phi \\ \varphi
\end{pmatrix}\neq 0$. We define the following function in $\Omega$
\begin{equation*}
u(\bm{x};\lambda_*):=
\left\{
\begin{aligned}
&-\mathcal{D}(\lambda_*;G_0)[\phi](\bm{x})+\mathcal{S}(\lambda_*;G_0)[\varphi](\bm{x}),\quad x_1>0, \\
&\mathcal{D}(\lambda_*;G_0)[\phi](\bm{x})-\mathcal{S}(\lambda_*;G_0)[\varphi](\bm{x}),\quad x_1<0.
\end{aligned}
\right.
\end{equation*}
Since $\begin{pmatrix}
    \phi \\ \varphi
\end{pmatrix}\in X$, the decomposition \eqref{eq_G0_decay_1} and \eqref{eq_G0_decay_2} indicate that $u(\bm{x};\lambda_*)$ exponentially decays as $|x_1|\to\infty$. On the other hand, the jump formula \eqref{eq_G0_jump_1}-\eqref{eq_G0_jump_4} and the assumption $\mathbb{T}^0(\lambda_*)\begin{pmatrix}
    \phi \\ \varphi
\end{pmatrix}=0$ show that for $\bm{x}\in\Gamma$,
\begin{equation*}
\begin{pmatrix}
\lim_{t\to 0_+}u(\bm{x};\lambda_*) \\ \lim_{t\to 0_+}\frac{\partial u}{\partial x_1}(\bm{x};\lambda_*)
\end{pmatrix}
=\begin{pmatrix}
    \frac{1}{2}\phi \\ \frac{1}{2}\varphi
\end{pmatrix}
=\begin{pmatrix}
\lim_{t\to 0_-}u(\bm{x};\lambda_*) \\ \lim_{t\to 0_-}\frac{\partial u}{\partial x_1}(\bm{x};\lambda_*)
\end{pmatrix}
\neq\bm{0}.
\end{equation*}
Finally, \eqref{eq_G0_fund} indicates that
\begin{equation*}
(\mathcal{L}^A-\lambda_*)u=0\quad \text{in}\,\,\Omega .
\end{equation*}
Hence, $u(\bm{x};\lambda_*)$ is a localized mode for the periodic operator $\mathcal{L}^A$ in $\Omega$, which contradicts to the absolute continuity of the spectrum $\mathcal{L}^A$ \cite{Sobolev02absolute}. This concludes the proof of \eqref{eq_app_D_7}.
\end{proof}

\subsection{Proof of Theorem \ref{thm_interface_mode}}
{\color{blue}Step 1: Existence of interface value.} We first show the existence of interface value by solving \eqref{eq_matrix_equation_2}. First, using \eqref{eq_matrix_con_1}, the characteristic values of the operator $\lim_{\delta\to 0}\Big(\mathbb{M}^{\delta}(\lambda_*+\delta \cdot h)+\mathbb{M}^{-\delta}(\lambda_*+\delta \cdot h)\Big)$ 
in $\mathcal{J}$ can be obtained by solving $\mathbb{E}(h){\Phi}^{(1)}=0$ or $\mathbb{F}(h){\Phi}^{(2)}=0$. This yields the following two characteristic values
\begin{equation*}
h_1=\frac{1}{\sqrt{2}}|t_*|,\quad
h_2=-\frac{1}{\sqrt{2}}|t_*|.
\end{equation*}
The geometric multiplicity of $h_1$ and $h_2$ both equal to two. The eigenvectors associated with $h_1$ are
\begin{equation} \label{eq_limit_root_function_1}
\begin{pmatrix}
{\Psi}_1 \\ {\Phi}_1^{(1)} \\ {\Phi}_1^{(2)}
\end{pmatrix}
=
\begin{pmatrix}
0 \\ (\partial_{\kappa_1}v_1(\bm{x};\pi),0)^T \\ 0
\end{pmatrix}
\quad\text{or}\quad
\begin{pmatrix}
0 \\ (v_2(\bm{x};\pi),0)^{T} \\ 0
\end{pmatrix}
\end{equation}
The eigenvectors associated with $h_2$ are
\begin{equation} \label{eq_limit_root_function_2}
\begin{pmatrix}
{\Psi}_2 \\ {\Phi}_1^{(2)} \\ {\Phi}_2^{(2)}
\end{pmatrix}
=\begin{pmatrix}
0 \\ 0 \\ (0,\frac{\partial}{\partial x_1}(\partial_{\kappa_1}v_2)(\bm{x};\pi))^T
\end{pmatrix}
\quad\text{or}\quad
\begin{pmatrix}
0 \\ 0 \\ (0,\frac{\partial v_1}{\partial x_1}(\bm{x};\pi))^T
\end{pmatrix}.
\end{equation}
Next, observe the following: 

(i) $\mathbb{M}^{\delta}(\lambda_*+\delta \cdot h)+\mathbb{M}^{-\delta}(\lambda_*+\delta \cdot h)$ is Fredholm with zero index for any $h\in\overline{\mathcal{J}}$. Indeed, since $\mathbb{T}^0(\lambda_*)$ is Fredholm, the limiting operator of $\mathbb{M}^{\delta}(\lambda_*+\delta \cdot h)+\mathbb{M}^{-\delta}(\lambda_*+\delta \cdot h)$ is Fredholm with zero index because it differs from $2\mathbb{T}^0(\lambda_*)$ by finite-rank operators. Consequently, $\mathbb{M}^{\delta}(\lambda_*+\delta \cdot h)+\mathbb{M}^{-\delta}(\lambda_*+\delta \cdot h)$ is Fredholm by \eqref{eq_matrix_con_1};

(ii) $\mathbb{M}^{\pm\delta}(\lambda_*+\delta \cdot h)$ are analytic for $h\in\mathcal{J}$ because the Green function $G^{\delta}(\bm{x},\bm{y};\lambda)$ is analytici when $\lambda$ lies in the band gap; 
%because $h$ lies in the resolvent set of $\mathcal{L}^{A\pm\delta\cdot B}_{\Omega,\pi}$.

(iii) $\lim_{\delta\to 0}\Big(\mathbb{M}^{\delta}(\lambda_*+\delta \cdot h)+\mathbb{M}^{-\delta}(\lambda_*+\delta \cdot h)\Big)$ is invertible for $h\in \overline{\mathcal{J}}\backslash \{h_1,h_2\}$. This follows from the facts that $\lim_{\delta\to 0}\Big(\mathbb{M}^{\delta}(\lambda_*+\delta \cdot h)+\mathbb{M}^{-\delta}(\lambda_*+\delta \cdot h)\Big)$ is Fredholm and is injective for $h\neq h_1$ and $h\neq h_2$ as proved previously.

Hence, the generalized Rouché theorem (cf. Chapter 1 of \cite{ammari2018mathematical}) concludes that $\mathbb{M}^{\delta}(\lambda_*+\delta \cdot h)+\mathbb{M}^{-\delta}(\lambda_*+\delta \cdot h)$ has four characteristic values $h_k(\delta)$, $1\leq k\leq 4$, counting the geometric multiplicity. The first two bifurcate from $\frac{1}{\sqrt{2}}|t_*|$, and the last two bifurcate from $-\frac{1}{\sqrt{2}}|t_*|$:
\begin{equation} \label{eq_sec5_1}
h_k(\delta)=\frac{1}{\sqrt{2}}|t_*|+o(1)\quad k=1,2; \quad
h_k(\delta)=-\frac{1}{\sqrt{2}}|t_*|+o(1)\quad k=3,4. 
\end{equation}
This completes the proof of Theorem \ref{thm_interface_mode}.

{\color{blue}Step 2: Number of interface eigenvalue.} As shown in \eqref{eq_sec5_1}, there are between two and four interface interface eigenvalues when $\delta$ is small. The additional one or two eigenvalues are consequences of the geometry multiplicity observed in \eqref{eq_limit_root_function_1} and \eqref{eq_limit_root_function_2}. This is due to the multiple choices of layer potential expressions of a single interface mode, as discussed in Section 9.1, and necessitates the additional condition imposed by equation \eqref{eq_existence_criterion_3_scaled}. Now we show how \eqref{eq_existence_criterion_3_scaled} eliminate the redundant multiplicity of the interface eigenvalue bifurcating from $\frac{1}{\sqrt{2}}|t_*|$. The discussion is similar for $-\frac{1}{\sqrt{2}}|t_*|$ and is therefore omitted.

Combining the two equations in \eqref{eq_existence_criterion_3_scaled} gives
\begin{equation} \label{eq_number_inteface_eigenvalue_1}
\big(1
-(\mathbb{M}^{\delta}(\lambda_*+\delta h)-\mathbb{M}^{-\delta}(\lambda_*+\delta h))
\big)
\begin{pmatrix}
{\Psi} \\ {\Phi}^{(1)} \\ {\Phi}^{(2)}
\end{pmatrix}=0.
\end{equation}
By Proposition \ref{prop_matrix_convergence}, the leading order equation of \eqref{eq_number_inteface_eigenvalue_1} is
\begin{equation*}
\begin{pmatrix}
1 & 0 & 0 \\
0 & -2\mathbb{F}^{\times}(h) & 1 \\
0 & 1 & -2\mathbb{E}^{\times}(h)
\end{pmatrix}
\begin{pmatrix}
{\Psi} \\ {\Phi}^{(1)} \\ {\Phi}^{(2)}
\end{pmatrix}=0,
\end{equation*}
which has a \textit{unique} solution $({\Psi}^{mode},{\Phi}^{(1),mode},{\Phi}^{(2),mode})$ at $h=\frac{1}{\sqrt{2}}|t_*|$ with
\begin{equation*}
{\Psi}^{mode}=0,\quad
{\Phi}^{(1),mode}=
\begin{pmatrix}
2^{\frac{1}{4}}\text{sgn}(t_*)\sqrt{\frac{|t_*|}{\gamma_*}}i\partial_{\kappa_1}v_1(\bm{x};\pi) \\ -2^{\frac{1}{4}}\text{sgn}(t_*)\sqrt{\frac{|t_*|}{\gamma_*}}i\frac{\partial}{\partial x_1}(\partial_{\kappa_1}v_2)(\bm{x};\pi)
\end{pmatrix},\quad
{\Phi}^{(2),mode}=\begin{pmatrix}
v_2(\bm{x};\pi) \\ \frac{\partial v_1}{\partial x_1}(\bm{x};\pi)
\end{pmatrix}.
\end{equation*}
This confirms there is a unique interface mode with its characteristic value near $\frac{1}{\sqrt{2}}|t_*|$. In fact, the complete solution (including the remainder) of \eqref{eq_number_inteface_eigenvalue_1} for small $\delta$ can be constructed by the Lyapunov-Schmidt argument, as detailed in Proposition 7.11 of \cite{li2024interface}.  We omit this construction here.

\appendix
\setcounter{secnumdepth}{0}
\section{Appendix A}
%: Connection between the convergence of layer potential operators and kernel functions}
\setcounter{equation}{0}
\setcounter{subsection}{0}
\setcounter{theorem}{0}
\renewcommand{\theequation}{A.\arabic{equation}}
\renewcommand{\thesubsection}{A.\arabic{subsection}}
\renewcommand{\thetheorem}{A.\arabic{theorem}}
\begin{lemma}
\label{lem_boundary_convergence_to_domain}
Let $K_n(\bm{x},\bm{y})$ ($n\geq 1$) and $K(x,y)$ be integral kernels such that $K_n(\bm{x},\bm{y})=\overline{K_n(\bm{y},\bm{x})}$ and the layer potential operators $\mathcal{S}(K_n),\mathcal{D}(K_n)$ and $\mathcal{S}(K),\mathcal{D}(K)$ (defined similarly as in \eqref{eq_layer_potential_operators}) satisfy the jump conditions \eqref{eq_layer_potential_operators_jump_1}-\eqref{eq_layer_potential_operators_jump_4}. Let $U\subset \Omega$ be a neighborhood of $\Gamma$. Suppose that the following holds uniformly for all $f\in (H^1(U))^*$ with $\|f\|_{(H^1(U))^*}=1$ and $\text{supp} f\subset\subset U$
\begin{equation}
\label{eq_app_C_4}
\int_{U}K_n(\bm{x},\bm{y})f(\bm{y})d\bm{y}
\overset{H^1(U)}{\longrightarrow}
\int_{U}K(\bm{x},\bm{y})f(\bm{y})d\bm{y}.
\end{equation}
Then the integral operators $\mathcal{S}(K_n)\in\mathcal{B}(\tilde{H}^{-\frac{1}{2}}(\Gamma),H^{\frac{1}{2}}(\Gamma))$, $\mathcal{N}(K_n)\in\mathcal{B}(H^{\frac{1}{2}}(\Gamma),\tilde{H}^{-\frac{1}{2}}(\Gamma))$, $\mathcal{K}(K_n)\in\mathcal{B}(H^{\frac{1}{2}}(\Gamma),H^{\frac{1}{2}}(\Gamma))$ and $\mathcal{K}^*(K_n)\in\mathcal{B}(\tilde{H}^{-\frac{1}{2}}(\Gamma),\tilde{H}^{-\frac{1}{2}}(\Gamma))$ converge to $\mathcal{S}(K),\mathcal{N}(K),\mathcal{K}(K)$ and $\mathcal{K}^*(K)$ in operator norm, respectively.
\end{lemma}
\begin{proof}
Here we prove $\mathcal{S}(K_n)\to \mathcal{S}(K)$, $\mathcal{K}^*(K_n)\to \mathcal{K}^*(K)$ and $\mathcal{K}(K_n)\to \mathcal{K}(K)$. One can prove $\mathcal{N}(K_n)\to \mathcal{N}(K)$ by a similar argument.

We first prove $\mathcal{S}(K_n)\to \mathcal{S}(K)$. For each $\varphi \in \tilde{H}^{-\frac{1}{2}}(\Gamma)$, \eqref{eq_app_C_4} implies that
\begin{equation} \label{eq_app_C_5}
Tr\int_{U}K_n(\bm{x},\bm{y})(Tr^*\varphi)(\bm{y})d\bm{y}
\overset{H^{\frac{1}{2}}(\Gamma)}{\longrightarrow}
Tr\int_{U}K(\bm{x},\bm{y})(Tr^*\varphi)(\bm{y})d\bm{y},
\end{equation}
where $Tr:H^{1}(U)\to H^{\frac{1}{2}}(\Gamma)$ is the trace operator, and $Tr^*$ denotes its adjoint. Note that $\mathcal{S}(K_n)[\varphi](\bm{x})=Tr\int_{U}K_n(\bm{x},\bm{y})(Tr^*\varphi)(\bm{y})d\bm{y}$ and $\mathcal{S}(K)[\varphi](\bm{x})=Tr\int_{U}K(\bm{x},\bm{y})(Tr^*\varphi)(\bm{y})d\bm{y}$. Thus it's clear that $\mathcal{S}(K_n)\to \mathcal{S}(K)$ by \eqref{eq_app_C_5}.

We next prove $\mathcal{K}^*(K_n)\to \mathcal{K}^*(K)$. For each $\varphi \in \tilde{H}^{-\frac{1}{2}}(\Gamma)$, \eqref{eq_app_C_4} implies that
\begin{equation*}
\frac{\partial}{\partial x_1}\int_{U}K_n(\bm{x},\bm{y})(Tr^*\varphi)(\bm{y})d\bm{y}
\overset{\tilde{H}^{-\frac{1}{2}}(\Gamma)}{\longrightarrow}
\frac{\partial}{\partial x_1}\int_{U}K(\bm{x},\bm{y})(Tr^*\varphi)(\bm{y})d\bm{y},\quad \bm{x}\in\Gamma .
\end{equation*}
By the jump condition \eqref{eq_layer_potential_operators_jump_3}, it is equivalent to
\begin{equation} \label{eq_app_C_6}
p.v.\int_{\Gamma}\frac{\partial}{\partial x_1}K_n(\bm{x},\bm{y})(Tr^*\varphi)(\bm{y})d\bm{y}
\overset{\tilde{H}^{-\frac{1}{2}}(\Gamma)}{\longrightarrow}
p.v.\int_{\Gamma}\frac{\partial}{\partial x_1}K(\bm{x},\bm{y})(Tr^*\varphi)(\bm{y})d\bm{y},
\end{equation}
which gives $\mathcal{K}^*(K_n)\to \mathcal{K}^*(K)$.

We finally prove $\mathcal{K}(K_n)\to \mathcal{K}(K)$. The condition $K_n(\bm{x},\bm{y})=\overline{K_n(\bm{y},\bm{x})}$ and \eqref{eq_app_C_6} implies that the following convergence holds uniformly for all $\phi\in H^{\frac{1}{2}}(\Gamma)$ and $\varphi\in \tilde{H}^{-\frac{1}{2}}(\Gamma)$ with unit norm
\footnotesize
\begin{equation*}
\begin{aligned}
\langle \varphi(\bm{y}),p.v.\int_{\Gamma}\frac{\partial}{\partial x_1}K_n(\bm{y},\bm{x})\phi(\bm{x})d\bm{x}  \rangle
=\overline{\langle \overline{\phi(\bm{x})},p.v.\int_{\Gamma}\frac{\partial}{\partial x_1}K_n(\bm{x},\bm{y})\overline{\varphi(\bm{y})}d\bm{y}  \rangle} 
&\longrightarrow \overline{\langle \overline{\phi(\bm{x})},p.v.\int_{\Gamma}\frac{\partial}{\partial x_1}K(\bm{x},\bm{y})\overline{\varphi(\bm{y})}d\bm{y}  \rangle} \\
&=\langle \varphi(\bm{y}),p.v.\int_{\Gamma}\frac{\partial}{\partial x_1}K(\bm{y},\bm{x})\phi(\bm{x})d\bm{x}  \rangle,
\end{aligned}
\end{equation*}
\normalsize
where $\langle\cdot,\cdot\rangle$ denotes the $\tilde{H}^{-\frac{1}{2}}-H^{\frac{1}{2}}$ pairing. Thus $\mathcal{K}(K_n)\to \mathcal{K}(K)$ is proved.
\end{proof}

\section{Appendix B}
\setcounter{equation}{0}
\setcounter{subsection}{0}
\setcounter{theorem}{0}
\renewcommand{\theequation}{A.\arabic{equation}}
\renewcommand{\thesubsection}{A.\arabic{subsection}}
\renewcommand{\thetheorem}{A.\arabic{theorem}}
Proof of Lemma \ref{lem_app_C_scalar_integrals_property}: The properties of the functions $f_*(\kappa_1, \delta)$ and $r(p, \delta)$
follow directly from their definitions.

We now prove \eqref{eq_app_C_scalar_integrals_property_1}. Note that each integral in \eqref{eq_center_scalar_integral} depends crucially on the integrand's behavior at infinity. Under a change of variable, the integrand is asymptotic of the order $\mathcal{O}(p^{k-2(\ell+1)})$ as $p\to\infty$ and the integral interval becomes $(-\delta^{1/9-1/2},\delta^{1/9-1/2})$, which yield the estimate by an elementary calculation. Specifically, 
by a change of variable $\kappa_1-\pi=\sqrt{\frac{2\delta}{\gamma_*}}p$, 
\footnotesize
\begin{equation*}
\begin{aligned}
&\int_{\pi-\delta^{\frac{1}{9}}}^{\pi+\delta^{\frac{1}{9}}}
\frac{(\kappa_1-\pi)^{k}f^{\ell}_{*}(\kappa_1;\delta)}{\delta\cdot h\pm\sqrt{\frac{1}{4}\gamma_*^2(\kappa_1-\pi)^4+t_*^2\delta^2}}\frac{d\kappa_1}{1+f^2_{*}(\kappa_1;\delta)} \\
&=(\frac{2}{\gamma_*})^{\frac{k+1}{2}}\delta^{\frac{k-1}{2}}
\int_{-\delta^{1/9-1/2}}^{\delta^{1/9-1/2}}
\frac{p^k}{h\pm\sqrt{p^4+t_*^2}}
\Big(\frac{t_*}{p^2+\sqrt{p^4+t_*^2}}\Big)^{\ell}
\frac{(p^2+\sqrt{p^4+t_*^2})^2}{(p^2+\sqrt{p^4+t_*^2})^2+t_*^2}dp.
\end{aligned}
\end{equation*}
\normalsize
The integrand is uniform bounded for $|p|\leq 1$ and of the order $O(p^{k-2(\ell+1)})$ as $p\to\infty$. Hence the integral is $\mathcal{O}(1)$ when $k\leq 2\ell$, or dominated by the by the polynomial integral $\int_{-\delta^{1/9-1/2}}^{\delta^{1/9-1/2}}p^{k-2(\ell+1)}$ when $k>2\ell$. Combining the $\delta^{\frac{k-1}{2}}$ factor, one obtains \eqref{eq_app_C_scalar_integrals_property_1}.

We next prove  \eqref{eq_app_C_scalar_integrals_property_2}. Direct calculation yields
\footnotesize
\begin{equation*}
\begin{aligned}
LHS=&\frac{1}{2\pi}\int_{\pi-\delta^{\frac{1}{9}}}^{\pi+\delta^{\frac{1}{9}}}
\frac{\delta\cdot h}{\delta^2\cdot h^2-(\frac{1}{4}\gamma_*^2(\kappa_1-\pi)^4+t_*^2\delta^2)}d\kappa_1  
\mp\frac{1}{2\pi}\int_{\pi-\delta^{\frac{1}{6}}}^{\pi+\delta^{\frac{1}{6}}}\frac{\sqrt{\frac{1}{4}\gamma_*^2(\kappa_1-\pi)^4+t_*^2\delta^2}}{\delta^2\cdot h^2-(\frac{1}{4}\gamma_*^2(\kappa_1-\pi)^4+t_*^2\delta^2)}\frac{1-f_{\delta}(\kappa_1)^2}{1+f_{\delta}(\kappa_1)^2}d\kappa_1.
\end{aligned}
\end{equation*}
\normalsize
Note that
\begin{equation*}
\frac{1-f_{\delta}(\kappa_1)^2}{1+f_{\delta}(\kappa_1)^2}
=\frac{\gamma_*(\kappa_1-\pi)^2}{2\sqrt{\frac{1}{4}\gamma_*^2(\kappa_1-\pi)^4+t_*^2\delta^2}}.
\end{equation*}

Thus
\footnotesize
\begin{equation*}
\begin{aligned}
LHS=&\frac{1}{2\pi}\int_{\pi-\delta^{\frac{1}{9}}}^{\pi+\delta^{\frac{1}{9}}}
\frac{\delta\cdot h}{\delta^2\cdot h^2-(\frac{1}{4}\gamma_*^2(\kappa_1-\pi)^4+t_*^2\delta^2)}d\kappa_1  
\mp\frac{1}{2\pi}\int_{\pi-\delta^{\frac{1}{9}}}^{\pi+\delta^{\frac{1}{9}}}\frac{\frac{1}{2}\gamma_*(\kappa_1-\pi)^2}{\delta^2\cdot h^2-(\frac{1}{4}\gamma_*^2(\kappa_1-\pi)^4+t_*^2\delta^2)}d\kappa_1 .
\end{aligned}
\end{equation*}
\normalsize
By direct calculation again, we have
\begin{equation*}
\begin{aligned}
&\frac{\delta^{\frac{1}{2}}}{2\pi}\int_{\pi-\delta^{\frac{1}{9}}}^{\pi+\delta^{\frac{1}{9}}}
\frac{\delta\cdot h}{\delta^2\cdot h^2-(\frac{1}{4}\gamma_*^2(\kappa_1-\pi)^4+t_*^2\delta^2)}d\kappa_1
\overset{\delta\to 0}{\longrightarrow}
-\frac{1}{2\gamma_*}\frac{\frac{h}{\gamma_*}}{((\frac{t_*}{\gamma_*})^2-(\frac{h}{\gamma_*})^2)^{\frac{3}{4}}}, \\
&\frac{\delta^{\frac{1}{2}}}{2\pi}\int_{\pi-\delta^{\frac{1}{9}}}^{\pi+\delta^{\frac{1}{9}}}
\frac{\frac{1}{2}\gamma_*(\kappa_1-\pi)^2}{\delta^2\cdot h^2-(\frac{1}{4}\gamma_*^2(\kappa_1-\pi)^4+t_*^2\delta^2)}d\kappa_1
\overset{\delta\to 0}{\longrightarrow}
-\frac{1}{2\gamma_*}\frac{1}{((\frac{t_*}{\gamma_*})^2-(\frac{h}{\gamma_*})^2)^{\frac{1}{4}}},
\end{aligned}
\end{equation*}
whence \eqref{eq_app_C_scalar_integrals_property_2} follows. The other estimates follow in a similar manner.

\footnotesize
\bibliographystyle{plain}
\bibliography{ref}

\end{document}